%% file: main.tex
\documentclass[12pt, a4paper, titlepage, headsepline, twoside]{book}

\usepackage{longtable}
\usepackage[small,bf,singlelinecheck=off]{caption}
\usepackage{pifont}

\usepackage{mathptmx}
\usepackage{bm}
\usepackage{helvet}
\usepackage{courier}
\usepackage{braket}

\usepackage[mathscr]{euscript}
\usepackage{amssymb}
\usepackage{type1cm}         
\usepackage{setspace}
\usepackage{makeidx}         
\usepackage{graphicx}        
\usepackage{multicol}        
\usepackage[bottom]{footmisc}
\usepackage[colorlinks=true,allcolors=blue]{hyperref}%

\newcommand{\varqbm}{VarQBM}

\usepackage{overpic}  

\usepackage{bbm}
\usepackage{nicefrac}
\usepackage{cite}

\usepackage{makecell}

\usepackage[usenames,dvipsnames]{xcolor}
\usepackage{colortbl}
\usepackage{color}
\usepackage{amsmath}
\usepackage{amsfonts}
\usepackage{amssymb}
\usepackage{bbm}
\usepackage{amsthm}
\usepackage{enumerate}

\usepackage{graphicx}	
\usepackage[caption=false]{subfig}
\usepackage{tikz}
\usetikzlibrary{calligraphy}
\usetikzlibrary{decorations.pathreplacing}
\usetikzlibrary{chains}
\usetikzlibrary{fit}
\usepackage{epsfig}
\usetikzlibrary{shapes.symbols,patterns} 
\usepackage{pgfplots}

\usepackage{makeidx}
\usepackage{dsfont}

\usepackage{afterpage}
\usepackage{soul}
\usepackage{algorithm}
\usepackage{algpseudocode}

\usepackage{wasysym}
\usepackage{wrapfig}
\usepackage{multirow}
\usepackage{mathtools}
\mathtoolsset{showonlyrefs=true}

\usepackage[framemethod=default]{mdframed}

\graphicspath{{figs/}}

\DeclareFontFamily{U}{mathx}{\hyphenchar\font45}
\DeclareFontShape{U}{mathx}{m}{n}{<-> mathx10}{}
\DeclareSymbolFont{mathx}{U}{mathx}{m}{n}
\DeclareMathAccent{\widebar}{0}{mathx}{"73}

\usepackage[margin=3cm]{geometry}

\theoremstyle{plain}
\newtheorem{theorem}{Theorem}[chapter]
\newtheorem{lemma}[theorem]{Lemma}

\newtheorem{proposition}[theorem]{Proposition}

\newtheorem{definition}[theorem]{Definition}

\newcommand{\norm}[1]{\left\lVert#1\right\rVert}
\newcommand{\proj}[1]{|#1\rangle\!\langle #1|}

\newcommand*{\di}{\mathrm{d}} 

\newcommand{\Tr}{\text{Tr}}

\DeclareMathOperator{\EX}{\mathbb{E}}
\DeclareMathOperator{\Var}{Var}
\DeclareMathOperator{\ee}{e}

\usepackage{environ}
\NewEnviron{es}{%
  \begin{equation}\begin{split}
     \BODY
  \end{split}\end{equation}
}

\allowdisplaybreaks    

\usepackage{fancyhdr}

\makeindex

\begin{document}

\thispagestyle{empty}
	\vspace*{-2cm} 
	
	\hspace{-0.0cm}

%
%
%
%

	\begin{minipage}[h]{\textwidth}
		\begin{center}\bf\Large
			\hspace{-1.0cm} 
			Generative Quantum Machine Learning
		\vspace{0.5cm}

		\end{center}

   \end{minipage}\\[2cm]

	\begin{minipage}[h]{\textwidth}

		\begin{center}\large

			\hspace{-1.0cm}  A thesis submitted to attain the degree of \\[0.5cm]
			
			\hspace{-1.0cm} DOCTOR OF SCIENCES of ETH ZURICH \\ [0.3cm]
			
			\vspace{0.3cm}
			
			\hspace{-1.0cm} (Dr.\ sc.\ ETH Zurich) \\ [1.9cm]

%
%
%


			\hspace{-1.0cm}    presented by\\[0.5cm]

			\hspace{-1.0cm}    \textsc  {Christa Anna Zoufal} \\[0.5cm]
			
			\vspace{1cm}
			\hspace{-1.0cm}    accepted on the recommendation of\\[0.5cm]
			\vspace{0.4cm}

			\hspace{-1.0cm}    Prof.\ Dr.\ Renato Renner, examiner \\
			
			\hspace{-1.0cm}    Dr.\ Stefan Woerner, co-examiner \\
			
			\hspace{-1.0cm}    Prof.\ Dr.\ Patrick Coles, co-examiner \\

			\vspace{8mm}

			\hspace{-1.0cm}   September 2021

		\end{center}

	\end{minipage}


\pagenumbering{gobble}
\thispagestyle{empty}
\newpage
\vspace*{\fill}
\copyright\ September 2021

Christa Zoufal

All Rights Reserved\\






\pagenumbering{roman}

\include{acknowledgments}

\include{abstract}

%
\csname @openrightfalse\endcsname
 \include{abstract_german}



\pagenumbering{arabic}


\setlength{\parskip}{0.51mm}
\begin{spacing}{0.95}
  \tableofcontents
\end{spacing}

\clearpage


\setlength{\parskip}{1.1mm}

\include{introduction}
\include{qml}
\include{methods}
\include{applications}
\include{conclusion_outlook}
\appendix
\include{appendix_varqte}

\include{appendix_qbm}

%
%
%
%
%
%
%



\typeout{}
\bibliographystyle{arxiv_no_month}
\bibliography{references}


\end{document}

%% file: acknowledgments.tex
\chapter*{Acknowledgements}

First and foremost, I would like to thank my IBM Research supervisor, Stefan Woerner, for his encouragement and support. Not only did he help me to figure out technical problems but he also acted as a role model with his intelligence, empathy, and management skills. He helped me to grow and always had an open door for me. I am also sincerely grateful that he accepts my odd dislike of the color green when it comes to making plots. Furthermore, I would like to thank my ETH Zurich supervisor, Renato Renner, who has supported my academic journey over many years. He helped me to find my passion in research and inspired me with his way of thinking.
Sincere thanks also go to my IBM Research manager, Walter Riess, who always had an open ear and a wise word for me.

Moreover, I want to thank my collaborators for inspiring discussions and interesting ideas: Amira Abbas, Hedayat Alghassi, Giuseppe Carleo, Amol Deshmukh, Daniel Egger, Alessio Figalli, Julien Gacon, Dmitry Grinko, Noelle Ibrahim, Raban Iten, Aur\'elien Lucchi, Ryan Mishmash, Nico Piatkowski, Nicolas Robles, and David Sutter. 

I was in the lucky position to find friends in many of my brilliant colleagues.
Pauline Ollitraut impressed me with her strength and determination. Julien Gacon was a great conversation partner when it came to technical topics, coding issues or – yes again – color scheme discussions. Guglielmo Mazzola and David Sutter would always share their wisdom and provide me with amusing jokes or help  with my bicycle. Amira Abbas and Elisa B\"aumer shared their endless motivation and energy with me. I am grateful to many other colleagues  for their friendship and for sharing their knowledge: Almudena Carrera V\'azquez, Panagiotis Barkoutsos, Ivano Tavernelli, Igor Sokolov, Raban Iten, Giulia Mazzola, Daniel Egger, Max Rossmanek, Fabio Scafirimuto, Matthias Mergenthaler, Bryce Fuller, Ryan Mishmash and the IBM Quantum team as well as the Quantum Information Theory group at ETH.

Furthermore, I want to express my gratitude to my loving and caring parents, brother and family who have encouraged me to go after my dreams and supported me in every possible way. The same is true for my partner, Julian Riebartsch, who has always been there for me and pushed me to be the best version of myself. Lastly, I would like to thank my friends including (but not limited to) Valerie Mertens, Marlene Rothe, Lisa G\"achter, Lisa Oberosler, Ann Michelle Mondragon, Laura B\'egon-Lours, Paul Vallaster, Thomas Spanninger, Martin Eichenhofer, Niko Peters, and Punit Mehra, who always helped me in so many ways.

I am grateful for the past years and the journey of my doctoral studies. I have met brilliant people, learned a lot, and had insightful as well as impactful experiences.

\vspace{0.4cm}


\hfill{Christa Zoufal}

\hfill{Zurich, September 2021}

%% file: abstract.tex
\chapter*{Abstract}

The goal of generative machine learning is to model the probability distribution underlying a given data set. This probability distribution helps to characterize the generation process of the data samples.
While classical generative machine learning is solely based on classical resources, generative quantum machine learning can also employ quantum resources -- such as parameterized quantum channels and quantum operators -- to learn and sample from the probability model of interest.

Applications of generative (quantum) models are multifaceted.
The trained model can generate new samples that are compatible with the given data and, thus, extend the data set.
Additionally, learning a model for the generation process of a data set may provide interesting information about the corresponding properties.
With the help of quantum resources, the respective generative models also have access to functions that are difficult to evaluate with a classical computer and may, thus, improve the performance or lead to new insights.
Furthermore, generative quantum machine learning can be applied to efficient, approximate loading of classical data into a quantum state which may help to avoid -- potentially exponentially -- expensive, exact quantum data loading.

The aim of this doctoral thesis is to develop new generative quantum machine learning algorithms, demonstrate their feasibility, and analyze their performance. Additionally, we outline their potential application to efficient, approximate quantum data loading.
More specifically, we introduce a quantum generative adversarial network and a quantum Boltzmann machine implementation, both of which can be realized with parameterized quantum circuits. These algorithms are compatible with first-generation quantum hardware and, thus, enable us to study proof of concept implementations not only with numerical quantum simulations but also real quantum hardware available today.


%% file: abstract_german.tex
\chapter*{Zusammenfassung}

Das Ziel von generativem maschinellen Lernen ist das Modellieren der \linebreak
Wahrscheinlichkeitsverteilung die einem gegebenen Datensatz zugrundeliegt und den Erzeugungsprozess der entsprechenden Daten charaktersiert. 
W\"ahrend klassisches generatives maschinelles Lernen ausschliesslich auf klassische Ressourcen zur\"uckgreifen kann, so stehen generativem Quanten-maschinellen Lernen, sowie dem entsprechendem Sampling, auch Quanten-Ressourcen -- wie beispielsweise parametrisierte Quantenkan\"ale und Quantenoperatoren -- zur Verf\"ugung. 

Anwendungen von generativen (Quanten-)Modellen sind vielseitig. Die \linebreak trainierten Modelle k\"onnen neue Daten generieren, die mit den gegebenen Daten insofern kompatibel sind, dass sie das Datenset effektiv vergr\"ossern. Ausserdem kann das Lernen von Modellen des Datenerzeugunsprozesses wertvolle Informationen \"uber entsprechenden Eigenschaften mit sich bringen.
Mit Hilfe von Quanten-Ressourcen haben die generativen Modelle Zugriff auf Funktionen, die f\"ur klassische Computer schwierig auszuwerten sind. Dementsprechend k\"onnen Leistungen verbessert oder neue Einsichten gewonnen werden.
Des Weiteren kann generatives Quanten-maschinelles Lernen auch f\"ur effizientes, approximatives Laden von klassischen Daten in Quantenzust\"ande verwendet werden und dabei helfen -- potentiell exponentiell -- aufwendiges, exaktes Quanten-Daten-Laden zu vermeiden.

Die Absicht dieser Doktorarbeit ist es neue generative Algorithmen fuer \linebreak Quanten-maschinelles Lernen zu entwickeln, ihre Umsetzbarkeit zu demonstrieren und ihre Leistung zu analysieren. Ausserdem f\"uhren wir deren potentielle Anwendung f\"ur effizientes, approximatives Quanten-Daten-Laden ein
Konkret f\"uhren wir ein Quanten Generative Adversarial Network und eine Quanten Boltzmann Maschine ein, welche beide mittels parametrisierter Quanten Circuits realisiert werden k\"onnen.
Die Algorithmen sind kompatibel mit Quanten Computern der ersten Generation und erm\"oglichen deshalb schon heute die Implementierung von ersten konzeptionellen Tests mit echter Hardware und nicht nur mit numerischen Quanten-Simulationen.

%% file: introduction.tex
\chapter{Introduction}
\label{introduction} 

\textbf{Abstract.} 
Generative quantum machine learning refers to algorithms which employ a combination of classical and quantum resources to solve generative machine learning tasks, i.e., learning a model for the probability distribution underlying a given training data set.
This chapter provides an introduction to the problems and topics touched upon in this thesis. 
Furthermore, we give an overview of the main results and present an outlook.

\index{introduction}

\vspace{8mm}

\noindent

Suppose a black-box that generates data samples without revealing information about the actual generation process. \emph{Generative machine learning} (ML) -- employing classical resources -- or potentially advantageous \emph{generative quantum machine learning} (QML) -- employing classical as well as quantum resources -- can be used to train a model for the black-box process w.r.t.~the given data. 


Generative ML offers a versatile tool that can be applied to image modification and generation \cite{karras2018progressiveGANS, goodfellow, radford2016unsupervised, ZhuImagetoImage17}, drug design \cite{PutinGenerativeDrugDesign18}, cosmology \cite{CosmoGan19Mustafa, CosmologiicalMassMapsLucchi21, Fagioli_2018Cosmo} and many other fields.
The models that are trained allow us to study the structure of the given data set and to generate new data samples that mimic the properties of the original data.
Quantum resources allow us to efficiently draw samples from probability distributions that are believed to be complicated to represent with a classical system \cite{Aaronson11ComputationComplexityOptics, ShepherdTemporallyUnstructuredQC09}. Generative QML has, thus, the potential to provide model improvements or new insights.
Additionally, it extends the applicability of generative models to quantum data.
Moreover, generative QML enables efficient loading of approximate quantum representations of given classical data. This, in turn, may help to circumvent expensive exact quantum data loading which could impair potential advantages of quantum applications.




The work presented in this thesis provides significant advances in the exploration of viable and useful (generative) quantum machine learning applications. 
Specifically, we introduce two algorithms -- quantum generative adversarial networks and quantum Boltzmann machines -- which are compatible with first-generation quantum hardware because they can be implemented with (shallow) parameterized quantum circuits.
We introduce new theoretical results, analyze their performance and provide a broad study of illustrative examples. The examples are evaluated with quantum simulators as well as on real quantum hardware to investigate the models in practice. This demonstrates that near-term quantum computers can be useful for generative quantum machine learning with carefully chosen parameterized quantum models.
Furthermore, we demonstrate that generative quantum machine learning can offer a flexible approach to efficiently load approximate quantum representations of classical data.

The remainder of this chapter presents the main results in Sec.~\ref{sec:main_results} and gives a detailed outline of the content included in this thesis in Sec.~\ref{sec:outline}.

\section{Main Results}
\label{sec:main_results}

The main contributions of this thesis are summarized in the following.
We discuss generative quantum machine learning and its application to efficient, approximate quantum data loading.
The respective content is focused on algorithms which employ parameterized (variational) quantum circuits and are compatible with current, respectively, near-term quantum hardware. 
To that end, we introduce a variety of methods that are not only of interest for generative quantum machine learning applications but also for the training of various other variational quantum algorithms.

More specifically, we present two generative QML algorithms: quantum generative adversarial networks (qGANs) \cite{Zoufal2019qGANs} and quantum Boltzmann machines (QBMs) \cite{VarQBMZoufal20}. These algorithms facilitate the exploration of near-term QML applications and enable the learning of efficient quantum approximations to probability distributions.
Given classical data samples, a parameterized quantum state can be trained to generate a discretized model of the data’s underlying distribution.
These generative QML algorithms are advantageous in the sense that the model for the data's underlying distribution is learned implicitly. In contrast, other schemes would firstly require that an explicit model to the distribution is loaded with possibly expensive operations and in general even exponentially many gates. Thus, it may render certain QML applications infeasible.

Generative adversarial networks employ an adversarial learning protocol which mimics the dynamics of a two-player game where the players correspond to a parameterized generator and a parameterized discriminator \cite{goodfellow, Kurach2018TheGL}.
Here we demonstrate that qGANs can be realized with parameterized quantum circuits which are compatible with first-generation quantum hardware.
Then, we present a thorough study of the feasibility and practicality of qGANs using illustrative examples. For instance, we investigate a real-world example where a qGAN is used to learn a model of the first two principle components of multivariate, constant maturity treasury rates of US government bonds. These experiments are either run with quantum simulations or actual quantum hardware.
Additionally, we demonstrate the exploitation of qGAN-based state preparation for a real-world example, i.e., financial derivative pricing with quantum amplitude estimation \cite{brassardQAE02}.

Boltzmann Machines \cite{HintonBM1985} are energy-based neural network models consisting of units which are either visible – defining the output – or hidden – acting as latent variables. These models feature interesting properties but are often difficult to train in practice. More specifically, the models require  a normalization factor, i.e., the partition function, which is expensive to evaluate.
The quantum counterpart, QBMs, are described by quantum mechanics, an inherently normalized theory. Furthermore, the Hamiltonian of a QBM can be hard to evaluate with classical computers due to the so-called \emph{sign problem} \cite{Troyer05}. Thus, QBMs are potentially more powerful.
Existing approaches to QBM implementations have difficulty with calculating analytic gradients of the loss function for certain parameterized Hamiltonians with hidden units. In contrast, our implementation employs variational quantum imaginary time evolution (VarQITE) to enable approximate Gibbs state preparation and facilitates the calculation of analytic gradients with automatic differentiation for generic Hamiltonians. 
We present results for variational Gibbs state approximation with VarQITE using numerical simulations and experiments run on real quantum hardware and apply this variational QBM approach to generative learning using quantum simulations. 
Due to the compatibility of VarQITE with shallow, parameterized circuits, it displays a promising approach to utilize near-term devices for solving practically relevant tasks.
At the same time, the resource restrictions such as limited number of qubits and gates, will inevitably result in an approximation error.   
We prove the first rigorous, a posteriori, phase-agnostic error bound for variational quantum imaginary time evolution \cite{zoufal2021errorBounds}. It should be noted that being able to practically quantify the approximation quality is crucial to understand the potential and improve the power of VarQITE.
\section{Outline}
\label{sec:outline}

In the following, we give a brief overview of each chapter of this thesis as well as of the Appendices.


\textbf{Chapter} \ref{sec:qml} provides an introduction to QML. 
First, Sec.~\ref{sec:quantum_computing} explains the basic concepts of quantum computing and Sec.~\ref{sec:machineLearning} gives an overview of ML paradigms.
Sec.~\ref{sec:qml_applications} presents various QML applications for fault-tolerant as well as near-term quantum computers. Then, Sec.~\ref{sec:qml_setting} explains the workflow of near-term QML algorithms that can already be realized with current technology. Furthermore, Sec.~\ref{sec:prospects_bottlenecks} discusses potentials and bottlenecks of QML algorithms.

\textbf{Chapter} \ref{methods} introduces a set of methods that are important for generative QML as well as other quantum algorithms based on variational quantum circuits. 
First, Sec.~\ref{sec:tools} introduces the components used in variational QML.
Furthermore, \linebreak Sec.~\ref{sec:data_encoding} discusses the considered classical data and how to load it into a quantum system using quantum data encodings.
Next, Sec.~\ref{sec:model} describes the computational information processing models used in QML and Sec.~\ref{sec:training} discusses how to train the respective algorithms.

Then, Sec.~\ref{sec:gradients} focuses on quantum gradients that are not only important for gradient-based optimization but also to investigate the properties of the training landscape. In that context, we discuss first-order and second-order gradients, the \emph{quantum Fisher information matrix} and \emph{quantum natural gradients} in Sec.~\ref{sec:quantum_grad_types}, as well as the analytic evaluation thereof in Sec.~\ref{sec:analytic_gradients}. 
Lastly, the potential problem of exponentially vanishing quantum gradients and promising mitigation strategies are reviewed in Sec.~\ref{sec:vanishing_grads}.

Then, Sec.~\ref{sec:varqite} discusses a variational implementation of quantum imaginary time evolution. The concept of quantum imaginary time evolution is introduced in Sec.~\ref{sec:qte} and a variational implementation based on McLachlan's variational principle in Sec.~\ref{sec:varqite_ground}. Furthermore, Sec.~\ref{sec:error_qite} introduces an efficient, a posteriori error bound which facilitates the estimation of the approximation accuracy of the variational approach. Applications of variational quantum imaginary time evolution include ground state preparation, see Sec.~\ref{app:ground_state_runtime}, and Gibbs state preparation, see Sec.~\ref{sec:varqite_gibbs}. Furthermore, this method can be used for automatic differentiation in the context of energy-based QML algorithms, see Sec.~\ref{sec:varqite_chainRule}. Following on,  Sec.~\ref{sec:methods} discuses the setting choices for variational quantum imaginary time evolution implementations. Finally, Sec.~\ref{sec:varqite_examples} presents various examples of variational quantum imaginary time evolution applications including ground state and Gibbs state preparation and shows the power of the discussed a posteriori error bounds.

\vspace{5mm}

\textbf{Chapter} \ref{sec:applications} explains generative QML algorithms which encode the trained information either in the parameters of a quantum channel, such as a parameterized quantum circuit, or in the parameters of a quantum operator, such as a Hamiltonian, and their application to approximate data encoding. This chapter focuses particularly on quantum generative adversarial networks, see Sec.~\ref{sec:qgan}, and quantum Boltzmann machines, see Sec.~\ref{sec:QBM}.
The introduction to both methods includes a discussion of prior research, see Sec.~\ref{sec:qgan_intro} and Sec.~\ref{sec:qbm_intro}. The classical machine learning algorithms, i.e., generative adversarial networks and Boltzmann machines, are explained in Sec.~\ref{sec:GANs} and Sec.~\ref{sec:BM}. The quantum counterparts as well as their variational implementations are introduced in Sec.~\ref{sec:quantumGAN} and Sec.~\ref{sec:qbm}. Then, we present a set of illustrative examples for both types of generative QML algorithms which are evaluated with numerical simulations or actual quantum hardware in Sec.~\ref{sec:qganApplications} and Sec.~\ref{sec:qbm_examples}, respectively.
Finally, Sec.~\ref{sec:qganconclusion_Outlook} and Sec.~\ref{sec:qbm_conclusion_outlook} conclude and give an outlook for both methods.

Furthermore, Sec.~\ref{sec:EuropeanCallOptionPricingexample} discusses an explicit example of generative QML for quantum data loading. More specifically, we present the pricing of European call options, see Sec.~\ref{sec:europeancallOpt}, with quantum amplitude estimation, see Sec.~\ref{sec:qae}, where a model for the option's underlying asset is loaded with a generative QML approach. For this example, we use a quantum generative adversarial network that is trained on actual quantum hardware. The respective results are presented in Sec.~\ref{subsec:europeanOptionPricingResults} and a discussion is provided in Sec.~\ref{sec:generativeQMLconclusionapplication}.

\vspace{5mm}

\textbf{Chapter} \ref{conclusion} presents a conclusion in Sec.~\ref{sec:conclusion} and gives an outlook in Sec.~\ref{sec:outlook}.

\vspace{5mm}

\textbf{Appendix} \ref{app:varqrte} explains variational quantum real time evolution which is the real counterpart to the method presented in Sec.~\ref{sec:varqite}. The variational implementation based on McLachlan's variational principle is explained in Sec.~\ref{sec:varqrte_implementation}. Furthermore, Sec.~\ref{subsec:error_qrte} introduces an efficient, a posteriori error bound based on the Bures metric, to estimate the accuracy of the variational implementation compared to the target evolution. Lastly, Sec.~\ref{app:varqrte_experiments} shows numerical experiments to illustrate the efficiency of variational quantum real time evolution and the respective error bounds.

\vspace{5mm}

\textbf{Appendix} \ref{app:qbms} discusses the application of quantum Boltzmann machines to discriminative learning. How to use QBMs for discriminative learning is explained in Sec.~\ref{appendix:discriminativeqbm}. An example is presented in Sec.~\ref{app:discr_qbm_results}, where the variational quantum Boltzmann machine, introduced in Sec.~\ref{sec:qbm}, learns to discriminate between fraudulent and non-fraudulent credit card transaction instances and, thereby, outperforms a set of classical standard classifiers.

\addtocontents{toc}{}

%% file: qml.tex
\chapter{Quantum Machine Learning}
\label{sec:qml}

\textbf{Abstract.} 
Quantum machine learning describes the use of quantum computing for machine learning problems. We present a variety of applications that are either compatible with fault-tolerant or with first-generation quantum technology. Since the latter can already be realized with current quantum hardware, it offers an accessible approach to study quantum machine learning today. Moreover, this chapter describes the workflow of quantum machine learning algorithms that employ parameterized quantum circuits and discusses promising as well as potentially problematic properties thereof. 

\index{quantum machine learning}

\vspace{8mm}

\noindent


In 1981, Richard Feynman inspired generations of researchers by pointing out that nature follows the rules of quantum mechanics and that we should rather investigate it using quantum computational systems \cite{FeynmanSimulatingPhysics1982}. Ever since, significant effort went into the development of universal quantum computers and the investigation of potential applications in natural science, e.g., in chemistry \cite{kassal2011simulating, Ganzhorn2011Eigenstates, NonadiaaticOllitrault2020} and material design \cite{BarkoutsosQAAlchemical21, MaQuantumSimMaterials2020, QuantumSimMaerialsBabbush18} but also other fields, e.g., in finance \cite{worQuantumRiskAnalysis19, Stamatopoulos_2020, Alcazar_2020qganFinance, Chakrabarti2021thresholdquantum}, optimization \cite{FarhiQAOA14, Gacon_2020QSBO, Barkoutsos20VQOCVaR} and machine learning \cite{Schuld2018SupervisedLearning, 2014QMLWittek, Zoufal2019qGANs, VarQBMZoufal20}.

Quantum machine learning \cite{2014QMLWittek, Schuld2018SupervisedLearning, Benedetti_2019PQC} generally refers to the exploration of synergies between quantum computing and machine learning. Typically, the QML workflow involves that given data -- the \emph{training data} -- is processed with an interplay of a classical and a quantum computer.
Various research projects reveal that the applications of quantum computing in the context of machine learning are broad. In fact, QML algorithms have been developed for all of the machine learning paradigms and problem classes presented in Fig.~\ref{fig:ml_diagram}. 

In the following, Sec.~\ref{sec:quantum_computing}, first, explains the basics of quantum computing and Sec.~\ref{sec:machineLearning} presents an introduction to ML.
Application details and examples are given in Sec.~\ref{sec:qml_applications}. A description of the components and workflow of QML algorithms that can be executed on current or near-term quantum hardware is given in Sec.\ref{sec:qml_setting}. Lastly, particularly promising as well as precarious aspects of QML are discussed in Sec.~\ref{sec:prospects_bottlenecks}.

\section{Quantum Computing}
\label{sec:quantum_computing}
In the following, we explain the basic concepts of quantum computing. We refer the interested reader to \cite{nielsen10, HayashiQuantumInfo06} for more information about quantum computation and information.

The basic unit of information in classical computing is given by a \emph{bit} which corresponds to a binary digit, i.e., a bit can take either of the values $\set{0, 1}$. The quantum counterpart is given by a \emph{qubit} that corresponds to a normalized, pure quantum state given by
\begin{align}
    \alpha \ket{0} + \beta \ket{1},
\end{align}
with $\alpha, \beta \in \mathbb{C}$ such that $\rvert\alpha\rvert^2 + \rvert\beta\rvert^2 = 1$ and ${\ket{0}, \ket{1}}$ forming a basis for the $2$ dimensional state space -- the Hilbert space $\mathcal{H}$. Furthermore, a pure quantum state that consists of $n$ qubits can be written as
\begin{align}
    \ket{\psi}=\bigotimes\limits_{i=0}^{n-1}\Big(\alpha_i \ket{0} + \beta_i \ket{1}\Big),
\end{align}
and lives in a Hilbert space that has dimension $2^n$.
A Hilbert space corresponds to a \emph{complex inner product space}, i.e., a complex vector space that is equipped with an inner product $\braket{x, y}\in\mathbb{C}$ for $x, y \in \mathcal{H}$.
Qubits can exist in a linear combination of basis states. This phenomenon is called \emph{superposition} and represents an important distinction from classical computation. Another important property that is inherent to quantum mechanics is called \emph{entanglement} and describes a form of correlation between two or more quantum systems that is stronger than any known correlation of classical systems.
A basic example of an entangled quantum state is given by a Bell state
 \begin{equation}
        \ket{\phi^{+}}=\frac{\ket{00} + \ket{11}}{\sqrt{2}}.
    \end{equation}
Besides \emph{pure} quantum states, which can be written as vectors, one can also define \emph{mixed} quantum states, which can only be represented by density matrices and correspond to a weighted sum of pure quantum states, i.e.,
\begin{align}
    \rho = \sum_j \gamma_j \ket{\psi_j}\bra{\psi_j},
\end{align}
where $\Tr\left[\rho\right]=1$.
Furthermore, quantum computers process information using \emph{quantum channels} \cite{nielsen10} which are defined as completely positive, trace preserving, linear transformations that map density matrices from $\mathcal{H}_A$ to a potentially different Hilbert space $\mathcal{H}_B$.
Finally, the measurement of a quantum state with respect to a Hermitian observable $\hat{O}$ corresponds to $\bra{\psi}\hat O\ket{\psi}$
or $\Tr\left[\hat O\rho\right]$.

\section{Machine Learning}
\label{sec:machineLearning}
ML \cite{Goodfellow2016_DeepL, murphy2012machinelearning} can generally be described as an approach that uses parameterized models to learn patterns and relations from given training data.
The variety of ML algorithms is diverse and includes, amongst others, classification, clustering and generative algorithms which employ models such as support vector machines \cite{Vapnik1995StatLearningTheory, BoserSV92}, Boltzmann machines  \cite{Hinton_BM_2011, RBM_Montufar_2018} and neural networks \cite{Rosenblatt58theperceptron, HintonDeepBeliefNets06, BengioConvexNNs05}.
An illustration of the main ML paradigms is depicted in Fig.~\ref{fig:ml_diagram}. In the following, we give a high-level discussion of these paradigms. We would like to refer the interested reader for further details on machine learning to \cite{Goodfellow2016_DeepL, Carbonell1983}.

First, \emph{supervised learning} refers to a class of ML algorithms focused on inference modelling, such as classification and regression. Given an input set $X$ and a corresponding output set $Y$, the goal of a supervised algorithm is to learn a representation of the correlations that exist between $X$ and $Y$ such that it can be used to find the correct $\tilde Y$ for unseen data $\tilde X$. 

Second, \emph{unsupervised} ML algorithms investigate the structural properties of a given data set. Some prominent examples are clustering, where the goal is to group data points according to structural similarities, dimensionality reduction, which aims to construct a data set with a smaller dimension than the original data set whilst preserving as much information as possible, and generative modelling, which focuses on learning the representation of a given data set's underlying probability distribution.

Third, \emph{reinforcement learning} refers to an ML paradigm, where an agent explores possible actions that are viable given an accessible (potentially action-dependent) state space and tries to learn the optimal action strategy -- the policy. The taken actions respectively reached states are associated with a certain score such that actions leading to a better score are rewarded while actions leading to a worse score are penalized.

\begin{figure}[h!]
\captionsetup{singlelinecheck = false, format= hang, justification=raggedright, font=footnotesize, labelsep=space}
\begin{center}
\begin{tikzpicture}
\node at (0, 0){\includegraphics[width=0.8\linewidth]{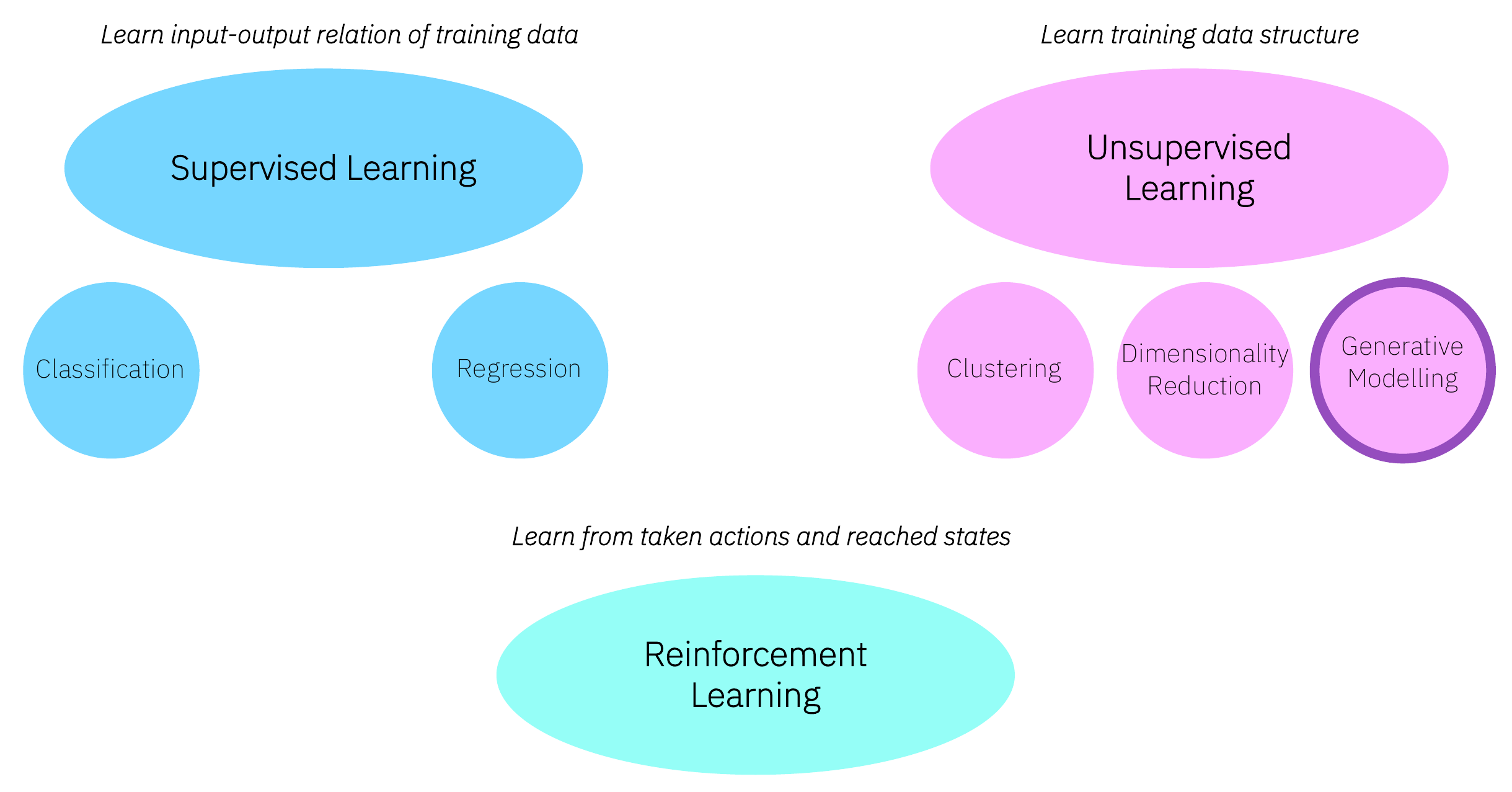}};
\node at (0, 3.8) {\textbf{Machine Learning Paradigms}};
\end{tikzpicture}
\end{center}
\caption{This figure illustrates an overview of the main machine learning paradigms and presents examples of related problem classes.}
\label{fig:ml_diagram}
\end{figure}

\section{Applications}
\label{sec:qml_applications}

This section introduces various examples of QML applications for the problem classes illustrated in Fig.~\ref{fig:ml_diagram}.
The variety of QML applications ranges from heuristic to provably advantageous applications and from near-term compatible algorithms to algorithms that require fault-tolerant quantum computers.
Although fault-tolerant QML is more likely to come with a provable advantage, it cannot be realized as of now.
The key components of those algorithms are either quantum phase estimation \cite{KitaevQPE}, quantum amplitude estimation \cite{brassardQAE02}, Grover's search algorithm \cite{grover_search96}, or methods for quantum linear systems based on the HHL algorithm \cite{hhl}\footnote{A unification for these types of algorithms has been proposed recently in \cite{martyn2021grandUnificationQAs}.}.
Near-term compatible QML, on the other hand, can already be implemented with currently available quantum hardware and, thus, enables us to understand and learn about the practical performance. 

Fault-tolerant, supervised learning algorithms for classification that employ quantum phase estimation to compute a distance measure between data instances is presented in \cite{QSlowFeatureAnalysis20Kerenidis}.
Another type of QML algorithms for classification corresponds to a quantum counterpart of support vector machines. To that end, one employs quantum kernels, see Def.~\ref{sec:quantum_feature}, which represent a specific mapping of classical data into a quantum feature space that is potentially hard to compute with a classical computer. 
The formulation presented in \cite{ Rebentrost_2014QSVM} relies on the HHL algorithm \cite{hhl}.
Moreover, quantum regression algorithms that are based on the HHL algorithm have the potential to enable a speed-up in the evaluation of systems of linear equations \cite{Wiebe_2012QuantumAlgorithmforDataFitting, Schuld_2016QLinRegression, Zhao_2019QReg, Dutta_2020QRegression, chakraborty_QRegression19}. An alternative  quantum regression approach which employs quantum amplitude estimation \cite{brassardQAE02} can lead to a quadratic speed-up in terms of the estimation accuracy \cite{Wang_2017QuantumRegression}.

Also near-term quantum classifiers have been studied in detail.
 The research presented in \cite{farhi2018classification} aims at training a parameterized quantum circuit to represent a classification function. The algorithm aims at matching the labels generated by the quantum circuit measurement to the labels given by the training data set. The basic idea behind this QML algorithm is that the quantum ansatz achieves a mapping which is more powerful than a comparable classical classifier. Similar approaches are studied in \cite{Cong_2019QCNN, 19HavSupervisedLearning, Schuld2018CircuitcentricQC, Schuld_2019QMLHilbertSpace, quantumCircuitLearnMitarai, TacchinoQNeuron19, GrantQuantumClassifiers18, Benedetti_2019PQC}. In fact, \cite{Liu_2021QuantumSpeedUpML} proves that there exist a certain problem class where quantum classifiers outperform their classical counterpart.
Another form of quantum classification employs near-term compatible distance measures which give rise to distance-based quantum classifiers \cite{johri2020nearestcentroidClass, Schuld_2017DistanceClassifierQuantum}.
Quantum kernels do not require fault-tolerant quantum computers but can already be evaluated with near-term quantum hardware. Respective quantum classifier implementations are given in \cite{19HavSupervisedLearning, Schuld_2019QMLHilbertSpace}.
Furthermore, we would like to mention that parameterized quantum circuits as well as quantum  kernels  can also be used to implement near-term quantum regression algorithms \cite{Benedetti_2019PQC}.


Much attention has also been paid to unsupervised QML algorithms.
Quantum
clustering based on Grover's search algorithm that can facilitate a quadratic speed-up for search problems \cite{grover_search96}, is investigated in \cite{QUnsupervisedAimeur13}. Other clustering examples given in \cite{Kerenidisqmeans19} and \cite{lloyd2013quantumClustering} are again based on quantum amplitude estimation and HHL, respectively. 
Under some assumptions on the available data, quantum principle component analysis (PCA) \cite{qPCAlloyd13, li2021resonantQPCA} can be implemented using density matrix exponentiation and quantum phase estimation and, thereby, achieve an exponential improvement compared to classical methods. 


Clustering algorithms that are compatible with near-term quantum hardware can use shallow quantum circuits to compute the distance between two data points \cite{khan2019kmeans}.
A near-term approach for dimensionality reduction relies on the use of variational autoencoders \cite{Kingma_2019VAE} that are implemented using variational quantum circuits \cite{QVAEDing19, Romero_2017QAE}.
Also, state diagonalization can be accomplished with variational quantum circuits \cite{LaRoseVarQDiag19}. Notably, diagonalization is one of the key components of PCA \cite{murphy2012machinelearning, bishop2006patternRecognition}. The theory developed in \cite{LaRoseVarQDiag19} is extended in \cite{cerezo2020vqse} to majorization that facilitates the investigation of the quantum state spectrum and, thus, enables a quantum PCA implementation.
A related but slightly different approach is taken in \cite{Bravo_Prieto_2020QSVD} where two parameterized quantum circuits are trained to find a singular value decomposition of a bipartite, pure quantum state \cite{SchmidtDecompositionEkert95}.
Generative QML algorithms\footnote{Notably, Sec.~\ref{sec:applications} includes further discussion on generative QML and illustrates the application of mapping a probability distribution underlying given training data into a quantum state.}, which train a quantum operator to represent the probability distribution underlying given training data, have mostly been explored for variational implementations.
Applications for the preparation of quantum channels whose output is an estimate to a quantum state \cite{killoran2018qgans, Benedetti_2019PureStateApproxqGANs, huqGANs2019, verdon2019quantumHbasedModels}, to generate classical data samples \cite{SITU2020193qGANs,romero,Zeng_2019LearningInferenceonqGANs} and to approximately load classical data into a quantum state \cite{Zoufal2019qGANs, VarQBMZoufal20, LiuDifferentiableLearning18, BornSupremacyCoyle2020, Hamilton_2019GenerativeModelBenchmarks, BenedettiGenModellin19} have been explored. We would like to point out that this class of QML algorithms is explicitly suitable for implementation on first-generation quantum hardware as is discussed in more detail later on in this thesis.
Intuitively, quantum resources may be useful, in particular in the context of high-dimensional probability distributions \cite{lloyd2qGANs18} because they enable the representation of an exponentially large state space.

Moreover, quantum reinforcement learning has been studied in various settings. An method based on Grover's algorithm is presented in \cite{QuantumReinforcementDong05}.
Near-term reinforcement learning may either be implemented using quantum Boltzmann machines, see Sec.~\ref{sec:QBM}, \cite{QuantumEnhancementsDeepReinforcementLearningJerbi21} or directly using parameterized quantum circuits \cite{lockwood2020qreinforcement, skolik2021quantumAgents}.

\section{Workflow}
\label{sec:qml_setting}

As mentioned before a QML workflow employs the interplay of classical and quantum computing resources to train a parameterized model that depends on the given learning task.
We can categorize QML algorithms into those that require fault-tolerant quantum computers and those that are compatible with first-generation quantum hardware, which we shall refer to as \emph{variational QML} because they typically employ shallow, parameterized quantum circuits as ans\"atze, which are also known as \emph{variational quantum circuits} or \emph{quantum neural networks}.
This thesis discusses the realization of generative QML algorithms, see Sec.~\ref{sec:applications}, with variational implementations. 
The remaining section introduces the respective workflow.


Given classical training data, the respective workflow, illustrated in Fig.~\ref{fig:qml_workflow_near_term}, proceeds as follows. First, the training data is loaded into the classical computer which initializes a set of parameters $\bm{\omega}$ according to the learning problem and sends them to the quantum circuit. The quantum circuit is, then, executed and measured. The measurement results enable the evaluation of a loss function $L\left(\bm{\omega}\right)$. which, in turn, is send to the classical computer that uses an optimization routine, such as gradient descent \cite{CauchyGD1847} or ADAM \cite{Kingmaadam14}, to update the parameters. This cycle is, now, repeated until a stopping criteria is sufficed. Notably, we can also compute the loss function gradient $\nabla_{\bm{\omega}}L\left(\bm{\omega}\right)$ with the help of parameterized quantum circuits to facilitate gradient-based optimization. A more formal definition of the components used in the described QML workflow is given in Sec.~\ref{sec:tools} and further details on quantum gradient computation can be found in Sec.~\ref{sec:analytic_gradients}.

\begin{figure}[h!t]
\captionsetup{singlelinecheck = false, format= hang, justification=raggedright, font=footnotesize, labelsep=space}
\begin{center}
\begin{tikzpicture}
\node at (0, 0){\includegraphics[width=0.9\linewidth]{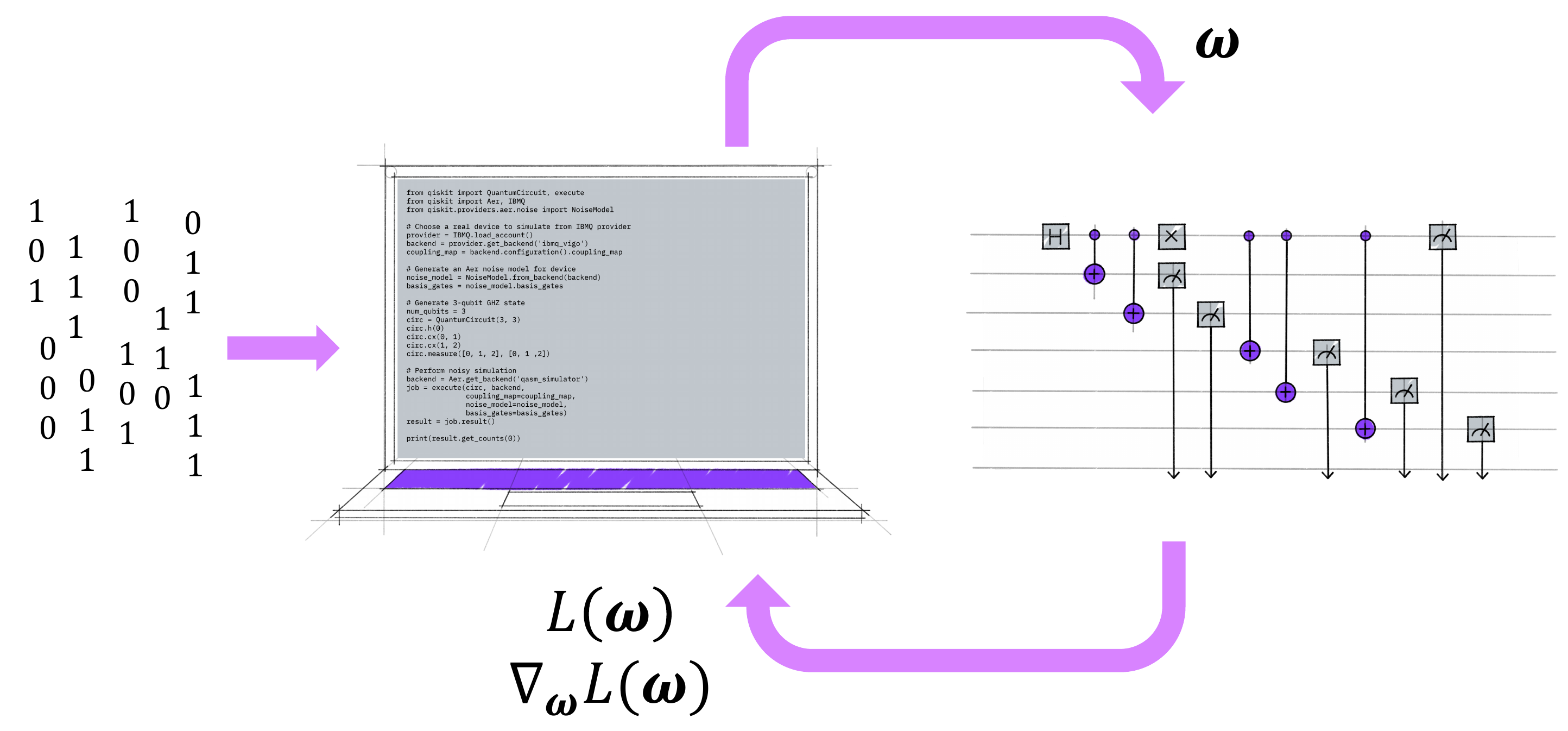}};
\node at (0, 3.8) {\textbf{Variational Quantum Machine Learning Workflow}};
\end{tikzpicture}
\end{center}
\caption{The workflow of variational quantum machine learning algorithms that are compatible with first-generation quantum computers is illustrated. This figure includes elements that are made available by \cite{ibmQX}.}
\label{fig:qml_workflow_near_term}
\end{figure}

\section{Prospects and Bottlenecks}
\label{sec:prospects_bottlenecks}

Quantum (machine learning) algorithms offer potential advantages compared to their classical counterparts for tasks such as prime factorization \cite{shorFactoring97}, unstructured search problems \cite{grover_search96}, solving linear systems of equations \cite{hhl}, Monte-Carlo sampling problems \cite{brassardQAE02}, semi-definite programming \cite{brandao2019quantumsdp}, classification \cite{Liu_2021QuantumSpeedUpML} and many others.
However, the underlying physics also lead to some bottlenecks in the practical realization of the respective algorithms.

A large motivation behind investigating QML is the exponentially large state space that a quantum system can represent as well as the representation of mappings that are classically difficult or even intractable to compute \cite{Bravyi20QAShallowCircuits, Bravyi_2018QuantumAdvantageShallowCircuits}. The goal is to use quantum channels to enable faster or more accurate model training.
Recent research supports the idea that quantum channels can offer certain benefits in an ML context. 
For example, \cite{Abbas_2021PowerofQNNs} investigates the effective dimension, a data-dependent capacity measure. 
The analysis reveals that the effective dimension and, therefore, the capacity, of a class of quantum neural networks is better than for a comparable class of classical neural networks and, therefore, indicates a potential advantage of quantum neural networks. 
Also \cite{DuExpressivePowerPQCs_2020} investigates the expressive power of variational quantum circuits in comparison to some classical ML models in the context of generative modelling. The work proves that the quantum models can outperform the investigated classical models for some generative ML problems.
Furthermore, research on quantum kernels \cite{19HavSupervisedLearning, Schuld_2019QMLHilbertSpace, glick2021covariantQuantumKernels} -- quantum kernels are formally defined in Def.~\ref{sec:quantum_feature} -- suggests that quantum circuits enable the representation of kernel functions \cite{HofmannKernelMethods08} which are classically intractable. In fact, \cite{Liu_2021QuantumSpeedUpML} shows that a quantum kernel method can achieve an exponential improvement compared to classical methods for a classification task that is based on the discrete logarithm problem.
All of these results are given for rather specific examples. One of the important open questions is whether we can formulate a general argument that certain quantum channels are beneficial for QML tasks.

A particularly important problem is that \emph{quantum data loading}, i.e., the mapping of classical data into a quantum representation, and \emph{data readout}, i.e., the mapping of a quantum representation into classical data can become exponentially expensive.
In fact, many potential advantages of quantum algorithms, particularly in quantum machine learning, are conditioned on the assumption that classical data can be efficiently loaded into a quantum state \cite{Aaronson2015_finePrint, DataLoadingTang21, Tang_2019QInspiredRecommendation, Huang_2021PowerDatainQML} but this may only be achieved for specific data structures \footnote{A common assumption is that the respective data is already loaded and accessible through a Quantum Random Access Memory (QRAM) \cite{qram}. While most of the QRAM literature, e.g., \cite{qram2} focuses on the description of data accessing logic, the original loading of the data is subject to the issues described above. Furthermore, we would like to point out that no physically stable QRAM structure is available as of now.}.
The exact preparation of a generic state in $n$ qubits can require $\mathscr{O}\left(2^n\right)$ gates \cite{Grover2000SynthesisSuperpos, Sanders2019StatePrep, Plesch2010StatePrep, Shende:2005:SQL:1120725.1120847} and, thus, render the application infeasible.
Data loading may, thus, easily dominate the complexity of an otherwise advantageous quantum algorithm.
A similar problem exists for the readout of data of a quantum algorithm. Suppose the result of the algorithm is encoded in a quantum state such as in \cite{hhl}. Then, one would have to apply state tomography which requires a number of measurements that is exponential in the number of qubits, see, e.g., \cite{Liu_2005Tomography} and, thus, impairs potential advantages of the application of the quantum algorithm.\footnote{If the algorithm target corresponds to an expectation value, then the readout may be conducted efficiently.} QML algorithms are, thus, more likely to be suitable for the processing of complex structures, rather than \emph{big data}.

In the early days of classical ML algorithms, researchers realised that due to the chain rule, deep neural networks are prone to exhibiting either very small or very big gradients depending on the activation functions \cite{hochreiter1991untersuchungenNN}. It was consequently found that special activation functions as well as certain neural network architectures actually help to avoid those disadvantageous gradients values and facilitate practical training.
Unfortunately, research has also revealed that exponentially small -- vanishing -- gradients can be a problem in the context of training parameterized quantum circuits. We discuss the occurrence of these vanishing gradients and potential mitigation strategies in Sec.~\ref{sec:vanishing_grads}.
Another trainability issue arises from the fact that the loss landscape is usually highly non-convex \cite{du2020learnabilityQNN, Stokes_2020QNG}. Thus, one should choose the optimization scheme carefully, e.g., by employing \emph{quantum natural gradients}, see Sec.~\ref{sec:qng}, which are aware of the problem geometry.

To summarize, there are many promising applications for QML and potential bottlenecks that need to be understood in more detail. 
This thesis investigates the practicality of generative QML, the impact of the discussed impediments, and the potential application of generative QML to circumvent the quantum data loading bottleneck.
\addtocontents{toc}{}

%% file: methods.tex

\chapter{Methods}
\label{methods} 

\textbf{Abstract.} 
Quantum machine learning relies on the rules of quantum mechanics. Thus, various methods that are common practice in classical machine learning, such as gradient-based optimization and data loading, need to be adapted accordingly. The following chapter provides an in-depth discussion of the basic building blocks and some of the key methods for the realization of (variational) quantum machine learning. We are going to introduce the basic concepts of variational circuits, loss functions and observables, different forms of data encoding in quantum systems, quantum gradients and variational quantum imaginary time evolution, a method which enables approximate Gibbs state preparation.

\index{methods}

\vspace{8mm}

\noindent

The potential of quantum computers stems from the fact that they are governed by quantum mechanics. Since the respective physical principles differ from those that govern classical computers some established machine learning methods need to be rethought.
These conceptual differences also offer new possibilities, e.g., for Gibbs states -- a crucial component in many probabilistic machine learning models \cite{GemanGibbs84, murphy2012machinelearning, KollerGraphicalModels09}.


In the following, we are going to discuss various methods that are relevant in the context of quantum machine learning and variational implementations thereof. First, we define the building blocks of variational QML in Sec.~\ref{sec:tools}. 
Next, Sec.~\ref{sec:gradients} introduces quantum gradients, outlines their evaluation and discusses settings where exponentially vanishing gradients occur.
Lastly, Sec.~\ref{sec:varqite} presents variational quantum imaginary time evolution, a method to solve optimization problems or to prepare Gibbs states which is compatible with automatic differentiation in QML algorithms.

\section{Setting}
\label{sec:tools}
This section introduces the components that form the basic building blocks of (variational) QML algorithms which employ the interplay of classical and quantum computers.
First, we discuss the type of data that we focus on in Sec.~\ref{sec:applications} and present different forms of quantum data encoding including potential bottlenecks thereof in Sec.~\ref{sec:data_encoding}.
Next, Sec.~\ref{sec:model} explains what models can be used in QML to process information.
Finally, Sec.~\ref{sec:training} discusses the training of QML algorithms.

\subsection{Quantum Data Encoding}
\label{sec:data_encoding}
In order to facilitate the realization of QML, it is particularly important to enable efficient representation of the given training data.
The QML approaches presented in this thesis consider classical training data sets which are denoted by $X$ and either consist of scalar data items $x\in\mathbb{R}$ or data vectors $\bm{x}\in\mathbb{R}^s$. Possibly available data labels $Y$ are given by items $y\in\mathbb{R}$. Notably, we refer to the mapping of classical data into a quantum state or quantum channel, i.e., 
\begin{align}
\mathcal{D}: X \mapsto \mathcal{H},
\end{align}
where $\mathcal{H}$ represents a Hilbert space of dimension $2^n$,  as \emph{quantum data loading}.
We would like to point out that this quantum data loading can become a severe bottleneck. In the general case, it may require exponential resources \cite{aharonov2013quantumpcp, Grover2000SynthesisSuperpos} and, thereby, impair potential quantum advantage. Furthermore, polynomial quantum data loading complexity can already impair potential quadratic advantage, e.g., of Grover-type \cite{grover_search96} quantum algorithms \cite{HerbertNoQuantumSpeedup21}.
Then, loading efficient data approximations can provide a reasonable remedy and offer a sufficient tool to study interesting system properties. 

Data loading may make use of different forms of data representations in a quantum channel or operator, i.e., different \emph{quantum data encoding} types, which can have a big impact on the training performance of a QML model \cite{LaRose_2020DataEncodings, Schuld_2021EffectDataEncodingQML}. 
The following section presents the most prominent quantum data encoding methods \cite{Schuld2018SupervisedLearning, Schuld_2019QMLHilbertSpace}. In those cases, where exact loading is generically inefficient, we also discuss efficient approximations as well as methods to estimate the approximation accuracy.



Basis encoding \cite{farhi2018classification, Schuld2018SupervisedLearning} refers to a binary representation of data points $x$ with quantum basis states.


\begin{definition}
\label{sec:basisEncoding}
Let $x$ be one of $k$ elements in a data set $X$. \textbf{Basis encoding} refers to the representation of the elements $x$ with an $n$-qubit basis state in $\{\ket{0}, \ldots, \ket{2^{n}-1}\}$ of a $2^n$ dimensional Hilbert space $\mathcal{H}$ with $k\leq 2^n$
\begin{align}
    \mathcal{E}\left(x\right) = \ket{x},
\end{align}
or to the representation of up to $2^n$ elements in $X$ as a superposition of the form
\begin{align}
    \mathcal{E}\left(X\right) = \frac{1}{\sqrt{k}}\sum\limits_{x}\ket{x}.
\end{align}
\end{definition}

\noindent
In other words, basis encoding employs binary representations of $n$ bits in $n$ qubits.
Consider, e.g., $x=5$ which corresponds to $x=101$ in binary encoding. Then, $x$ can be mapped onto the pure quantum state $\ket{101}$. 
Loading single data points comes with a computational cost that scales as $\mathscr{O}\left(n\right)$ and loading a superposition of data items of the form given in Def.~\ref{sec:basisEncoding} comes with a cost that scales as $\mathscr{O}\left(kn\right)$ \cite{Schuld2018SupervisedLearning}.


In amplitude encoding \cite{Zoufal2019qGANs, VarQBMZoufal20, Wiebe_2012QuantumAlgorithmforDataFitting, Schuld_2017DistanceClassifierQuantum, CarreraEfficientStatePreparation21}, normalized classical data data sets $X$ are represented by the continuous amplitudes of an $n$-qubit quantum state.

\begin{definition}
\label{sec:amplitudeEncoding}
Given a data set $X=\set{x_0, \ldots, x_{k-1}}$ such that 
\begin{align}
    \sum\limits_{j=0}^{k-1}\rvert x_j\rvert^2=1
\end{align}
and a quantum state consisting of $n$ qubits with $k\leq 2^n$.\textbf{ Amplitude encoding} refers to
\begin{equation}
\mathcal{E}\left(X\right)=\sum\limits_{j=0}^{k-1}\sqrt{x_{j}}\ket{j},
\end{equation}
where the states $\ket{j}$ are basis states of the underlying Hilbert space.
\end{definition}

\noindent
This type of data encoding enables the representation of up to $2^n$ scalar data samples with $n$ qubits.
Furthermore, it is explicitly suitable to load discrete probability distributions as
\begin{equation}
\sum\limits_{j=0}^{k}\sqrt{p^{j}}\ket{j},
\end{equation}
where $p^{j}$ corresponds to the sampling probabilities of $\ket{j}$.
For efficiently integrable probability distributions \cite{groverSuperpositionseffintegrablepdfs02}, such as log-concave distributions, amplitude encoding can be implemented with polynomial complexity.
However, the respective encoding method may have to rely on quantum arithmetic \cite{Vedral_1996QuantumArithmetic, Beckman_1996QuantumFactoring, Van_Meter_2005QuantumExponentiation, SvoreLookaheadAdder06, Ruiz_Perez_2017QuantumArithmetic} and can, thus, lead to long quantum circuits which are not well suited for first-generation quantum hardware.
In the generic case, amplitude encoding for $2^n$ data points requires either $\mathscr{O}\left(2^n\right)$ gates and $n$ qubits \cite{Grover2000SynthesisSuperpos, Sanders2019StatePrep, Plesch2010StatePrep, Shende:2005:SQL:1120725.1120847} or $\mathscr{O}\left(n^2\right)$ gates and $\mathscr{O}\left(2^{2n}\right)$ qubits \cite{zhang2021lowdepth}.

If the exact data encoding is exponentially expensive or too costly one can employ approximate state preparation, such as the generative QML models discussed in Sec.~\ref{sec:applications}, instead. 
Suppose, we employ approximate loading of a probability distribution and we have either access to a functional description of the target distribution or to samples thereof. Then, classical statistics such as the Kolmogorov-Smirnov statistic \cite{kolmogorov, JUSTEL1997_multiKS} may provide valuable information about the approximation accuracy. 
Given two (empirical) probability distributions $P\left(x\right)$ and $Q\left(x\right)$, the Kolmogorov-Smirnov statistic is based on the (empirical) cumulative distribution functions $P\left(X\leq x\right)$ and $Q\left(X\leq x\right)$ and is given by
\begin{equation}
\label{eq:KS}
D_{\text{KS}}\left(P||Q\right)  = \underset{x\in X}{sup}\:\vert P\left(X\leq x\right) - Q\left(X\leq x\right)\vert.
\end{equation}
Suppose, the null-hypothesis that $P\left(x\right)$ and $Q\left(x\right)$ 
are equivalent. If the \linebreak Kolmogorov-Smirnov statistic $D_{\text{KS}}$ is bigger than a certain critical value that depends a chosen significance level and the number of available data samples, then the null-hypothesis is rejected \cite{JUSTEL1997_multiKS}.
Another measure that characterizes the closeness of (empirical) discrete probability distributions $P\left(x\right)$ and $Q\left(x\right)$ is the relative entropy, also called Kullback-Leibler divergence \cite{kullback1951}.
This entropy-related measure is given by
\begin{equation}
\label{eq:relEntr}
D_{\text{RE}}\left(P||Q\right) = \sum\limits_{x\in X}P\left(x\right)\log\left(\frac{P\left(x\right)}{Q\left(x\right)}\right)\: .
\end{equation}
The relative entropy represents a non-negative quantity, i.e., $D_{\text{RE}}\left(P||Q\right)\geq 0$, where \\ $D_{\text{RE}}\left(P||Q\right) = 0$ holds if and only if $P\left(x\right) = Q\left(x\right),$ for all values of $x$.


Furthermore, feature maps are a widely employed concept from classical machine learning. The goal is to map given data to a higher-dimensional feature space which can then facilitate, e.g., classification with support vector machines \cite{Vapnik1995StatLearningTheory, BoserSV92}.
The quantum counterpart -- \emph{quantum feature maps} \cite{19HavSupervisedLearning, glick2021covariantQuantumKernels, Schuld_2021EffectDataEncodingQML} -- refer to a transformation of given classical data into a quantum operation. This in turn can be used to apply a non-linear mapping of the respective data into a quantum state.

\begin{definition}
\label{sec:quantum_feature}
Given elements $x, \tilde x$ of a data set $X$, a \textbf{quantum feature map} corresponds to a parameterized unitary $U$ that acts on a $2^n$ dimensional Hilbert space $\mathcal{H}$
\begin{align}
\mathcal{E}\left(x\right)=U\left(x\right).
\end{align}
The quantum feature map prepares a \textbf{quantum feature vector} as
\begin{align}
    \ket{\phi\left(x\right)}=U\left(x\right)\ket{0},
\end{align}
which in turn facilitates the evaluation of \textbf{quantum kernels}
\begin{align}
    \lvert \braket{\phi\left(\tilde x\right)|\phi\left(x\right)}\rvert^2,
\end{align}
corresponding to the inner product of two quantum feature vectors.
\end{definition}

\noindent
It has been argued in \cite{19HavSupervisedLearning, Schuld_2019QMLHilbertSpace, glick2021covariantQuantumKernels} that quantum kernels are particularly interesting because they can represent transformations which are intractable for classical computers.
Furthermore, this type of quantum data encoding enables a simple implementation of a potentially, non-linear data transformation which can be further processed by quantum circuits \cite{19HavSupervisedLearning, Schuld_2019QMLHilbertSpace}. In fact, it is shown in \cite{Schuld_2021EffectDataEncodingQML} that this form of data encoding has the potential to represent universal function approximators.



Moreover, Hamiltonian encoding \cite{Schuld2018CircuitcentricQC, VarQBMZoufal20} describes the embedding of classical data into a quantum Hamiltonian.
\begin{definition}
\label{sec:gibbs_related_work}
Let $x$ be an element in the data set $X$ such that \textbf{Hamiltonian encoding} refers to 
\begin{align}
     \mathcal{E}\left(x\right) = H\left(x\right) = \sum\limits_{i}f_i\left(x\right)h_i,
\end{align}
where the $n$-qubit Hamiltonian $H$ corresponds to a weighted sum of Pauli operators $h_i=\bigotimes_{j=0}^{n-1}\sigma_i^j$ for $\sigma_i^j\in\set{I, X, Y, Z}$ acting on the $j^{\text{th}}$ qubit whose weights are given by functions $f_i\left(x\right)$.
\end{definition}

\noindent
This Hamiltonian can then be used to transform a quantum state via quantum real time evolution \cite{Schuld2018SupervisedLearning} or Gibbs state preparation \cite{VarQBMZoufal20}.
Quantum real time evolution describes the propagation of an initital state $\ket{\psi_0}$ for time $t\in\mathbb{R}$ with respect to $H\left(x\right)$ according to the Schr\"odinger equation
\begin{align}
    \ket{\psi_t\left(x\right)}=\ee^{-iH\left(x\right)t/\hslash}\ket{\psi_0}.
\end{align}
This evolution can be approximately implemented on a gate-based quantum computer with Hamiltonian simulation techniques, such as Trotterization \cite{Lloyd1073UniversalQuantumSim96}, linear combination of unitaries \cite{LCUHamiltonianSimulaitonWiebe12, low2019wellconditionedHSimulation} and variational \cite{Simon18TheoryVarQSim, barison2021efficientQRTE, AdaptiveVarQTEYao21} approaches. One of the variational implementations is explained in detail in Appendix~\ref{app:varqrte}.

A Gibbs state describes the probability density operator of the configuration space corresponding to the system Hamiltonian in thermal equilibrium \cite{gibbs_2010, Boltzmann1877, gibbs02} at temperature $T$ and its formal definition is provided in the following.
\begin{definition}
\label{def:gibbs_state}
Given the Hamiltonian $H\left(x\right)$ and a system temperature $T$, the corresponding \textbf{Gibbs state} reads
\begin{align}
    \rho^{\text{Gibbs}}\left(x\right) = \frac{e^{-\frac{H\left(x\right)}{k_BT}}}{Z},
\end{align}
where $k_B$ represents the Boltzmann constant and $Z=\Tr\left[e^{-\frac{H\left(x\right)}{k_BT}} \right]$ corresponds to the partition function.
\end{definition}

Originally, Gibbs states were studied in the context of statistical mechanics but, as shown in \cite{Pauli1927}, the density operator also describes interesting statistics which are often employed in machine learning algorithms \cite{GemanGibbs84, murphy2012machinelearning, KollerGraphicalModels09}.
Furthermore, Gibbs states are an important component in quantum Boltzmann machines, see Sec.~\ref{sec:QBM}, and quantum enhanced semi-definite programming algorithms \cite{brandao2017quantumsdp, brandao2019quantumsdp, vanApeldoorn2020quantumsdpsolvers}.
Gibbs states can only be prepared efficiently in some special cases \cite{brandaoFiniteCorrLengthEfficientPrep19, kastoryano2016quantum_gibbs_samplers, PoulinThermalQGibbs09} but not in general.
In fact, it is shown in \cite{aharonov2013quantumpcp} that the preparation of a Gibbs state at low temperatures can be as difficult as finding a Hamiltonian's ground state, i.e., QMA\footnote{QMA is the quantum complexity counterpart to NP \cite{Watrous2012QuantumCompComplexity}.} hard.

Thus, one may have to rely on one of the approximation techniques which are discussed in the following.
The schemes presented in \cite{Temme2011QuantumMS, YungQuantumMetropolis12} employ quantum variants of the well-known Metropolis algorithm to prepare Gibbs states. However, they are based on the use of quantum phase estimation \cite{AbramsQPE99} and, thus, likely to require error-corrected quantum computers. 
Thermalization-based approaches are presented in \cite{PoulinThermalQGibbs09, Anschtz2019RealizingQB}, where the aim is to prepare a quantum Gibbs state by coupling the state register to a heat bath given in the form of an additional working quantum register.
However, it is a priori unknown what environmental system would be suitable. Thus, one may have to conduct a study of potential candidate systems. Furthermore, the method discussed in \cite{PoulinThermalQGibbs09} also requires quantum phase estimation, i.e., fault-tolerant quantum computers.
Variational Gibbs state preparation methods which are based on the fact that Gibbs states minimize the free energy of a system at constant temperature are presented in \cite{VarThermofieldDSJingxiang19, WiebeVariationalGibbs2020, wang2020variationalGibbs, guo2021thermalqState}. 
 The underlying difficulty is the estimation of the von Neumann entropy in every step of the state preparation. To resolve this problem \cite{WiebeVariationalGibbs2020} suggests to use a Fourier series expansion combined with quantum amplitude estimation \cite{brassardQAE02}. Similarly to quantum phase estimation, this algorithm is not well suited for near-term quantum computing applications. A simpler approximation is given in \cite{wang2020variationalGibbs} which uses a second-order Taylor approximation of the free energy. This can be efficiently computed with a SWAP test \cite{EkertSwapTest02, Cincio_2018SwapTest, Suba__2019_swap_test}.
Furthermore, the work presented in \cite{MottaQITE20} employs an efficient implementation of quantum imaginary time evolution, see Sec.~\ref{sec:varqite}, for states with finite correlations and, thereby, enables the estimation of quantum thermal averages, i.e., sampling  with respect to a given observable in accordance with the probability distribution of the respective Gibbs state.
An alternative approach which prepares an approximation to the given Gibbs state with a variational quantum imaginary time evolution is discussed in detail in Sec.~\ref{sec:varqite_gibbs}. This Gibbs state approximation technique is also used in Sec.~\ref{sec:QBM} to realize a variational quantum Boltzmann machine.

One of the biggest challenges in Hamiltonian encoding is the quantification of the approximation accuracy of $\ket{\psi_t\left(x\right)}$ respectively $\rho^{\text{Gibbs}}\left(x\right)$. We do not necessarily have access to an analytic description or even samples of these quantum states\footnote{Notably, there also exist problems where samples of the resulting quantum state are given and the goal is to learn the corresponding Hamiltonian \cite{AnshuSample-efficientQManyBody, verdon2019quantumHbasedModels}.}. In this case, statistical metrics such as the Kullback-Leibler divergence or the Kolomogorov-Smirnov statistic cannot be applied.
Instead, one has to rely on method-specific error bounds, such as the efficient, a posteriori error bounds for variational quantum time evolution presented in Sec.~\ref{sec:error_qite} and Appendix \ref{subsec:error_qrte} which are applicable to the variational Gibbs state preparation presented in Sec.~\ref{sec:varqite_gibbs} respectively the variational Hamiltonian simulation presented in Appendix \ref{app:varqrte}.

\subsection{Information Processing}
\label{sec:model}
In order to train a QML algorithm, we need to process information which can, e.g., be given as training data $X$ or parameters $\bm{\omega}$. 
Information processing can involve data preparation, dimensionality reduction, quantum data loading, classification, Gibbs state preparation, post-processing, measurements, etc. 
Since QML models have access to classical as well as quantum computational resources, data can be processed using classical functions
\begin{align}
    f: \mathbb{C}\mapsto \mathbb{C},
\end{align}
functions that map classical information to quantum resources
\begin{align}
    f: \mathbb{C}\mapsto \mathcal{H},
\end{align}
such as the data encoding maps discussed in Sec.~\ref{sec:data_encoding}, quantum channels
\begin{align}
    f: \mathcal{H}_A \mapsto \mathcal{H}_B,
\end{align}
and measurements of quantum states
\begin{align}
    f:  \mathcal{H} \mapsto \mathbb{C}.
\end{align}
Furthermore, QML algorithms that are suitable for execution with current respectively near-term quantum hardware usually employ quantum channels that are  given as \emph{variational quantum circuits}, which are also known as parameterized quantum circuits or quantum neural networks.
A variational quantum circuit represents a parameterized unitary $U(\bm{\omega})$ with $\bm{\omega}\in\mathbb{R}^k$ acting on  $\mathcal{H}$
\begin{align}
    \tilde\rho\left(\bm{\omega}\right)=U(\bm{\omega})\rho U^{\dagger}(\bm{\omega}).
\end{align}
The respective circuits often comprise of alternating layers of $d$ repeating layers of parameterized single qubit gates, i.e.,
\begin{align}U_r(\bm{\omega}_r)=\bigotimes\limits_{j=0}^{n-1}U^j_r({\omega}^j_r),
\end{align}
with $U^j_r({\omega}^j_r)$ acting on the $j^{\text{th}}$ qubit, ${\omega}^j_r$ corresponding to the $nr+j^{\text{th}}$ element in $\bm{\omega}$, $k=nd$,
and so-called \emph{entangling blocks} $V_r$ which apply multi-qubit gates to generate entanglement between the different qubits such that
\begin{align}
\label{eq:variational_qc}
    U\left(\bm{\omega}\right)=\prod\limits_{r=0}^{d-1} V_rU_r(\bm{\omega}_r).
\end{align}
In the remainder of this chapter, we simplify the respective notation as follows: The parameters ${\omega}^j_r$ are enumerated as $\omega_i$ for $i\in\{0, \ldots, k-1\}$ and the corresponding gates $U^j_r({\omega}^j_r)$ are referred to as $U_i(\omega_i)$. Then, we can write
\begin{align}
\label{eq:simplified_var_qc}
    U\left(\bm{\omega}\right)=\prod\limits_{i=0}^{k-1} V_iU_i(\omega_i),
\end{align}
where $V_i$ either corresponds to $\mathds{1}$ or a multi-qubit gate.

Similar to increasing the number of layers in deep neural networks \cite{Goodfellow2016_DeepL}, increasing the number of layer repetitions increases the number of parameters and, thus, also the \emph{expressivity} of the quantum circuit\footnote{A study of expressivity and entanglement capabilities for different circuit architectures may be found in \cite{SukinExpressibilityqCircuits19}.}. In this context, expressivity refers to the fraction of the $2^n$-dimensional Hilbert space that can be represented by the circuit.
One would expect that quantum circuits with more parameters are preferable over those with less. However, as is discussed in more detail in Sec.~\ref{sec:vanishing_grads}, highly expressive circuits can lead to training difficulties due to so-called \emph{barren plateaus}.

\subsection{Training}
\label{sec:training}
Next, we can define a loss function that depends on the parameters $\bm{\omega}$ of a variational quantum circuit that prepares a parameterized quantum state as \linebreak $\ket{\psi\left(\bm{\omega}\right)} = U(\bm{\omega})\ket{0}^{\otimes n}$.
Usually, we consider cost functions that are either given as expectation values
\begin{align}
\label{eq:exp_val_fun}
    L\left(\ket{\psi\left(\bm{\omega}\right)}\right) = \langle\psi(\bm{\omega})|\hat O|\psi(\bm{\omega})\rangle,
\end{align}
with $\hat O$ representing a Hermitian observable, e.g., given as a Hamiltonian of the form
\begin{equation}
    H\left({\bm{\theta}}\right)=\sum_{c=0}^{p-1}\theta_ch_c,
\end{equation}
with $h_c=\bigotimes_{j=0}^{n-1}\sigma_c^j$ for $\sigma_c^j\in\set{I, X, Y, Z}$ acting on the $j^{\text{th}}$ qubit and $\bm{\theta}\in\mathbb{R}^p$, or alternatively as
\begin{align}
\label{eq:generic_cost_fun}
    L\left(\ket{\psi\left(\bm{\omega}\right)}\right) = \sum\limits_x p_x(\bm{\omega})f\left(x\right),
\end{align}
where $ p_x(\bm{\omega}) = |\langle\psi(\bm{\omega})|x\rangle|^2$ corresponds to the sampling probabilities of the computational basis states $\ket{x}$.
For instance, setting $f\left(x\right) = -\log\left(p_x^{\text{target}}\right)$, where $p_x^{\text{target}}$ represents the target probability for $x$, results in the averaged negative log-likelihood
\begin{equation}
	L\left(\bm{\omega}\right) = -\sum\limits_x p_x(\bm{\omega})\log\left(p_x^{\text{target}}\right).
\end{equation}
The goal of the QML algorithms considered in this thesis can be written as optimization problems of the form
\begin{align}
\label{eq:cost_function_quantum}
    \underset{\bm{\omega}}{\min}\,L\left(\ket{\psi\left(\bm{\omega}\right)}\right).
\end{align}
The loss landscape of this type of optimization problem is typically highly non-convex. Therefore, the optimization method must be chosen carefully. It has been shown that gradient-based optimizers are not only suitable for convex problems but can also offer a promising approach for non-convex loss functions \cite{GradientBasedNon-convexPaquette2018, Non-convexStochasticGhadimi26, NaturalGradientDescentNon-Convex2019Zhang}.

\section{Quantum Gradients}
\label{sec:gradients}
Gradients are a particularly useful tool in the context of classical as well as quantum machine learning as they provide valuable information about the loss landscape.
Furthermore, the learning performance of a classical as well as quantum machine learning algorithms strongly relies on the chosen optimization scheme. Gradient-based optimizers offer a versatile optimization approach that  enable informed parameter update choices. 
The training of a QML loss function often requires the evaluation of gradients with respect to the parameters $\bm{\omega}$ of a quantum channel.


The remaining section is structured as follows. First, Sec.~\ref{sec:quantum_grad_types} presents various types of quantum gradients.
Next, we are going to introduce two approaches which enable the evaluation of analytic quantum gradients in Sec.~\ref{sec:analytic_gradients}. Finally, Sec.~\ref{sec:vanishing_grads} discusses different settings that lead to exponentially vanishing gradients and can, therefore, cause critical difficulties in QML training.

\subsection{Quantum Gradient Types}
\label{sec:quantum_grad_types}

In the remainder of this section, we discuss different types of quantum gradients, i.e., \emph{first-order gradients}, \emph{second-order gradients}, the \emph{quantum Fisher information} and \emph{quantum natural gradients}. Notably, the latter corresponds to gradients that are rescaled by the quantum Fisher information and, thus, consider the information geometry of the loss landscape. This, in turn, can help to improve the training convergence. 

\subsubsection{First-Order Gradients}
First-order gradients of the loss function $L\left(\ket{\psi\left(\bm{\omega}\right)}\right)$ w.r.t.~the circuit parameters, i.e., $
    \nabla_{\bm{\omega}} L\left(\ket{\psi\left(\bm{\omega}\right)}\right)$
reflect the slope of the cost function and, thus, enable informed parameter updates during optimization routines, e.g., using gradient descent \cite{CauchyGD1847}. In fact, it was shown in \cite{HarrowGradient21} that gradient-based optimization can lead to a speed-up in convergence compared to gradient free or finite difference-based optimization for variational quantum algorithms.
\pagebreak 
\begin{definition}
\label{def:first_order_exp_val_grads}
Given a pure quantum state $\ket{\psi\left(\bm{\omega}\right)} = U(\bm{\omega})\ket{0}^{\otimes n}$, where $\bm{\omega}\in\mathbb{R}^k$ and a Hermitian operator $\hat O$, the respective \textbf{first-order quantum gradient} reads
\begin{align}
   \frac{\partial\langle\psi(\bm{\omega})|\hat O|\psi(\bm{\omega})\rangle}{\partial\omega_i} = 2\text{Re}\left(\bra{\psi(\bm{\omega})}\hat O\, \frac{\partial\ket{\psi(\bm{\omega})}}{\partial\omega_i}\right),
\end{align}
for $i\in\set{0, \ldots, k-1}$.
\end{definition}

First-order gradients for cost functions of the form given in Eq.~\eqref{eq:exp_val_fun} can be directly computed with the gradients from  Def.~\ref{def:first_order_exp_val_grads}.
Derivatives of the loss functions given in Eq.~\eqref{eq:generic_cost_fun} require the evaluation of probability gradients, i.e.,
\begin{align}
\label{eq:probability_gradients}
\nabla_{\bm{\omega}}\sum_x  p_x(\bm{\omega})f\left(x\right) = \sum\limits_x \nabla_{\bm{\omega}} p_x(\bm{\omega})f\left(x\right),
\end{align}
where $ p_x(\bm{\omega}) = | \langle\psi(\bm{\omega})|x\rangle|^2$ corresponds to the measurement probability of a basis state $\ket{x}$ for $ x\in\set{0, \ldots, 2^n-1}$ forming an orthonormal basis of the underlying $2^n$ dimensional Hilbert space.
The respective derivatives are based on the probability quantum gradients $\textstyle{\nabla_{\bm{\omega}} p_x(\bm{\omega})}$ which are computed with Def.~\ref{def:first_order_exp_val_grads} where $\hat{O}=\proj{x}$. In practice, we do not have to measure all projectors but can employ a simple sampling scheme where we assume that those basis states that are not sampled have measurement probability $0$.

\subsubsection{Second-Order Gradients}

Furthermore, second-order gradients of the loss function $L\left(\ket{\psi\left(\bm{\omega}\right)}\right)$ w.r.t.~the circuit parameters, i.e., $
    \nabla^2_{\bm{\omega}}L\left(\ket{\psi\left(\bm{\omega}\right)}\right)$ reflect the curvature of the loss landscape \cite{kim2021quantumLandscape, Cerezo_2021higher_order_bps, Huembeli_2021Hessian_Loss} and, thus, provide valuable information about the training abilities of a model.
\begin{definition}
\label{def:second_order_exp_val_grads}
Given a pure quantum state $\ket{\psi\left(\bm{\omega}\right)} = U(\bm{\omega})\ket{0}^{\otimes n}$, where $\bm{\omega}\in\mathbb{R}^k$ and a Hermitian operator $\hat O$, the respective \textbf{second-order quantum gradient} reads
\begin{align}
    &\frac{\partial^2\langle\psi(\bm{\omega})|\hat O|\psi(\bm{\omega})\rangle}{\partial \omega_i\partial \omega_j}= 2\text{Re}\left(\frac{\partial\bra{\psi(\bm{\omega})}}{\partial \omega_i}\hat O\, \frac{\partial\ket{\psi(\bm{\omega})}}{\partial\omega_j} + \bra{\psi(\bm{\omega})}\hat O \, \frac{\partial^2\ket{\psi(\bm{\omega})}}{\partial\omega_i\partial\omega_j}\right),
\end{align}
for $i,j \in\set{0, \ldots, k-1}$.
\end{definition}

Second-order gradients of the cost function form given in Eq.~\eqref{eq:exp_val_fun} are directly related to the second-order quantum gradients from  Def.~\ref{def:second_order_exp_val_grads}, whereas the ones in Eq.~\eqref{eq:generic_cost_fun} use probability gradients
\begin{align}
 \nabla^2_{\bm{\omega}}\sum_x  p_x(\bm{\omega})f\left(x\right) &= \sum\limits_x   \nabla^2_{\bm{\omega}} p_x(\bm{\omega})f\left(x\right),
\end{align}
with $\nabla^2_{\bm{\omega}} p_x(\bm{\omega})$ being computed using Def.~\ref{def:second_order_exp_val_grads} by setting $\hat{O} = \proj{x}$. As in the first-order case, we do not need to evaluate all $2^n$ projectors $\proj{x}$ but can employ a practical sampling scheme that limits the evaluation to the observed $x$.
We would like to point out that other higher-order quantum gradients can be evaluated in a similar fashion \cite{Mari_2021Gradients}.

\subsubsection{Quantum Fisher Information}
The quantum Fisher information matrix (QFIM) \cite{QFIMBraunstein94, meyer2021fisher} is an interesting metric that reflects the information geometry induced by a parameterized quantum state. As such it can be used to define promising optimization routines, such as quantum natural gradients which are introduced in Sec.~\ref{sec:qng}, and to investigate the properties of QML models \cite{haug2021capacity}.

\begin{definition}
\label{def:qfi}
Let, $\ket{\psi\left(\bm{\omega}\right)} = U(\bm{\omega})\ket{0}^{\otimes n}$ be a pure quantum state, where $\bm{\omega}\in\mathbb{R}^k$. The \textbf{quantum Fisher Information matrix} is given by $4\mathcal{F}^Q\left(\bm{\omega}\right)$ with $\mathcal{F}^Q\left(\bm{\omega}\right)$ representing the Fubini-Study metric whose entries read
\begin{align}
\label{eq:qfim_def}
    \mathcal{F}^Q_{ij}\left(\bm{\omega}\right) &= -\frac{1}{2}\frac{\partial^2 |\braket{\psi(\bm{\omega'})|\psi(\bm{\omega})}|^2}{\partial\omega_i\partial\omega_j}
    \Big\rvert_{\bm{\omega'}=\bm{\omega}} \nonumber \\
    &= -\frac{1}{2}\frac{\partial^2 |\braket{\psi(\bm{\omega}')|\psi(\bm{\omega})}|^2}{\partial\omega_i\partial\omega_j}
    \Big\rvert_{\bm{\omega}'=\bm{\omega}} \nonumber \\
    &= \text{Re}\left(\frac{\partial \bra{\psi(\bm{\omega})}}{\partial \omega_i}\frac{\partial \ket{\psi(\bm{\omega})}}{\partial \omega_j} - \frac{\partial \bra{\psi(\bm{\omega})}}{\partial \omega_i}\proj{\psi(\bm{\omega})}\frac{\partial \ket{\psi(\bm{\omega})}}{\partial \omega_j}\right).
\end{align}
To simplify the notation we are going to use $\mathcal{F}^Q_{ij} := \mathcal{F}^Q_{ij}\left(\bm{\omega}\right) $.
\end{definition}

\subsubsection[Quantum Natural Gradient]{Quantum Natural Gradient\footnote{This section is reproduced in part, with permission, from J.~Gacon, C.~Zoufal, G.~Carleo, S.~Woerner, "Simultaneous Perturbation Stochastic Approximation of the Quantum Fisher Information", \textit{Quantum}, 5(567) 2021}}
\label{sec:qng}
In the following, we present the quantum natural gradient (QNG) that has favorable properties compared to standard gradients and can be advantageous for the training of QML models.
First, we introduce the classical natural gradient\cite{NatGradAmari98}. The special feature of this gradient form is that it considers the curvature of the model space by rescaling the gradient with the metric tensor underlying the parameter space.
The key idea is to take the information geometry \cite{Amari2016InfoGeometry} into account and, thus, improve the convergence of gradient-based optimization techniques. 

Standard gradient descent attempts to minimize the loss function $L\left(\bm{\omega}\right)$ by choosing the parameter update step
proportional to the negative gradient 
\begin{equation}
\label{eq:vanilla_gradient_descent_update}
    \bm{\omega}^{(j+1)} = \bm{\omega}^{(j)} -\eta \nabla_{\bm{\omega}}L\left(\bm{\omega}\right),
\end{equation}
with $\eta>0$ the learning rate and $\bm{\omega}^{(j)}$ corresponding to the parameters at the $j^{\text{th}}$ optimization iteration.
From the geometric perspective, this means selecting the direction of the steepest descent of the loss landscape -- w.r.t.~the $\ell_2$ norm -- that induces the smallest possible change in the parameter space. This is also described by the following optimization problem
\begin{equation}\label{eq:vanilla_argmin}
 \bm{\omega}^{(j + 1)} = \underset{\bm{\omega} \in \mathbb{R}^d}{\text{argmin} } \bigg(\big( \bm{\omega} - \bm{\omega}^{(j)}\big)^{T}\nabla_{\bm{\omega}} L\left(\bm{\omega}\right)+ \frac{1}{2\eta} \|\bm{\omega} - \bm{\omega}^{(j)}\|_2^2\bigg),
\end{equation}
where $||\cdot||_2$ corresponds to the $\ell_2$ norm.

A problem of the standard gradient descent is that all parameters $\bm{\omega}$ are rescaled by a constant factor which is directly reflected in the magnitude of the gradient descent step.
If the learning rate is not properly adjusted, multiplying
the parameters by a large constant may lead to overshooting the desired values while multiplying by a small constant can severely slow down the speed of convergence.
This rescaling problem may be resolved by explicit consideration of the information geometry induced by the parameter space. 
More specifically, we want to find the parameters that minimize $L\left(\bm{\omega}\right)$ while keeping the change in the parameter space -- w.r.t.~the Riemannian metric -- as small as possible.
To this end, we replace the norm $||\cdot||_2$ in Eq.~\eqref{eq:vanilla_argmin} by $||\cdot||_{g(\bm{\omega})}$ where $g(\bm{\omega}) \in \mathbb{R}^{d \times d}$ denotes the Riemann metric tensor induced by the paramterized model and get the corresponding optimization problem
\begin{align}
   \bm{\omega}^{(j + 1)} = \underset{\bm{\omega} \in \mathbb{R}^d}{\text{argmin} }\bigg(\big( \bm{\omega} - \bm{\omega}^{(j)}\big)^{T}\nabla_{\bm{\omega}}L\left(\bm{\omega}\right)+ \frac{1}{2\eta} \|\bm{\omega} - \bm{\omega}^{(j)}\|_{g(\bm{\omega}^{(k)})}^2\bigg),
\end{align}
which results in the natural gradient descent \cite{NatGradAmari98}
\begin{equation}\label{eq:natural_gradient}
    \bm{\omega}^{(j+1)} = \bm{\omega}^{(j)} - \eta g^{-1}(\bm{\omega}^{(j)})  \nabla_{\bm{\omega}}L\left(\bm{\omega}\right).
\end{equation}
As proven in \cite{NatGradAmari98}, not the standard gradient but the natural gradient corresponds to the actual steepest descent if the parameter space is described by a Riemannian instead of a Euclidean metric.

Now, suppose the cost function depends on a parameterized quantum state $\ket{\psi\left(\bm{\omega}\right)}$, i.e., $
    L\left(\ket{\psi\left(\bm{\omega}\right)}\right).$
Then, the information geometry underlying the parameterized model corresponds to the Fubini-Study metric $\mathcal{F}^{Q}$ given in  Def.~\ref{def:qfi} and $g(\bm{\omega}) = \mathcal{F}^{Q}$.
The resulting QNG reads
\begin{align}
    \left(\mathcal{F}^{Q}\right)^{-1}  \nabla_{\bm{\omega}} L\left(\ket{\psi\left(\bm{\omega}\right)}\right).
\end{align}
A drawback of this gradient method is that $\mathcal{F}^{Q}$ generally requires the evaluation of $\mathscr{O}(k^2)$ expectation values, where $k$ denotes the number of ansatz parameters. Thus, it could happen that a QNG-based optimizer converges better than gradient descent with respect to the number of optimization iterations but worse in terms of the total computational resources.
It has, therefore, been suggested to use approximations to the Fubini-Study metric to reduce the computational cost. In this context, \cite{Stokes_2020QNG} proposes the use of a diagonal respectively block-diagonal approximation which reduces the complexity to $\mathscr{O}(k)$. Furthermore, \cite{gacon2021simultaneous} introduces a well-performing stochastic approximation whose complexity is constant, i.e., independent of $k$.
Potential advantages of QNG-based optimization are investigated, e.g., in \cite{Stokes_2020QNG, gacon2021simultaneous, yamamoto2019naturalVQE, lopatnikova2021QNG_VarBayes}.

\subsection{Analytic Quantum Gradients}
\label{sec:analytic_gradients}

As discussed before, quantum gradients are vital to facilitate gradient-based optimization for QML and to study the respective loss landscape.
An approximation to the quantum gradients defined in Sec.~\ref{sec:quantum_grad_types} can be computed using, e.g., forward finite differences where
\begin{equation}
    \frac{\partial L\left(\bm{\omega}\right)}{\partial\omega_i} \approx  \frac{L\left(\bm{\omega} + {\epsilon}\bm{e_i} \right) - L\left( \bm{\omega} \right)}{\epsilon},
\end{equation}
for $\bm{e_i}$ representing the unit vector with $0$ everywhere except for a $1$ at the $i^{\text{th}}$ position and a small $\epsilon > 0$. However, finite difference gradients are volatile to the choice of $\epsilon$. It is particularly difficult to find a suitable $\epsilon$ that enables a reasonable quantum gradient approximation given the typical non-convex QML loss landscape and evaluation with noisy quantum hardware. Thus, analytic quantum gradients are a crucial part of QML applications and their evaluation received a significant amount of attention. 

In this section, the \emph{parameter shift} \cite{izmaylov2021analyticqgradients, hubregtsen2021singlecomponent, kyriienko2021generalizedQCDiff, Mari_2021Gradients, SchuldQuantumGradients19, Zoufal2019qGANs, Banchi2021measuringanalytic} as well as the \emph{linear combination} \cite{SchuldQuantumGradients19, VarSITEMcArdle19, VarQBMZoufal20} techniques for analytic quantum gradient computation are presented, illustrative examples are shown and the respective advantages and disadvantages are discussed.
To that end, we consider the following setting.
Suppose an observable $\hat O$ and a pure, parameterized quantum state that is generated with a variational quantum circuit as
\begin{equation}
    \ket{\psi(\bm{\omega})}=U\left(\bm{\omega}\right)\ket{0}^{\otimes n},
\end{equation}
where $\bm{\omega}\in\mathbb{R}^k$ and $U\left(\bm{\omega}\right)$ is given by the $n$-qubit variational ansatz defined in Eq.~\eqref{eq:simplified_var_qc}.

\subsubsection{Parameter Shift Gradients}

First, we consider the simple case of Pauli generators $\sigma_i\in\set{I, X, Y, Z}$ whose eigenvalues are $\pm 1$ and corresponding unitaries $U_i({\omega}_i)=\ee^{-i\frac{\omega_i}{2}\sigma_i}$ acting on qubit $i$ such that
\begin{align}
    U_i\left(\omega_i\right)=\ee^{i\frac{\omega_i}{2}{\sigma_i}}=\cos{\big(\frac{\omega_i}{2}\big)}\mathds{1} + i \sin{\big(\frac{\omega_i}{2}\big)}\sigma_i.
\end{align}
It follows from trigonometric identities \cite{Li_2017QuantumOptimalControl} that 
\begin{align}
        \frac{\partial\langle\psi(\bm{\omega})|\hat O|\psi(\bm{\omega})\rangle}{\partial \omega_i}  = \frac{1}{2}\big(\langle\psi(\bm{\omega}_{+\frac{\pi}{2}\bm{e_i}})|\hat O|\psi(\bm{\omega}_{+\frac{\pi}{2}\bm{e_i}})\rangle - \langle\psi(\bm{\omega}_{-\frac{\pi}{2}\bm{e_i}})|\hat O|\psi(\bm{\omega}_{-\frac{\pi}{2}\bm{e_i}})\rangle \big),
\end{align}
with $\bm{\omega}_{\pm \gamma\bm{e_i}}$ denoting a parameter vector which is equivalent to $\bm{\omega}$ except for the parameter $\omega_i$ which is shifted by $\pm\gamma$.

More generally, suppose Hermitian generators ${M_i}$ with $M_i^\dagger M_i =\mathds{1}$ and two eigenvalues $\pm q_i$. Then, it also holds that the corresponding unitary operator 
\begin{align}
    U_i\left(\omega_i\right)=\ee^{i\omega _i{M_i}}=\cos{(\omega_i)}\mathds{1} + i \sin{(\omega_i)}M_i,
\end{align}
and that \cite{SchuldQuantumGradients19}
\begin{align}
        \frac{\partial\langle\psi(\bm{\omega})|\hat O|\psi(\bm{\omega})\rangle}{\partial \omega_i}   = q_i\big(\langle\psi(\bm{\omega}_{+\frac{\pi}{4q_i}\bm{e_i}})|\hat O|\psi(\bm{\omega}_{+\frac{\pi}{4q_i}\bm{e_i}})\rangle - \langle\psi(\bm{\omega}_{-\frac{\pi}{4q_i}\bm{e_i}})|\hat O|\psi(\bm{\bm{\omega}}_{-\frac{\pi}{4q_i}\bm{e_i}})\rangle \big).
\end{align}
Interestingly, it was derived in \cite{Mari_2021Gradients} that the shift may also be conducted independently of the eigenvalues $\pm q_i$, i.e.,
\begin{align}
    \frac{\partial\langle\psi(\bm{\omega})|\hat O|\psi(\bm{\omega})\rangle}{\partial \omega_i}  = \frac{\langle\psi(\bm{\omega}_{+s\bm{e_i}})|\hat O|\psi(\bm{\omega}_{+s\bm{e_i}})\rangle - \langle\psi(\bm{\omega}_{-s\bm{e_i}})|\hat O|\psi(\bm{\omega}_{-s\bm{e_i}})\rangle}{2\sin{(s)}},
\end{align}
where  $s\in\mathbb{R}$ such that it is not an integer multiple of $\pi$. This offers additional flexibility and may potentially help to reduce the sampling variance.
The gradient evaluations derived above are generally referred to as \emph{parameter shift} gradients.
Furthermore, these parameter shift gradients can also be easily extended to second-order gradients 
by applying two shifts for $\omega_i$ and $\omega_j$
\begin{align}
    &\frac{\partial^2\langle\psi(\bm{\omega})|\hat O|\psi(\bm{\omega})\rangle}{\partial \omega_i\partial\omega_j} \nonumber \\
    &= \frac{\big(\langle\psi(\bm{\omega}_{+s_i\bm{e_i},+s_j\bm{e_j}})|\hat O|\psi(\bm{\omega}_{+s_i\bm{e_i},+s_j\bm{e_j}})\rangle +\langle\psi(\bm{\omega}_{-s_i\bm{e_i},-s_j\bm{e_j}})|\hat O|\psi(\bm{\omega}_{-s_i\bm{e_i},-s_j\bm{e_j}})\rangle \big)}{4\sin(s_i\bm{e_i})\sin(s_j\bm{e_j})} \nonumber \\
    &\hspace{5mm}- \frac{\big( \langle\psi(\bm{\omega}_{+s_i\bm{e_i},-s_j\bm{e_j}})|\hat O|\psi(\bm{\omega}_{+s_i\bm{e_i},-s_j\bm{e_j}})\rangle + \langle\psi(\bm{\omega}_{-s_i\bm{e_i},+s_j\bm{e_j}})|\hat O|\psi(\bm{\omega}_{-s_i\bm{e_i},+s_j\bm{e_j}})\rangle\big)}{4\sin(s_i\bm{e_i})\sin(s_j\bm{e_j})}
\end{align}
and similarly to the evaluation of
\begin{align}
   \mathcal{F}^Q_{ij} = \frac{\partial^2\langle{\tilde\psi(\bm{\omega}, \bm{\omega}')}\proj{0}^{\otimes n}{\tilde\psi(\bm{\omega}, \bm{\omega}'))}\rangle}{\partial \omega_i\partial\omega_j}\Big\rvert_{\bm{\omega}'=\bm{\omega}}
\end{align}
by employing second-order parameter shift gradients with $\hat O = \proj{0}^{\otimes n}$ and $\ket{\tilde\psi(\bm{\omega}, \bm{\omega}')}=U^\dagger(\bm{\omega}')\ket{\psi(\bm{\omega})}$.

\subsubsection{Linear Combination Gradients}

In the following, we use the simplified notation 
\begin{align*}
    \tilde U_{a,b}=\prod\limits_{i=a}^{b} V_iU_i(\omega_i).
\end{align*}
Now, consider a parameterized unitaries that can be written as
\begin{align}
    U_i\left(\omega_i\right) = e^{i\omega_i M_i},
\end{align}
where $M_i=\bigotimes\limits_{j=0}^{n-1}\sigma_i^{j}$ for $\sigma_i^{j}\in\set{I, X, Y, Z}$ \cite{nielsen10} acting on the $j^{\text{th}}$ qubit, denotes a Hermitian matrix. 
It follows that
\begin{equation}
\frac{\partial U_i\left(\omega_i\right)}{\partial\omega_i} = iM_iU_i\left(\omega_i\right),
\end{equation}
and, thus,
\begin{align}
\frac{\partial\ket{\psi(\bm{\omega})}}{\partial\omega_i} =
i\tilde U_{i+1,k-1}M_i\tilde U_{0,i}\ket{0}^{\otimes n} .
\end{align}
This directly implies that first-order gradients can be written as
\begin{align}
\label{eq:first_order_exp_val_grad}
\frac{\partial\langle\psi(\bm{\omega})|\hat O|\psi(\bm{\omega})\rangle}{\partial \omega_i} &= 2\text{Re}\Big(i\bra{\psi(\bm{\omega})}\hat O\, \tilde U_{i+1,k-1}M_i\tilde U_{0,i}\ket{0}^{\otimes n}\Big),
\end{align}
second-order gradients as 

\pagebreak
\begin{align}
\label{eq:second_order_exp_val_grad}
    \frac{\partial^2\langle\psi(\bm{\omega})|\hat O|\psi(\bm{\omega})\rangle}{\partial \omega_i\partial \omega_j}\nonumber &=2\text{Re}\Big(\bra{0}^{\otimes n}\tilde U_{0,j}^{\dagger}M_j \tilde U^{\dagger}_{j+1, k-1}\, \hat O \, \tilde U_{i+1, k-1}M_i\tilde U_{0,i}\ket{0}^{\otimes n}   \nonumber \\
    &\hspace{5mm}- \bra{\psi(\bm{\omega})}\,\hat O \,\tilde U_{j+1,k-1}M_j\tilde U_{i+1,j}M_i\tilde U_{0,i}\ket{0}^{\otimes n} \Big),
\end{align}
and the Fubini-Study metric entries as
\begin{align}
\label{eq:qfi_eval}
    \mathcal{F}^Q_{ij}
&= \text{Re}\Big(\bra{0}^{\otimes n}\tilde U_{0,j-1}^{\dagger}M_j\tilde U_{i+1,j-1}M_i\tilde U_{0,i}\ket{0}^{\otimes n} \nonumber \\
&\hspace{5mm}-\bra{0}^{\otimes n}\tilde U_{0,j-1}^{\dagger}M_j\tilde U_{0,j-1}\ket{0}\bra{0}^{\otimes n}\tilde U_{0,i-1})^{\dagger}M_i\tilde U_{0,i-1}\ket{0}^{\otimes n} \Big),
\end{align}
for $i<j$.

Now, we are ready for the introduction of \emph{linear combination} (LC) quantum gradients.
There are two possibilities to compute the first-order gradients using either the quantum circuit structure shown in Fig.~\ref{fig:laflamme} introduced by \cite{LaflammeSimulatingPhysPhenom02} -- that we refer to as \emph{first-order implicit LC circuit} -- or the quantum circuit structure illustrated in  Fig.~\ref{fig:lin_comb} -- which we denote as \emph{first-order explicit LC circuit}. Both cases enable the evaluation of 
 \begin{equation}
\text{Re}\Big(i\bra{\psi(\bm{\omega})}\hat O\, \tilde U_{i+1,k-1}M_i\tilde U_{0,i}\ket{0}^{\otimes n}\Big),
\end{equation}
by setting $\alpha = \frac{\pi}{2}$, the input state to $\tilde U_{0,i}\ket{0}^{\otimes n}$, $V=M_i$ and $W=\tilde U_{i+1,k-1}$. Furthermore, the former needs to make use of the fact that every complex square matrix can be written as a weighted sum of unitaries \cite{SchuldQuantumGradients19}, leading to
\begin{align}
\label{eq:o_decomposition}
    \hat O = \sum_{c=0}^{p-1} \theta_c O_c,
\end{align}
with $\theta_c\in\mathbb{R}$ and $O_c$ corresponding to unitaries. Then, $p$ separate first-order explicit LC circuits are evaluated where the $O_c$ are implemented with the $U$ gate. Due to linearity the results can be summed as
\begin{align}
    \sum_{c=0}^{p-1} \theta_c\text{Re}\Big(i\bra{\psi(\bm{\omega})}O_c\, \tilde U_{i+1,k-1}M_i\tilde U_{0,i}\ket{0}^{\otimes n}\Big).
\end{align}
On the other hand, the structure of the first-order explicit LC circuit does not require the decomposition of $\hat O$, such that $U=\mathds{1}$, but enables a direct evaluation of an expectation value with respect to the extended observable $\mathscr{B}\otimes\hat O$. Naively, the evaluation of this expectation value also requires to decompose $\hat O$, to measure all $p$ terms $O_c$ independently and, then, to sum the results according to Eq.~\eqref{eq:o_decomposition}. However, as showed in \cite{HamamuraEfficientQuantumObservables2020, Crawford2021efficientquantum, bravyi2017tapering, McClean_2016VariationalQCAlgorithms} efficient measurement strategies can be exploited to reduce the number of independent measurements. A detailed presentation of the terms for both approaches is given in Tbl.~\ref{tbl:lin_comb_grad}.

\begin{figure}[h!]
\captionsetup{singlelinecheck = false, format= hang, justification=raggedright, font=footnotesize, labelsep=space}
\begin{center}
\begin{tikzpicture}
\node at (0,0) {\includegraphics[width=0.61
\linewidth]{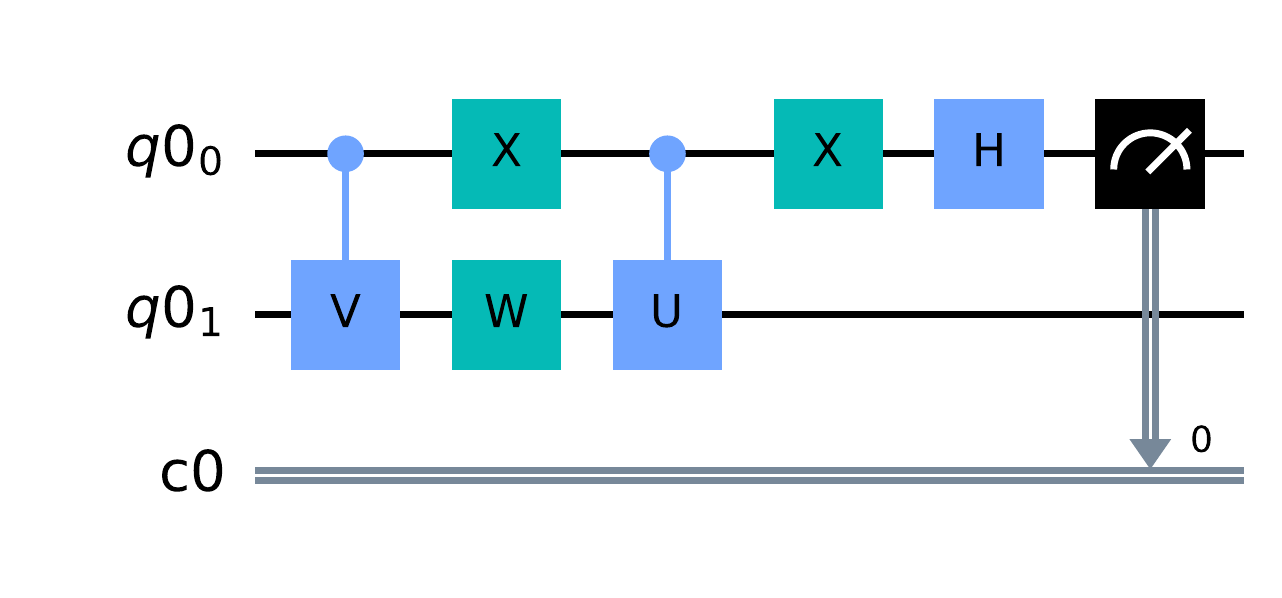}};
\node at (-5.8,1.1) {$\ket{0}+e^{i \alpha} \ket{1}$};
\node at (-5.3,-0.3) {$\ket{\psi_{\textnormal{in}}}$};
\node at (-4.7,-1.45) {$c$};
\node at (3.5, 1.9) {\small{$\mathscr{B}$}};
\end{tikzpicture}
\end{center}
\caption{\textbf{Linear combination gradients.} The auxiliary working qubit is measured with respect to the observable $\mathscr{B}$. The measurement outcome which is stored in the classical register $c$ corresponds to the real respectively imaginary part of $\ee^{i\alpha}\bra{\psi_{\text{in}}}W^{\dagger}\,U^{\dagger}\,W\,V\ket{\psi_{\text{in}}}$ for $\mathscr{B}=Z$ respectively $\mathscr{B}=-Y$.}
\label{fig:laflamme}
\end{figure}
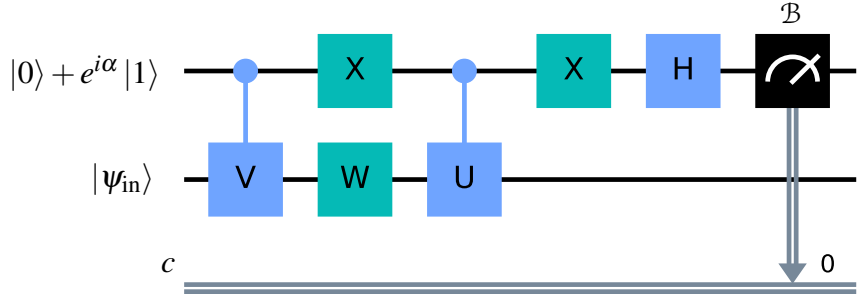

\begin{figure}[h!]
\captionsetup{singlelinecheck = false, format= hang, justification=raggedright, font=footnotesize, labelsep=space}
\begin{center}
\begin{tikzpicture}
\node at (0,0) {\includegraphics[width=0.61
\linewidth]{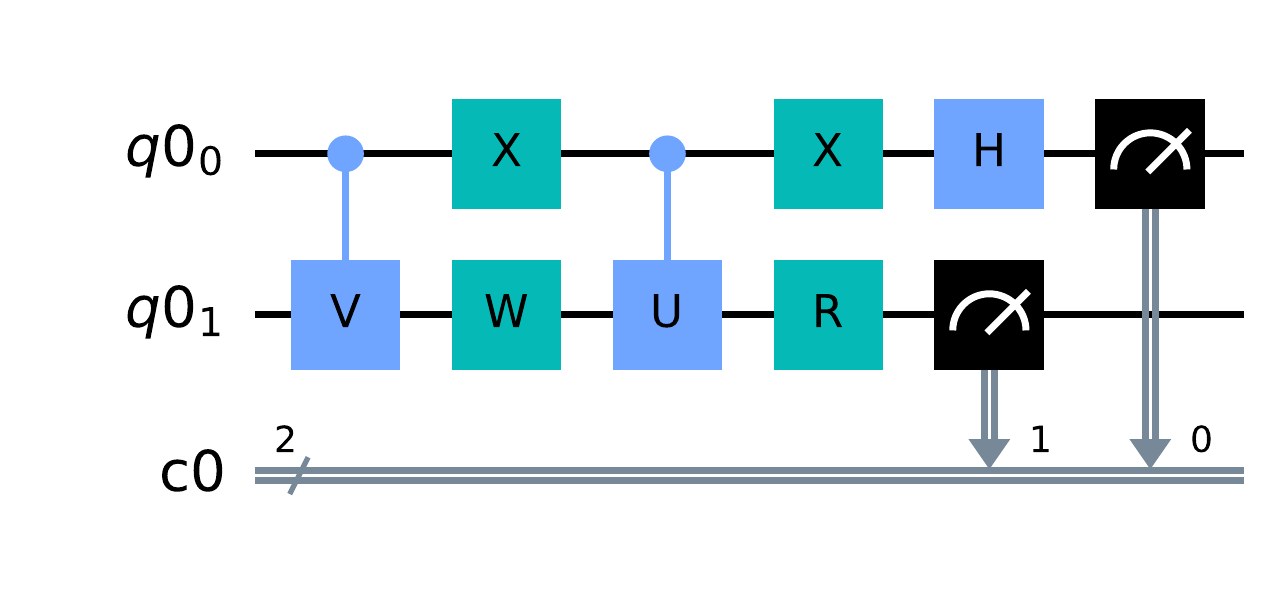}};
\node at (-5.8,1.2) {$\ket{0}+e^{i \alpha} \ket{1}$};
\node at (-5.3,-0.3) {$\ket{\psi_{\textnormal{in}}}$};
\node at (-4.65,-1.55) {$c$};
\node at (3.4,1.85) {\small{$\mathscr{B}$}};
\node at (2.05,0.35) {\small{$\hat{O}$}};
\end{tikzpicture}
\end{center}
\caption{\textbf{Linear combination gradients with explicit measurement of} $\boldsymbol{\hat O}$. The auxiliary working qubit is measured with respect to the observable $\mathscr{B}$. The measurement outcome which is stored in the classical register $c$ corresponds to the real respectively imaginary part of $\ee^{i\alpha}\bra{\psi_{\text{in}}}W^{\dagger}U^{\dagger}R^{\dagger}\,\hat O\,RWV\ket{\psi_{\text{in}}}$ for $\mathscr{B}=Z$ respectively $\mathscr{B}=-Y$.}
\label{fig:lin_comb}
\end{figure}

Furthermore, second-order gradients can either be computed with the first-order explicit LC circuit or with the circuit shown in Fig.~\ref{fig:killoran_second_order_lin_comb}. The latter is a simplified version of a circuit given in \cite{killoran2018qgans} and referred to as \emph{second-order explicit LC circuit}.
In the former case, we evaluate 
\begin{align}
\text{Re}\Big(\bra{0}^{\otimes n}\tilde U_{0,j}^{\dagger}M_j \tilde U^{\dagger}_{j+1, k-1}\, \hat O \, \tilde U_{i+1, k-1}M_i\tilde U_{0,i}\ket{0}^{\otimes n} \Big),
\end{align}
with the presented circuit and 
\begin{align}
\text{Re}\Big(\bra{\psi(\bm{\omega})}\,\hat O \,\tilde U_{j+1,k-1}M_j\tilde U_{i+1,j}M_i\tilde U_{0,i}\ket{0}^{\otimes n} \Big),
\end{align}
with almost the same circuit, i.e., we drop the $X$ gates.
The evaluation of the second-order explicit LC circuit enables the evaluation of both terms with one circuit. The gates are set in a similar fashion to the first-order gradients. The exact settings are listed in Tbl.~\ref{tbl:lin_comb_grad}.

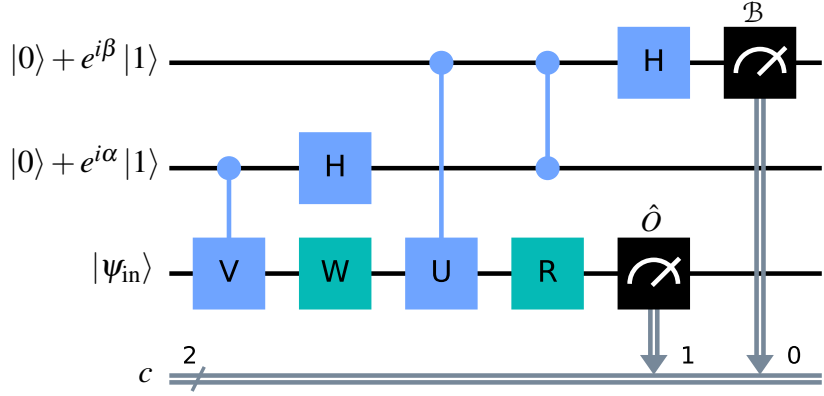
\begin{figure}[h!]
\captionsetup{singlelinecheck = false, format= hang, justification=raggedright, font=footnotesize, labelsep=space}
\begin{center}
\begin{tikzpicture}
\node at (0,0) {\includegraphics[width=0.6
\linewidth]{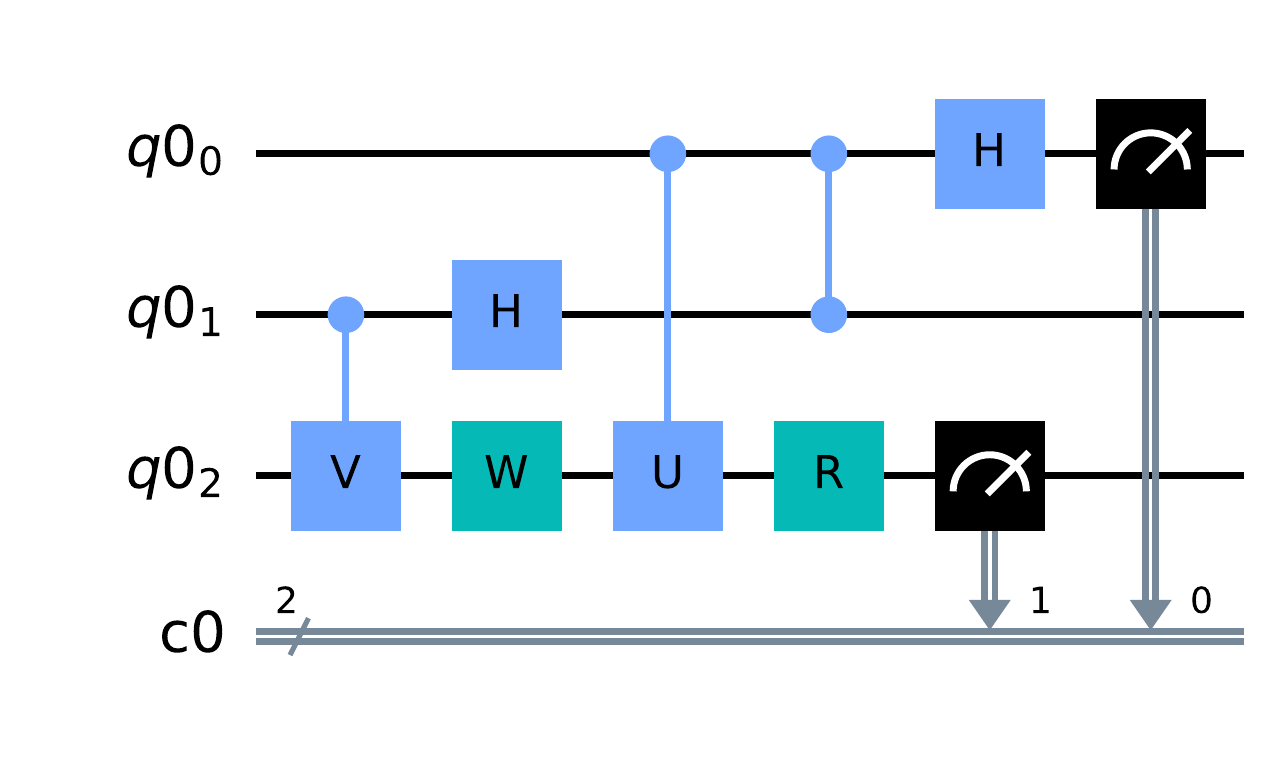}};
\node at (-5.45,  2.) {$\ket{0}+e^{i \beta} \ket{1}$};
\node at (-5.45,0.6) {$\ket{0}+e^{i \alpha} \ket{1}$};
\node at (-4.95,-0.8) {$\ket{\psi_{\textnormal{in}}}$};
\node at (-4.65,-2.22) {$c$};
\node at (3.35, 2.6) {\small{$\mathscr{B}$}};
\node at (2.,-0.13) {\small{$\hat{O}$}};
\end{tikzpicture}
\end{center}
\caption{\textbf{Linear combination second-order gradients with explicit measurement of} $\boldsymbol{\hat O}$. The auxiliary working qubit is measured with respect to the observable $\mathscr{B}$. The measurement outcome which is stored in the classical register $c$ corresponds to the real respectively imaginary part of $\ee^{i\alpha}\bra{\psi_{\text{in}}}\left(\ee^{-i\beta}W^{\dagger}U^{\dagger}R^{\dagger}\,\hat O\,RWV - \ee^{i\beta}W^{\dagger}R^{\dagger}\,\hat O\,RUWV\right)\ket{\psi_{\text{in}}}$ for $\mathscr{B}=Z$ respectively $\mathscr{B}=-Y$.}
\label{fig:killoran_second_order_lin_comb}
\end{figure}

Lastly, the matrix entries $\mathcal{F}^Q_{ij}$ may be evaluated with the first-order implicit LC circuit. The first term
\begin{align}
 \text{Re}\Big(\bra{0}^{\otimes n}\tilde U_{0,j-1}^{\dagger}M_j\tilde U_{i+1,j-1}M_i\tilde U_{0,i}\ket{0}^{\otimes n}\Big),
\end{align}
is evaluated by setting $V=M_i$ and $U=M_j$ and measuring w.r.t.~the observable $\mathscr{B}=Z$.
The second term 
\begin{align}\text{Re}\Big(\bra{0}^{\otimes n}\tilde U_{0,j-1}^{\dagger}M_j\tilde U_{0,j-1}\ket{0}\bra{0}^{\otimes n}\tilde U_{0,i-1}^{\dagger}M_i\tilde U_{0,i-1}\ket{0}^{\otimes n} \Big),
\end{align}
requires the independent evaluation of 
\begin{align}
\bra{0}^{\otimes n}\tilde U_{0,j-1}^{\dagger}M_j\tilde U_{0,j-1}\ket{0}^{\otimes n},
\end{align}
respectively
\begin{align}\bra{0}^{\otimes n}\tilde U_{0,i-1}^{\dagger}M_i\tilde U_{0,i-1}\ket{0}^{\otimes n},
\end{align}
by setting $V=M_j$ respectively $V=M_i$ and measuring w.r.t.~$\mathscr{B}=Z-iY$. Then, one can multiply the terms and evaluate the real part.

Alternatively, one may also evaluate $\mathcal{F}^Q_{ij}$
directly by employing the second-order LC circuit with $V=M_i$, $U=M_j$ and $\hat O = \proj{\psi(\omega')}$.
We refer, the reader to Tbl.~\ref{tbl:lin_comb_grad} for the remaining settings.

\begin{table}[!htb]
\captionsetup{singlelinecheck = false, format=hang, justification=raggedright, font=footnotesize, labelsep=space}
\small{
\begin{center}
\begin{tabular}{c|c|c|c|c|c|c}
 & \textbf{Method} & $\bm{\alpha}$ & $\bm{W}$ &$\bm{U}$ & $\bm{R}$ & $\bm{B}$\\
 \hline
\multirow{2}{*}{$\nabla_{\bm{\omega}}L$} &  implicit LC & \multirow{2}{*}{$\frac{\pi}{2}$} & \multirow{2}{*}{$\tilde U_{i+1, k-1}$} & $O_c$ & - & \multirow{2}{*}{$Z$}\\
 &   explicit LC &   &  & $\mathds{1}$ & $\mathds{1}$ & \\
 \hline
\multirow{2}{*}{$\nabla^2_{\bm{\omega}}L$} & explicit LC  & \multirow{2}{*}{0} &  \multirow{2}{*}{$\tilde U_{i+1, j}$} & \multirow{2}{*}{$M_j$} & \multirow{2}{*}{$\tilde U_{j+1, k-1}$} & \multirow{2}{*}{$Z$}\\
&  explicit $2$-LC & & & & & \\
\hline
\multirow{3}{*}{$\mathcal{F}^Q$} & implicit LC& \multirow{3}{*}{0} & $\tilde U_{i+1, j}$ & $M_j$ & \multirow{2}{*}{-} & $Z$ \\
&  implicit LC &  & $\mathds{1}$ & $\mathds{1}$ &  & $Z-iY$\\
& explicit $2$-LC  &  & $\tilde U_{i+1, j}$ & $M_j$ & $\tilde U_{j+1, k-1}$ & $Z$
\end{tabular}
\end{center}
}
\caption{This table lists the different parameter settings of linear combination quantum gradients for first-order and second-order gradients as well as the QFIM by using a first-order implicit, first-order explicit or second-order explicit circuit structure for the case that $i<j$. Additionally, it always holds that the initial state ${\ket{\mathbf\psi}_{\text{\textbf{in}}}}=\textstyle{\tilde U_{0, i}\ket{0}^{\otimes n}}$, $V=M_i$ and ${\beta}=0$.
Furthermore, the computation of the QFIM with first-order implicit circuits requires the summation of two different terms which are evaluated according to the settings given in the two rows for the 'Implicit LC' method.}
\label{tbl:lin_comb_grad}
\end{table}

\subsubsection{Comparison of Analytic Gradient Methods}

The different approaches for analytic quantum gradient calculation have their advantages and disadvantages.
Firstly, the quantum circuit implementation of the parameter shift evaluation as well as the linear combination method based on the second-order LC circuit for $\mathcal{F}^Q$ is twice as long as the linear combination method based on the the first-order implicit LC circuit.
In contrast to the parameter shift method, the linear combination gradients require additional qubits: one for first-order and two for second-order.
Secondly, the parameter shift gradients have the advantage that only the parameter values change but not the circuit structure. We may, therefore, store a pre-compiled circuit to increase the evaluation efficiency. The linear combination gradients, on the other hand, differ for all gradients with respect to different parameters.
Notably, one particular disadvantage of the first-order implicit gradient evaluation is that the number of expectation values that have to be evaluated scale with the number of terms $p$ in $\hat{O}$. The number of expectation values that have to be computed for the explicit schemes are independent of $p$. 
Tbl.~\ref{tbl:grad_comparison} summarizes the comparison for the evaluation of the gradient $\nabla_{\bm{\omega}}\langle\psi(\bm{\omega})|\hat O|\psi(\bm{\omega})\rangle$ 
 with $\bm{\omega}\in\mathbb{R}^k$, $\textstyle{\hat O = \sum_{c=0}^{p-1} \theta_c O_c}$ and assuming that $\ket{\psi(\bm{\omega}}$ is prepared with a quantum circuit that has depth $d$.

\begin{table}[h!]
\captionsetup{singlelinecheck = false, format=hang, justification=raggedright, font=footnotesize, labelsep=space}
\small{
\begin{center}
\begin{tabular}{c|c|c|c|c|c}
\textbf{Gradient} & \textbf{Method} & \textbf{Depth} & $\#$ \textbf{Qubits} & $\#$ \textbf{Circuits}  & $\#$ \textbf{Expectation Values}\\
 \hline
 \multirow{3}{*}{$\nabla_{\bm{\omega}}L$} & param. shift & $d$ & $n$ & $1$ & $2k$ \\
 &  implicit LC & $d+2$ & $n+1$ & $pk$ & $pk$ \\
 & explicit LC & $d+1$ & $n+1$ & $k$ & $k$\\

\hline
 \multirow{2}{*}{$\nabla^2_{\bm{\omega}}L$} & param. shift & $d$  & $n$ & $1$ & $2k^2$\\
 & explicit LC & $d+2$ & $n+1$ & $k^2$ & $k^2$\\
&  explicit $2$-LC & $d+3$ & $n+2$  &  ${k^2}/{2}$ &${k^2}/{2}$\\

\hline
 \multirow{2}{*}{$\mathcal{F}^Q$} & param. shift & $2d$ & $n$ & $1$ &  $2k^2$ \\
 &  implicit LC & $d+2$  & $n+1$ & ${k^2}/{2}$ &  ${k^2}/{2}+k$\\
& explicit $2$-LC & $2(d+2)$ & $n+2$ & ${k^2}/{2}$&  ${k^2}/{2}$\\

\end{tabular}
\end{center}
}
\caption{The table illustrates the differences between the quantum gradient methods presented in this section: parameter shift, first-order implicit LC, first-order explicit LC and second-order explicit LC. To that end, we compare the quantum circuit depth, the number of qubits, the number of individual circuits and the number of different expectation values that need to be evaluated.}
\label{tbl:grad_comparison}
\end{table}

The variance of the estimated gradients depends on various factors such as the number of shots, the number of expectation values required to compute the estimate as well as the chosen hyper-parameters. It remains an open topic for future research to investigate the variance behavior for different gradient method choices. 

Finally we would like to point out that the linear combination approach offers a lot of flexibility. Unlike the parameter shift gradients it also enables us to directly evaluate imaginary instead of only real parts, by changing the phase $\alpha$ of the top working qubit.

\subsection{Vanishing Gradients}
\label{sec:vanishing_grads}

Quantum states that are represented by $n$ qubits live in a Hilbert space $\mathcal{H}$ of dimension $2^n$. The exponential size of the state space represents a particularly interesting property as well as one of the biggest obstacles for QML. More explicitly, random unitaries acting on the respective quantum states can suffer from a \emph{concentration of measure} \cite{Ledoux2001TheCO} and, thereby, lead to exponentially vanishing gradients also known as \emph{barren plateaus}\cite{Clean_2018_BarrenPlateaus}.
Suppose, you are given an $n$-qubit quantum state 
\begin{align}
    \ket{\psi\left(\bm{\omega}\right)} = U\left(\bm{\omega}\right)\ket{0}^{\otimes n},
\end{align}
where $U\left(\bm{\omega}\right)$ corresponds to a universal ansatz that can represent all elements in the unitary group $\text{U}\left(N\right)$ with $N=2^n$.
More specifically, let the ansatz correspond to an item of $\text{U}\left(N\right)$ that is sampled uniformly at random w.r.t.~the Haar measure $\mu_N$ \cite{HaarMeasure1933, holmes2021AnsatzExpressBarrenPlateaus}. This measure describes the volume of all items in the unitary group $\text{U}\left(N\right)$.
Given an integrable function $g$ acting on elements $u\in\text{U}\left(N\right)$, the respective integral is left and right invariant, i.e.,
\begin{align}
    \int g\left(u\right)d\mu_N\left(u\right) = \int g\left(vu\right)d\mu_N\left(u\right) = \int g\left(uv\right)d\mu_N\left(u\right),
\end{align}
for $v\in \text{U}\left(N\right)$.

The unitary group suffers from a phenomenon called \emph{concentration of measure}. Quantum states that are prepared by elements of the unitary group -- pure quantum states $\ket{\psi}$ -- correspond to points on the $2n-1$ dimensional unit sphere $S^{2n-1}$ in $\mathbb{R}^{2n}$. Given a Lipschitz-continuous function $f\left(\ket{\psi}\right): S^{2n-1} \mapsto \mathbb{R}$ with Lipschitz constant $\eta$, the probability to sample a value
$f\left(\ket{\psi}\right)$ that is is at most $\epsilon$-close to the expectation value $E\left(f\right)$ with $\epsilon > 0$ decreases exponentially in the system size $n$, i.e.,
\begin{align}
    P\left(\rvert E\left(f\right) - f\left(\ket{\psi}\right) \rvert \geq \epsilon \right) \leq 2\ee^{-\frac{n\epsilon^2}{9\pi^2\eta^2}}.
\end{align} This is known as Levy's lemma \cite{Ledoux2001TheCO}.
Although ans\"atze that are representative for the Haar measure require exponential resources and are, thus, unlikely to be used in practice, \emph{$t$-designs}, which correspond to a set of unitaries whose moments are equivalent to the Haar distribution up to order $t$, fulfill the same properties on some functions $f$ \cite{Hayashi_2005tdesigns, Ambainis07t-designs, Oliveira_2007GenericEntanglement06, Dahlsten_2007t-designs}.
In fact, the above problem already occurs if the given ansatz forms a $2$-design.
Furthermore, we would like to point out that $t$-designs \cite{Harrow09_2-designs, harrow2018approximatet_designs} can already be represented by quantum circuits of depth $\text{poly}\left(n, t\right)$.

Suppose, further, that you are given an ansatz $U\left(\bm{\omega}\right)$ corresponding to Eq.~\eqref{eq:simplified_var_qc} such that $ \ket{\psi\left(\bm{\omega}\right)} = U\left(\bm{\omega}\right)\ket{0}^{\otimes n}$. As shown in \cite{Clean_2018_BarrenPlateaus}, $
   \big \langle f_i\big(\ket{\psi\left(\bm{\omega}\right)}\big) \big\rangle_{\bm{\omega}}= 0,\:\forall i $ for
\begin{align}
    f_i\big(\ket{\psi\left(\bm{\omega}\right)}\big)&=\frac{\partial\langle\psi(\bm{\omega})|\hat O|\psi(\bm{\omega})\rangle}{\partial\omega_i} \nonumber \\
    &= 2\text{Re}\Big(i\bra{\psi(\bm{\omega})}\hat O\, \left(\prod\limits_{j=i+1}^{k-1} V_kU_j(\omega_j)\right)^{\dagger}M_i\prod\limits_{j=0}^{i} V_jU_j(\omega_j)\ket{0}^{\otimes n}\Big),
\end{align}
if $\textstyle{\prod\limits_{j=0}^{i} V_jU_j(\omega_j)}$ and $\textstyle{\prod\limits_{j=i+1}^{k-1} V_kU_j(\omega_j)}$ are independent and either of them forms at least a $2$-design. The combination of these results leads to the vanishing gradient phenomenon.

To sum up, a parameterized unitary $U\left(\bm{\omega}\right)$ which is sufficiently expressive in the sense that it enables the occupation of a large fraction of the full Hilbert space can lead to exponentially small gradients which in turn do not provide useful information for QML training tasks. 
The direct connection between ansatz expressibility and barren plateau magnitudes is investigated in \cite{holmes2021AnsatzExpressBarrenPlateaus}.
To scale up QML algorithms, it may, thus, become important to design problem-specific ans\"atze which only enable to access a limited but suitable subspace of the full Hilbert space.
Further literature on mitigation strategies for ansatz induced barren plateaus are, e.g., given by \cite{Volkoff_2021BP_mitigation_correlation, SkolikLayerwiseLearningQNNs2021, Grant_2019BPInitialization, pesah2020absenceQCNNs}.

We would like to point out that barren plateaus manifest themselves also in higher-order gradients \cite{Cerezo_2021higher_order_bps} as well as in the Fisher information \cite{Abbas_2021PowerofQNNs}.
This problem not only occurs in gradient-based optimization. Even finite-difference and gradient-free optimization suffer from the concentration of measure \cite{arrasmith2020effectbps_grad_freeOpt} since also the difference between two loss function values computed for different parameter values becomes exponentially small in the system size. In other words, the loss landscape becomes increasingly flat for bigger $n$.

The chosen ansatz itself is not the only source for potential barren plateaus.
Vanishing gradients can also be due to cost functions depending on a global observable which acts on the full quantum state instead of a $k$-local fraction thereof. This means that QML loss functions must be designed carefully. For a formal definition, we refer the interested reader to \cite{CerezoCostFunctDependentBarrenPlats21}. 
Furthermore, using an ansatz or a model Hamiltonian that contains too much entanglement combined with the use of a partial trace may lead to vanishing gradients \cite{Wiebe2020Barren, sharma2020trainability, Holmes_2021BPs_scramblers}.
Limiting the entanglement in an ansatz can already help to mitigate this issue \cite{Patti_2021EntanglementDevisedBPs}.
Moreover, locally acting hardware noise \cite{wang2021noiseinducedbps} may lead to exponentially vanishing gradients as well. Considering the mitigation of noise induced barren plateaus, we suggest the interested reader the following literature on error mitigation \cite{ErrorMitigationTemme17} respectively error correction \cite{ShorErrorCorrection95} strategies. 

The known sources of barren plateaus, potential mitigation strategies and related open questions are listed in Tbl.~\ref{tbl:vanishing_gradients}.
In summary, the trainability of a QML models strongly depends on the chosen ansatz, cost function and backend type.

\begin{table}[h!]
\captionsetup{singlelinecheck = false, format= hang, justification=raggedright, font=footnotesize, labelsep=space}
\footnotesize{
\begin{tabular}{c|c|c}
\textbf{Cause} & \textbf{Mitigation Strategies }&\textbf{ Open Questions }\\
 \hline
 \makecell{ansatz forming \\ (approximate) $2-$design }& limit expressivity & \makecell{Can we employ \\ problem-specific ans\"atze?} \\
 \hline
 global cost function  & local cost function & \makecell{Can we avoid this with ans\"atze which\\ allow only a small but suitable fraction\\ of the Hilbert space to be occupied?} \\
 \hline
 entanglement & restrict entanglement & \makecell{Do vanishing gradients occur  \\ if hidden units are measured \\ instead of traced out?}\\
  \hline
 noise & error mitigation/correction & \makecell{Can we show that special types of\\  noise exist which are beneficial?} \\

\end{tabular}
}
\caption{This table lists various causes for vanishing gradients and presents potential mitigation strategies as well as related open questions.}
\label{tbl:vanishing_gradients}
\end{table}















\section[Variational Quantum Imaginary Time Evolution]{Variational Quantum Imaginary Time Evolution\footnote{This section is reproduced in part, with permission, from C.~Zoufal, D.~Sutter, S.~Woerner, "Error Bounds for Variational Quantum Time Evolution", Preprint available at  	arXiv:2108.00022, 2021}}
\label{sec:varqite}

Simulating quantum systems was not only the original motivation for building a quantum computer, it still presents one of the most promising applications.
\emph{Quantum time evolution} (QTE) describes the propagation of a quantum state according to a given Hamiltonian and applications thereof are numerous. A few examples which are relevant for quantum machine learning are given by combinatorial optimization problems \cite{gacon2021simultaneous}, the simulation of Ising models \cite{barison2021efficientQRTE} and the approximation of quantum Gibbs states \cite{VarQBMZoufal20, MottaQITE20, Temme2011QuantumMS, YungQuantumMetropolis12, WiebeVariationalGibbs2020}.

In order to implement QTE on a gate-based quantum computer, the respective evolution needs to be translated into quantum gates. This translation can, e.g, be approximated with Trotterization \cite{Lloyd1073UniversalQuantumSim96, MottaQITE20}. Depending on the evolution time and expected accuracy, this approach may lead to deep quantum circuits and is, thus, not well-suited for near-term quantum computers.
\emph{Variational quantum time evolution} (VarQTE)
\cite{VarSITEMcArdle19, Simon18TheoryVarQSim} is a powerful and versatile approach to simulate quantum time dynamics with parameterized quantum circuits. 
In the context of QML, VarQTE enables optimization via ground state preparation and approximate Gibbs state preparation that is compatible with automatic differentiation. Furthermore, it allows us to do these simulations on current and near-term quantum devices -- at least approximately using shallow, parameterized quantum circuits.
The usage of short circuits results in a limited expressivity of the available parameterized unitaries and implies that VarQTE generally comes with an approximation error. It is crucial to quantify this error, to be able to interpret the results, and to possibly rerun simulations with more expressive parameterized circuits in case the error is too large.

In the following, we first introduce quantum time evolution in Sec.~\ref{sec:qte}. Then, Sec.~\ref{sec:varqite_ground} presents a variational quantum imaginary time evolution (VarQITE) implementation based on McLachlan's variational principle \cite{McLachlan64}. \linebreak Sec.~\ref{sec:error_qite} introduces an efficient, a posteriori error bound for the error accumulated by VarQITE.
Next, we discuss two applications of VarQITE, namely ground state search, see Sec.~\ref{app:ground_state_runtime}, and approximate Gibbs state preparation, see Sec.~\ref{sec:varqite_gibbs}. 
Then, Sec.~\ref{sec:varqite_chainRule} discusses how VarQITE can facilitate automatic differentiation for QML algorithms. Furthermore, details which are important to enable a successful VarQITE implementation are described in Sec.~\ref{sec:methods}. Lastly, Sec.~\ref{sec:varqite_examples} presents various examples run with numerical simulations as well as actual quantum hardware to illustrate the power of VarQITE and the respective error bounds.

\subsection{Quantum Time Evolution}
\label{sec:qte}
QTE describes the process of evolving a quantum state over time with respect to a Hamiltonian $H$ \footnote{
Notably, we can also define a generalized time evolution which can be used to do matrix multiplications, solve systems of linear equations and combine real and imaginary time evolution~\cite{VarQSGeneralEndo20}.}. 
Real time evolution enables the study of unitary quantum dynamics, e.g., of fermionic models such as the Hubbard model \cite{BarendsFermionicModels2015}, of a state $\ket{\psi^*_t}$  described by the time-dependent Schr\"odinger equation
\begin{align}
\label{eq:schroedingerRTE}
    i\hslash\ket{\dot\psi^*_t}= H\ket{\psi^*_t},
\end{align}
where the time derivative is denoted as $\ket{\dot\psi^*_t} = \frac{\partial \ket{\psi^*_t}}{\partial t}$.
In the following, the notation is simplified by setting $\hslash=1$
The resulting time-dependence of the state reads
\begin{align}
    \ket{\psi^*_t} = \ee^{-iHt}\ket{\psi^*_0}.
\end{align}

If the time parameter $t$ is replaced by an imaginary time $i t$ the system dynamics change to a non-unitary evolution which is mathematically described by the normalized, Wick-rotated Schr\"odinger equation
\begin{align}
\label{eq:wick_schroedinger}
    \ket{\dot\psi^*_t} = \left( E^*_t\mathds{1} - H\right)\ket{\psi^*_t},
\end{align}
where $E^*_t = \bra{\psi^*_t} H\ket{\psi^*_t}$ corresponds to the system energy.
In the remainder of this work, the notation for $E^*_t\mathds{1} - H$ is simplified to $E^*_t - H$.
The resulting state evolution reads
\begin{align}
    \ket{\psi^*_t} = \frac{\ee^{-Ht}\ket{\psi^*_0}}{\sqrt{\bra{\psi^*_0}\ee^{-2Ht}\ket{\psi^*_0}}}\, .
    \end{align}
Suppose that the initial state has a non-zero overlap with the ground state of $H$, then all components that do not correspond to the ground state are damped exponentially in time during imaginary time evolution. This form of time evolution is, thus, a particularly useful tool to find the ground state of $H$ \cite{VarSITEMcArdle19}, see Sec.~\ref{app:ground_state_runtime}.
Furthermore, imaginary time evolution can be used to solve partial differential equations \cite{gonzalezconde2021pricing, fontanela2021quantumpdes, Kubo_2021StochastidDifferential} or to prepare quantum Gibbs states \cite{VarQBMZoufal20, Simon18TheoryVarQSim}, see Sec.~\ref{sec:varqite_gibbs}. \
For the remainder of this section, we are going to focus on imaginary time evolution as well as a its variational implementations. 
Further details on quantum real time evolution and its the variational implementation are given in Appendix \ref{app:varqrte}.

\subsection{Variational Approach}
\label{sec:varqite_ground}

VarQTE approximates the target state $\ket{\psi^*_t} $ with a state $\ket{\psi^{\omega}_t}$ whose time dependence is projected onto the parameters $\bm{\omega_t}$ of a variational ansatz. To simplify the notation, the time parameter $t$ is dropped from $\bm{\omega}=(\omega_0, \ldots, \omega_{k-1})\in\mathbb{R}^{k}$ in the remainder of this work when referring to the ansatz parameters. 
More specifically, we consider a formulation for pure states based on McLachlan's variational principle \cite{McLachlan64} with a global phase-agnostic evolution \cite{Simon18TheoryVarQSim}.
The state evolution described by this variational principle corresponds to an initial value problem where the underlying \emph{ordinary differential equation} (ODE)~\cite{Tahir-Kheli2018ODEs} is derived from McLachlan's variational principle \cite{McLachlan64}.
We simulate the time evolution by numerically solving the ODE for a set of initial parameter values.

In the following, we discuss the idea of a variational quantum imaginary time evolution implementation which is agnostic to a potential time-dependent global phase \cite{Simon18TheoryVarQSim, VarSITEMcArdle19}.
Consider the imaginary time evolution of a parameterized state with an explicit time-dependent global phase parameter $\nu$, i.e., $\ket{\psi^{\nu}_t} = \ee^{-i\nu}\ket{\psi^{\omega}_t}$ for $\nu = \nu_t \in \mathbb{R}$,
where
\begin{align}
\ket{\dot\psi^\nu} = -i\dot\nu\ee^{-i\nu}\ket{\psi^{\omega}_t}- \ee^{-i\nu}\ket{\dot{\psi}^{\omega}_t}.
\end{align}
The normalized, Wick-rotated Schr\"odinger equation of an evolution of the state $\ket{\psi^{\nu}_t}$ reads
\begin{align}
\label{eq:global_phase_wick}
         \ket{\dot\psi^\nu} = \left( E_t^{\omega} - H\right)\ket{\psi^\nu}.
\end{align}
Thus,
\begin{align}
     \ee^{-i\nu}\ket{\dot\psi^{\omega}_t} = \left(E_t^{\omega} - H+  i\dot\nu\right) \ee^{-i\nu}\ket{\psi^{\omega}_t},
\end{align}
and division by $\ee^{-i\nu}$ leads to
\begin{align}
     \ket{\dot\psi^{\omega}_t} = \left(E_t^{\omega} - H+  i\dot\nu\right)\ket{\psi^{\omega}_t}.
\end{align}
As we employ a variational ansatz $\ket{\psi^{\omega}_t}$ with the time-dependence being encoded in the parameters $\bm{\omega}$, we cannot necessarily find parameter updates  $\bm{\dot\omega}$ such that
\begin{align}
    \ket{\dot\psi^{\omega}_t} = \sum_i\dot\omega_i\frac{\partial\ket{\psi^{\omega}_t}}{\partial\omega_i},
\end{align}
suffices
Eq.~\eqref{eq:global_phase_wick} exactly. The aim of McLachlan's variational principle is to find $\ket{\dot\psi^{\omega}_t}$ that minimizes a potential error in Eq.~\eqref{eq:global_phase_wick} w.r.t.~the  $\ell_2$-norm $\norm{x}_2 = \sqrt{\langle x,x\rangle}$, i.e.,
\begin{align}
     \delta\norm{\ket{\dot\psi^{\omega}_t} - \left(E_t^{\omega}-H+i\dot\nu\right)\ket{\psi^{\omega}_t}}_2&= 0.
\end{align}
Next, we evaluate the variational principle with respect to $\dot\nu$, i.e.,
\begin{align}
     \delta_{\dot\nu}\norm{\ket{\dot\psi^{\omega}_t} - \left(E_t^{\omega}-H+i\dot\nu\right)\ket{\psi^{\omega}_t} }_2 &= 0,
\end{align}
to find $\dot\nu = -\text{Im}(\braket{\dot\psi^{\omega}_t|\psi^{\omega}_t})$.
This facilitates a variational time evolution of $\ket{\psi^{\omega}_t}$ that simulates a global phase degree of freedom $\nu$ without actual implementation or tracking of $\ee^{-i\nu}$, i.e.,
\begin{align}
     \label{eq:MacLachlan_phase_agnostic_imag}
     \ket{\dot\psi^{\dot\nu}_t}= \left(E_t^{\omega}-H\right)\ket{\psi^{\omega}_t},
\end{align}
with the effective state gradient 
\begin{align} \label{eq_forDavid}
\ket{\dot\psi^{\dot\nu}_t} = \ket{\dot\psi^{\omega}_t} + i\text{Im}(\braket{\dot\psi^{\omega}_t|\psi^{\omega}_t})\ket{\psi^{\omega}_t}.
\end{align}
Rewriting the variational principle accordingly gives
\begin{align}
\label{eq:VarQITE_phase_agnostic}
     \delta\norm{\ket{\dot\psi^{\dot\nu}_t} - \Big(E_t^{\omega}-H\Big)\ket{\psi^{\omega}_t} }_2&= 0.
\end{align}

Plugging the chain rule given above into Eq.~\eqref{eq:VarQITE_phase_agnostic} and evaluating the variational principle with respect to $\bm{\dot\omega}$ leads to the following system of linear equations (SLE) \cite{Liesen2015LinAlg}
\begin{align}
\label{eq:McLachlanVarQITE}
\sum\limits_{j=0}^k \mathcal{F}^Q_{ij} \dot\omega_j= - \text{Re}\left(C_i\right),
\end{align}
where $C_i = \frac{\partial \bra{\psi^{\omega}_t}}{\partial \omega_i}H\ket{\psi^{\omega}_t}$ and
$\mathcal{F}_{ij}^Q$ denotes the $(i,j)$-entry of the Fubini-Study metric from Def.~\ref{def:qfi}, i.e.,
\begin{align}
\mathcal{F}^Q_{ij} =
\text{Re}\left(\frac{\partial \bra{\psi^{\omega}_t}}{\partial \omega_i}\frac{\partial \ket{\psi^{\omega}_t}}{\partial \omega_j} - \frac{\partial \bra{\psi^{\omega}_t}}{\partial \omega_i}\proj{\psi^{\omega}_t}\frac{\partial \ket{\psi^{\omega}_t}}{\partial \omega_j}\right). 
\end{align}
As pointed out before, the Fubini-Study metric is proportional to the quantum Fisher Information matrix for pure quantum states.
Further details on the evaluation of the terms in Eq.~\eqref{eq:McLachlanVarQITE} are given in Sec.~\ref{sec:analytic_gradients}.

Solving Eq.~\eqref{eq:McLachlanVarQITE} for $\bm{\dot\omega}$ leads to an ODE
that describes the evolution of the ansatz parameters with respect to Eq.~\eqref{eq:MacLachlan_phase_agnostic_imag}, i.e.,
\begin{align}
\label{eq:standardODE_varqite}
        f_{\text{std}}\left(\bm{\omega}\right) = -\left(\mathcal{F}^Q\right)^{-1}\text{Re}\left( \bm{C}\right),
\end{align}
with $\bm{C} = \left(C_0, \ldots, C_k\right)$.

As mentioned before, this approximate time evolution implementation can lead to inexact state gradients such that the gradient error
\begin{align}
\label{eq:grad_error_varqite}
    \ket{e_t}:=\ket{\dot\psi^{\dot\nu}_t}  - \Big(E_t^{\omega}-H\Big)\ket{\psi^{\omega}_t},
\end{align}
may give $\|\ket{e_t}\|>0$.
Eq.~\eqref{eq:grad_error_varqite} motivates an alternative ODE for VarQITE which is given as the following optimization problem
\begin{align}
\label{eq:argminODE_varqite}
  f_{\text{min}}\left(\bm{\omega}\right)=\underset{\boldsymbol{\dot{\omega}}\in\mathbb{R}^{k}}{\text{argmin}} \, \|\ket{e_{t}}\|_2 ^2,
\end{align}
for
\begin{align}
\label{eq:et_imag}
    \|\ket{e_{t}}\|_2  ^2=   \Var(H)_{\psi^{\omega}_t} 
    &+ \braket{\dot\psi^{\omega}_t|\dot\psi^{\omega}_t} - \braket{\dot\psi^{\omega}_t|\psi^{\omega}_t}\braket{\psi^{\omega}_t|\dot\psi^{\omega}_t} \nonumber  \\
    & +2\mathrm{Re}\big(\!\bra{\dot\psi^{\omega}_t}H\ket{\psi^{\omega}_t}\!\big),
\end{align}
following from
$
    \textstyle{2\text{Re}\left(\braket{\psi^{\omega}_t|\dot\psi^{\omega}_t}\right) = \frac{\partial\braket{\psi^{\omega}_t|\psi^{\omega}_t}}{\partial t} = 0},
$ 
 and 
$
\textstyle{2\text{Im}\left(\braket{\dot\psi^{\omega}_t|\dot\psi^{\omega}_t}\right) = 0}.
$
We would like to point out that solving Eq.~\eqref{eq:standardODE_varqite} with a least square solver is analytically equivalent to solving Eq.~\eqref{eq:argminODE_varqite}. However, as the simulation results in Sec.~\ref{sec:varqite_examples} show, the numerical behavior of the latter is more stable.

Since the time-dependence of $\ket{\psi^{\omega}_t}$ is encoded in the parameters $\bm{\omega}$, we can rewrite Eq.~\eqref{eq:et_imag} as 
\begin{align}
\|\ket{e_{t}}\|_2  ^2 
= \Var(H)_{\psi^{\omega}_t} \!+\! \sum_{i,j}\dot \omega_i \dot \omega_j \mathcal{F}^Q_{ij} \!+\! 2\sum_i\dot\omega_i\text{Re}(C_i), \nonumber
\end{align}
where it is used that
\begin{es}
\text{Re}\big(\!\bra{\dot\psi^{\omega}_t}H\ket{\psi^{\omega}_t}\!\big) = \sum_i\dot\omega_i\text{Re}(C_i),
\end{es}
as well as that
\begin{align} \label{eq_Christa_reform1}
   \braket{\dot\psi^{\omega}_t|\dot\psi^{\omega}_t} - \braket{\dot\psi^{\omega}_t|\psi^{\omega}_t}\braket{\psi^{\omega}_t|\dot\psi^{\omega}_t} = \sum_{i,j}\dot \omega_i \dot \omega_j \mathcal{F}^Q_{ij}.
\end{align}
This facilitates the efficient evaluation of $\|\!\ket{e_t}\!\|_2^2 $.
Examples for VarQITE implementations are presented in Sec.~\ref{sec:varqite_examples}.

\subsection{Error Bound}
\label{sec:error_qite}
In this section, we derive an a posteriori error bound for variational quantum imaginary time evolution.
The obtained error bound allows us to efficiently, quantify the approximation error with respect to the exact quantum imaginary time evolution. The bound is phrased in terms of the Bures metric ~\cite{HayashiQuantumInfo06} and is, hence, agnostic to physically irrelevant global phase mismatches -- a feature that is crucial for practical applications. This represents an improvement compared to existing results for variational quantum real time evolution~\cite{MartinazzoErrorVarQuantumDyn20}.
Thus, it directly defines a lower bound on the fidelity between prepared and target state.
The power of the error bound is demonstrated on various numerical examples, in Sec.~\ref{sec:varqite_examples}, where we also investigate the performance of VarQITE and illustrate the application to concrete settings.

Let $\ket{\psi^{\omega}_t}$ be the state prepared by the variational algorithm at time $t$ and denote the ideal target state by $\ket{\psi^*_t}$. 
To formalize an error bound, we want to use a metric which describes the distance between two quantum states. A popular distance measure is the fidelity \cite{nielsen10} given by $|\braket{\psi^{\omega}_t|\psi^*_t}|^2$. Unlike the $\ell2$ norm, the fidelity is invariant to changes in the global phases.  Since the global phase is physically irrelevant, 
this is a desired property for a meaningful quantum state distance measure. Although the fidelity itself is not a metric, it may be used to define the Bures metric~\cite{HayashiQuantumInfo06}, given by
\begin{align}
\label{eq:Bures}
	&B\left(\proj{\psi^{\omega}_t}, \proj{\psi^*_t} \right)= \nonumber \\ 
	&\hspace{20mm} \sqrt{\braket{\psi^{\omega}_t|\psi^{\omega}_t}+\braket{\psi^*_t|\psi^*_t}-2|\braket{\psi^{\omega}_t|\psi^*_t}|},
 \end{align}
 where the states $\ket{\psi^{\omega}_t}$ and $\ket{\psi^*_t}$ are not necessarily normalized.
 If the states are normalized then the Bures metric simplifies to 
 \begin{align}
  \label{eq_buresPhase}
	B\left(\proj{\psi^{\omega}_t}, \proj{\psi^*_t} \right)&= \sqrt{2-2|\braket{\psi^{\omega}_t|\psi^*_t}|} \nonumber \\
	&=\!\min_{\phi \in [0,2\pi]}\!\norm{\ee^{i\phi}\ket{\psi^{\omega}_t} \!-\!\ket{\psi^*_t}}_2.
\end{align}
Unlike the $\ell2$ norm, the Bures metric is invariant to changes in the global phases.

Our goal is, now, to prove an error bound $\epsilon_t$ of the form
\begin{align}
	B\left(\proj{\psi^{\omega}_t}, \proj{\psi^*_t} \right) \leq \epsilon_t,
\end{align}
that can be evaluated efficiently in practice.
We would again like to point out the aforementioned relation of the Bures metric to the fidelity, which leads to
\begin{align}
    |\braket{\psi^{\omega}_t|\psi^*_t}|^2 \geq 1 - \frac{\epsilon_t^2}{2},
\end{align}
and implies that the relevant range of $\epsilon_t$ is $\left[0, \sqrt{2}\right]$ for normalized $\ket{\psi^{\omega}_t}$ and $\ket{\psi^*_t}$.
If the error bound estimate lies outside of this interval, then the fidelity and error can be clipped to $0$ and $\sqrt{2}$, respectively.

To prove an upper bound to the Bures metric between the target state $\ket{\psi^*_t}$ given by Eq.~\eqref{eq:wick_schroedinger} and $\ket{\psi^{\omega}_t}$ prepared with VarQITE, we need two preparatory lemmas. The first one quantifies the energy difference between $E_t^{\omega}$ and $E^*_t$.
\begin{lemma}[Energy difference] \label{lem_energyDiff}
For $t>0$, let $\ket{\psi^*_t}$ be the exact solution to Eq.~\eqref{eq:MacLachlan_phase_agnostic_imag} and $\ket{\psi^{\omega}_t}$ be the variationally prepared state.
Suppose that 
\begin{align}
    B\left(\proj{\psi^*_t}, \proj{\psi^{\omega}_t}\right) \leq \epsilon_{t},
\end{align} 
then, $|E_t^{\omega} - E^*_t|  \leq \zeta\left(\bm{\omega_t}, \epsilon_t\right)$ for
\begin{eqnarray}
\zeta\left(\bm{\omega_t}, \epsilon_t\right) &=& \epsilon_t^2 \norm{H}_{\infty} \nonumber  \\
&& + 2\!\!\!\!\!\!\max_{\alpha \in [0,\min\{\epsilon_t^2/2,1\}]}\Big\rvert\alpha E_t^{\omega} \!-\!\sqrt{1\!-\!(1\!-\!\alpha)^2} \sqrt{ \Var(H)_{\psi^{\omega}_t}}\Big\rvert. \label{eq_TS_EnergyDiff}
\end{eqnarray}
\end{lemma}
\begin{proof}
\label{app_proof_ED}

Let us denote by $\phi \in [0,2\pi]$ the optimal phase induced by $B\left(\ket{\psi^*_t}, \ket{\psi^{\omega}_t}\right)$ according to Eq.~\eqref{eq_buresPhase} and set $\ket{\psi^{\phi}_t} = \ee^{i \phi}\ket{\psi^{\omega}_t}$.
Next, let $\ket{\mathcal{E}_t}:=\ket{\psi^*_t}-\ket{\psi^\phi_t}$,  which by definition satisfies
\begin{align}
    \norm{\ket{\mathcal{E}_t}}_2 = 
    B\left(\ket{\psi^*_t}, \ket{\psi^{\omega}_t}\right) \leq \epsilon_t.
\end{align}
Then,
\begin{align}
|E_t^{\omega} - E^*_t|
&= |\langle \psi^{\omega}_t | H | \psi^{\omega}_t \rangle - \langle \psi^*_t | H | \psi^*_t \rangle| \nonumber \\
&= |\langle \psi^\phi_t | H | \psi^\phi_t \rangle - \langle \psi^*_t | H | \psi^*_t \rangle| \nonumber \\
&= | \langle \psi^\phi_t | H | \mathcal{E}_t \rangle + \langle \mathcal{E}_t | H | \psi^\phi_t \rangle + \langle \mathcal{E}_t | H | \mathcal{E}_t \rangle |\nonumber \\
&\leq |\langle \mathcal{E}_t | H | \mathcal{E}_t \rangle | + 2 |\langle \mathcal{E}_t | H | \psi^\phi_t \rangle|.
\end{align}
From Cauchy-Schwarz, we find
\begin{align}
  |\langle\mathcal{E}_t | H | \mathcal{E}_t \rangle |
  \leq \norm{\mathcal{E}_t}_2 \norm{H \mathcal{E}_t }_2 
  \leq \epsilon_t^2 \norm{H}_{\infty}.
\end{align}
Furthermore,
\begin{align}
 |\langle \mathcal{E}_t | H | \psi^\phi_t \rangle|
&= |\langle \psi^{\phi}_t | H | \psi^\phi_t \rangle-\langle \psi^*_t | H | \psi^\phi_t \rangle|\nonumber \\
&\leq \left \lbrace
\begin{array}{rl}
\max \limits_{\ket{\psi^*_t}}     & |\langle \psi^\phi_t | H | \psi^\phi_t \rangle-\langle \psi^*_t | H | \psi^\phi_t \rangle| \\
\textnormal{s.t.}     & \langle \psi^*_t|\psi^*_t\rangle =1  \\
                    & |\langle \psi^\phi_t|\psi^*_t\rangle| \geq 1-\epsilon_t^2/2 .
\end{array} \right.
\end{align}
Here, the optimizer $\ket{\psi_t^*}$ can absorb the phase $\ee^{i \phi}$, and thus, the optimization problem is equivalent to
\begin{align}
\left \lbrace
\begin{array}{rl}
    \max \limits_{\ket{\psi^*_t}}     & |\langle \psi^{\omega}_t | H | \psi^{\omega}_t \rangle-\langle \psi^*_t | H | \psi^{\omega}_t \rangle| \\
\textnormal{s.t.}     & \langle \psi^*_t|\psi^*_t\rangle =1  \\
                    & |\langle \psi^{\omega}_t|\psi^*_t\rangle| \geq 1-\epsilon_t^2/2,
\end{array}                    \right .
\end{align}
with the optimizer having the form
\begin{align}
    \ket{\psi^*_t} \!=\! (1\!-\!\alpha) \ket{\psi^{\omega}_t} \!\pm\! \sqrt{1\!-\!(1\!-\!\alpha)^2} \frac{(H-E_t^{\omega})\ket{\psi^{\omega}_t}}{\sqrt{\langle \psi^{\omega}_t|(H-E_t^{\omega})^2\ket{\psi^{\omega}_t}}}
\end{align}
for $\alpha \in [0,1]$.
Inserting this above and noting that $|\langle \psi^{\omega}_t|\psi^*_t\rangle| =1-\alpha \geq 1-\epsilon_t^2/2$ proves Eq.~\eqref{eq_TS_EnergyDiff}.
\end{proof}
In practice, $\norm{H}_{\infty}$ may be replaced with an efficient approximation thereof.
The second lemma bounds the distance between the states resulting when applying the approximate, non-unitary dynamics to both the approximate and the exact state.
To shorten the notation, we use  $\langle H^2 \rangle:=\bra{\psi^{\omega}_t}H^2\ket{\psi^{\omega}_t}$ in the following.

\begin{lemma} \label{lem_secondTerm}
For sufficiently small $\delta_t$ such that terms of order $\mathscr{O}(\delta_t^2)$ are negligible, we can upper bound 
\begin{align} \label{eq_optProblem_EB}
 \left \lbrace \begin{array}{rl}
    \max \limits_{\ket{\psi^*_t}}  & B\big((\mathds{1}\!+\!\delta_t (E_t^{\omega} \!-\! H))\ket{\psi^{\omega}_t},(\mathds{1}\!+\!\delta_t (E_t^{\omega} \!-\! H))\ket{\psi^*_t} \big)  \\
     \textnormal{s.t.} &  \braket{\psi^*_t|\psi^*_t}=1 \,\,\,\, \& \,\,\,\, B(\ket{\psi^{\omega}_t},\ket{\psi^*_t}) \leq \varepsilon_t
 \end{array} \right.
\end{align}
with 
\begin{align}
    \sqrt{2 + 2 \delta_t \zeta\left(\bm{\omega_t}, \epsilon_t\right) - 2 \chi\left(\bm{\omega_t}, \epsilon_t\right)},
\end{align}
where $\zeta\left(\bm{\omega_t}, \epsilon_t\right)$ is defined in~\eqref{eq_TS_EnergyDiff} and
\begin{align} \label{eq_optProblem_EB_simple}
 \chi\left(\bm{\omega_t}, \epsilon_t\right) =\left \lbrace \begin{array}{rl}
    &\min \limits_{\alpha \in [-1,1]} \frac{1}{c_\alpha}\Big|(1+2\delta_t E_t^{\omega})(1-|\alpha|+\alpha E_t^{\omega})   \\
    & \hspace{17mm}- 2 \delta_t\big((1-|\alpha|)E_t^{\omega} + \alpha \langle H^2 \rangle \big)\Big|\\
    & \textnormal{s.t.} \: |1-|\alpha| + \alpha E_t^{\omega}| \geq c_{\alpha} (1-\frac{\varepsilon_t^2}{2}),
 \end{array} \right . 
\end{align}
where 
\begin{align}
    c_{\alpha}=\sqrt{(1-|\alpha|)^2 +2\alpha(1-|\alpha|)E_t^{\omega} + \alpha^2 \langle H^2 \rangle }\, .
\end{align}
\end{lemma}
\begin{proof}
\label{app_proof_second_term}
We start by noting that
\begin{align}
    \norm{(\mathds{1}+\delta_t (E_t^{\omega}-H))\ket{\psi^{\omega}_t}}^2_2 = 1 + \mathscr{O}(\delta_t^2),
\end{align}
and
\begin{align}
    \norm{(\mathds{1}\!+\!\delta_t (E_t^{\omega}\!-\!H))\ket{\psi^*_t}}^2_2 
    &= 1 + 2 \delta_t(E_t^{\omega} - E^*_t) + \mathscr{O}(\delta_t^2)\nonumber \\
    & \leq\! 1 \!+\! 2\delta_t \zeta\!\left(\omega_t, \epsilon_t\right) \!+\! \mathscr{O}(\delta_t^2),
\end{align}
where the final step uses 
Lemma \ref{lem_energyDiff}.
Furthermore, by neglecting terms of order $\mathscr{O}(\delta_t^2)$, we have
\begin{align}
 &\big|\bra{\psi^{\omega}_t}(\mathds{1}\!+\!\delta_t (E_t^{\omega} \!-\! H))(\mathds{1}\!+\!\delta_t (E_t^{\omega} \!-\! H))\ket{\psi^*_t} \big|^2 \nonumber \\
 &\hspace{5mm}=|(1+2\delta_t E_t^{\omega})\braket{\psi^{\omega}_t|\psi^*_t} - 2\delta_t \bra{\psi^{\omega}_t} H \ket{\psi^*_t}|^2\, .
\end{align}

By definition of the Bures metric and again neglecting terms of order $\mathscr{O}(\delta_t^2)$, we find
\begin{align}
    \eqref{eq_optProblem_EB} \leq \sqrt{2 + 2 \delta_t \zeta\left(\bm{\omega_t}, \epsilon_t\right) - 2 \xi}
\end{align}
for 
\begin{align}  \label{eq_problem_simplified}
 \xi = \left \lbrace \begin{array}{rl}
    \min \limits_{\ket{\psi^*_t}}  & |(1+2\delta_t E_t^{\omega})\braket{\psi^{\omega}_t|\psi^*_t} - 2\delta_t \bra{\psi^{\omega}_t} H \ket{\psi^*_t}|  \\
     \textnormal{s.t.} &  \braket{\psi^*_t|\psi^*_t}=1 \,\,\,\, \& \,\,\,\, |\braket{\psi^{\omega}_t|\psi^*_t}| \geq 1-\varepsilon_t^2/2.
 \end{array} \right .
\end{align}
The optimizer for~\eqref{eq_problem_simplified} is of the form
\begin{align} \label{eq_optimizer}
    \ket{\psi_t^*} = \frac{(1-|\alpha|)\ket{\psi^{\omega}_t} + \alpha H\ket{\psi^{\omega}_t}}{c_{\alpha}}, \quad  \alpha \in [-1,1] .
\end{align}
Plugging~\eqref{eq_optimizer} into~\eqref{eq_problem_simplified} shows that $\xi = \chi\left(\bm{\omega_t}, \epsilon_t\right)$ from Eq.~\eqref{eq_optProblem_EB_simple} because
\begin{align}
|\braket{\psi^{\omega}_t|\psi^*_t}|=\frac{1}{c_{\alpha}}(1-|\alpha|+\alpha E_t^{\omega})
\end{align}
and
\begin{align}
|\braket{\psi^{\omega}_t|H|\psi^*_t}|=\frac{1}{c_{\alpha}}\big((1-|\alpha|)E_t^{\omega} + \alpha \langle H^2 \rangle \big).
\end{align}
This then proves the assertion.
\end{proof}
The optimization problems given in Eqs.~\eqref{eq_TS_EnergyDiff} and \eqref{eq_optProblem_EB_simple} can be solved efficiently in practice as they correspond to a one-dimensional search over a closed interval.
Furthermore, we would like to point out that $\chi\left(\bm{\omega_t}, \epsilon_t\right)$ bounds the overlap between the two arguments of the objective function in Eq.~\eqref{eq_optProblem_EB}. 
We are finally ready to state the error bound for VarQITE. 
\begin{theorem} \label{thm_VQITE}
For $T>0$ and $\varepsilon_0=0$, let $\ket{\psi^*_T}$ be the exact solution to Eq.~\eqref{eq:MacLachlan_phase_agnostic_imag} and $\ket{\psi^{\omega}_T}$ be the simulation implemented using VarQITE. Then
\begin{align}
 B\left(\proj{\psi^*_T}, \proj{\psi^{\omega}_T}\right) \leq \epsilon_{T},
\end{align}
for $\epsilon_{T} = \int_{0}^{T}\dot\varepsilon_t \di t$, where
\begin{align}
\label{eq:qite_error}
\varepsilon_{t+\delta_t} &= \delta_t\norm{\ket{e_{t}}}_2 + \delta_t \zeta\left(\bm{\omega_t}, \epsilon_t\right) \nonumber \\
&\hspace{10mm}+\sqrt{2 + 2 \delta_t \zeta\left(\bm{\omega_t}, \epsilon_t\right) - 2\chi\left(\bm{\omega_t}, \epsilon_t\right)},
\end{align}
with $\zeta\left(\bm{\omega_t}, \epsilon_t\right)$ and $\chi\left(\bm{\omega_t}, \epsilon_t\right)$ as given in Eq.~\eqref{eq_TS_EnergyDiff} and Eq.~\eqref{eq_optProblem_EB_simple}, respectively,
allows to define
\begin{align}
\label{eq:qite_error_grad}
    \dot\varepsilon_t &= \lim_{\delta_t\rightarrow 0} \frac{\varepsilon_{t+\delta_t} - \varepsilon_{t}}{\delta_t}\nonumber \\
    &=\norm{\ket{e_{t}}}_2 + \zeta\left(\bm{\omega_t}, \epsilon_t\right) \nonumber \\
    &\hspace{5mm}+\lim_{\delta_t\rightarrow 0} \frac{\sqrt{2 + 2 \delta_t \zeta\left(\bm{\omega_t}, \epsilon_t\right) - 2\chi\left(\bm{\omega_t}, \epsilon_t\right)} -  \varepsilon_{t}}{\delta_t}.
\end{align}
\end{theorem}
\begin{proof}\label{app:proofThm2}
For $\delta_t > 0$, Eq.~\eqref{eq:MacLachlan_phase_agnostic_imag} gives
    \begin{align} 
  \ket{\psi^{\omega}_{t+\delta_t}  }&= \ket{\psi^{\omega}_t  } + \delta_t\ket{\dot\psi^{\dot\nu}_t} \nonumber   \\
  &=\ket{\psi^{\omega}_t  } + \delta_t\left(\ket{\dot\psi^{\omega}_t} + i\text{Im}(\braket{\dot\psi^{\omega}_t|\psi^{\omega}_t})\ket{\psi^{\omega}_t}\right). 
    \end{align}
Combining this with the triangle inequality gives 
\pagebreak
\begin{align}
    &B\left(\ket{\psi^{\omega}_{t+\delta_t}}, \ket{\psi^*_{t+\delta_t}}\right) \nonumber \\ 
      &\hspace{0mm}\leq B\Big(\!\ket{\psi^{\omega}_t  } \!+\! \delta_t\ket{\dot\psi^{\dot\nu}_{t}}, \big(\mathds{1}\!+\delta_t\left(E_t^{\omega}-H\right)\big)\ket{\psi^{\omega}_t}\! \Big)\nonumber  \\
      &\hspace{2mm}+B\Big(\!\big(\mathds{1}\!+\delta_t\left(E_t^{\omega}-H\right)\big)\ket{\psi^{\omega}_t}, \big(\mathds{1}\!+\delta_t\left(E_t^{\omega}-H\right)\big)\ket{\psi^*_{t}}\!\Big) \nonumber \\
      &\hspace{2mm}+B\Big(\big(\mathds{1}\!+\delta_t\left(E_t^{\omega}-H\right)\big)\ket{\psi^*_{t}}, \ket{\psi^*_{t+\delta_t}}\Big). \label{eq_triangle_start_qite}
\end{align}
We next bound all three terms separately.
Using Eq.~\eqref{eq_buresPhase} and neglecting terms of order $\mathscr{O}(\delta_t^2)$ gives
\begin{align}
     &B\Big(\ket{\psi^{\omega}_t  } \!+\! \delta_t\ket{\dot\psi^{\dot\nu}_{t}}, \big(\mathds{1}\!+\delta_t\left(E_t^{\omega}-H\right)\big)\ket{\psi^{\omega}_t} \Big) \nonumber \\
     &=\min_{\phi \in [0,2\pi]}\norm{\ee^{i\phi}(\ket{\psi^{\omega}_t  } \!+\! \delta_t\ket{\dot\psi^{\dot\nu}_{t}}) - \big(\mathds{1}\!+\delta_t\left(E_t^{\omega}-H\right)\big)\ket{\psi^{\omega}_t}}_2 \nonumber\\
     &\leq\norm{\ket{\psi^{\omega}_t} + \delta_t\ket{\dot\psi^{\dot\nu}_{t}}  - \left(\mathds{1}+\delta_t\left(E_t^{\omega}-H\right)\right)\ket{\psi^{\omega}_t}}_2  \nonumber\\
     &= \delta_t \norm{\ket{\dot\psi^{\dot\nu}_{t}}  - \left(E_t^{\omega}-H\right)\ket{\psi^{\omega}_t}}_2  \nonumber\\
     &=\delta_t \norm{\ket{\dot\psi^{\omega}_t}  - \left(E_t^{\omega}-H-i\text{Im}\left(\braket{\dot\psi^{\omega}_t|\psi^{\omega}_t}\right)\right)\ket{\psi^{\omega}_t}
    }_2 \nonumber \\
     &= \delta_t\norm{\ket{e_t}}_2 , \label{eq_triangle_part1_qite}
\end{align}
where the penultimate step uses Eq.~\eqref{eq_forDavid}.
The second term in Eq.~\eqref{eq_triangle_start_qite} is bounded from above by
Lemma \ref{lem_secondTerm}.
It, thus, remains to bound the third term in Eq.~\eqref{eq_triangle_start_qite}.
With the help of Eq.~\eqref{eq_buresPhase} we find
\begin{align}
     &B\Big(\big(\mathds{1}\!+\delta_t\left(E_t^{\omega}-H\right)\big)\ket{\psi^*_{t}}, \ket{\psi^*_{t+\delta_t}}\Big) \nonumber \\
     &= \min_{\phi \in [0,2\pi]}\norm{\ee^{i\phi}\big(\big(\mathds{1}\!+\delta_t\left(E_t^{\omega}-H\right)\big)\ket{\psi^*_{t}} \big)-\ket{\psi^*_{t+\delta_t}}}_2 \nonumber \\
     &\leq \|\big(\mathds{1}\!+\delta_t\left(E_t^{\omega}-H\right)\big)\ket{\psi^*_{t}}-\big(\mathds{1}\!+\delta_t\left(E^*_t-H\right)\big)\ket{\psi^*_{t}}\|_2\nonumber \\
     &=\delta_t\|\big(E_t^{\omega} - E^*_t \big) \ket{\psi^*_{t}}\|_2 \nonumber \\
     &\leq \delta_t |E_t^{\omega} - E^*_t |\|\ket{\psi^*_{t}}\|_2  \nonumber \\
     &=\delta_t |E_t^{\omega} - E^*_t | \nonumber \\
     &\leq \delta_t \zeta,
\end{align}
where $\zeta\left(\bm{\omega_t}, \epsilon_t\right)$ is defined in Eq.~\eqref{eq_TS_EnergyDiff}.
\end{proof}
This error bound is not only independent of a potential physically irrelevant global phase mismatch between prepared and target state but also compatible with  implementations which use numerical techniques to solve the SLE given in Eq.~\eqref{eq:McLachlanVarQITE}.
It should be noted that the rather complex form of the VarQITE error bound is due to the non-unitary nature of imaginary time evolution.
Lastly, we would like to point out that a similar -- in fact even easier -- bound can be derived for variational quantum real time evolution. The details are presented in Appendix \ref{app:varqrte}.
\subsection{Ground State Preparation}
\label{app:ground_state_runtime}

Searching for the ground state $\ket{\psi^{(0)}}$ of a given Hamiltonian $H$ is one possible application of VarQITE.
Suppose the initial state $\ket{\psi_0}$ has a non-zero overlap with the ground state of $H$. Then, the propagation under imaginary time evolution for $t\rightarrow \infty$ will lead to an exponential damping of all components in $\ket{\psi_0}$ which do not correspond to $\ket{\psi^{(0)}}$.
If we can approximate the quantum imaginary time evolution process sufficiently well, we can exploit VarQITE to search for the ground state.
In fact, it is shown in \cite{VarSITEMcArdle19} that also with the variational implementation the average energy monotonically decreases for sufficiently small time steps if $ \mathcal{F}^Q$ is invertible.
Notably, a ground state preparation with VarQITE corresponds to a special case of QNG optimization, see Sec.~\ref{sec:qng}, where
\begin{align}
   L\left(\ket{\psi\left(\bm{\omega}\right)}\right) = \frac{1}{2} \braket{\psi(\bm{\omega})|\hat{O}|\psi(\bm{\omega})}.
\end{align}

In practice, we cannot propagate the system for $t \rightarrow \infty$. Thus, we want to understand what finite evolution time is sufficient to reach the ground state. 
Suppose the initial state for $t=0$ has a non-zero overlap with the ground state and that we evolve the state according to QITE.
Then, the required evolution time for the ground state to dominate scales as $t = \mathcal{O}\left(n/G\right)$, where $G$ is the spectral gap of $H$ and $n$ denotes the number of qubits, i.e., linearly in the number of qubits, which is formally proven in the following proposition. 


\begin{proposition}
\label{proposition:ground_stateruntime} \footnote{This proposition is reproduced, with permission, from \cite{NeuweilerCOVarQITE}.}
Suppose a Hamiltonian $H$ on $n$ qubits with a non-degenerate ground state such that the first and second eigenvalues are given by $E_0$ and $E_1$. Furthermore, denote the respective spectral gap by $G = E_1 - E_0 > 0$, and the ground state by $\ket{\psi^{(0)}}$, respectively. 
Now, assume a state $\ket{\psi_t}$ where the time dependence is described by quantum imaginary time evolution and where the initial state $\ket{\psi_0}$ has an overlap with the ground state that may be at least exponentially small, i.e., $\rvert\braket{\psi^{(0)}|\psi_0}\rvert^2 \geq c 2^{-n}$, for some constant $c > 0$.
Then, the evolution time necessary for the probability of sampling the ground state, i.e., $\rvert\braket{\psi^{(0)}|\psi_t}\rvert^2$, to reach a target value $p^{(0)}$, can be upper bounded by $\frac{\log{(2)}n - \log\left(1-p^{(0)}\right) - \log(c) + \log\left(p^{(0)}\right)}{2G}$.
\end{proposition}
\begin{proof}
    Given the Hamiltonian $H$ with eigenstates $\ket{\psi^{(j)}}$ and corresponding eigenvalues $E_j$, for $j = 0, \ldots, 2^n-1$, we can write the initial state as
    \begin{align}
         \ket{\psi_0} = \sum_{j} \alpha_j \ket{\psi^{(j)}}, 
    \end{align}
    where $\alpha_j = \braket{\psi^{(j)}|\psi_0}$ and $\sum_j \|\alpha_j\|_2 = 1$. This directly implies that
    \begin{align}
         \bra{\psi_0} \ee^{-2Ht}\ket{\psi_0} = \sum\limits_j \|\alpha_j\|_2^2\ee^{-2E_jt},
    \end{align}
    which leads to
        \begin{align}
        \ket{\psi_t} = \frac{\ee^{-Ht}\ket{\psi_0}}{\sqrt{\bra{\psi_0} \ee^{-2Ht}\ket{\psi_0}}} = \frac{\sum_j \alpha_j \ee^{-E_jt}\ket{\psi^{(j)}}}{\sqrt{\sum_j \|\alpha_j\|_2^2\ee^{-2E_jt}}}.
    \end{align}
    Now, the probability to sample the ground state at time $t$ is given by
    \begin{align}
         \rvert\braket{\psi^{(0)}|\psi_t}\rvert^2 &=  \frac{\|\alpha_0\|_2^2\ee^{-2E_0t}}{\sum_j \|\alpha_j\|_2^2e^{-2E_jt}}  \\
         &= \frac{\|\alpha_0\|_2^2\ee^{-2E_0t}}{\|\alpha_0\|_2^2\ee^{-2E_0t} + \sum_{j>0} \|\alpha_j\|_2^2 \ee^{-2E_jt}} \\
          &= \frac{\|\alpha_0\|_2^2}{\|\alpha_0\|_2^2 + \sum_{j>0} \|\alpha_j\|_2^2\ee^{-2(E_j-E_0)t}}  \\
            &\geq \frac{\|\alpha_0\|_2^2}{\|\alpha_0\|_2^2 + \ee^{-2Gt} \sum_{j>0} \|\alpha_j\|_2^2}  \\
            &\geq \frac{\|\alpha_0\|_2^2}{\|\alpha_0\|_2^2 + \ee^{-2Gt}} \\
            &= \frac{\rvert\braket{\psi^{(0)}|\psi_0}\rvert^2}{\rvert\braket{\psi^{(0)}|\psi_0}\rvert^2 + \ee^{-2Gt}}.
    \end{align}
This can be reformulated to
\begin{align}
     \ee^{-2Gt} \geq \frac{\rvert\braket{\psi^{(0)}|\psi_0}\rvert^2}{\rvert\braket{\psi^{(0)}|\psi_t}\rvert^2} - \rvert\braket{\psi^{(0)}|\psi_0}\rvert^2,
\end{align}
leading to 
\begin{eqnarray}
     {-2Gt} 
     &\geq& 
     \log{\left(
        \rvert\braket{\psi^{(0)}|\psi_0}\rvert^2
        \left(
        \frac{1}{\rvert\braket{\psi^{(0)}|\psi_t}\rvert^2}-1
        \right)
     \right)} \\
     &=& \log{\left(c 2^{-n}\right)} + \log \left(\frac{1}{p^{(0)}}-1\right) \\
     &=& \log(c) - \log{(2)}n + \log\left(1-p^{(0)}\right) - \log\left(p^{(0)}\right),
\end{eqnarray}
where it is used that $\rvert\braket{\psi^{(0)}|\psi_t}\rvert^2 \geq c 2^{-n}$ and that we are looking for $t$ such that $\rvert\braket{\psi^{(0)}|\psi_t}\rvert^2 = p^{(0)}$.
Now, this results in
\begin{align}
     t \leq \frac{\log{(2)}n - \log\left(1-p^{(0)}\right) - \log(c) + \log\left(p^{(0)}\right)}{2G}.
\end{align}
\end{proof}
Note that for $c=1$, as achieved by the equal superposition state, and for a target of $p^{(0)} = 1/2$, the bound simplifies to $log(2)n/2G$.

Furthermore, we can investigate the error bound and the variance to find out whether the ground state has been reached. More explicitly, if $t$ is sufficiently large, $\epsilon_t \leq 2\sin(\pi/8)$ such that the prepared state is reasonably accurate and $ \Var(H)_{\psi_t}\approx 0$, then the prepared state can be assumed to be close to the ground state because $ \Var(H)_{\psi_t}=0$ for all eigenstates of $H$.
Sec.~\ref{sec:varqite_examples} includes an illustrative VarQITE example for ground state preparation and presents the respective error bounds and variance.

\subsection{Gibbs State Preparation}
\label{sec:varqite_gibbs}

We, now, introduce a Gibbs state approximation technique that is based on VarQITE. The method is applicable to generic Hamiltonians, compatible with near-term quantum computers and also suitable for states with long-range interactions. 
Suppose the Gibbs state, see Def.~\ref{def:gibbs_state},
\begin{align}
        \rho = \frac{e^{-H/\left(\text{k}_{\text{B}}\text{T}\right)}}{Z},
\end{align}
for an $n$-qubit Hamiltonian with the Boltzmann constant $\text{k}_{\text{B}}$, the system temperature $\text{T}$ and the partition function
\begin{equation}
    Z=\text{Tr}\left[e^{-H/\left(\text{k}_{\text{B}}\text{T}\right)}\right].
\end{equation}
The first step is the preparation of a purification of the maximally mixed state. To that end, we choose an ansatz acting on $2n$ qubits given in the form of a variational quantum circuit $U\left(\bm{\omega}_t\right)$ for $\bm{\omega}_t \in \mathbb{R}^{k}$ with two $n$-qubit registers $a$ and $b$ that prepares
\begin{equation}
    \ket{\psi\left(\bm{\omega}_t\right)} = U\left(\bm{\omega}_t\right)\ket{0}^{\otimes 2n}.
\end{equation}
Now, the initial parameters $\bm{\omega}_0$ are fixed such that the initial state is
\begin{equation}
\label{eq:state_init}
     \ket{\psi\left(\bm{\omega}_{0}\right)} = \ket{\phi^{+}}^{\otimes n},
\end{equation}
where $\ket{\phi^{+}} = \textstyle{\frac{\ket{00}+\ket{11}}{\sqrt{2}}}^{\otimes n}$ corresponds to a Bell state and the first respectively second qubit of each Bell state are in register $a$ and $b$, respectively. This already corresponds to the desired purification. As we can see, tracing out sub-system $b$ results in an $n$-dimensional maximally mixed state
\begin{equation}
		  \text{Tr}_{b}\left[\ket{\phi^{+}}^{\otimes n} \right] = \frac{\mathds{1}}{2^n}.
\end{equation}
Secondly, we define an effective Hamiltonian on $2n$ qubits, i.e., $H^{ab} = H^a\otimes \mathds{1}^b,$
where $H$ and $\mathds{1}$ act on system $a$ and $b$, respectively.
\begin{algorithm}[t!]
 \caption{VarQITE for Gibbs State Preparation} \label{algo:VarQITE_Gibbs}
\begin{algorithmic}[0]
\\
    \State  \textbf{input}
    \State $H^{ab} = H^a \otimes I^b$
	\State $\ket{\psi\left( \bm{\omega}_0\right)} = U\big( \bm{\omega}_0\big)\ket{0}^{\otimes 2n} = \ket{\phi^{+}}^{\otimes 2n}$
	\State$\tau = \textstyle{\frac{1}{2\text{k}_{\text{B}}\text{T}}}$
	\\
    \State \textbf{procedure}
	\For{$t\in\left[0,  \ldots, \tau\right)$}
		\State Use ODE solver to propagate for time step $\delta t$: $\bm{\omega}_{t+\delta t}$
	\EndFor
	\\
	\State \textbf{return} 	$\rho\left(\bm{\omega}\right)=\text{Tr}_{b}\Big[\proj{ \psi\big({\bm{\omega}_{\tau}}\big)} \Big]$
\end{algorithmic}
\end{algorithm}
Now, we are ready to generate a purification of the desired Gibbs state by using VarQITE and an arbitrary ODE solver to propagate the trial state $\ket{\psi\left(\bm{\omega}_t\right)}$ with respect to $H^{ab}$ for time $t=\textstyle{\frac{1}{2\text{k}_{\text{B}}\text{T}}}$. 
Notably, the action on subsystem $a$ reads
\begin{align}
    \frac{\ee^{-H^at}\mathds{1}\ee^{-H^at}}{\Tr\left[\ee^{-2H^at}  \right]}= \frac{\ee^{-2H^at}}{\Tr\left[\ee^{-2H^at}  \right]}.
\end{align}
Thus, the Gibbs state approximation is given by 
\begin{equation}	
	\rho\left(\bm{\omega}\right)=\text{Tr}_{b}\Big[\proj{ \psi\big({\bm{\omega}_{\frac{1}{2\text{k}_{\text{B}}\text{T}}}}\big)} \Big] \approx  \rho.
\end{equation}
An outline of the Gibbs state preparation with VarQITE is presented in Alg.~\ref{algo:VarQITE_Gibbs} and application examples are given in Sec.~\ref{sec:varqite_examples}.

Notably, the register $b$ can be fully omitted if $H$ is diagonal \cite{shingu2020boltzmann}. More specifically, we can use $n$ instead of $2n$ qubits to generate a state whose diagonal approximates the desired Gibbs state.
The respective algorithm is described in Alg.~\ref{algo:VarQITE_Gibbs_diag}.

\begin{algorithm}[!htb]
 \caption{VarQITE for diagonal Gibbs State Preparation} \label{algo:VarQITE_Gibbs_diag}
\begin{algorithmic}[0]
\\
       \State  \textbf{input}
    \State $H^a$
	\State $\ket{\psi\left( \bm{\omega}_0\right)} = U\big( \bm{\omega}_0\big)\ket{0}^{\otimes n} = \frac{\ket{0}+\ket{1}}{\sqrt{2}}^{\otimes n}$
	\State$\tau = \textstyle{\frac{1}{2\text{k}_{\text{B}}\text{T}}}$
	\\
    \State \textbf{procedure}
	\For{$t\in\left[0,  \ldots, \tau\right]$}
		\State Use ODE solver to propagate for time step $\delta t$: $\bm{\omega}_{t+\delta t}$
	\EndFor
	\\
	\State \textbf{return} 	$\rho\left(\bm{\omega}\right)=\text{diag}\big(\proj{ \psi\big({\bm{\omega}_{\tau}}\big)}\big)$
\end{algorithmic}
\end{algorithm}

\subsection{Automatic Differentiation}
\label{sec:varqite_chainRule}

Several generative QML algorithms -- such as quantum Boltzmann machines which are presented in Sec.~\ref{sec:QBM} -- train the parameters $\bm{\theta}\in\mathbb{R}^p$ of a Hamiltonian of the form
\begin{equation}
    H\left({\bm{\theta}}\right)=\sum_{i=0}^{p-1}\theta_ih_i,
\end{equation}
with $h_i=\bigotimes_{j=0}^{n-1}\sigma_i^j$ for $\sigma_i^j\in\set{I, X, Y, Z}$ acting on the $j^{\text{th}}$ qubit, such that the statistics generated by the respective Gibbs state $\rho(\bm{\theta})$ approximate a target distribution.
Suppose training data samples $\set{x_0, \ldots, x_{2^n-1}}$ which are distributed according to the probability distribution $p^{\text{data}}$ and a Gibbs state approximation prepared with VarQITE $\rho(\bm{\omega}\left(\bm{\theta}\right))$ whose sampling statistics are given by 
\begin{equation}
    p_x\left(\bm{\omega}\left(\bm{\theta}\right)\right) = \text{Tr}\textstyle{\left[\proj{x}\rho(\bm{\omega}\left(\bm{\theta}\right)) \right]},
\end{equation}
where we assume an affine mapping from $\{x_0, \ldots, x_{2^n-1}\}$ to $\ket{x}$ for $x\in{{0}, \ldots, {2^n-1}}$.
To achieve the training goal, one can use, e.g., the cross entropy as loss function
\begin{equation}
	L\left(\bm{\omega}\left(\bm{\theta}\right)\right) = -\sum\limits_x p_x\left(\bm{\omega}\left(\bm{\theta}\right)\right)\log\left(p_x^{\text{data}}\right),
\end{equation}
 with
\begin{equation}
\label{eq:auto_diff_loss}
	\nabla_{\bm{\theta}}L\left(\bm{\omega}\left(\bm{\theta}\right)\right) = - \sum\limits_{x}\nabla_{\bm{\theta}}p_x\left(\bm{\omega}\left(\bm{\theta}\right)\right)\log\left(p_x^{\text{data}}\right).
\end{equation}
The critical step in the calculation of the above gradient is the 
evaluation of $\nabla_{\bm{\theta}}p_x\left(\bm{\omega}\left(\bm{\theta}\right)\right)$ which can be achieved using VarQITE-based automatic differentiation with
\begin{align}
    \nabla_{\bm{\theta}}p_x\left(\bm{\omega}\left(\bm{\theta}\right)\right) =  \nabla_{\bm{\omega}}p_x\left(\bm{\omega}\left(\bm{\theta}\right)\right) \nabla_{\bm{\theta}}\bm{\omega}.
\end{align}
Firstly, the probability quantum gradient $\nabla_{\bm{\omega}}p_x $ may be calculated with the quantum gradient methods introduced in Sec.~\ref{sec:analytic_gradients}.
Secondly, the calculation of $\nabla_{\bm{\theta}}\bm{\omega}$ can be evaluated with the following SLE \footnote{The SLE stems from applying the gradient $\nabla_{\bm{\theta}}$ to both sides of Eq.~\eqref{eq:standardODE_varqite}.}
\begin{align}
\label{eq:sle_chain_rule}
\sum\limits_{j=0}^k \left(\nabla_{\bm{\theta}}\mathcal{F}^Q_{ij} \dot\omega_j + \mathcal{F}^Q_{ij}\nabla_{\bm{\theta}} \dot\omega_j \right)= - \nabla_{\bm{\theta}}\text{Re}\left(C_i\right),
\end{align}
where
\begin{align}
    \frac{\partial\mathcal{F}^Q_{ij}}{\partial\theta_m} &= \sum\limits_{l=0}^k\frac{\partial\omega_l}{\partial\theta_m}  \frac{\partial\mathcal{F}^Q_{ij}}{\partial\omega_l} = \frac{1}{2}\frac{\partial^3 |\braket{\psi(\bm{\omega'})|\psi(\bm{\omega})}|^2}{\partial\omega_i\partial\omega_j\partial\omega_l}
    \Big\rvert_{\bm{\omega'}=\bm{\omega}},
\end{align}
and
\begin{align}
    \frac{\partial\text{Re}\left(C_i\right)}{\partial\theta_m} &= \sum\limits_{l=0}^k\frac{\partial\omega_l}{\partial\theta_m}  \frac{\partial\text{Re}\left(C_i\right)}{\partial\omega_l}=\frac{1}{2}\sum\limits_{l=0}^k\frac{\partial\omega_l}{\partial\theta_m}\frac{\partial^2 \bra{\psi^{\omega}_t}H\ket{\psi^{\omega}_t}}{\partial \omega_i\partial \omega_l},
\end{align}
can be computed with the quantum gradient methods introduced in Sec.~\ref{sec:analytic_gradients}.
Solving the SLE given in Eq.~\eqref{eq:sle_chain_rule} w.r.t.~$\nabla_{\bm{\theta}}\bm{\dot\omega}$ defines the following ODE
\begin{align}
\label{eq:ode_chain}
     f_{\text{chain}}\left(\bm{\omega}, \bm{\theta}\right) = -{\mathcal{F}^Q}^{-1} \Big(\nabla_{\bm{\theta}}\mathcal{F}^Q \bm{\dot\omega} +\nabla_{\bm{\theta}}\text{Re}\left( \bm{C}\right)\Big).
\end{align}
Now, using numerical integration to solve this ODE -- where by definition all items in $ \nabla_{\bm{\theta}}\bm{\dot\omega_0}$ are $0$ -- gives an estimate for $\nabla_{\bm{\theta}}\bm{\omega}$.

Notably, the gradient of the loss function  $ \nabla_{\bm{\theta}}L\left(\bm{\theta}\right)$ could also be approximated with a finite difference method \cite{Kardestuncer1975}. If the number of Hamiltonian parameters $p$ is smaller than the number of trial state parameters $k$ this requires less evaluation circuits. However, given a trial state that has less parameters than the respective Hamiltonian, the automatic differentiation scheme presented in this section is favorable in terms of the number of evaluation circuits.
Suppose that the Gibbs state preparation with VarQITE employs $t$ time steps, then the number of circuits that need to be evaluated for Gibbs state preparation with VarQITE are $\mathscr{O}\left(tk(k+p)\right)$.
Now, computing the gradient with forward finite differences reads
\begin{equation*}
    \frac{\partial L\left(\bm{\theta}\right)}{\partial\theta_i} \approx  \frac{L\left(\bm{\theta} + {\epsilon}\bm{e_i} \right) - L\left( \bm{\theta} \right)}{\epsilon},
\end{equation*}
for $i\in\set{0, \ldots, p-1}$, $\bm{e_i}$ representing the unit vector as given before. 
For this purpose, VarQITE must be run once with $\bm{\theta}$ and $p$ times with an $\epsilon$-shift which leads to a total number of $\mathscr{O}\left(tpk(k+p)\right)$ circuits.
For the automatic differentiation gradient, given in Eq.~\eqref{eq:auto_diff_loss},
VarQITE needs to be run once to prepare $\rho\left(\bm{\omega}\left(\bm{\theta}\right)\right) \approx \rho\left(\bm{\theta}\right)$. Furthermore, the evaluation of $\nabla_{\bm{\theta}}\bm{\omega}$ requires $\mathscr{O}\left(tk^2(k+p)\right)$ circuits.
The resulting overall complexity of the number of circuits is $\mathscr{O}\left(tk^2(k+p)\right)$.
These results are summarized in Tbl.~\ref{tbl:complexity}.

\begin{table}[h!]
\begin{center}
\captionsetup{singlelinecheck = false, format= hang, justification=raggedright, font=footnotesize, labelsep=space}
{\renewcommand{\arraystretch}{1.2}
\begin{tabular}{ c | c }
\textbf{Method} & \textbf{Number Circuits} \\ 
\hline
finite diff. &  $\mathscr{O}\left(tkp(k+p)\right)$\\
automatic diff.   & $\mathscr{O}\left(tk^2(k+p)\right)$\\
\end{tabular}
}
\end{center}
\caption{Comparing the number of circuits needed to compute the loss function gradient $\nabla_{\theta}L$ using either finite differences or automatic differentiation. The number of Hamiltonian parameters is $p$, the number of ansatz parameters is $k$ and the number of time steps during the Gibbs state preparation is $t$.}
\label{tbl:complexity}
\end{table}

Automatic differentiation is more efficient than finite differences if $k < p$. For $k > p$, on the other hand, the computational complexity of finite differences is better. 
Considering, e.g., a $k$-local Ising model that corresponds to a Hamiltonian with $\mathscr{O}\left(n^k\right)$ parameters. Suppose that we can find a reasonable variational $n$-qubit trial state with $\mathscr{O}\left(n\right)$ layers of parameterized and entangling gates, which results in $k = \mathscr{O}\left(n^2\right)$ parameters. Then, automatic differentiation would outperform finite differences for $k>2$. Furthermore, as discussed in Sec.~\ref{sec:gradients}, finite differences may result in numerical instabilities. Automatic differentiation might, thus, be generally preferable and achieve higher quality results when $k > p$.

\subsection{Implementation}
\label{sec:methods}
To ensure a stable VarQITE implementation, it is vital to choose the correct settings. The possible choices with their advantages and disadvantages are explained next.

\subsubsection{ODE Solvers}
\label{sec:odesolvers}

The ODE underlying VarQITE is solved using numerical integration. This can lead to an additional error term. 
Let $\ket{\psi_t^*}$ denote the target state, $\ket{\psi^{\omega}_t}$ the prepared state, and $\ket{\psi_t'}$ the state that we would prepare if we could take infinitesimally small time steps and, thus, integrate the ODE exactly.
Then, the error bound derived in Sec.~\ref{sec:error_qite} captures the error induced by the variational method, i.e.,
\begin{align}
	B\left(\proj{\psi_t^*}, \proj{\psi_t'}\right) \leq \epsilon_t.
\end{align}
The triangle inequality gives
\begin{align}
	B\left(\proj{\psi^*_t}, \proj{\psi^{\omega}_t}\right) &\leq B\left(\proj{\psi^*_t}, \proj{\psi'_t}\right)\nonumber \\
	&+B\left(\proj{\psi'_t}, \proj{\psi^{\omega}_t}\right).
\end{align}
The term $B(\proj{\psi'_t}, \proj{\psi^{\omega}_t})$ is generally unknown and the error bound from Sec.~\ref{sec:error_qite} only holds if $B(\proj{\psi'_t}, \proj{\psi^{\omega}_t}) \ll 1$ such that 
\begin{align}
	B\left(\proj{\psi^*_t}, \proj{\psi^{\omega}_t}\right)  
	\!\approx\! B\left(\proj{\psi^*_t}, \proj{\psi'_t}\right) 
	\!\leq\! \epsilon_t.
\end{align}

ODE solvers, such as the \emph{forward Euler} \cite{Griffiths2010ODE} method, which operate with a fixed step size may induce large errors in the numerical simulations if the time steps are not chosen sufficiently small.
The forward Euler method evaluates the gradient $\bm{\dot\omega_t}$ and propagates the underlying variable for $n_T$ time steps according to a predefined step size, i.e.,
\begin{align}
   \bm{\omega_{T}} = \bm{\omega_0} + \sum\limits_{k=0}^{n_T}  \delta_t \bm{\dot{\omega}_{t_k}},
\end{align}
with $t_{n_T}=T$ and the step-size $\delta_t=t_{k+1}-t_k$.
In contrast, \emph{Runge-Kutta} \cite{Runge1895_RK, kutta1901beitrag} methods evaluate additional supporting points and compute a parameter update using an average of these points, thereby, truncating the local update error. 
Combining two Runge-Kutta methods of different order but using the same supporting points allows to define efficient adaptive step-size ODE solvers which ensure that the local step-by-step error is small and, thus, that the dominant part of the error is coming from the variational approximation.
The results in Sec.~\ref{sec:varqite_examples} illustrate this aspect on the example of the forward Euler method with fixed step  size and an explicit Runge-Kutta method of order 5(4) (RK54) method from SciPy \cite{2020SciPy-NMeth} that uses additional interpolation points as well as an adaptive step size to minimize the step-by-step integration errors.
We refer the interested reader to an introductory book on numerical ODE solvers such as \cite{Griffiths2010ODE} for further information.

\subsubsection{ODE Definition}
The SLE underlying McLachlan's variational principle, given in 
Eq.~\eqref{eq:McLachlanVarQITE}, is prone to being ill-conditioned and may, thus, only be solvable approximately with a numerical technique such as regularized least squares or pseudo-inversion. The commonly used regularization schemes as well as the pseudo-inversion can be seen as small perturbations which are not necessarily in accordance with the physics of the system. This, in turn, may lead to inappropriate parameter updates. In the following, we shall refer to the ODE definition based on Eq.~\eqref{eq:standardODE_varqite}, i.e., $f_{\text{std}}$ as \emph{standard ODE}.
The alternative ODE definition  given in Eq.~\eqref{eq:argminODE_varqite}, i.e., $f_{\text{min}}$ which shall be referred to as \emph{argmin ODE}, is analytically equivalent to solving $f_{\text{std}}$ with a least square solver. However, the simulation results in Sec.~\ref{sec:varqite_examples} show that the numerical behavior differs.
In fact, the experiments reveal that the argmin ODE can lead to significantly better numerical stability. 
The simulations employ the SciPy COBYLA optimizer \cite{2020SciPy-NMeth} to find $\bm{\dot\omega}$ in $f_{\text{min}}$ where the initial point is chosen
as the numerical solution to the SLE given in 
Eq.~\eqref{eq:McLachlanVarQITE}.

\subsubsection{Error Bound Evaluation}
\label{sec:ode_form}
To enable a reliable error bound evaluation, we jointly evolve the state parameters and the error bounds. 
More explicitly, we extend the parameter  ODE to
\begin{align}
        \begin{pmatrix}
       \boldsymbol{\dot \omega}_t \\
    \dot \epsilon_t
    \end{pmatrix}= 
        \tilde f\left(\bm{\omega_t}, \epsilon_t\right),
\end{align}
with  $\bm{\omega_0}$ being set, $\epsilon_0 = 0$ by assumption, and 
\begin{align}
    \tilde f\left(\bm{\omega_t}, \epsilon_t\right)=\begin{pmatrix}
        f\left(\bm{\omega_t}\right) \\
        \dot\varepsilon_t
    \end{pmatrix},
\end{align}
with $\dot\varepsilon_t$ from 
Eq.~\eqref{eq:qite_error_grad}.
Furthermore, $ f\left(\bm{\omega_t}\right)$ is either chosen as $f_{\text{std}}\left(\bm{\omega_t}\right)$ or $f_{\text{min}}\left(\bm{\omega_t}\right)$.
This formulation has the advantage that the error bound directly reflects the propagation of the evolution and that adaptive step size ODE solvers also consider the changes in the error bounds.

\subsubsection{Chain Rule Evaluation}
\label{sec:ode_form_chain}
If a QML algorithm requires the chain rule evaluation described in Sec.~\ref{sec:varqite_chainRule}, we can extend the parameter ODE in a similar fashion to the error bound evaluation.
The system parameters are evolved according to the following  ODE
\begin{align}
        \begin{pmatrix}
       \boldsymbol{\dot \omega}_t \\
    \nabla_{\bm{\theta}}\boldsymbol{\dot \omega}_t
    \end{pmatrix}= 
        f'\left(\bm{\omega_t}, \bm{\theta}\right),
\end{align}
with  $\bm{\omega_0}$ and $\bm{\theta}$ as defined, all elements in $ \nabla_{\bm{\theta}}\bm{\dot\omega_0}$ being $0$, and 
\begin{align}
     f'\left(\bm{\omega_t}, \bm{\theta}\right)=\begin{pmatrix}
        f_{\text{std}}\left(\bm{\omega_t}\right) \\
        f_{\text{chain}}\left(\bm{\omega_t}, \bm{\theta}\right)
    \end{pmatrix},
\end{align}
with $ f_{\text{chain}}\left(\bm{\omega_t}, \bm{\theta}\right)$ from 
Eq.~\eqref{eq:ode_chain}.

\subsection{Illustrative Examples}
\label{sec:varqite_examples}

The following section presents a set of illustrative examples of VarQITE applications.
First, we introduce three examples that we will use to demonstrate the efficiency of VarQITE and the error bounds derived in Sec.~\ref{sec:error_qite} as well as the impact of the implementation details discussed in Sec.~\ref{sec:methods}. 
All experiments prepare $\ket{\psi^{\omega}_t}$ with an ansatz as shown in Fig.~\ref{fig:effsu2}, adjusted to the number of qubits $n$ given by the respective Hamiltonian:
\begin{enumerate}[(i)]
\label{enumerate:Examples}
    \item An illustrative example is considered with
\begin{align}
    H_{\text{illustrative}} = Z\otimes X + X\otimes Z + 3 \,Z\otimes Z.
\end{align}
Hereby, the evolution time is $T=1$ and the initial parameters are chosen such that all parameters are set to $0$ except for the parameters of the last layer of Pauli $Y$ rotations $\left(RY\right)$ \cite{nielsen10} gates which are chosen to be $\pi/2$. This gives $\ket{\psi_0}=\ket{++}$.
    \item The well-studied Ising model with a transverse magnetic field on an open chain with $3$ qubits is investigated, see, e.g., \cite{Calabrese_2012TransveresIsing}, i.e.,
\begin{align}
\label{eq:ising}
     H_{\text{Ising}} = -J\left(\sum\limits_{i,j}Z_i\otimes  Z_j + g\sum\limits_j X_j\right),
\end{align}
where $J=-\frac{1}{2}$ and $g=-\frac{1}{2}$.
The evolution time is again set to $T=1$ and the initial parameters are all $0$ except for the parameters of the last layer of $RZ$ gates which are chosen at random in $(0, \frac{\pi}{2}]$ such that $\ket{\psi_0}=\ee^{-i\gamma}\ket{000}$ with $\gamma\in\mathbb{R}$.
Notably, we avoid the initial state $\ket{\psi_0}=\ket{000}$ to circumvent getting stuck in a local minima.
    \item The two qubit hydrogen molecule approximation given in \cite{VarSITEMcArdle19} is studied, with
    \begin{align}
        H_{\text{hydrogen}} 
        &= 0.2252\,I\otimes I + 0.5716\,Z\otimes Z \nonumber \\
        &\hspace{3mm}+ 0.3435\,I\otimes Z  - 0.4347\,Z\otimes I \nonumber \\
        &\hspace{3mm}+ 0.0910\,Y\otimes Y + 0.0910\,X\otimes X.     \label{eq:hydrogen}
    \end{align}
    Again, the evolution time is set to $T=1$ and the initial parameters are chosen such that the initial state is $\ket{\psi_0}=\ket{++}$, i.e., all parameters are $0$ except for the last layer of $RY$ rotations which are given as $\pi/2$.
\end{enumerate}
\begin{figure}[h!]
\captionsetup{singlelinecheck = false, format= hang, justification=centerlast, font=footnotesize, labelsep=space}
\begin{center}
\begin{tikzpicture}
\node at (0,0){\includegraphics[width=0.6\linewidth]{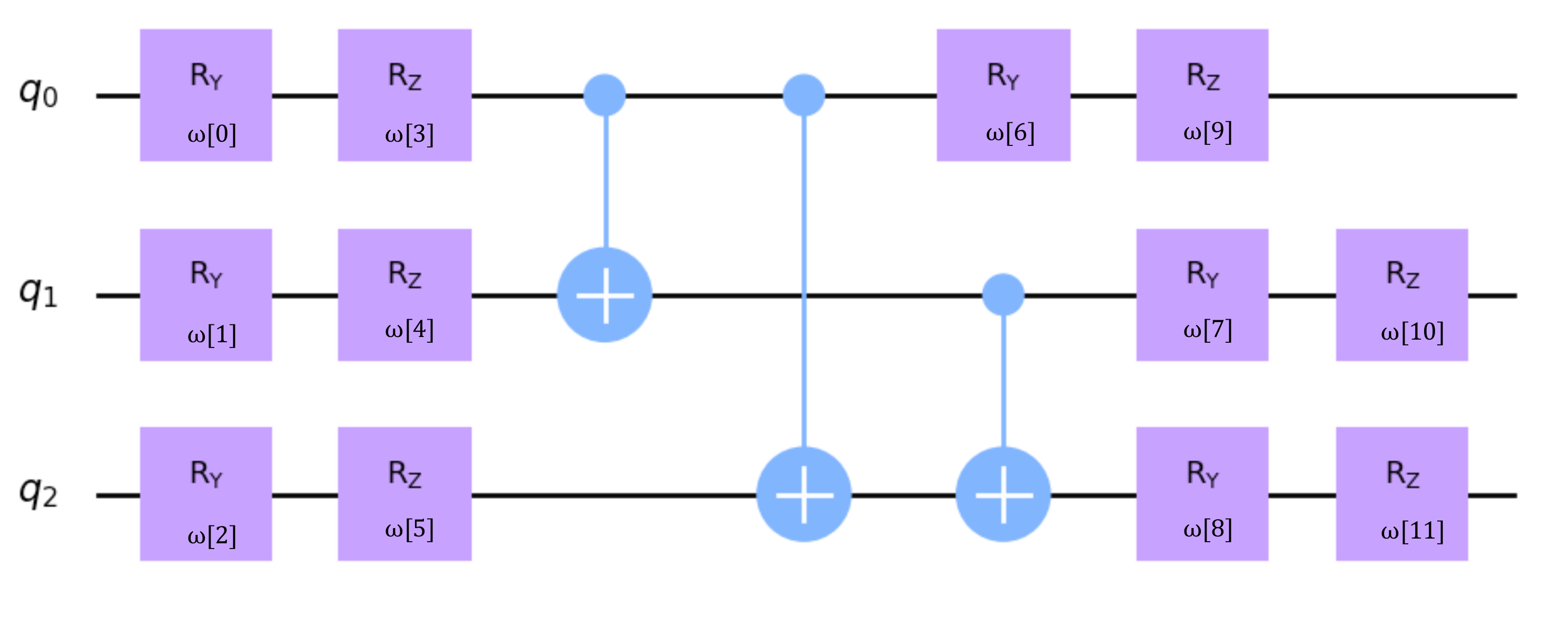}};
\node at (-5.8, 0) {$\ket{0}^{\otimes 3}$};
\draw[decorate, thick, decoration = {brace, amplitude=15pt}] (-4.7,-1) --  (-4.7,1);
\end{tikzpicture}
\end{center}
\caption{This quantum circuit, called \emph{EfficientSU2}, is part of the Qiskit \cite{qiskit} circuit library and chosen as ansatz for the presented VarQITE experiments. It consists of two layers of $RY$ and parameterized single-qubit $RZ$ rotations and one $CX$ entanglement block. }
\label{fig:effsu2}
\end{figure}

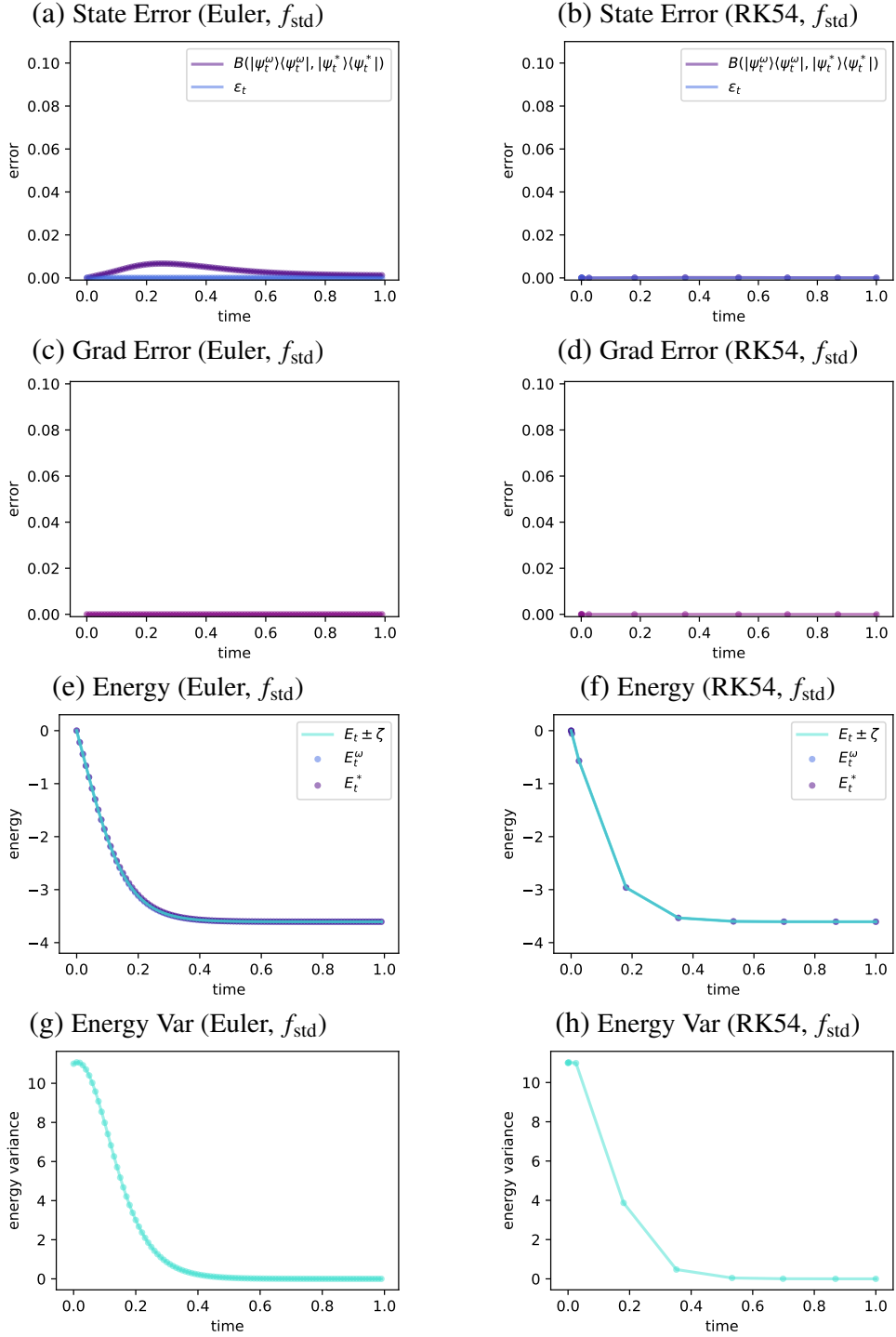
\begin{figure}[!ht]
    \centering
    \begin{tikzpicture}
 \node at (1, 0.9) {\textbf{VarQITE: $H_{\text{illustrative}}$ with different ODE solvers}};
\node[anchor=north west] at (-6, 0) {\includegraphics[width=0.4\textwidth]{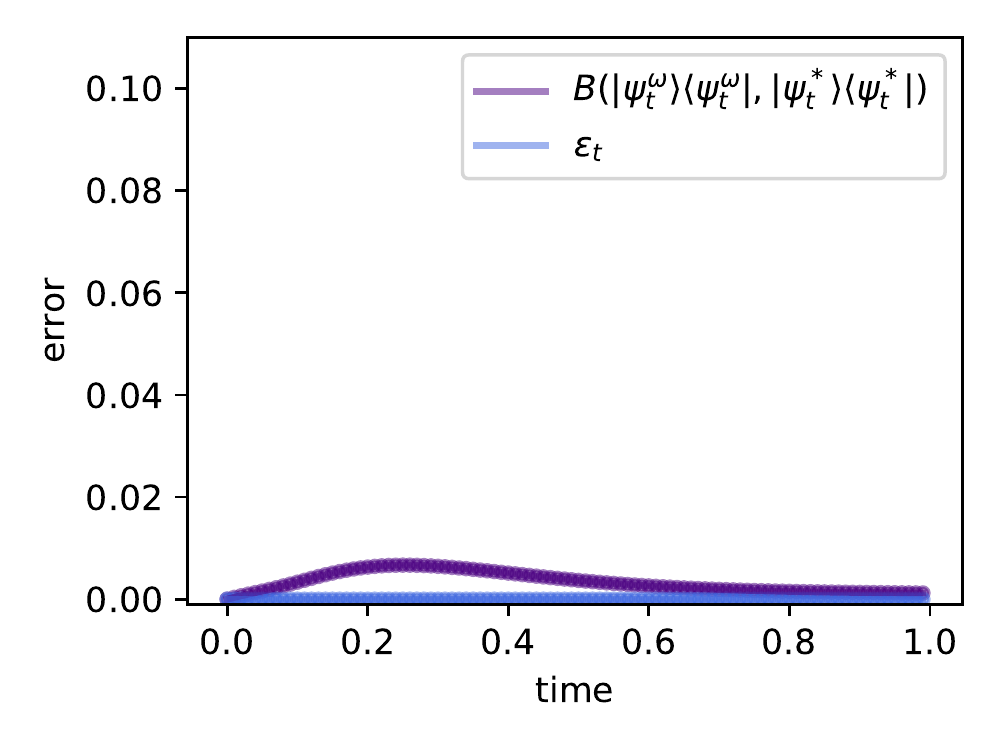}};
\node[anchor=north west] at (-6, -4.8) {\includegraphics[width=0.4\textwidth]{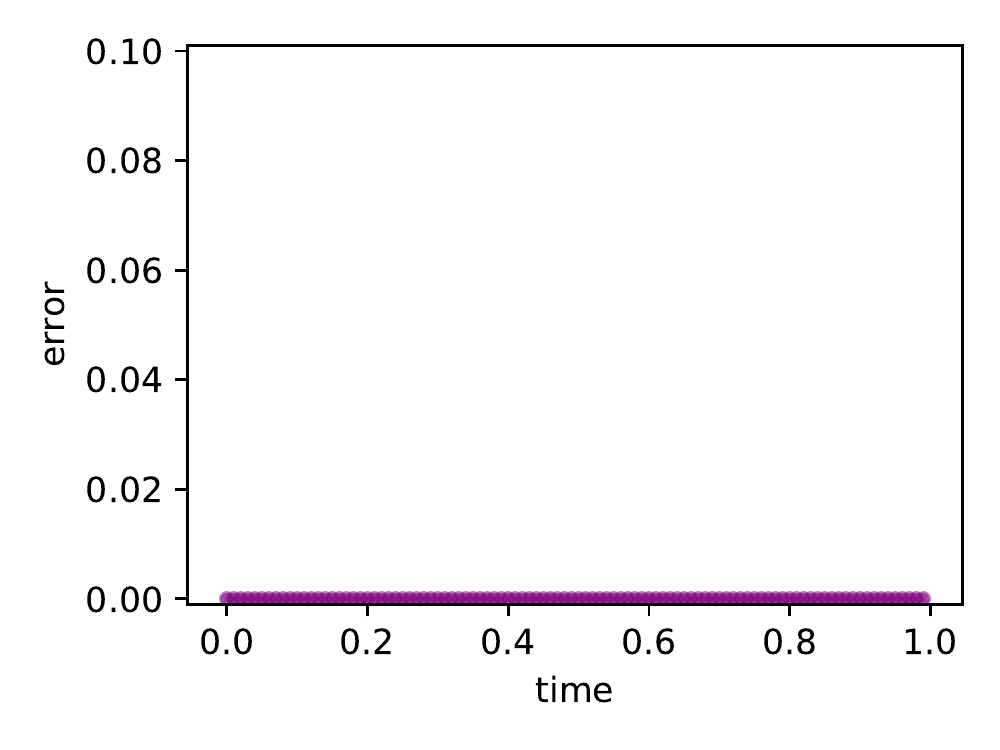}};
\node[anchor=north west] at (-6, -9.6) {\includegraphics[width=0.4\textwidth]{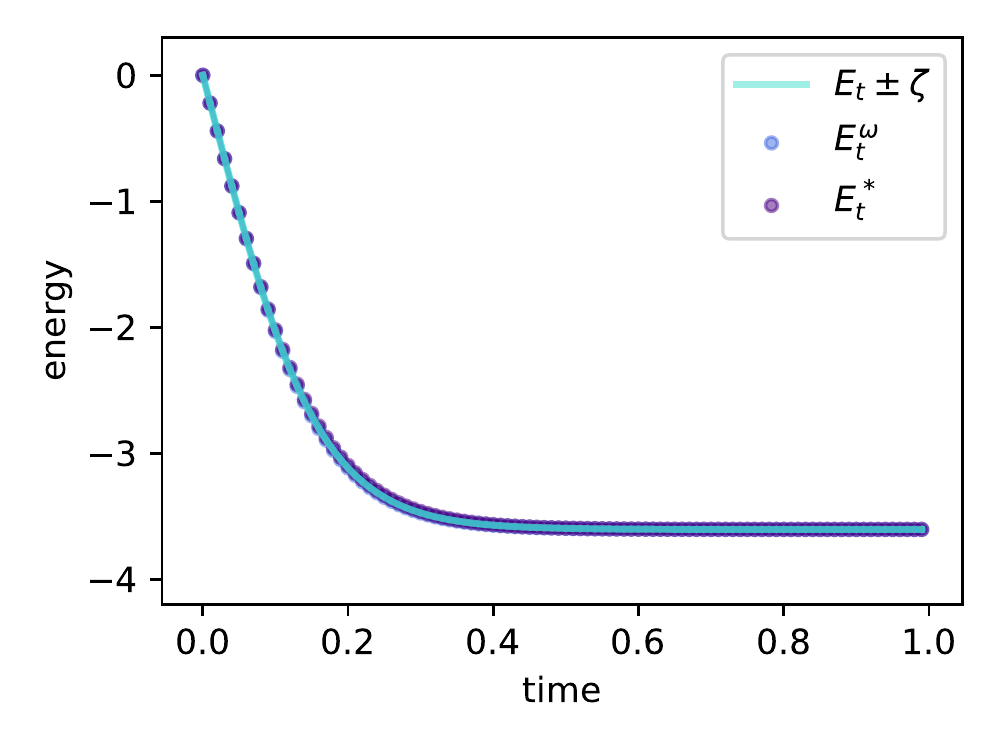}};
\node[anchor=north west] at (-6, -14.4) {\includegraphics[width=0.4\textwidth]{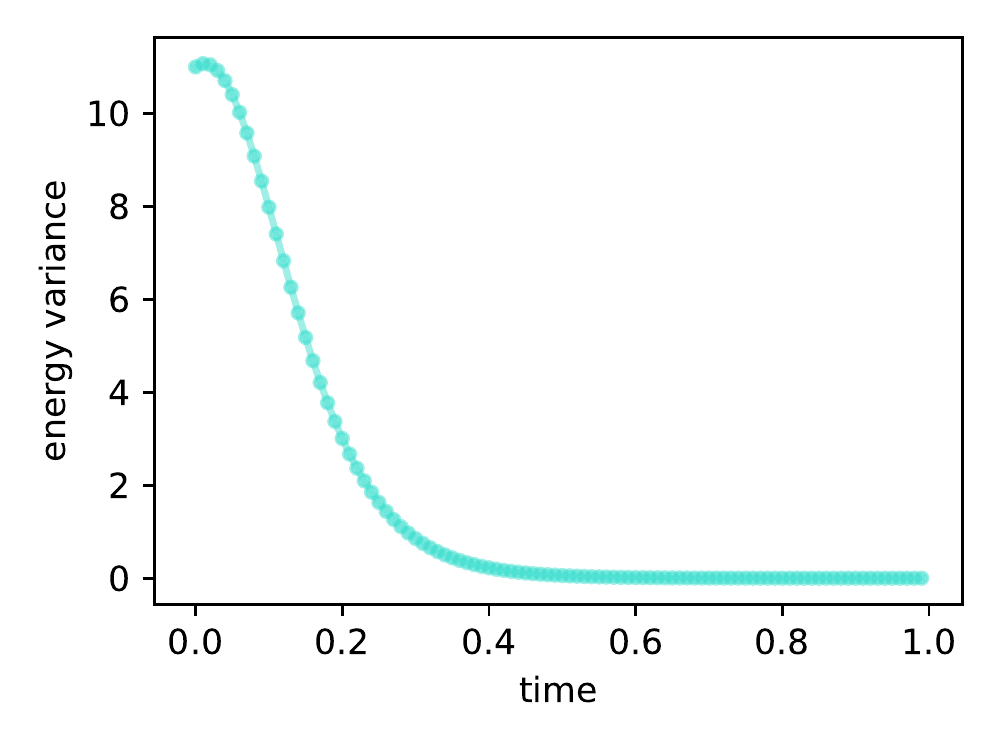}};
\node[anchor=north west] at (1, -0) {\includegraphics[width=0.4\textwidth]{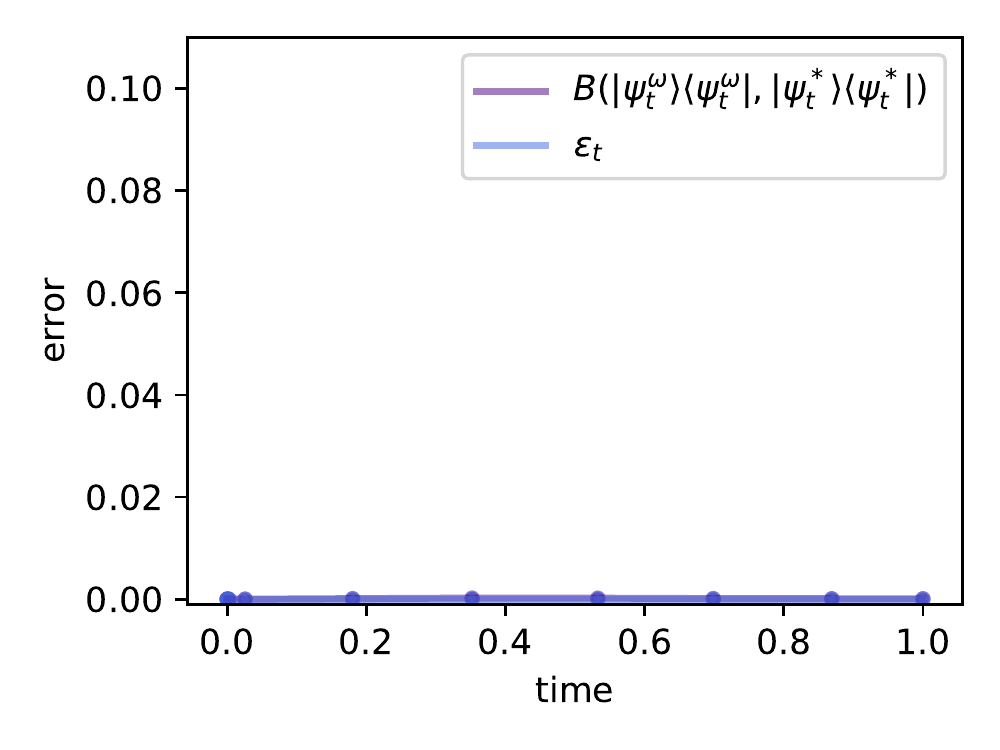}};
\node[anchor=north west] at (1, -4.8) {\includegraphics[width=0.4\textwidth]{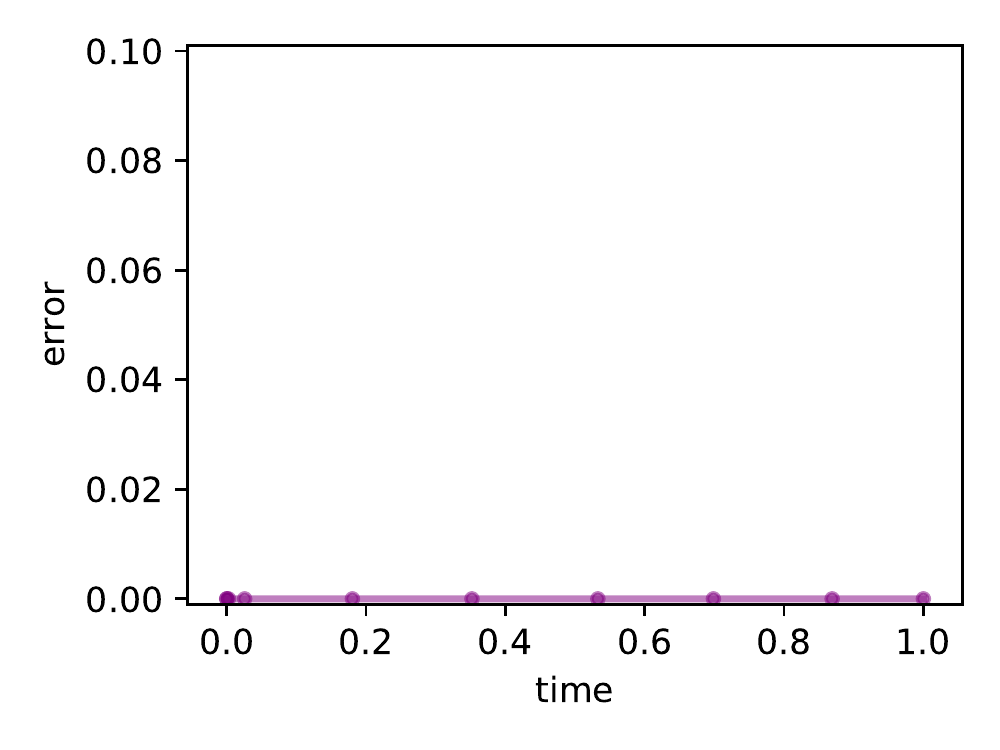}};
\node[anchor=north west] at (1,-9.6) {\includegraphics[width=0.4\textwidth]{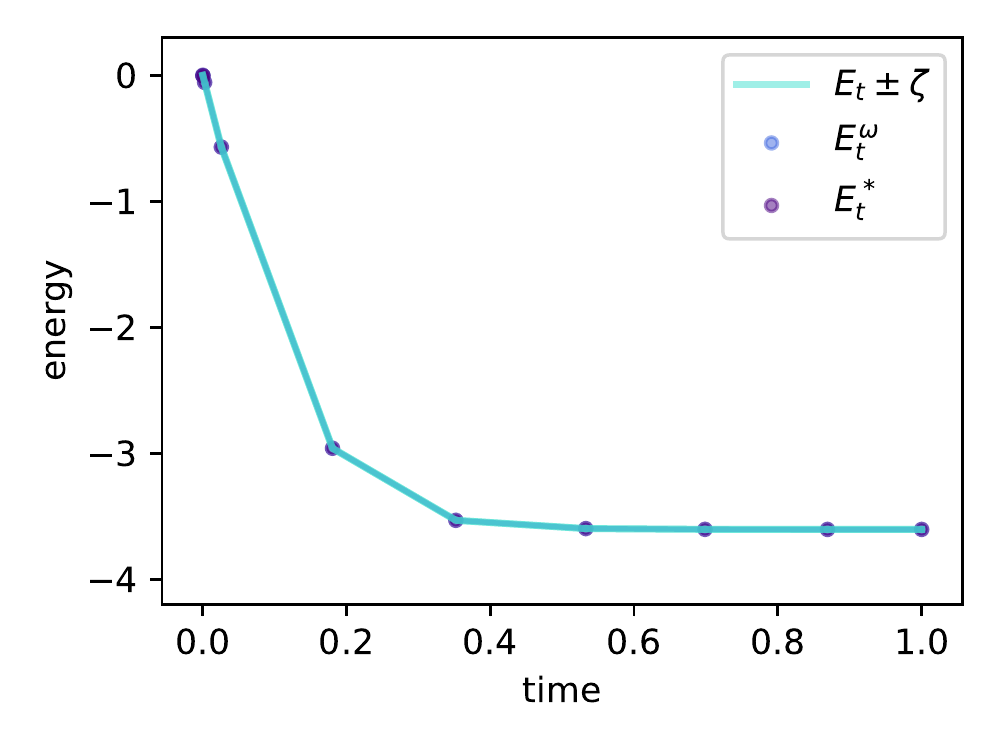}};
\node[anchor=north west] at (1, -14.4) {\includegraphics[width=0.4\textwidth]{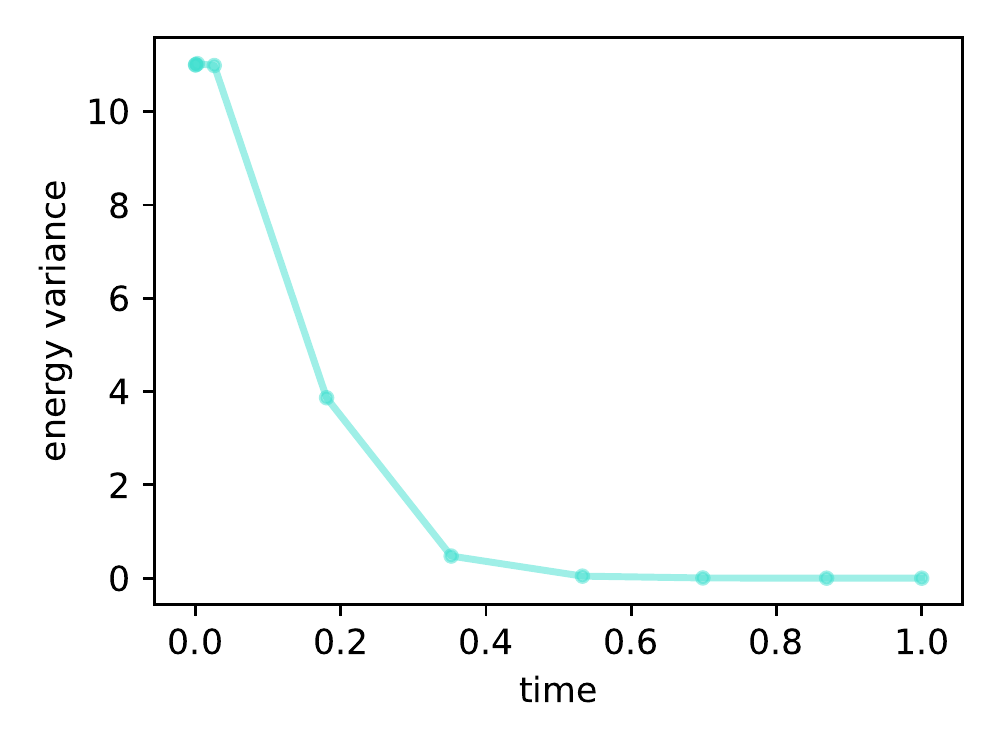}};
\node at (-3.2, 0) {(a)  \small{State Error (Euler, $f_{\text{std}}$)}};
\node at (-3.2, -4.8) {(c)  \small{Grad Error (Euler, $f_{\text{std}}$)}};
\node at (-3.2, -9.6) {(e)  \small{Energy (Euler, $f_{\text{std}}$)}};
\node at (-3.2, -14.4) {(g)  \small{Energy Var (Euler, $f_{\text{std}}$)}};
\node at (4.3, 0) {(b)  \small{State Error (RK54, $f_{\text{std}}$)}};
\node at (4.3, -4.8) {(d)  \small{Grad Error (RK54, $f_{\text{std}}$)}};
\node at (4.3, -9.6) {(f) \small{Energy (RK54, $f_{\text{std}}$)}};
\node at (4.3, -14.4) {(h)  \small{Energy Var (RK54, $f_{\text{std}}$)}};
\end{tikzpicture}
        \captionsetup{singlelinecheck = false, format= hang, justification=centerlast, font=footnotesize, labelsep=space}
        \caption{ VarQITE for $\ket{\psi_0}=\ket{++}$, $H_{\text{illustrative}}$ and $T=1$  with the standard ODE  (a), (c), (e), (g) are computed using forward Euler. (b), (d), (f), (h) employ  RK54. - (a), (b) illustrate the error bound $\epsilon_{t}$ and the true Bures metric. Furthermore, (c), (d) show the gradient errors $\|\ket{e_t}\|_2$.  (e), (f) present the energies as well as the error bound to $E_t^{\omega}$, i.e., $E_t^{\omega}\pm \zeta\left(\bm{\omega_t}, \epsilon_t\right)$.  Lastly, (g), (h) present the evolution of the energy variance \smash{$\Var(H)_{\psi^{\omega}_t}$}.}
     \label{fig:illustrative_qite}
\end{figure}

Now, we investigate the power VarQITE and of the introduced error bound.
To that end, the optimizations for $\zeta\left(\bm{\omega_t}, \epsilon_t\right)$ and $\chi\left(\bm{\omega_t}, \epsilon_t\right)$, see Lemmas \ref{lem_energyDiff} and \ref{lem_secondTerm}, respectively, use a grid search which is viable since the problems are cheap to evaluate. 
Furthermore, we would like to point out that $\delta_t$, used to compute $\dot\varepsilon_t \approx \frac{\varepsilon_t - \varepsilon_{t-1}}{\delta_t}$, does not need to correspond to the step  size of the ODE solver but may be chosen such that a reasonable finite difference approximation is obtained. We set $\delta_t = 10^{-4}$.

\begin{figure}[!h]
    \centering
    \begin{tikzpicture}
 \node at (1, 1.1) {\textbf{VarQITE: $H_{\text{Ising}}$ with different ODE definitions}};
\node[inner sep=0pt, anchor=north west] at (-6, -0.12) {\includegraphics[width=0.4\textwidth]{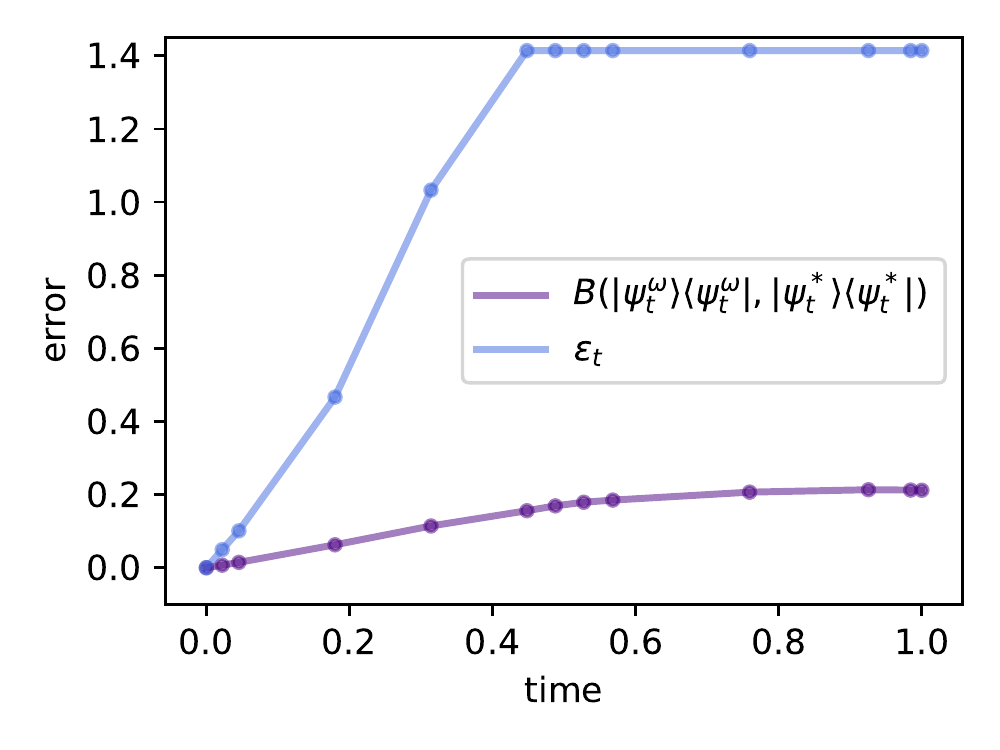}};
\node[anchor=north west] at (1,0) {    \includegraphics[width=0.4\textwidth]{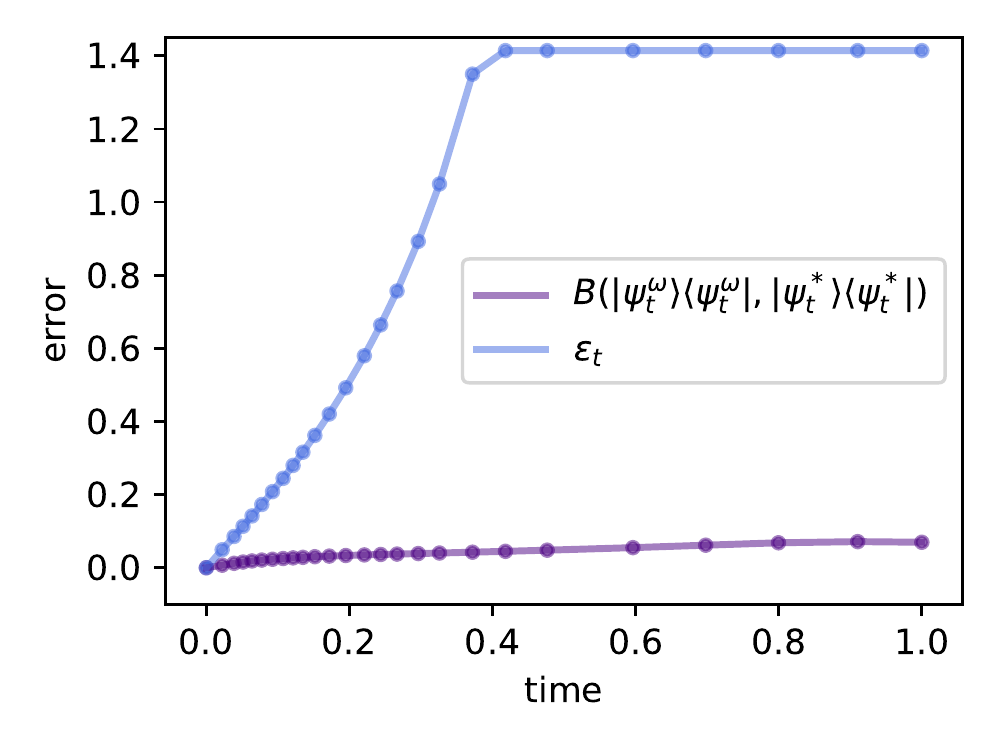}};
\node[anchor=north west] at (-6, -4.7) {\includegraphics[width=0.4\textwidth]{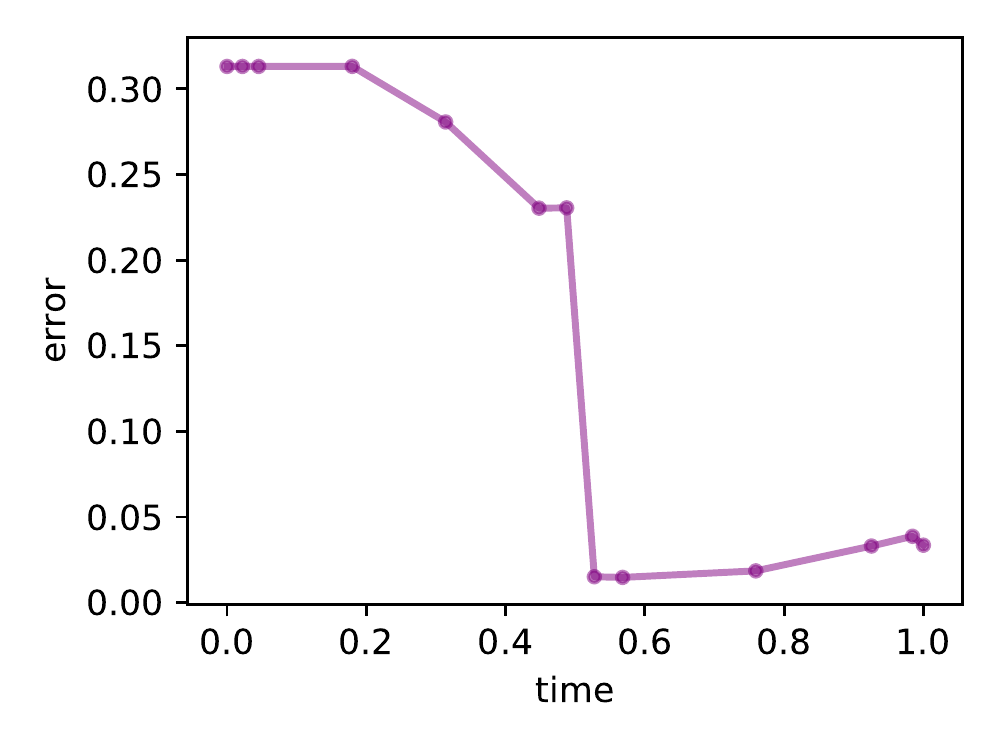}};
\node[anchor=north west] at (1,-4.7) {\includegraphics[width=0.4\textwidth]{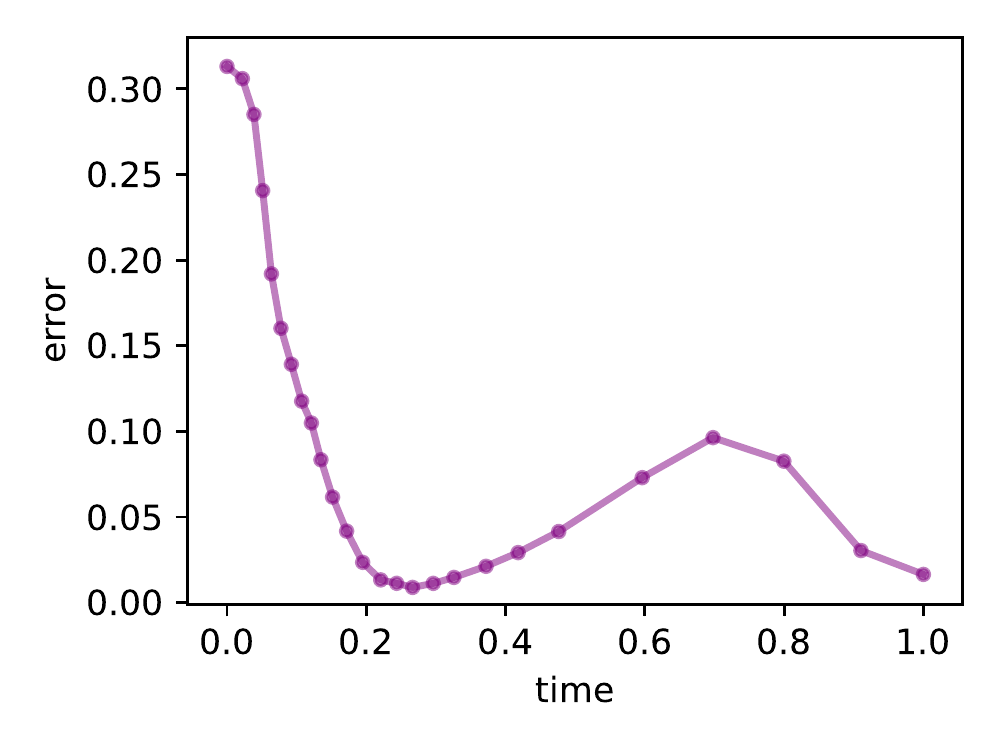}};
\node at (-3.2, 0) {(a) \small{State Error (RK54, $f_{\text{std}}$)}};
\node at (4.3, 0) {(b) \small{State Error (RK54, $f_{\text{argmin}}$)}};
\node at (-3.2, -4.7) {(c) \small{Grad Error (RK54, $f_{\text{std}}$)}};
\node at (4.3, -4.7) {(d) \small{Grad Error (RK54, $f_{\text{argmin}}$)}};
\end{tikzpicture}
        \captionsetup{singlelinecheck = false, format= hang, justification=centerlast, font=footnotesize, labelsep=space}
        \caption{VarQITE for $\ket{\psi_0}=\ee^{-i\gamma}\ket{000}$, $H_{\text{Ising}}$ and $T=1$ with RK54. (a), (b) employ the standard ODE. (c), (d) use the argmin ODE. (a), (c) illustrate the error bound $\epsilon_{t}$ as well as the actual Bures metric. (b), (d) present the corresponding evolution of the gradient errors $\|\ket{e_t}\|_2$.}
            \label{fig:qite_ising}
\end{figure}
\begin{figure}[!h]
    \centering
    \begin{tikzpicture}
 \node at (1, 1.1) {\textbf{VarQITE: $H_{\text{hydrogen}}$}};
\node[anchor=north west] at (-6, -0.) {\includegraphics[width=0.4\textwidth]{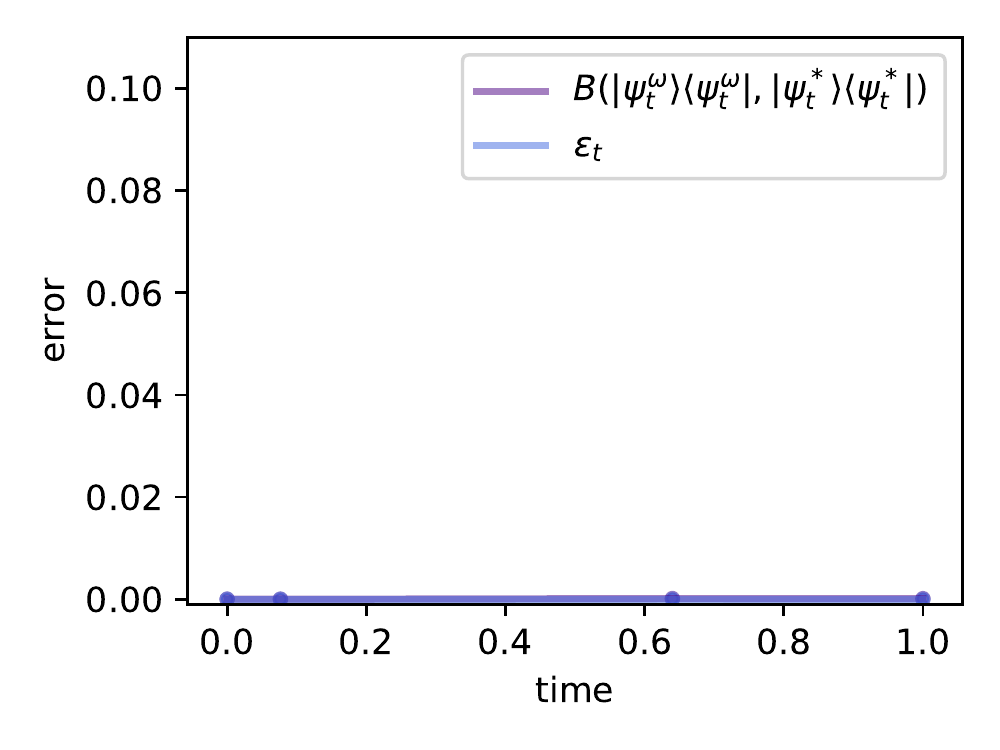}};
\node[anchor=north west]  at (1,0) {    \includegraphics[width=0.4\textwidth]{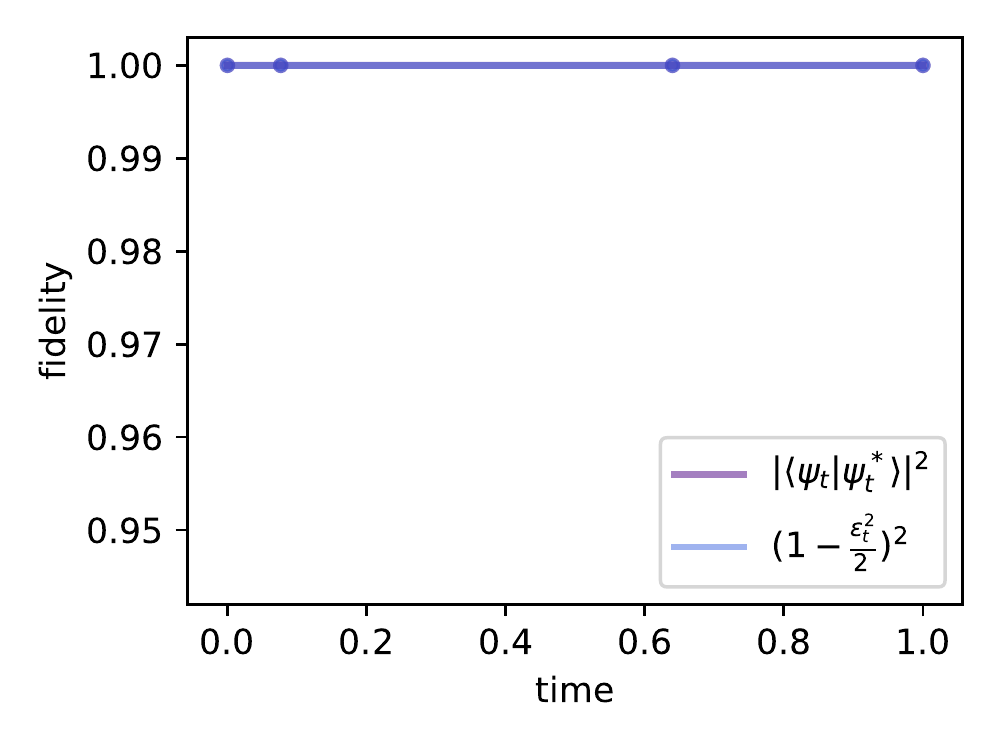}};
\node[anchor=north west] at (-6, -4.7) {\includegraphics[width=0.4\textwidth]{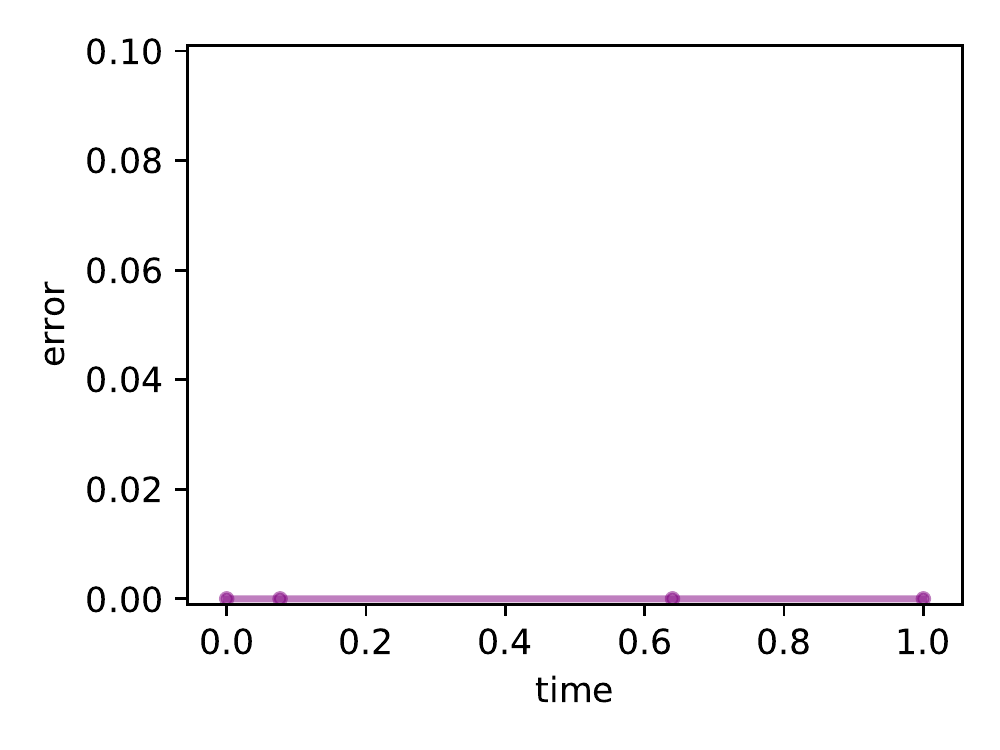}};
\node[anchor=north west] at (1,-4.7) {\includegraphics[width=0.4\textwidth]{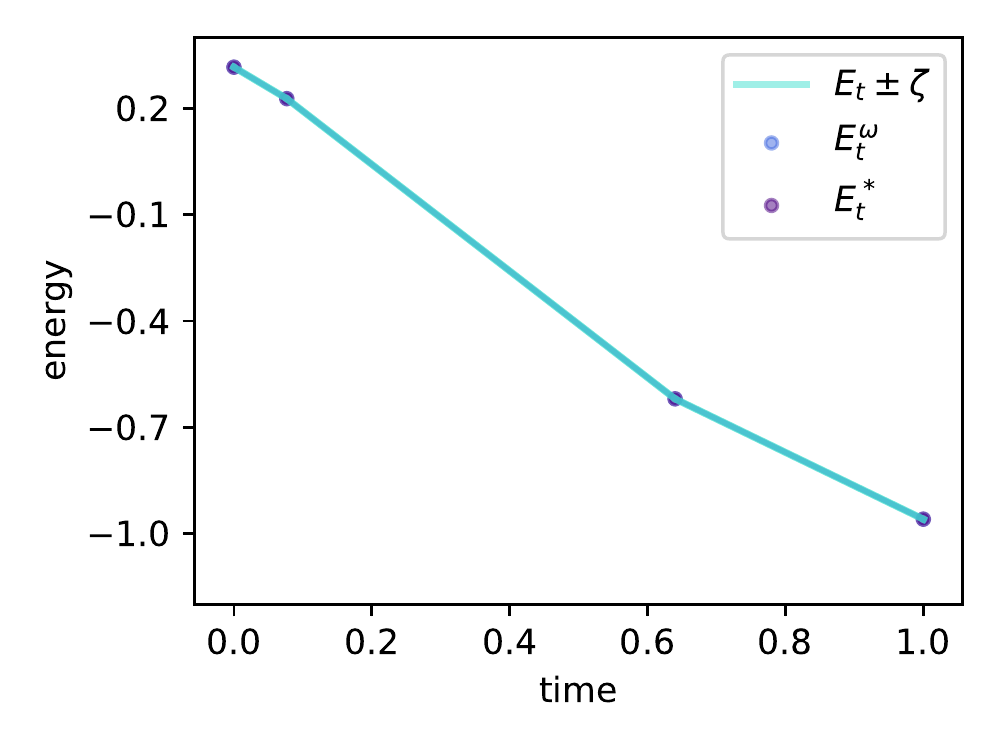}};
\node at (-3.2, 0) {(a) \small{State Error (RK54, $f_{\text{argmin}}$)}};
\node at (4.3, 0) {(b) \small{Fidelity (RK54, $f_{\text{argmin}}$)}};
\node at (-3.2, -4.7) {(c) \small{Grad Error (RK54, $f_{\text{argmin}}$)}};
\node at (4.3, -4.7) {(d) \small{Energy (RK54, $f_{\text{argmin}}$)}};
\end{tikzpicture}
        \captionsetup{singlelinecheck = false, format= hang, justification=centerlast, font=footnotesize, labelsep=space}
        \caption{VarQITE for $\ket{\psi_0}=\ket{++}$, $H_{\text{hydrogen}}$ and $T=1$.  All plots are based on RK54 and the argmin ODE. - (a) illustrates the error bound $\epsilon_{t}$ as well as the actual Bures metric. (b) shows the corresponding fidelity bound $\left(1+\epsilon^2/2\right)^2$ as well as the actual fidelity $|\langle\psi^{\omega}_t|\psi^*_t\rangle|^2$. Furthermore, (c) illustrates the evolution of the gradient errors $\|\ket{e_t}\|_2$.   (d) presents the system energy $E_t^{\omega}$ corresponding to the prepared state, the energy $E_t^*$ corresponding to the target state and the error bound to $E_t^{\omega}$, i.e., $E_t^{\omega}\pm \zeta\left(\bm{\omega_t}, \epsilon_t\right)$.}
            \label{fig:hydrogen}
\end{figure}

First, the outcomes using the standard ODE with forward Euler as well as RK54 are compared for $H_{\text{illustrative}}$.
The results shown in Fig.~\ref{fig:illustrative_qite} provide an example of the potentially insufficient numerical integration accuracy of forward Euler. More explicitly, although the gradient errors and, thus, also the error bounds are all $0$, the actual error between prepared and target state, i.e., $ B(\proj{\psi^*_T}, \proj{\psi^{\omega}_t})$, shown in purple, is positive. The application of RK54 in comparison reduces the error in the integration. Furthermore, the plotted energies $E_t^{\omega}$ and $E_t^*$ are very close and show a convergence behavior. Since the state error is sufficiently small and the energy variance $\Var(H)_{\psi^{\omega}_t} = 
\bra{\psi^{\omega}_t} H^2 \ket{\psi^{\omega}_t} - (E_t^{\omega})^2$ 
shown in Fig.~\ref{fig:illustrative_qite} (d), (h) converges to $0$, we can conclude that the VarQITE is run for a sufficiently long time and actually reached the ground state. 

Moreover, the performance using RK54 and the standard respectively argmin ODE formulation is compared on the example of $H_{\text{Ising}}$. Fig.~\ref{fig:qite_ising} illustrates the sensitivity of the error bounds for VarQITE with respect to the gradient errors $\|\ket{e_t}\|_2$. Since the gradient errors are rather large especially at the beginning of the evolution, the error bounds reach the maximum value of $\sqrt{2}$ throughout the evolution. At this point, they are not representative for the actual error $B\left(\proj{\psi^*_T}, \proj{\psi^{\omega}_t}\right)$ anymore and clipped to $\sqrt{2}$. Notably, for this example the argmin ODE only works marginally better than the standard ODE.

For VarQITE applied to $H_{\text{hydrogen}}$, we focus on RK54 and the argmin ODE formulation.
The results are visualized in Fig.~\ref{fig:hydrogen}. The SLE is solved with a least squares method from NumPy \cite{Numpy2020}.
It can be seen that VarQTE is able to represent the state exactly and that all gradient errors $\|\ket{e_t}\|_2$ are $0$. This is directly reflected by the respective error bound $\epsilon_t$ and fidelity bound $(1-\epsilon_t^2/2)^2$. 
Furthermore, the energy with respect to the prepared state $E_t^{\omega}$ is equal to $E_t^*$ at all times of the numerical ODE integration.

The presented experiments show that the error bounds are good approximations to the actual error if the gradient errors are small. In the case of larger gradient errors, the error bounds can grow fast and exceed the meaningful range.
Moreover, the results reveal that the numerical integration error introduced by forward Euler can easily exceed the error bounds, rendering them useless.

\begin{figure}[h!t]
\captionsetup{singlelinecheck = false, format= hang, justification=raggedright, font=footnotesize, labelsep=space}
\begin{center}
\begin{tikzpicture}
\node at (0, -1.5) {
\includegraphics[width=0.5\linewidth]{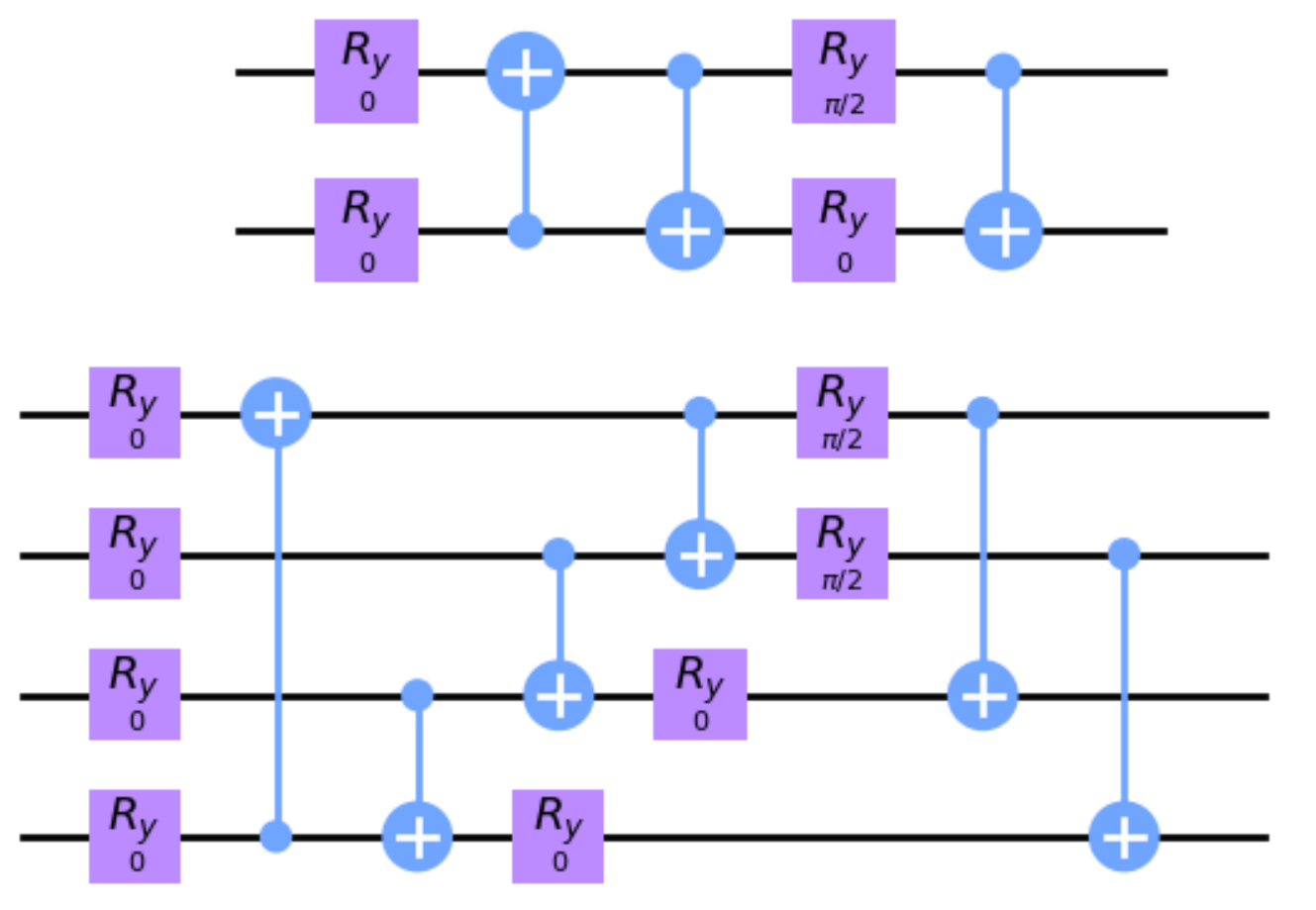}};
\node[rectangle,
    fill = white,
    minimum width = 2cm, 
    minimum height = 2cm] (r) at (-3.3, 0.5) {};
    \node at (-2.9, 1) {(a)};
\node at (-4.3, -.9) {(b)};
\node at (-6.5, -4.3) {(c)};
\draw[decorate, thick, decoration = {brace, mirror, amplitude=5pt}] (-2.5,0.6) --  (-2.5,-.26);
\draw[decorate, thick, decoration = {brace, mirror, amplitude=13pt}] (-3.8,-1.3) --  (-3.8,-3.6);
    \node at (0,-6) {\includegraphics[width=0.8\linewidth]{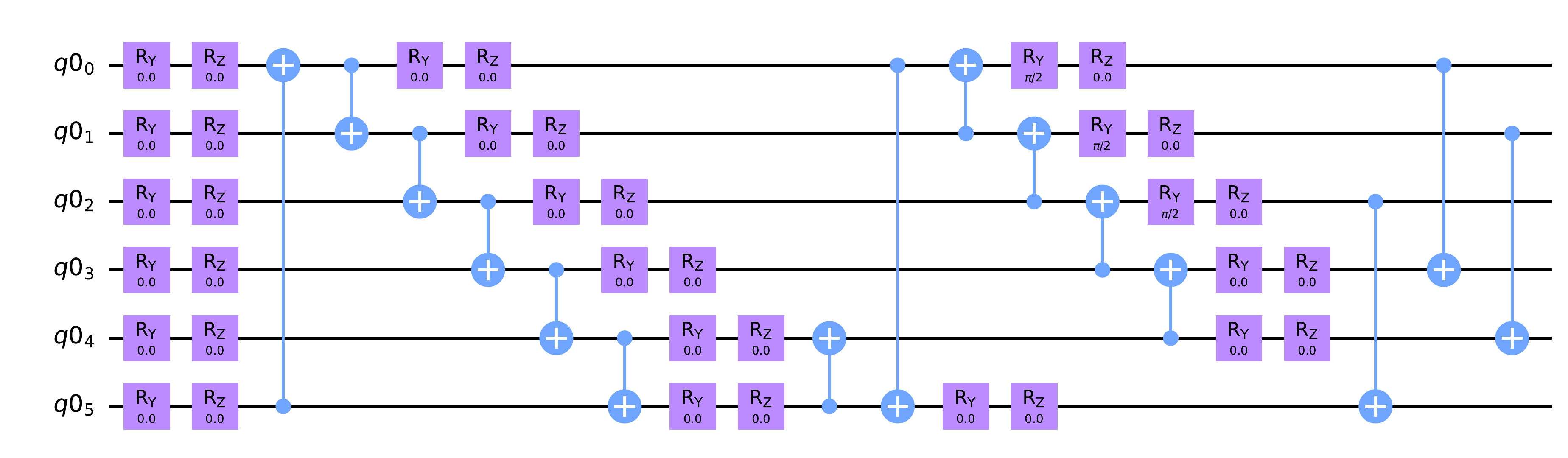}};
\node at (-3.5,.3) {$\ket{0}^{\otimes 2}$};
\node at (-5.,-2.4) {$\ket{0}^{\otimes 4}$};
    \draw[decorate, thick, decoration = {brace, mirror, amplitude=13pt}] (-6.2,-4.7) --  (-6.2,-7.4);
\node at (-7.4,-6) {$\ket{0}^{\otimes 6}$};
\end{tikzpicture}
\end{center}
\caption{The depicted circuits illustrate the initial trial state for the Gibbs state preparation with VarQITE corresponding to (a) $H_1 $, (b) $H_2$ and (c) $H_3$.
}
\label{fig:ansaetze}
\end{figure}
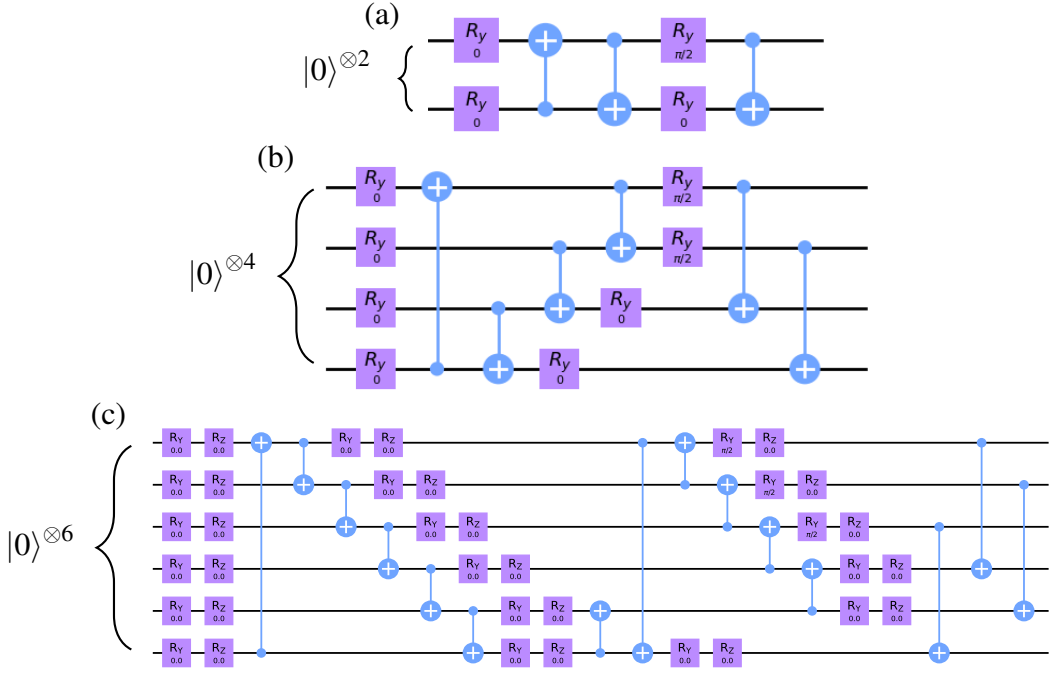

Furthermore, we demonstrate the feasibility of approximate Gibbs state preparation with VarQITE. 
First, we illustrate the convergence of the state fidelity with respect to the target state for the following two simple one- and two-qubit Hamiltonians
\begin{align}
H_1 &= Z, \\
H_2 &= Z\otimes Z - 0.2 \,Z\otimes I  - 0.2 \, I\otimes Z + 0.3 \, X\otimes I + 0.3 \, I\otimes X,
\end{align}
and temperature $k_BT=1$.
The results are computed using the parameterized quantum circuit shown in Fig.~\ref{fig:ansaetze} (a) and (b).
The algorithm is executed with the standard ODE, forward Euler and $10$ time steps on different backends: a statevector simulator -- which is based on matrix-vector multiplication -- and the IBM Quantum Johannesburg $20$-qubit backend \cite{ibmQX}.
In the latter case, readout error-mitigation \cite{qiskit, dewes2012readout} is used to obtain the final results run on real quantum hardware. Fig.~\ref{fig:VarQITE} (a) and (b) depicts the results considering the fidelity between the trained and the target Gibbs state for each time step.
It should be noted that the fidelity for the quantum backend evaluations employ state tomography.
The plots illustrate that the method approximates the states, we are interested in, reasonably well and that also the real quantum hardware achieves fidelity values over $0.99$ and $0.96$, respectively.

\begin{figure}[!htb]
\captionsetup{singlelinecheck = false, format= hang, justification=raggedright, font=footnotesize, labelsep=space}
\begin{center}
    \begin{tikzpicture}
\node at (0,0) {
\includegraphics[width=0.95\linewidth]{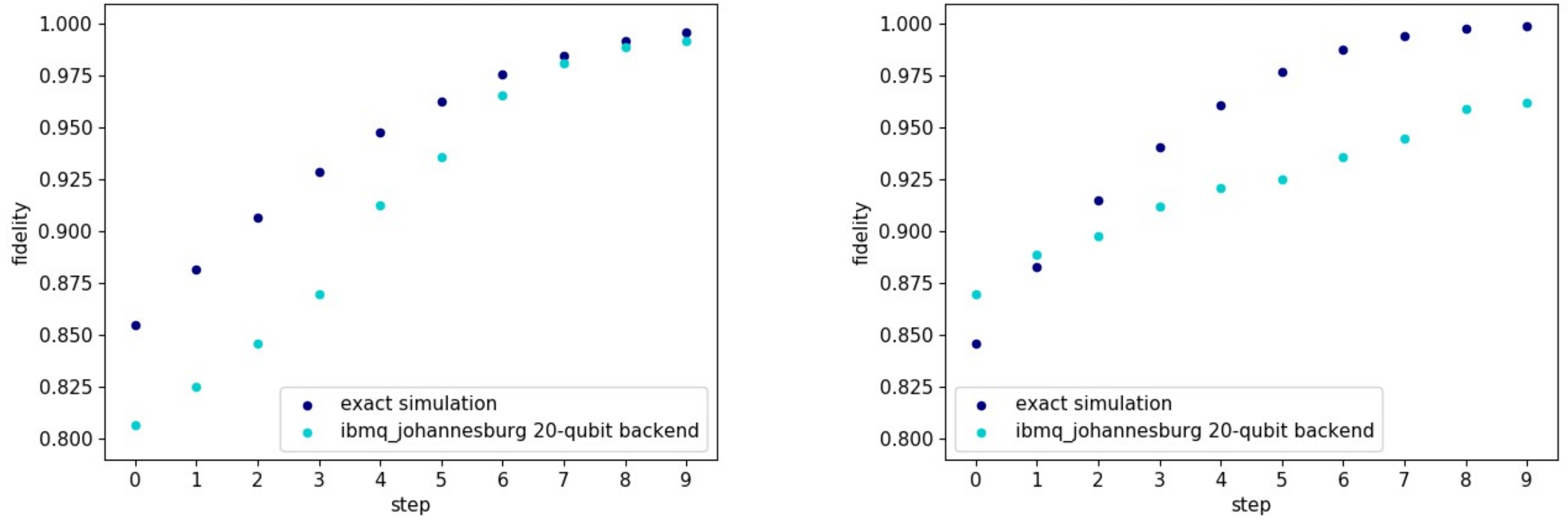}};
 \node at (0,3.5) {\textbf{Gibbs State Preparation - Fidelity}};
 \node at (0,-5.5) {\includegraphics[width=0.5
\linewidth]{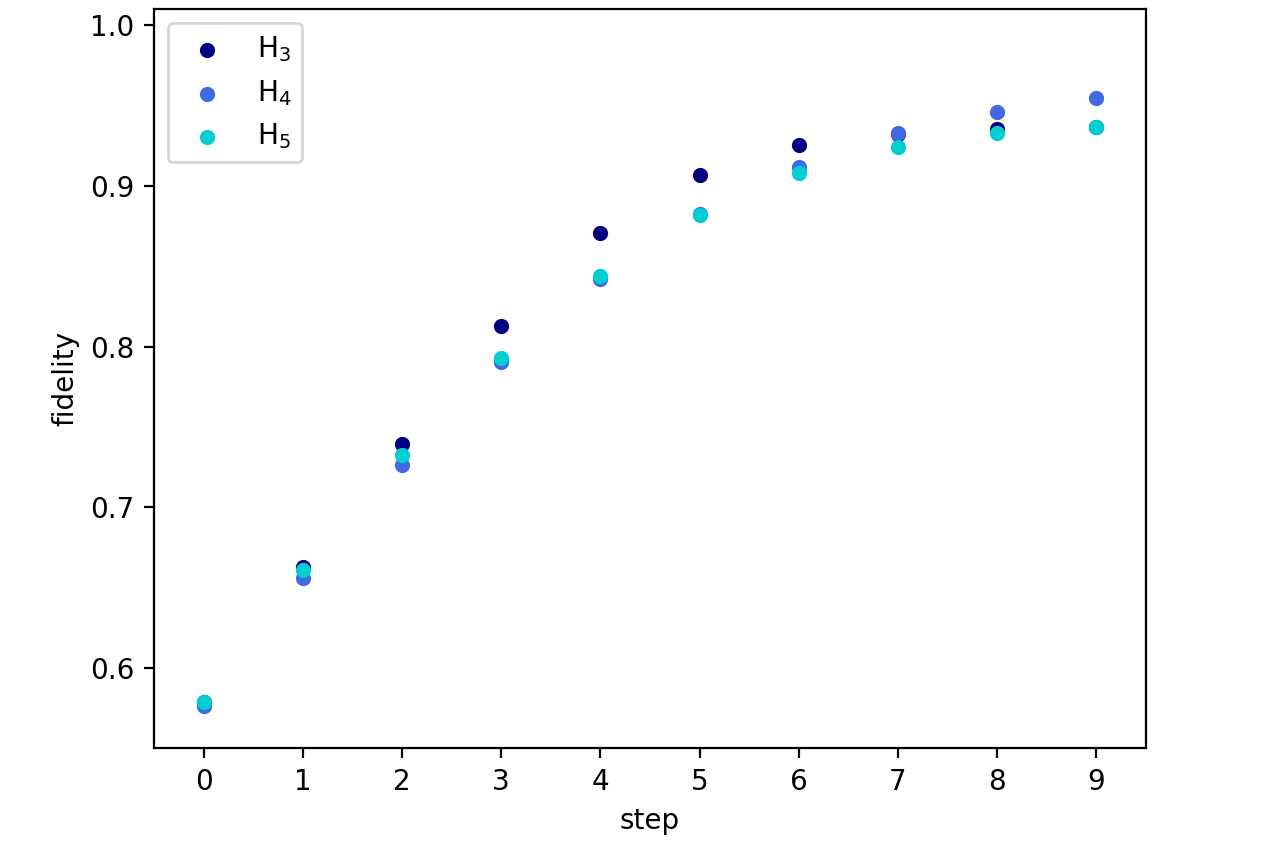}};
     \node at (-6., 2.7) {(a)};
\node at (1, 2.7) {(b)};
\node at (-3, -2.7) {(c)};
\end{tikzpicture}
\end{center}
\caption{The fidelity between approximate VarQITE and target Gibbs state is shown for (a) $H_1$ respectively (b) $H_2$ which are trained with a statevector simulator and real quantum hardware, i.e., the IBM Quantum Johannesburg $20$-qubit backend \cite{ibmQX}. The fidelity between the state that is prepared using VarQITE and the target Gibbs state preparation using exact statevector simulation is illustrated in (c) for $H_3$, $H_4$ and $H_5$. Each VarQITE simulation uses forward Euler with $10$ time steps.
}
\label{fig:VarQITE}
\end{figure}
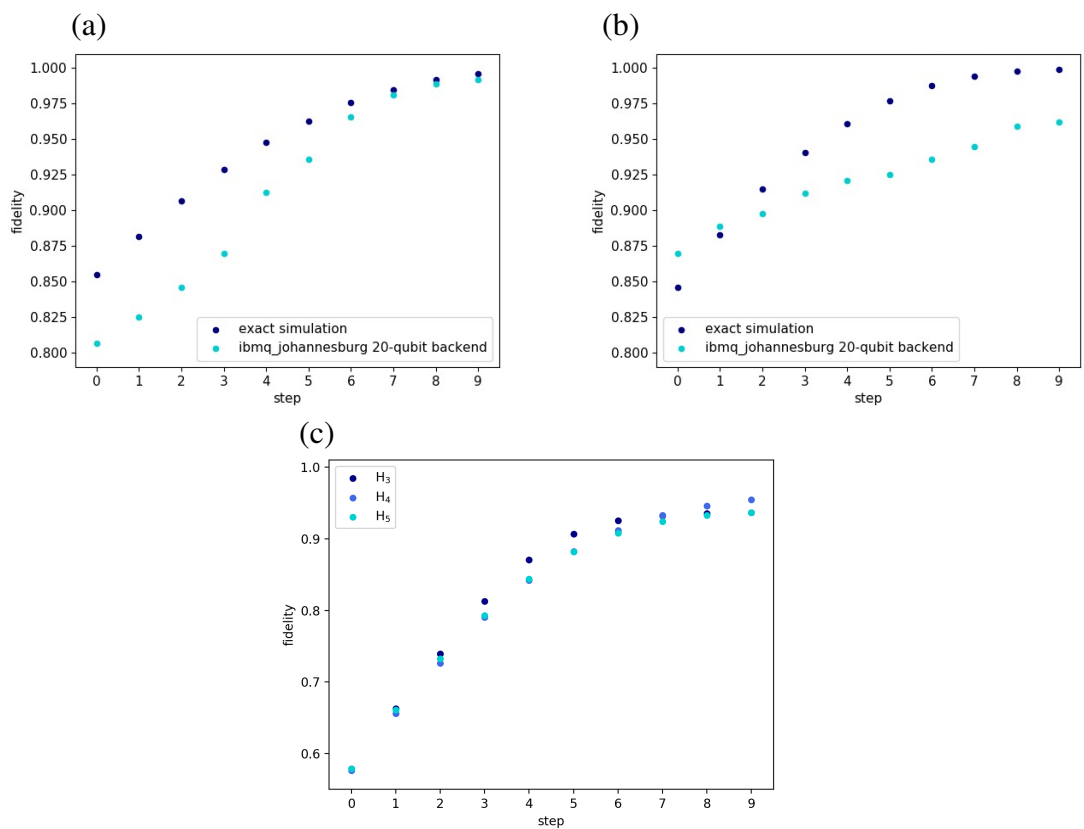


Next, the scalability of VarQITE for approximate Gibbs state preparation is investigated. To that end, we conduct further experiments for Hamiltonians with $n>2$ qubits using only statevector simulation.
The states are prepared for the following Hamiltonians
\begin{align}
    H_3 &= 2 \, Z\otimes Z\otimes I + I\otimes Z\otimes Z - 0.5 \, I\otimes Z\otimes I, \\
    H_4 &= 2\, Z\otimes Z\otimes I\otimes I +  I\otimes Z\otimes Z\otimes I - 0.5 \, I\otimes Z\otimes I\otimes Z, \\
    H_5 &= 2 \, Z\otimes Z\otimes I\otimes I\otimes I +  I\otimes Z\otimes Z\otimes I\otimes I - 0.5 \, I\otimes Z\otimes I\otimes I\otimes Z,
\end{align}
again with temperature $k_BT=1$ and the standard ODE, forward Euler and $10$ time steps.
All Gibbs states are prepared with a shallow, depth $2$ ansatz that consists of parameterized $RY$ and $RZ$ gates as well as controlled $X$ gates $\left(CX\right)$ \cite{nielsen10}. Fig.~\ref{fig:ansaetze} (c) illustrates the respective quantum circuit for $H_3$, whereby the given parameters prepare the initial state. The states for $H_4$ and $H_5$ are prepared equivalently.
The effectiveness of the respective VarQITE Gibbs state approximation is confirmed with a plot, see Fig.~\ref{fig:VarQITE} (c), of the fidelities between the trained and the target Gibbs state throughout the simulated time evolution. It should be noted that all final states approximate the targets with fidelity bigger than $0.93$.


\addtocontents{toc}{}

%% file: applications.tex

\chapter{Generative Quantum Machine Learning}
\label{sec:applications} 

\textbf{Abstract.} 

 The goal of generative quantum machine learning is to learn the random distribution, underlying given training data by using quantum resources such as quantum channels or quantum operators. In many cases, the respective probability distribution is modelled by the sampling statistics of a parameterized quantum state which, in turn, may be employed as approximate but efficient state preparation scheme.
 In this section, we introduce the concept of generative quantum machine learning on the examples of quantum generative adversarial networks and quantum Boltzmann machines.
Furthermore, we present variational implementations of these generative QML algorithms and demonstrate their feasibility on illustrative examples which are either executed with quantum simulations or quantum hardware.
\index{generative quantum machine learning}

\vspace{8mm}

\noindent

Generative quantum machine learning aims at training a parameterized model such that it represents the probability distribution underlying given training data. Notably, this distribution characterizes the generation process of the corresponding data samples. There are different approaches to tackle this problem.

Generally, this QML application is compatible with generic classical or quantum data and parameterized models which comprise of classical as well as quantum computing resources.
In the following, we demonstrate the workflow of generative QML on the example of algorithms that consider classical data samples and employ parameterized quantum channels or operators to model the underlying generation process. To that end, the respective distribution is represented approximately by the sampling probabilities of a quantum state, see amplitude encoding in Def.~\ref{sec:amplitudeEncoding}. 

The first approach that we discuss employs a direct training of the parameters of a quantum channel, i.e., a variational quantum circuit. To begin with, the parameters of a quantum circuit are initialized. Then, the quantum circuit is measured and the respective measurement samples are mapped to a relevant feature space. The samples are, now, used to compute the parameter updates with respect to a given loss function. This cycle is repeated until a certain stopping criteria is sufficed. There exist different types of this approach, e.g., quantum Born machines \cite{ LiuDifferentiableLearning18, BornSupremacyCoyle2020, Hamilton_2019GenerativeModelBenchmarks, BenedettiGenModellin19} and quantum generative adversarial networks \cite{lloyd2qGANs18, killoran2018qgans, SITU2020193qGANs, Benedetti_2019PureStateApproxqGANs, Zeng_2019LearningInferenceonqGANs, huqGANs2019, Alcazar_2020qganFinance} which differ in their learning strategy.
The former employs direct fitting of the trained sampling probabilities to a given probability distribution using, e.g., the negative log-likelihood. The latter uses an indirect learning procedure which can be described as a two-player game between a generator and a discriminator. 
We would like to point out that for multi-modal distributions log-likelihood optimization tends to spread the mass of the learned distribution over all modes whereas GANs tend to focus the mass \cite{goodfellow, metz2017unrolled}.
GAN-based learning is, thus, explicitly suitable for capturing not only uni-modal but also multi-modal distributions.

The second approach that we consider trains the parameters of a quantum operator in the form of a Hamiltonian using the following steps: First, a suitable parameterized Hamiltonian is defined - see Hamiltonian encoding in Def.~\ref{sec:gibbs_related_work}. After choosing a set of initial parameters, the Hamiltonian is mapped onto a quantum state which in turn is measured. The resulting measurement statistics are used to update the operator parameters. The training goal is to optimize the operator parameters such that the sampling statistics reflect the probability distribution underlying the given training data. One prominent example, of this kind of algorithms is given by
quantum Boltzmann machines \cite{Khoshaman_2018QVAE, QBMAmin18, Anschtz2019RealizingQB, QBMWiebe17, Kappen18QBM}. These represent a structured generative QML approach, in the sense that the model Hamiltonian can be directly related to the conditional independence structure of the training data \cite{piatkowski2019exponential}.


This learning approach facilitates \emph{implicit} learning of the data generation process. Standard schemes, on the other hand, may require to first model an \emph{explicit} functional description of the respective distribution whose loading potentially requires expensive quantum arithmetic \cite{Vedral_1996QuantumArithmetic, Beckman_1996QuantumFactoring, Van_Meter_2005QuantumExponentiation, SvoreLookaheadAdder06, Ruiz_Perez_2017QuantumArithmetic} or exponentially many gates. 
Furthermore, we would like to point out that generative QML represents a flexible learning algorithm. More explicitly, it is compatible with online or incremental learning, i.e., the model can be updated if new training data samples become available and the model ansatz can be adapted with respect to the properties of the given training data.
Moreover, it is suitable for execution on current or near-term quantum hardware and may, thus, already be studied in practice.
Since generative QML leads to an approximate model of the probability distribution underlying given training data, it enables, e.g., structural analysis of the given data, the creation of new data samples which are in accordance with the given training data and the loading of an approximate but efficient quantum representation of the respective random distribution.
The loading use-case\footnote{Once the parameters are trained, the quantum state which approximate the probability distribution can be trivially generated as often as needed.} is important for noisy as well as fault-tolerant quantum hardware. In noisy systems, it is of utter importance to reduce the circuit depth for quantum data loading to limit the impact of hardware noise and approximate schemes can help to enable that. Furthermore, due to theoretical limitations, the exact loading of some classical data structures into quantum states cannot be accomplished efficiently but only approximately. In this context, generative quantum machine learning offers a promising avenue.


The remainder of this chapter presents a detailed discussion on quantum generative adversarial networks, see Sec.~\ref{sec:qgan}, and quantum Boltzmann machines, see Sec.~\ref{sec:QBM}. The feasibility of variational implementations is illustrated on examples which are either simulated numerically or executed on actual quantum hardware. Furthermore, Sec.~\ref{sec:EuropeanCallOptionPricingexample} shows how generative QML can be used for approximate quantum data loading and, thereby, enable the application of quantum algorithms to real-world problems, e.g., in the realm of finance. 

\section[Quantum Generative Adversarial Networks]{Quantum Generative Adversarial Networks\footnote{This section is reproduced in part, with permission, from C.~Zoufal, A.~Lucchi, S.~Woerner, "Quantum Generative Adversarial Networks for Learning and Loading Random Distributions", \textit{npj Quantum Information}, vol. 5, Article Nr. 103, 2019}}

\label{sec:qgan}

Quantum generative adversarial networks (qGANs) employ the interplay of a generator and discriminator to map an approximate representation of a probability distribution underlying given data samples into a quantum channel. 
This section particularly focuses on a qGAN implementation where the generator is given by a quantum channel, i.e., a variational quantum circuit, and the discriminator by a classical neural network, and discusses the application of efficient learning and loading of generic probability distributions -- implicitly given by data samples -- into quantum states.
The loading requires $\mathscr{O}\left(poly\left(n\right)\right)$ gates and can, thus, enable the use of potentially advantageous quantum algorithms, such as quantum amplitude estimation (QAE) \cite{brassardQAE02}.
The efficiency of variational qGANs is illustrated using numerical simulations as well as an actual quantum processor provided by the IBM Quantum Experience.

\subsection{Introduction}
\label{sec:qgan_intro}

In classical machine learning, generative adversarial networks (GANs) \cite{goodfellow, Kurach2018TheGL}  have proven useful for generative modeling. These algorithms employ two competing neural networks - a generator and a discriminator - which are trained alternately to learn a probability distribution that is implicitly given by data samples. 
The quantum counterpart, quantum generative adversarial networks -- originally suggested in \cite{lloyd2qGANs18, killoran2018qgans} -- either employ a generator, or a generator and a discriminator which are given as parameterized quantum channels. This form of generative QML is, hence, compatible with classical as well as quantum data. 

Firstly, qGANs may be used to train a quantum state that replicates the statistics of a given quantum state \cite{lloyd2qGANs18, huqGANs2019, killoran2018qgans, Benedetti_2019PureStateApproxqGANs, chakrabarti2019wassersteinqGAN}. To this end, a quantum generator generates quantum states and a quantum discriminator applies measurements to given quantum states. The training goal of the generator is to minimize the distance between the generated and the target quantum states, whereas the discriminator tries to discriminate correctly between these states.
Secondly, this type of generative QML can learn a quantum representation of a probability distribution underlying classical training data. In this context, the generator is given as a parameterized quantum channel and the discriminator as a classical neural network. The aim of the generator is to construct a quantum state whose measurement statistics resemble the statistics of the given training data and the discriminator aims to correctly classify between generated and training data samples.
In \cite{romero}, this setting is used to load continuous probability distributions and generate samples from it. Similarly, \cite{Zeng_2019LearningInferenceonqGANs} employs this approach to model the probability distribution of discrete classical data and generates corresponding samples.
The work presented in \cite{Zoufal2019qGANs} shows that this method may also facilitates efficient learning and loading of probability distributions into quantum states.
The remainder of this section focuses on qGANs that consist of quantum generators and classical discriminators and, more specifically, on their application to probability distribution learning and loading.

The structure is given as follows.
First, we introduce the classical GAN concept in Sec.~\ref{sec:GANs}.
Then, we discuss the idea behind the quantum counterpart, qGANs, outline the implementation, explain the respective application for probability distribution learning and loading in Sec.~\ref{sec:quantumGAN}.
Next, Sec.~\ref{sec:qganApplications} presents different test cases for qGAN-based approximate quantum data loading results obtained with a quantum simulator and the IBM Quantum Boeblingen superconducting quantum computer with 20 qubits, both accessible via the IBM Quantum Experience \cite{ibmQX} 
Finally, the conclusions and a discussion on open questions are presented in Sec.~\ref{sec:qganconclusion_Outlook}.

\subsection{Generative Adversarial Networks}
\label{sec:GANs}

The generative models considered in this work, GANs \cite{goodfellow, Kurach2018TheGL}, employ two neural networks - a generator and a discriminator - to learn random distributions that are implicitly given by training data samples.
Originally, GANs were used in the context of image generation and modification. 
Other generative models, such as variational auto encoders \cite{Kingma2013AutoEncodingVB, Burda2015ImportanceWA},  rely on log-likelihood optimization and are, therefore, prone to generating blurry images. In contrast, GANs manage to generate sharp images and are consequently popular in the machine learning community \cite{Dumoulin2017AdversLearnedInference}.

Suppose a classical training data set $X=\set{\bm{x}_0, \ldots, \bm{x}_{s-1}} \subset \mathbb{R}^{d_{\text{out}}}$ sampled from a potentially unknown probability distribution $p_{\text{real}}$. Let $G^{\omega}: \mathbb{R}^{d_{\text{in}}} \rightarrow \mathbb{R}^{d_{\text{out}}}$ and $D^{\phi}:\mathbb{R}^{d_{\text{out}}} \rightarrow \{0, 1\}$ denote the generator and the discriminator networks, respectively. The corresponding network parameters are given by $\bm{\omega} \in \mathbb{R}^{k_{\text{g}}}$ and $\bm{\phi} \in \mathbb{R}^{k_{\text{d}}}$. 
The generator $G^{\omega}$ aims to translate samples from a fixed prior distribution $p_{\text{prior}}$ in $\mathbb{R}^{d_{\text{in}}}$ into samples which are indistinguishable from samples of the real distribution $p_{\text{real}}$ in $\mathbb{R}^{d_{\text{out}}}$. 
The discriminator $D^{\phi}$, on the other hand, tries to distinguish between data from the generator and from the training set.
The training process is illustrated in Fig.~\ref{fig:gan}.	
\begin{figure}
\begin{tikzpicture}
\node at(0,3.5){\textbf{Generative Adversarial Network}};
\node at(0,0){
\includegraphics[width=0.95\linewidth]{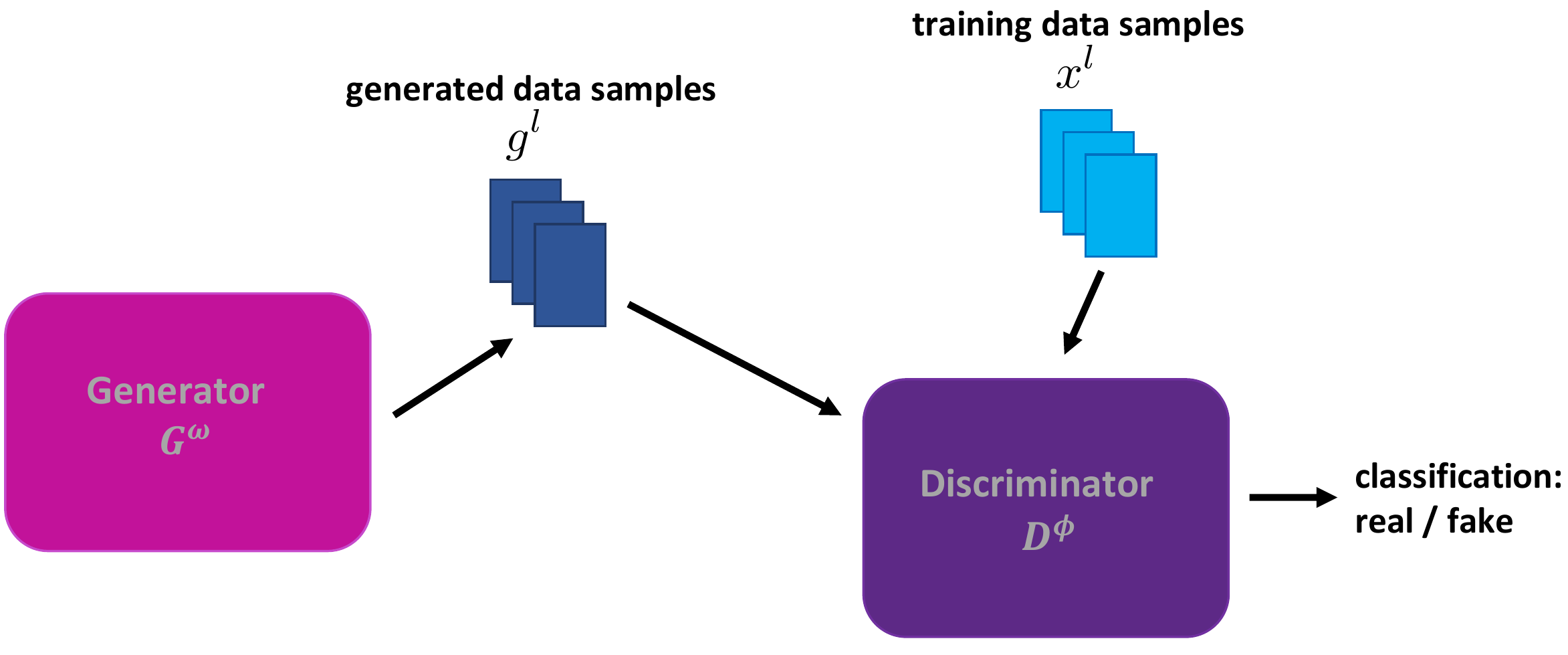}};
\end{tikzpicture}
\caption{This figure illustrates the setting employed during a GAN training. First, the generator creates data samples. Second, the discriminator tries to differentiate between the generated samples and the training samples. The generator and discriminator are trained alternately. \label{fig:gan}}

\end{figure}

The optimization objective of classical GANs may be defined in various ways. In this work, we consider the non-saturating loss \cite{fedus2018many} which is also used in the original GAN work \cite{goodfellow}. The generator's loss function

\begin{equation}
\begin{split}
	L_G\left(\bm{\phi},\bm{\omega}\right) = \thinspace -\mathbb{E}_{z\sim p_{\text{prior}}}\left[\log\left(D^{\phi}\left(G^{\omega}\left(z\right)\right)\right)\right]
	\end{split},
\end{equation}
aims at maximizing the likelihood that the generator creates samples which are labeled as real data samples.
On the other hand, the discriminator's loss function
\begin{equation}
\begin{split}
	L_D\left(\bm{\phi},\bm{\omega}\right) = \mathbb{E}_{x\sim p_{\text{real}}}\left[\log D^{\phi}\left(x\right) \right] +
	\mathbb{E}_{z\sim p_{\text{prior}}}\left[\log\left(1-D^{\phi}\left(G^{\omega}\left(z\right)\right)\right)\right],
	\end{split}
\end{equation}

aims at maximizing the likelihood that the discriminator labels training data samples as training data samples and generated data samples as generated data samples.
In practice, the expected values are approximated with data batches of size $m$ \footnote{Notably, the respective training often \emph{clips} the values in the logarithm to ensure that the loss function is well-defined, i.e., $\log{\left(\max(x, \epsilon)\right)}$ for a small $\epsilon>0$.}, i.e.,
\begin{equation}
\begin{split}
	L_G\left(\bm{\phi},\bm{\omega}\right) = -\frac{1}{m}\sum\limits_{j=1}^{m}\left[\log\left(D^{\phi}\left(G^{\omega}\left(\bm{z_j}\right)\right)\right)\right],
	\end{split}
\end{equation}and
\begin{equation}
	L_D\left(\bm{\phi},\bm{\omega}\right) = \frac{1}{m}\sum\limits_{j=1}^{m}\left[\log D^{\phi}\left(\bm{x}_j\right)  + \right.
	\left.\log\left(1-D^{\phi}\left(G^{\omega}\left(\bm{z_j}\right)\right)\right)\right],
\end{equation}
for $\bm{x}_j \in X$ and $\bm{z_j} \sim p_{\text{prior}}$.
Training the GAN is equivalent to searching for a Nash-equilibrium of a two-player game:
\begin{align}
 \label{eq:minmaxGenerator}
\underset{\bm{\omega}}{\max}&\: L_G\left(\bm{\phi},\bm{\omega}\right) \\
\label{eq:minmax}
 \underset{\bm{\phi}}{\max}&\: L_D\left(\bm{\phi},\bm{\omega}\right).
\end{align}

Typically, the optimization of Eq.~\eqref{eq:minmaxGenerator} and Eq.~\eqref{eq:minmax} employs alternating update steps for the generator and the discriminator. These alternating steps lead to non-stationary objective functions, i.e., an update of the generator's (discriminator's) network parameters also changes the discriminator's (generator's) loss function.
Common choices to perform the update steps are adaptive learning-rate, gradient-based optimizers, such as ADAM \cite{Kingmaadam14} and AMSGRAD \cite{amsgrad}, that are well suited for solving non-stationary objective functions \cite{Kingmaadam14}.

\subsection{Quantum Algorithm}
\label{sec:quantumGAN}

Next, we discuss how a qGAN implementation that uses a \textbf{quantum} generator and a \textbf{classical} discriminator can be employed to capture the probability distribution of \textbf{classical} training samples.
In this setting, a parameterized quantum channel, i.e., the quantum generator, is trained to transform a given $n$-qubit input state $\ket{\psi_{\text{in}}}$ to an $n$-qubit output state
\begin{equation}
G^{\omega}\ket{\psi_{\text{in}}} = \ket{g^{\omega}} = \sum\limits_{b=0}^{2^n-1}\sqrt{p_b(\bm{\omega})}\ket{b},
\end{equation}
where $p_b(\bm{\omega})$ describes the occurrence probability of a basis state $\ket{b}$. 
This parameterized quantum channel may, e.g., be chosen as parameterized quantum circuit\footnote{Parameterized quantum circuits have the potential to act as universal function approximators \cite{Schuld_2019QMLHilbertSpace}.}, see Sec.~\ref{sec:model}, or as a quantum Boltzmann machine, see Sec.~\ref{sec:qbm}.
For simplicity, we assume that the domain of $X$ is $\{0, \ldots, 2^n-1\}$ because the states that are represented by the generator can, then, naturally be mapped to the sample space of the training data.
This assumption may be easily relaxed, for instance, by introducing an affine mapping between $\{0, \dots, 2^n-1\}$ and an equidistant grid suitable for $X$.
In this case, it might be necessary to map points in $X$ to the closest grid point to allow for an efficient training.
The number of qubits $n$ determines the resolution of the qGAN output -- the number of discrete values $2^n$ that can be represented.
During the training, this affine mapping can be applied classically after measuring the quantum state. However, when the resulting quantum channel is used within another quantum algorithm the mapping must be executed as part of the quantum circuit.
Such an affine mapping can be implemented with a gate-based quantum circuit using linearly many gates \cite{worQuantumRiskAnalysis19}. 
Furthermore, this generative QML algorithm can be easily extended to $d$-dimensional random distributions by choosing $d$ qubit registers with $n_i$ qubits each, for $i = 1, \ldots, d$.

The classical discriminator, a neural network consisting of several linear layers with non-linear activation functions layers, processes the data samples and labels them either as being real  or generated. Notably, the topology of the networks, i.e., number of nodes and layers, needs to be carefully chosen to ensure that the discriminator does not overpower the generator and vice versa.

To train the qGAN, we need to generate classical data samples from the quantum generator. This is achieved by measuring the output state $\ket{g^{\omega}}$ in a particular basis, e.g., the computational basis given as $\ket{b},\:b \in \set{0, \ldots, 2^n-1}$. 
Unlike in the classical case, the sampling does not require a stochastic input but is based on the inherent stochasticity of quantum measurements. 
In the following, it is assumed that the measurement frequency of $\ket{b}$ represents the sampling probability $p_b(\bm{\omega})$.

Given $m$ data samples $\bm{g_j}$  from the quantum generator which are sampled with probability $p_j(\bm{\omega})$ and $m$ randomly chosen training data samples $\bm{x}_j$ for $j\in\set{0, \ldots, m}$, the loss functions of the qGAN are
\begin{align}
\label{eq:lossqGanG}
	 L_G\left(\bm{\phi}, \bm{\omega}\right) = -\sum\limits_{j=1}^{m}p_j(\bm{\omega})\log\left( D^{\phi}\left(\bm{g}_{j}\right) \right),
\end{align}
for the generator, and
\begin{align}
\label{eq:lossqGanD}
	 L_D\left(\bm{\phi}, \bm{\omega}\right) = \sum\limits_{j=1}^{m}\frac{1}{m}\log\left( D^{\phi}\left(\bm{x}_j\right)\right)  +  
	 p_j(\bm{\omega})\log\left(1-D^{\phi}\left(\bm{g}_{j}\right)\right),
\end{align}
for the discriminator, respectively.
The corresponding gradients
\begin{align}
\label{eq:gradqGanG}
	 \nabla_{\bm{\omega}} L_G\left(\bm{\phi}, \bm{\omega}\right) = -\sum\limits_{j=1}^{m}	 \nabla_{\bm{\omega}}p_j(\bm{\omega})\log \left(D^{\phi}\left(\bm{g}_{j}\right) \right),
\end{align}
and
\begin{align}
\label{eq:gradqGanD}
	 	 \nabla_{\bm{\phi}} L_D\left(\bm{\phi}, \bm{\omega}\right) =\sum\limits_{j=1}^{m}\nabla_{\bm{\phi}}\Big[\frac{1}{m}\log\left( D^{\phi}\left(\bm{x}_j\right)\right)  +  
	p_j(\bm{\omega})\log\left(1-D^{\phi}\left(\bm{g}_{j}\right)\right)\Big],
\end{align}
may be evaluated using the analytic quantum gradient techniques described in Sec.~\ref{sec:analytic_gradients}.
As in the classical case, see Eq.~(\ref{eq:minmaxGenerator}) and (\ref{eq:minmax}), the loss functions are optimized alternately with respect to the generator's parameters $\bm{\omega}$ and the discriminator's parameters $\bm{\phi}$.

We would like to point out that a carefully chosen input state $\ket{\psi_{\text{in}}}$ can help to reduce the complexity of the quantum generator and the number of training epochs as well as avoid local optima in the quantum circuit training.
Since the preparation of $\ket{\psi_{\text{in}}}$ should not dominate the overall gate complexity, the input state must be efficiently loadable.
This is feasible, e.g., for efficiently integrable probability distributions, such as log-concave distributions \cite{groverSuperpositionseffintegrablepdfs02}.
In practice, statistical analysis of the training data can guide the choice for a suitable $\ket{\psi_{\text{in}}}$ from the family of efficiently loadable distributions, e.g., by matching expected value and variance.
Later in this section, we present numerical experiments that analyze the impact of $\ket{\psi_{\text{in}}}$.

\subsection{Illustrative Examples}
\label{sec:qganApplications}

Next, we present the results of a broad simulation study on training qGANs with different settings for different target distributions. This study is going to put a focus on the application to approximate quantum data loading for generic probability distributions.
The training data for the following examples are given by $20000$ samples of
\begin{enumerate}
    \item a log-normal distribution with $\mu = 1$ and $\sigma = 1$,
    \item a triangular distribution with lower limit $l=0$, upper limit $u=7$ and mode $m=2$,
    \item and a bimodal distribution consisting of two superimposed Gaussian distributions with $\mu_1=0.5$, $\sigma_1=1$ and $\mu_2=3.5$, $\sigma_2=0.5$.
\end{enumerate}
All distributions are truncated to $\left[0, 7\right]$ and the samples were rounded to integer values.

We consider a quantum generator acting on $n=3$ qubits, which can represent $2^3=8$ values such as $(0, 1, \ldots, 7)$.
The quantum generator, shown in Fig.~\ref{fig:varForm}, is implemented with a parameterized quantum circuit that consists of parameterized single-qubit rotations and blocks of two-qubit gates.
More specifically, the circuit consists of a first layer of $RY$, and then $l$ alternating repetitions of $CZ$ gates and further layers of $RY$ gates.
Similarly to increasing the number of layers in deep neural networks \cite{Goodfellow2016_DeepL}, increasing the depth $l$ enables the circuit to represent more complex structures and increases the number of parameters.\footnote{As discussed in Sec.~\ref{sec:vanishing_grads}, an increase in the quantum circuit depth may also lead to training issues due to the barren plateau phenomenon.} We test quantum generators with depths $l \in \set{1, 2, 3}$.
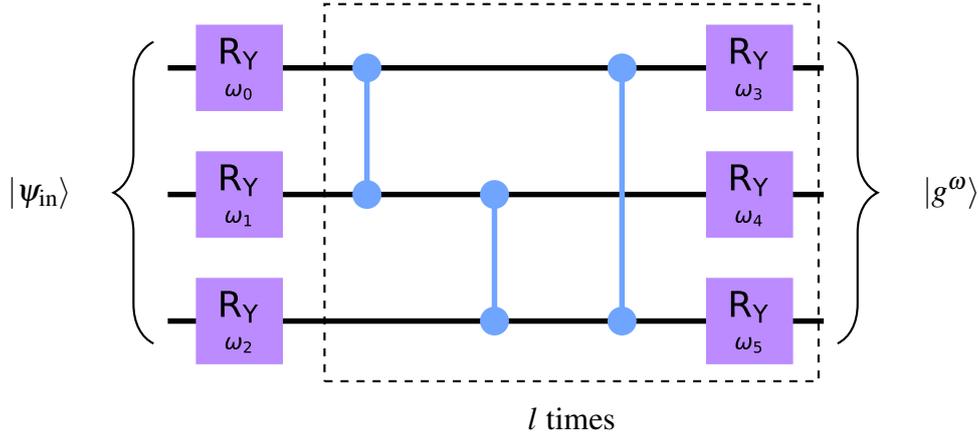
\begin{figure}
   \centering{
   \begin{tikzpicture}
    \node at (0,0) {\includegraphics[width=0.6\linewidth]{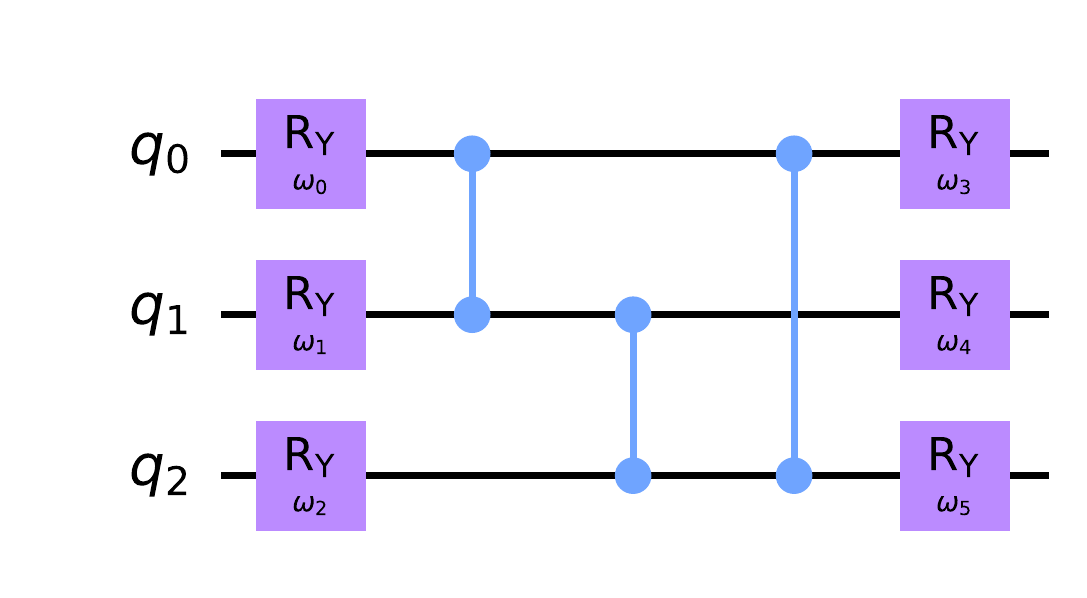}};
\node at (-6,-0) {$\ket{\psi_{\textnormal{in}}}$};
\draw[decorate, thick, decoration = {brace, amplitude=15pt}] (-4.5,-2) --  (-4.5,2);
\node[rectangle, thick, dashed,
    draw = black,
    minimum width = 6.5cm, 
    minimum height = 5cm] (r) at (1,0) {};
\draw[decorate, thick, decoration = {brace, mirror, amplitude=15pt}] (4.5,-2) --  (4.5,2);
\node at (6,-0) {$\ket{g^{\omega}}$};
\node at (1, -3) {$l$ times};
\end{tikzpicture}
\caption{The variational quantum circuit, depicted acts on $n=3$ qubits acts as quantum generator. It is composed of $l+1$ layers of single-qubit $RY$ gates and $l$ entangling blocks. Each entangling block applies $CZ$ gates from qubit $i$ to qubit $\left(i+1\right) \mod\: 3,\:i \in \set{0, \ldots, 2}$ to create entanglement between the different qubits.
 \label{fig:varForm}
 }}
\end{figure}

Furthermore, the initial state of the generator $\ket{g^{\omega}}$ is prepared according to a discrete uniform distribution, a truncated and discretized normal distribution with $\mu$ and $\sigma$ being empirical estimates of mean and standard deviation of the training data samples as well as a randomly chosen initial distribution.
To prepare a uniform distribution on $3$ qubits, we set the generator's input state $\ket{\psi_{\text{in}}}$ to $\ket{+}^{\otimes 3}$ with $\ket{+} = \frac{\ket{0}+\ket{1}}{\sqrt{2}}$ and the initial parameters $\bm{\omega}$ close to $0$, i.e., we draw the parameter values from a uniform distribution on $[-\delta, +\delta]$, for $\delta = 10^{-1}$. This leads to a state which is close to a uniform distribution but slightly perturbed.
Adding small random perturbations helps to break symmetries and can, thus,  improve the training performance \cite{LEHTOKANGAS1998265, Thimm_percept97, Chen2019}.
Loading a normal distribution can be conducted efficiently  \cite{groverSuperpositionseffintegrablepdfs02} but may require the use of involved quantum arithmetic techniques. 
Instead, we fit the parameters of a shallow quantum circuit to represent a normal distribution, use this state as $\ket{\psi_{\text{in}}}$ and, then, choose the initial parameters $\bm{\omega}$ again by drawing from a uniform distribution on $[-\delta, +\delta]$, for $\delta = 10^{-1}$. This results in a state $\ket{g^{\omega}}$ which represents a slightly perturbed normal distribution. 
The fitting employs a $3$-qubit variational quantum circuit with depth $1$depicted in in Fig.~\ref{fig:normal_init} and is based on a least squares loss function.
More specifically, we minimize the distance between the measurement probabilities $p_i\left(\bm{\omega}\right)$ of the circuit output and the probability density function (PDF) of a discretized normal distribution $q_i$ as
\begin{equation}
	\min_{\bm{\omega}} \sum\limits_i \left\|p_i\left(\bm{\omega}\right) - q_i \right\|^2.
\end{equation}
The trained circuit parameters for the log-normal distribution are,
\begin{equation*}
	\bm{\omega}_\text{{ln}}= \left[0.3580, 1.0903, 1.5255, 1.3651, 1.4932, -0.9092\right], 
\end{equation*}
for the triangular distribution read,
\begin{equation*}
\begin{split}
	\bm{\omega}_\text{{tr}} = \left[1.5343, 1.6183, 0.8559, -0.4041, 0.4953, 1.2238\right]
	\end{split}
\end{equation*}
 and for the bimodal distribution correspond to,
 \begin{equation*}
 	\bm{\omega}_\text{{bm}} = \left[0.4683, 0.8200, 1.4512, 1.1875, 1.3883, -0.8418\right].
 \end{equation*}
\begin{figure}[h!]
\captionsetup{singlelinecheck = false, format= hang, justification=raggedright, font=footnotesize, labelsep=space}
\begin{center}
\begin{tikzpicture}
   \node at (-5.5,2) {$\ket{0}^{\otimes 3}$};
\draw[decorate, thick, decoration = {brace, amplitude=15pt}] (-4,0.2) --  (-4,3.7);
\node at (0,2){\includegraphics[width=0.5\textwidth]{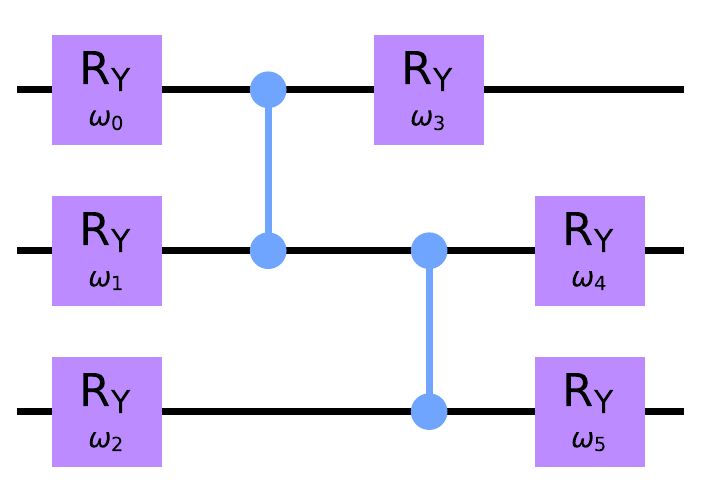}};
\end{tikzpicture}
\end{center}
\caption{This variational quantum circuit is used for approximate quantum loading of a normal distribution. To reduce the number of gates required for initialization with a normal state, the parameters of the illustrated circuit are fitted to load an approximate discretized normal distribution into a quantum state.}
\label{fig:normal_init}
\end{figure}
To create a randomly chosen distribution, we set $\ket{\psi_{\text{in}}}=\textstyle{\ket{0}^{\otimes 3}}$ and initialize the parameters of the variational quantum circuit following a uniform distribution on $[-\pi, \pi]$.
From now on, we refer to these three cases as \emph{uniform}, \emph{normal}, and \emph{random} initialization.
Since, we chose a variational quantum circuit with $RY$ and $CZ$ gates the circuit does not have any effect on the state amplitudes of the $\ket{\psi_{\text{in}}}$ but only flips the phases for $\omega_{j} = 0, \:\forall j$. This is particularly interesting if the initial state can be expected to be close to the target state, e.g., because it is initialized with respect to the first and second momentum of the training data. In other words, this feature is important to exploit the potential training advantage of informed initial state choices.

In the following examples, the discriminator is implemented with PyTorch \cite{pytorch}.
The neural network consists of a $3$-node input layer, a $50$-node hidden-layer layer, a $20$-node hidden-layer and a single-node output layer. 
First, the input and the hidden layer apply linear transformations followed by Leaky ReLU functions \cite{goodfellow, Pedamonti_2018_Non-linear_Activation_Functions_NN, HeRectifiers15}.
Then, the output layer implements another linear transformation and applies a sigmoid function.
The network should neither be too weak nor too powerful to ensure that neither the generator nor the discriminator overpowers the other network during the training.
The choice for the discriminator topology is based on empirical tests.

The qGAN is trained for $2000$ iterations using AMSGRAD \cite{amsgrad} with the initial learning rate being $10^{-4}$.
Due to the utilization of first and second momentum terms, this is a robust optimization technique for non-stationary objective functions as well as for noisy gradients \cite{Kingmaadam14} which makes it particularly suitable for evaluations with real quantum hardware.
The training stability is improved further by applying a gradient penalty on the discriminator's loss function \cite{Kodali2017OnCA, Roth2017StabilizingTO}.
In each training epoch, the training data is shuffled and split into batches of size $2000$.
The generated data samples are created by preparing and measuring the quantum generator $2000$ times.
Then, the batches are used to update the parameters of the discriminator and the generator in an alternating fashion. After the updates are completed for all batches, a new epoch starts.

Fig.~\ref{fig:benchmark} illustrates the trained and the target probability distributions as well as the loss function throughout the training for 
\begin{enumerate}
    \item the log-normal distribution with normal initialization and a depth $2$ variational circuit acting as generator,
    \item the triangular distribution with random initialization and a depth $2$ quantum generator,
    \item and the bimodal distribution with uniform initialization and a depth $3$ ansatz.
\end{enumerate}
The plots show that the investigated qGAN trainings lead to generator statistics which are close to a discretized form of the target distribution and, thus, perform well. 
Furthermore, the plots reveal that the loss functions not always converge to the same value but still result in good models.

\begin{figure}[!h]
    \centering{
    \begin{tikzpicture}
\node at (1, 0.5) {\textbf{qGAN Benchmarking}};
\node[inner sep=0pt, anchor=north west] at (-5.5, -0.12) {\includegraphics[width=0.35\textwidth]{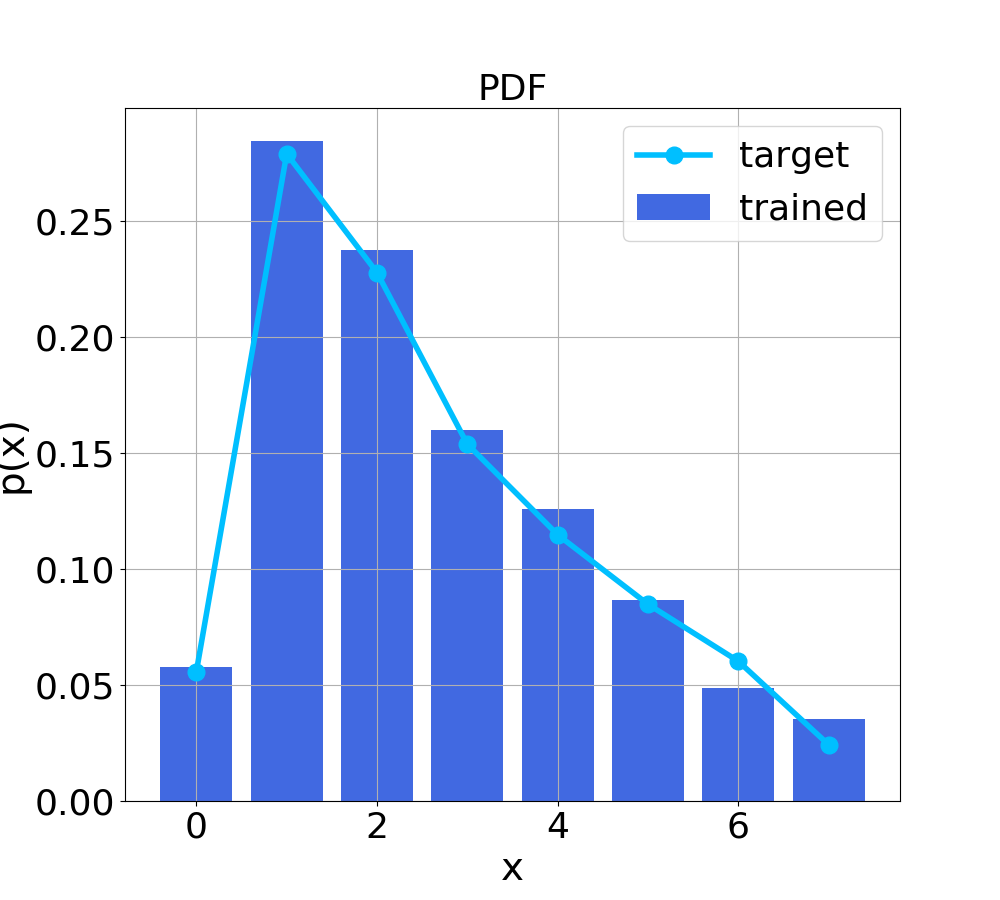}};
\node[anchor=north west] at (1,0) {    \includegraphics[width=0.417\textwidth]{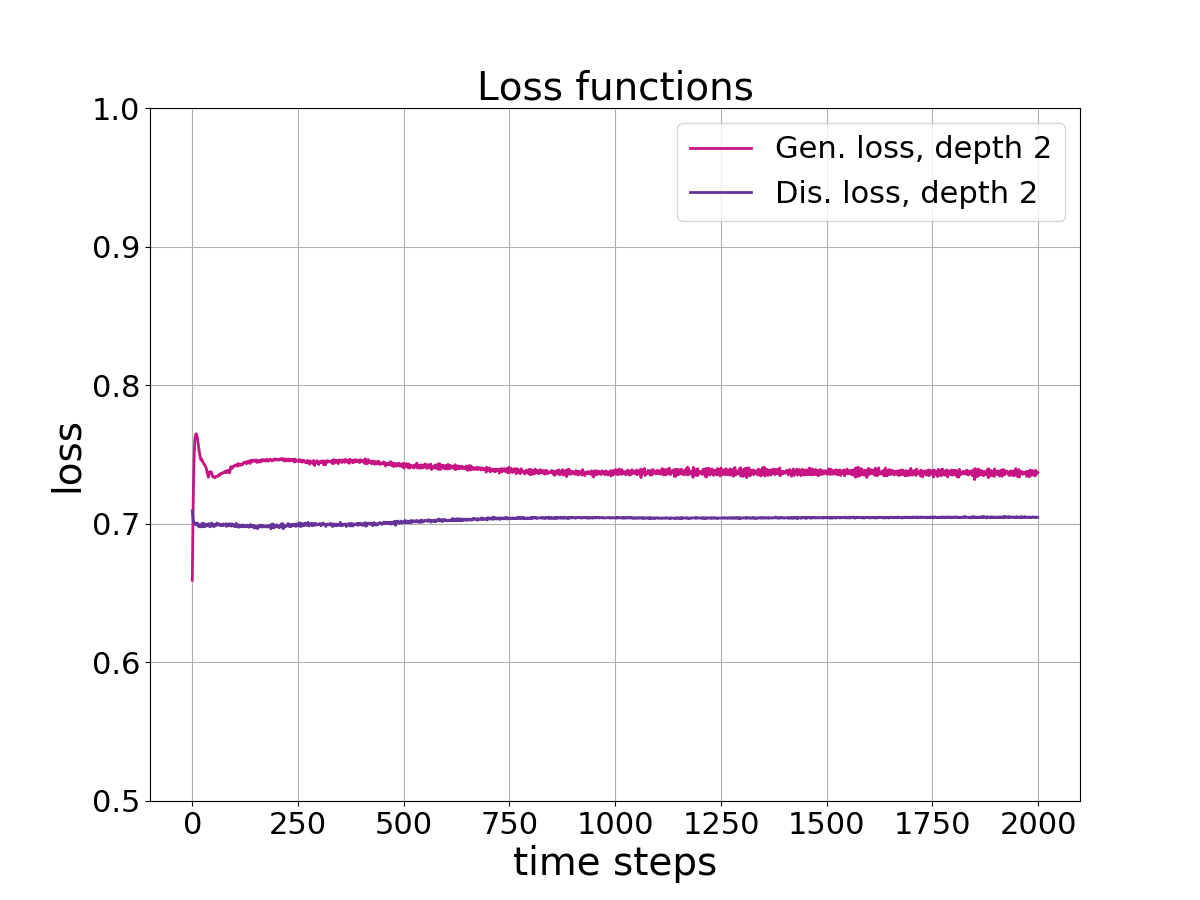}};
\node[anchor=north west] at (-5.5, -5) {\includegraphics[width=0.35\textwidth]{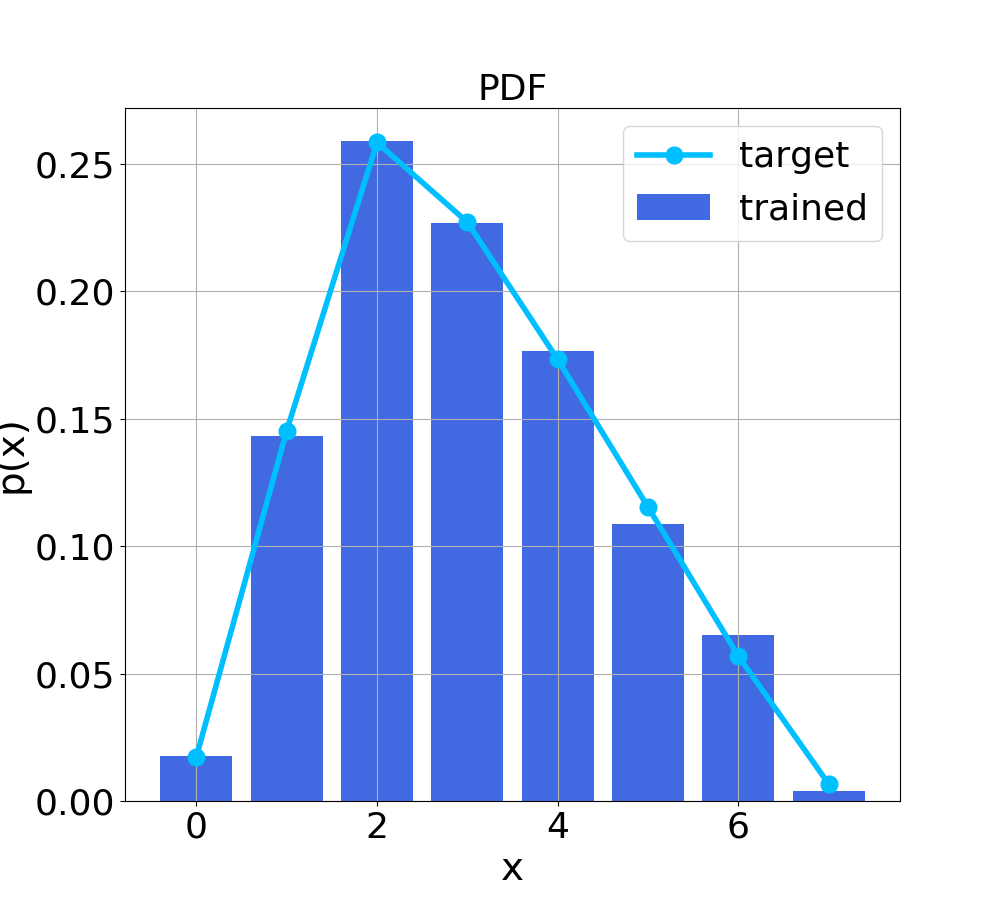}};
\node[anchor=north west] at (1, -5.03) {\includegraphics[width=0.417\textwidth]{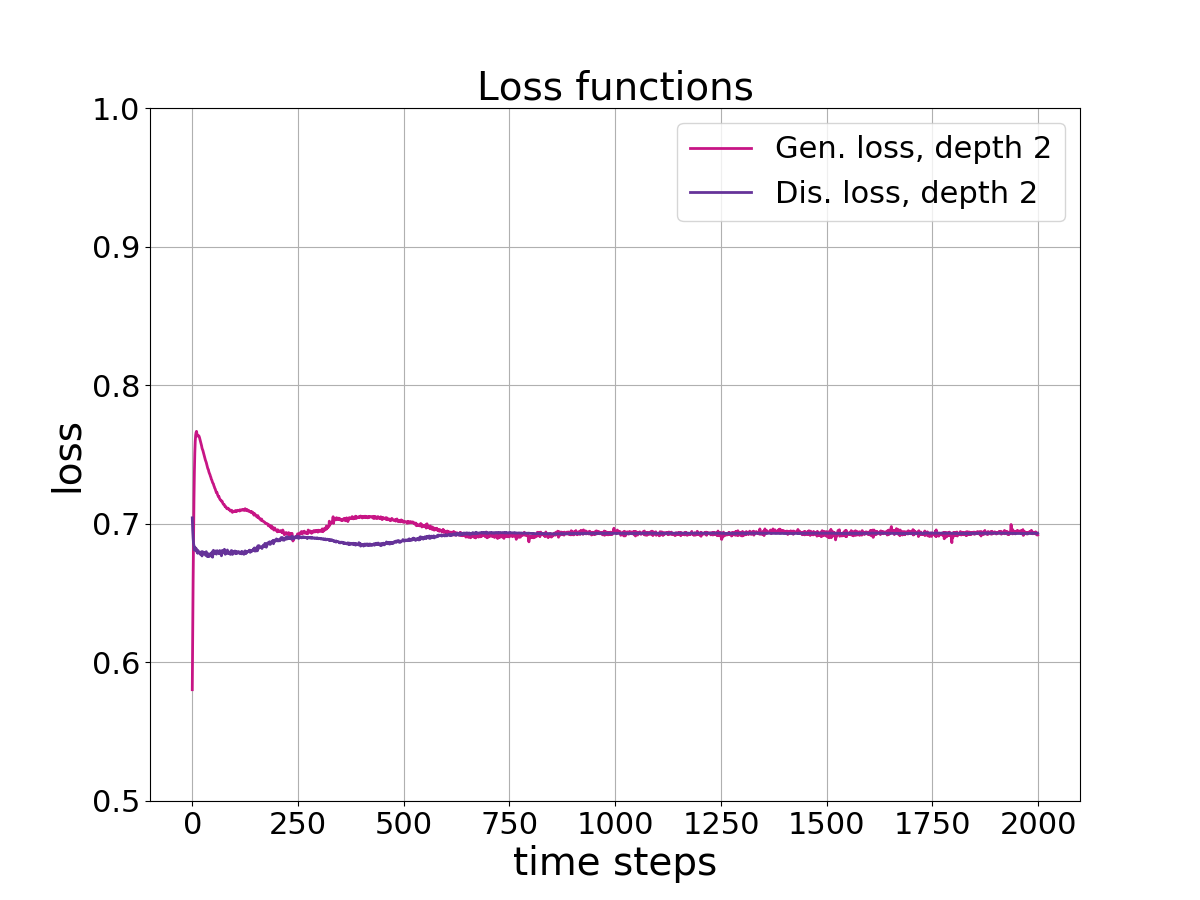}};
\node[anchor=north west] at (-5.5, -10)
{\includegraphics[width=0.35\textwidth]{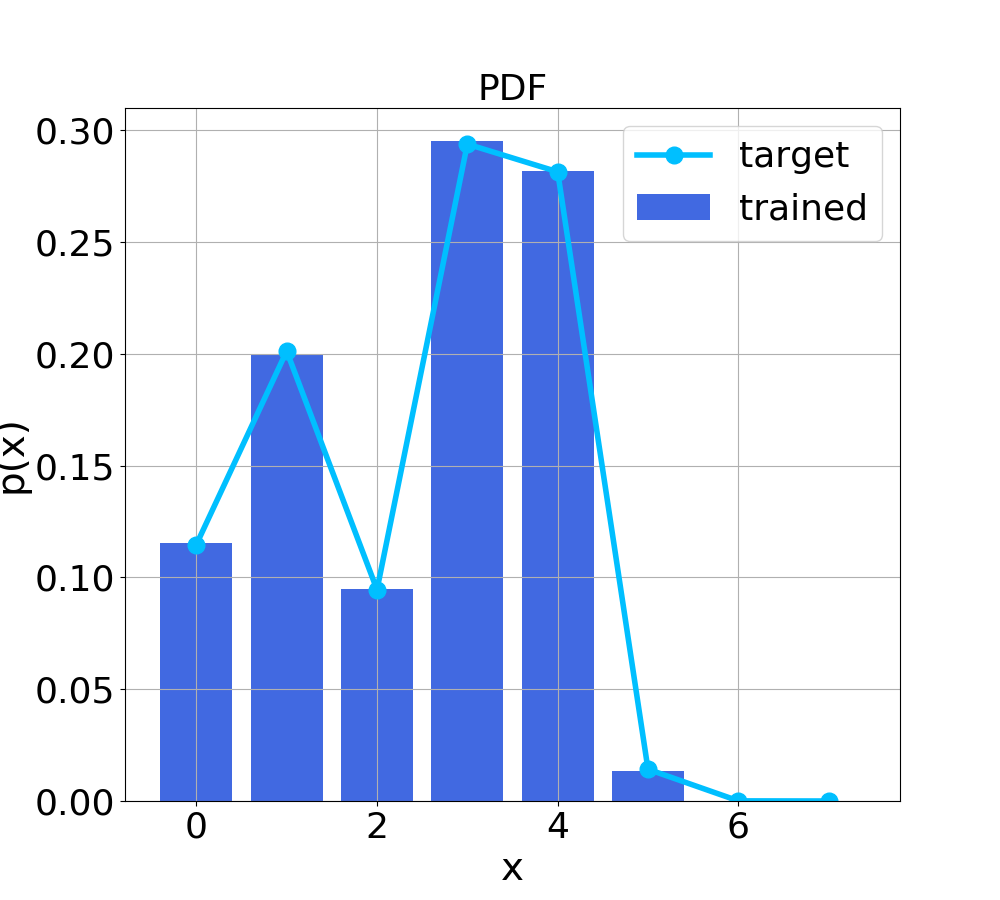}};
\node[anchor=north west] at (1,-10.03) {\includegraphics[width=0.417\textwidth]{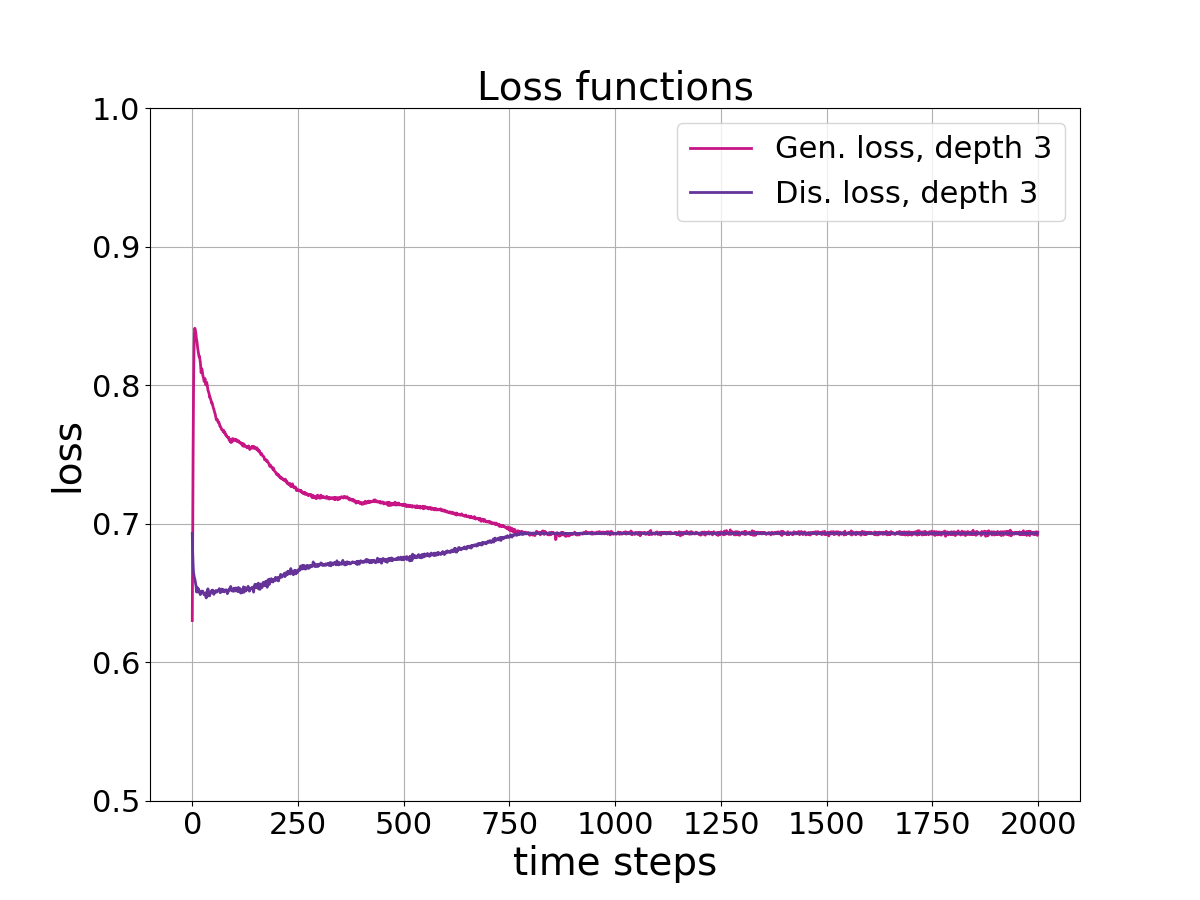}};
\node at (-5, -0.2) {(a)};
\node at (1.7, -0.2) {(b)};
\node at (-5, -5.2) {(c)};
\node at (1.7, -5.2) {(d)};
\node at (-5, -10.2) {(e)};
\node at (1.7, -10.2) {(f)};
\end{tikzpicture}
}
    \caption{
    (a, b) Log-normal target distribution with normal initialization and a depth $2$ generator. 
    (c, d) Triangular target distribution with random initialization and a depth $2$ generator.
    (e, f) Bimodal target distribution with uniform initialization and a depth $3$ generator.
    (a, c, e) The presented probability density functions correspond to the trained $\ket{g^{\omega}}$. 
    (b, d, f) The loss function progress is illustrated for the generator as well as for the discriminator.}
    \label{fig:benchmark}
\end{figure}
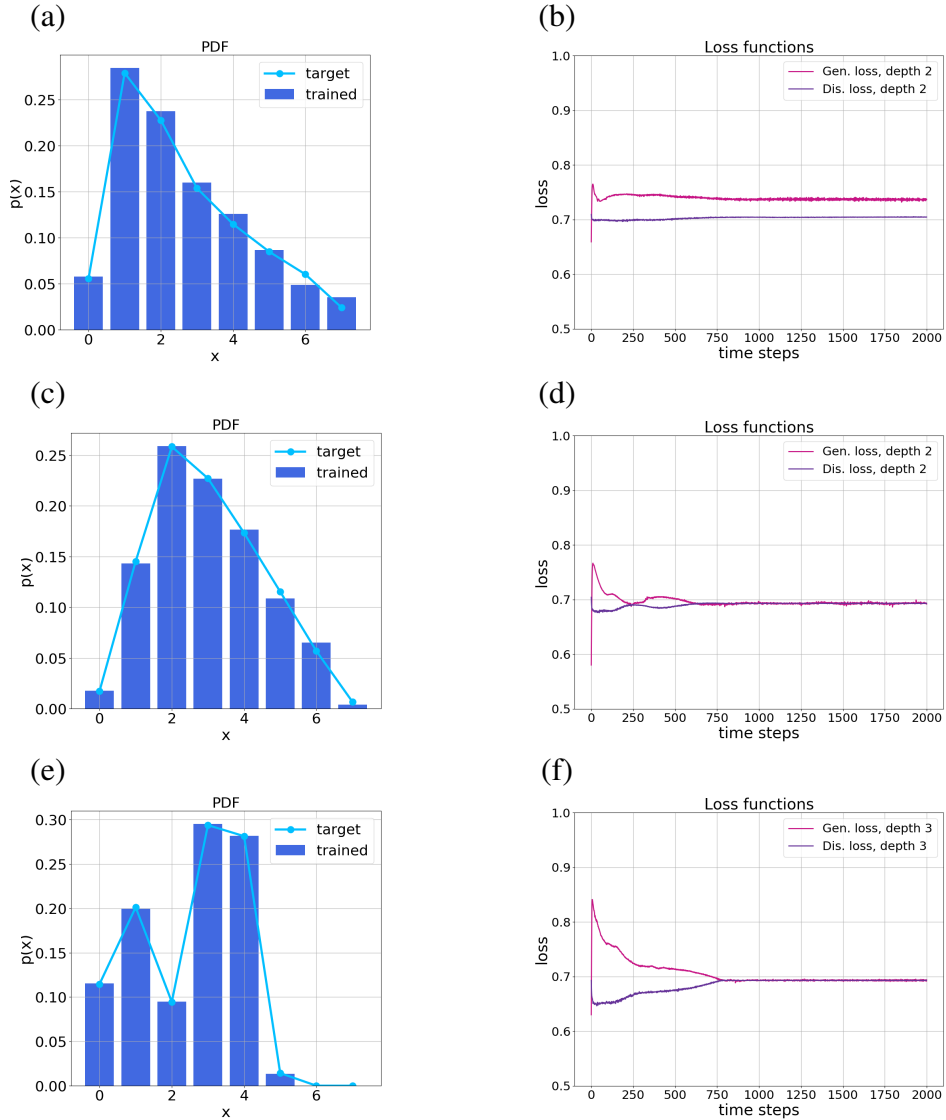


In fact, classical GAN literature discusses that the loss functions do not necessarily reflect whether the method converges \cite{Grnarova2018EvaluatingGV}. Instead, measures such as the Kolmogorov-Smirnov statistic and the relative entropy, which are introduced in Sec.~\ref{sec:data_encoding}, can provide information about the closeness of the trained and the target probability distribution. 
Here, the Kolmogorov-Smirnov statistic is used as a goodness-of-fit test for a confidence level of $95\%$.
Given the null-hypothesis that the prepared and target distribution are equivalent, we draw $s=500$ samples from both distributions and choose a confidence level $(1 - \alpha)$ with $\alpha = 0.05$.
The null-hypothesis is rejected if the corresponding Kolmogorov-Smirnov statistic is bigger than $0.0859$.

\begin{table}[!htb]
\centering{
\small{
\begin{tabular}{c|c|c|c|c|c|c|c}
\textbf{Data }& \textbf{Init} & $\bm{k}$	& $\bm{\mu_{KS}}$ &	$\bm{\sigma_{KS}}$	& $\bm{n_{\leq b}}$ & $\bm{\mu_{RE}}$	& $\bm{\sigma_{RE}}$ \\
\hline
\multirow{9}{* } {log-normal}& \multirow{3}{*}{uniform} & 1	& 0.0522	& 0.0214 &	9	&0.0454	&0.0856\\
& & 	2 &	0.0699  & 0.0204	 & 7	 & 0.0739	 & 0.0510 \\
 & &	3 &0.0576	 &0.0206	 &9	 &0.0309	 &0.0206 \\

& \multirow{3}{*}{normal}  &	1	 &0.1301	 &0.1016 &	5	 &0.1379 &	0.1449\\
& &2	 &0.1380 &	0.0347	 &1 &	0.1283	 &0.0716\\
& &	3	 &0.0810 &	0.0491	 & 7 &	0.0435 &	0.0560 \\
&\multirow{3}{*}{random}	 &1	 &0.0821	 &0.0466	& 7	 &0.0916 &	0.0678\\
& &	2	 &0.0780 &	0.0337	 &6	 &0.0639	 &0.0463\\
& &	3	 &0.0541 &	0.0174 &	10 &	0.0436	 &0.0456\\
\hline
\multirow{9}{* } {triangular	}&\multirow{3}{*}{uniform}&	1&	0.0880&	0.0632	&6	&0.0624&	0.0535\\
&&	2&	0.0336	&0.0174	&10&	0.0091&	0.0042\\
&&	3	&0.0695	&0.1028&	9	&0.0760 &	0.1929\\
&\multirow{3}{*}{normal} 	&1&	0.0288	&0.0106	&10	&0.0038&	0.0048\\
&&	2&	0.0484	&0.0424&	9	&0.0210&0.0315\\
&&	3&	0.0251	&0.0067& 10&	0.0033	&0.0038 \\
&\multirow{3}{*}{random}	 &	1	&0.0843&	0.0635	&7&0.1050	&0.1387\\
&&	2&0.0538&	0.0294	&9&	0.0387	&0.0486\\
&&	3	&0.0438	&0.0163&	10	&0.0201	&0.0194\\
\hline
\multirow{9}{* } {bimodal}&\multirow{3}{*}{uniform}	&1	&0.1288	&0.0259	&0&0.3254	&0.0146\\
&&	2	&0.0358	&0.0206	&10	&0.0192&	0.0252\\
&&	3	&0.0278	&0.0172&	10	&0.0127	&0.0040\\
&\multirow{3}{*}{normal} 	&1&	0.0509	&0.0162	&9&	0.3417&	0.0031\\
&&	2	&0.0406	&0.0135	&10	&0.0114&	0.0094\\
&&	3	&0.0374&	0.0067	&10	&0.0018&0.0041\\
&\multirow{3}{*}{random}		&1	&0.2432&	0.0537	&0	&0.5813&	0.2541\\
&&	2	&0.0279&0.0078	&10	&0.0088&	0.0060\\
&&	3 &	0.0318&0.0133&	10&	0.0070&	0.0069\\

\end{tabular}
}}
\caption{Benchmarking the qGAN Training \\
The table presents results for training a qGAN for log-normal, triangular and bimodal target distributions, uniform, normal and random initializations, and variational circuits with depth $1, 2$ and $3$. The tests were repeated $10$ times using quantum simulation. The table shows the mean $\left(\mu\right)$ and the standard deviation $\left(\sigma\right)$ of the Kolmogorov-Smirnov statistic $\left(\text{KS}\right)$ as well as of the relative entropy $\left(\text{RE}\right)$ between the generator output and the corresponding target distribution. Furthermore, the table shows the number of runs which are not rejected according to the Kolmogorov Smirnov statistic $\left(n_{\leq b}\right)$ with confidence level $95\%$, i.e., with acceptance bound $b=0.0859$.}
\label{tbl:trainingBenchmarking}
\end{table}

For each setting, we repeat the training $10$ times to get a better understanding of the robustness of the results.
Tbl.~\ref{tbl:trainingBenchmarking} shows aggregated results over all $10$ runs and presents the mean $\mu_{\text{KS}}$, the standard deviation $\sigma_{\text{KS}}$ and the number of runs $n_{\leq b}$ that are not rejected by the Kolmogorov-Smirnov statistic as well as the mean $\mu_{\text{RE}}$ and standard deviation $\sigma_{\text{RE}}$ of the relative entropy outcomes between the generator output and the corresponding target distribution.
The data shows that increasing the quantum generator depth $l$ usually improves the training outcomes.
Furthermore, the table illustrates that a carefully chosen initialization can have favorable effects, as can be seen especially well for the bimodal target distribution with normal initialization.
Since the standard deviations are relatively small and the number of not rejected results is usually close to $10$, at least for depth $l \geq 2$, we conclude that the presented approach is quite robust and also applicable to more complicated distributions.

Next, the log-normal example is investigated in more detail. We compare the training results using a shot-based quantum simulation given a quantum generator with depth $l=1$ which is sufficient for this small example, for uniform, normal and random initializations. Fig.~\ref{fig:training} shows the PDFs corresponding to the trained $\ket{g^{\omega}}$ and the target PDF as well as the progress of the relative entropy. The results illustrate that both uniform and normal initialization perform better than the random initialization. Furthermore, the relative entropy plot highlights the fact that a carefully chosen initialization such as the normal initialization can lead to a better starting point for the generator training.
Tbl.~\ref{tbl:ks} presents the Kolmogorov-Smirnov statistics of the experiments.
The results again confirm that initialization impacts the training performance.
The statistics for the normal initialization are better than for the uniform initialization, which itself outperforms random initialization. It should be noted that the null-hypothesis is not rejected for any of the settings.


\begin{figure}
   \centering{
   \begin{tikzpicture}
   \node at (0, 3) {\textbf{Log-Normal Distribution Simulation}};
   \node at (-5,0) {\includegraphics[width=0.3\linewidth]{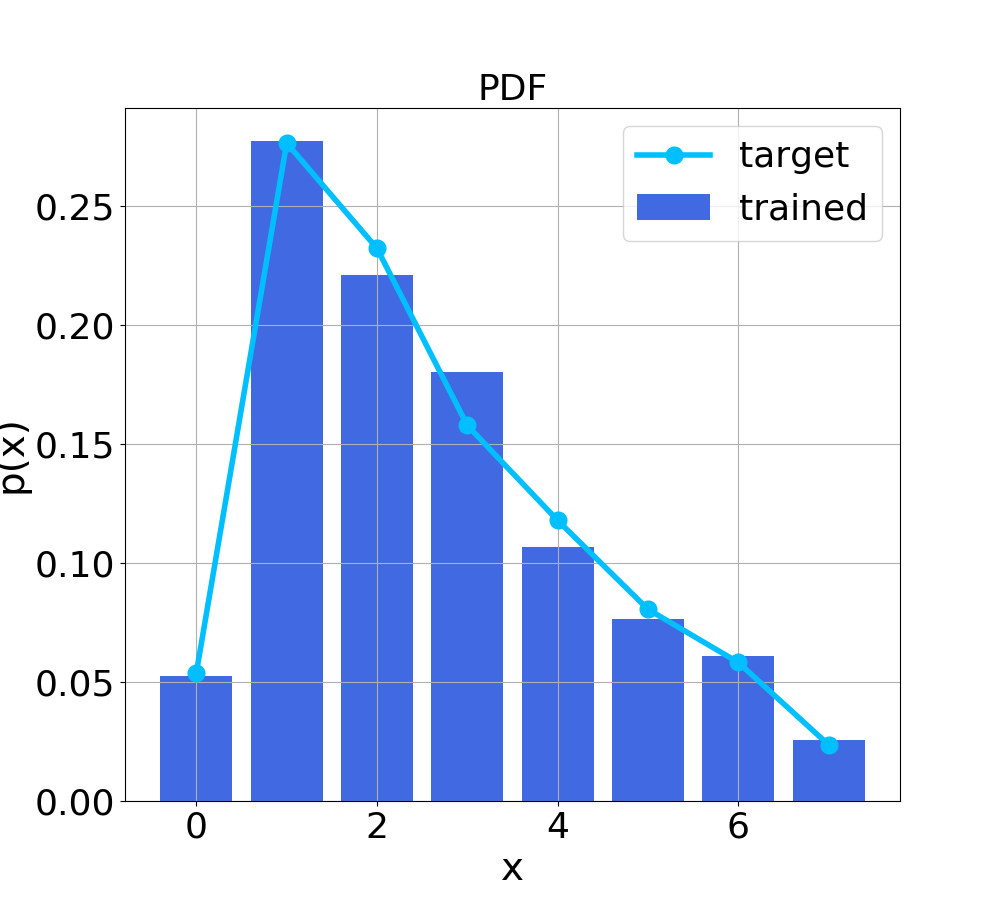}};
   \node at (0,0) {\includegraphics[width=0.3\linewidth]{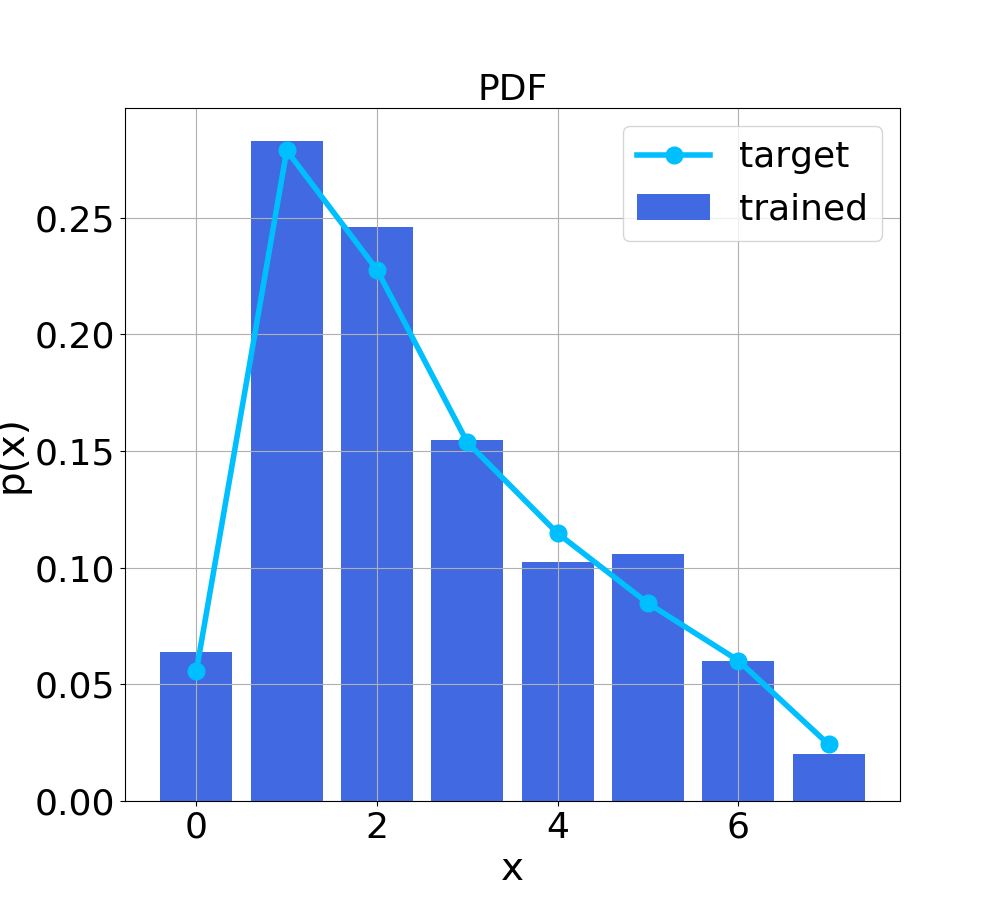}};
   \node at (5,0) {\includegraphics[width=0.3\linewidth]{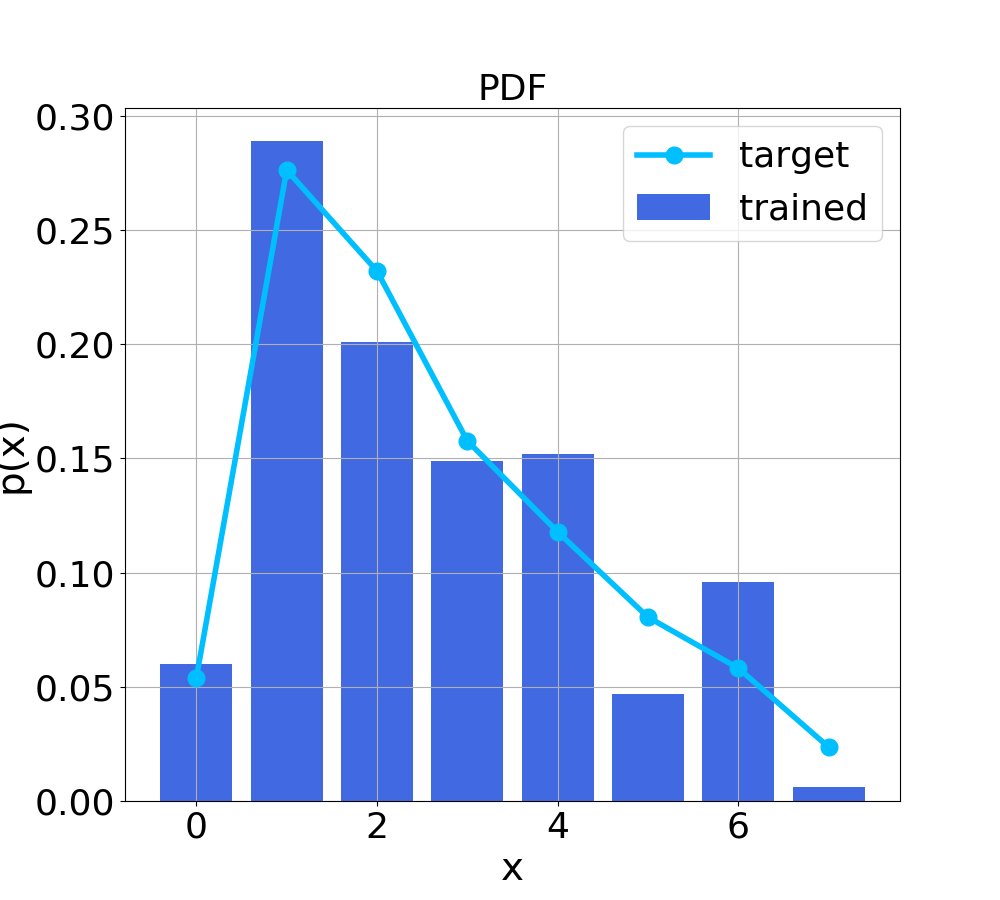}};
    \node at (0,-6) {\includegraphics[width=0.8\linewidth]{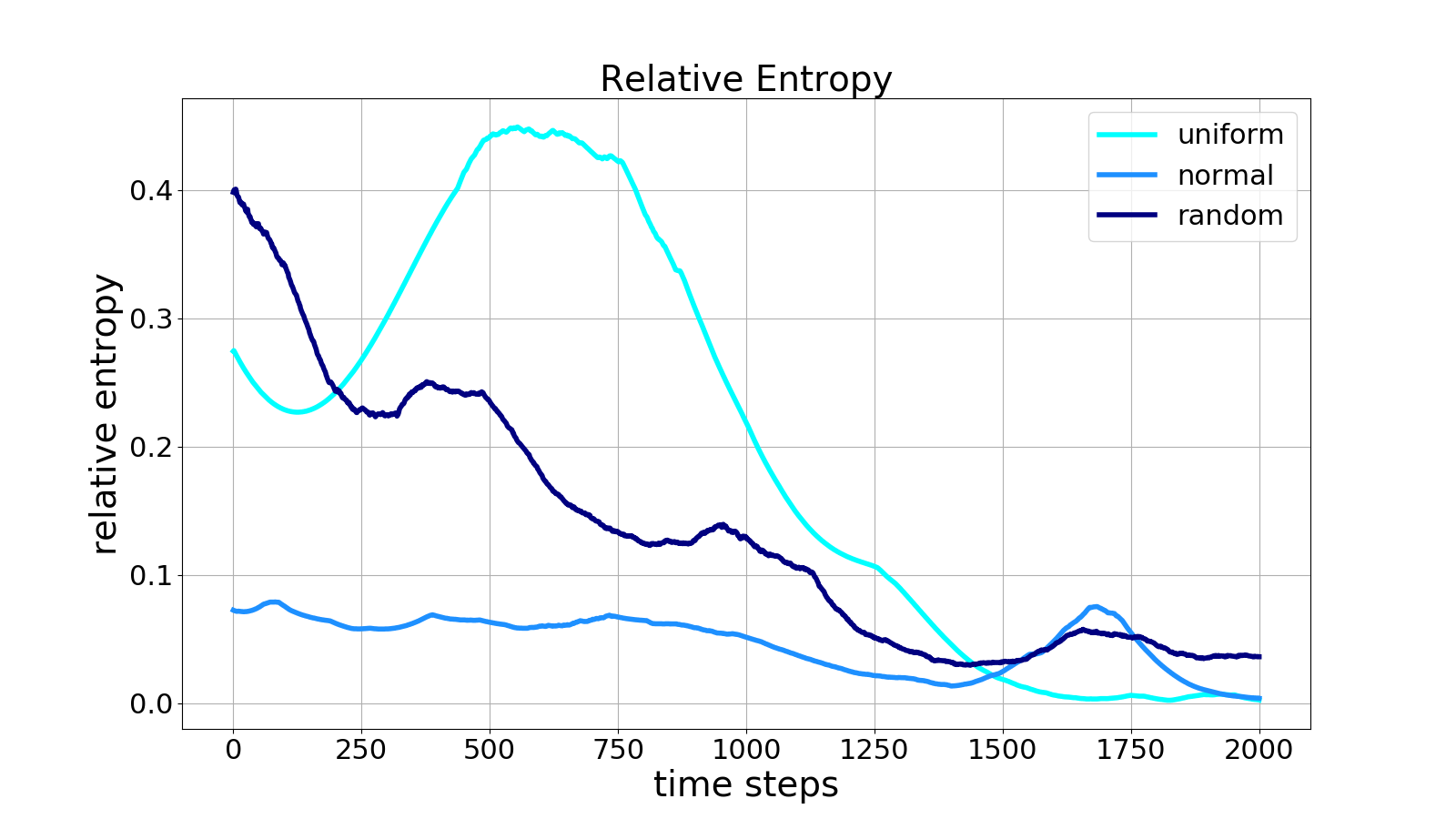}};
\node at (-7,2.) {(a)};
\node at (-2, 2.) {(b)};
\node at (3, 2.) {(c)};
\node at (-4.5, -3) {(d)};
\end{tikzpicture}
}
\caption{The figure illustrates the PDFs corresponding to $\ket{g^{\omega}}$ trained on samples from a log-normal distribution using a (a) uniformly, (b) randomly, and (c) normally initialized quantum generator. Furthermore, (d) presents the  convergence of the relative entropy for the various initializations over $2000$ training epochs.}
 \label{fig:training}
\end{figure}

\begin{table}
\centering{
\begin{tabular}{c|c|c}
\textbf{Initialization} &  $\bm{D}_{\text{\textbf{KS}}}$ & \textbf{Reject} \\
\hline
uniform & $0.0369$ & No  \\
normal  & $0.0320$ & No  \\
random  & $0.0560$ & No
\end{tabular}
}
\caption{
The Kolmogorov-Smirnov statistic is computed for randomly chosen samples from $\ket{g^{\omega}}$ and from the discretized, truncated log-normal distribution ${X}$ using shot-based quantum simulations.}
\label{tbl:ks}
\end{table}


Now, we present a qGAN training run with an actual quantum processor -- the IBM Quantum Boeblingen chip \cite{ibmQX} -- and compare the results with a simulation that implements an idealized model of the hardware noise.
The same log-normal training data, quantum generator and discriminator as before are used. To improve the robustness against the hardware noise, we set the optimizer's learning rate to $10^{-3}$. 
The initialization is chosen according to the random setting because it requires the least gates. 
Due to the increased learning rate, it is sufficient to run the training for $200$ optimization epochs.
Fig.~\ref{fig:pdf_real} presents the PDF corresponding to $\ket{g^{\omega}}$ trained with IBM Quantum Boeblingen respectively with the noisy simulation as well as the respective comparison of the progress of the loss functions and the relative entropy.
The figure illustrates that the generated probability distributions converge towards the random distribution underlying the training data samples for both, the simulation and the real quantum hardware. Notably, some of the more prominent fluctuations might be due to the fact that the IBM Quantum Boeblingen chip is calibrated on a daily basis which is, due to the queuing, circuit preparation, and network communication overhead, shorter than the overall training time of the qGAN.
$D_{\text{KS}}$ determines whether the null-hypothesis is rejected using a confidence level of $95\%$.
The results presented in Tbl.~\ref{tbl:ks_hw} confirm that we are able to train an appropriate model on the actual quantum hardware.

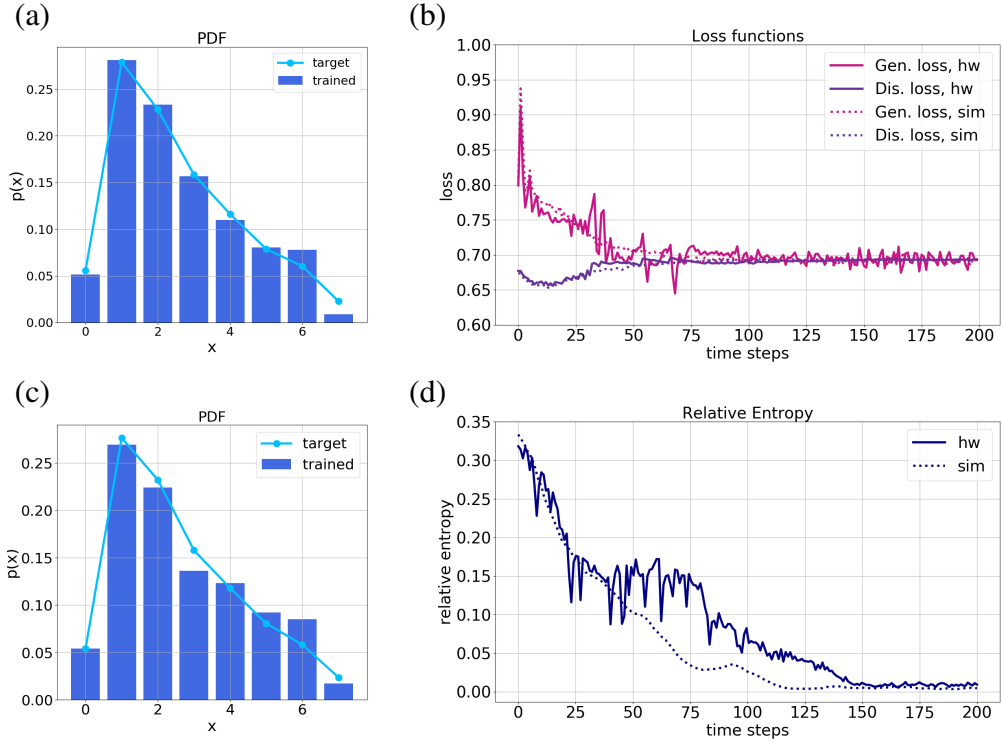
\begin{figure}[h!]
   \centering{
   \begin{tikzpicture}
   \node at (-1, 3.3) {\textbf{Log-Normal Distribution Hardware}};
   \node at (-5,0) {\includegraphics[width=0.35\linewidth]{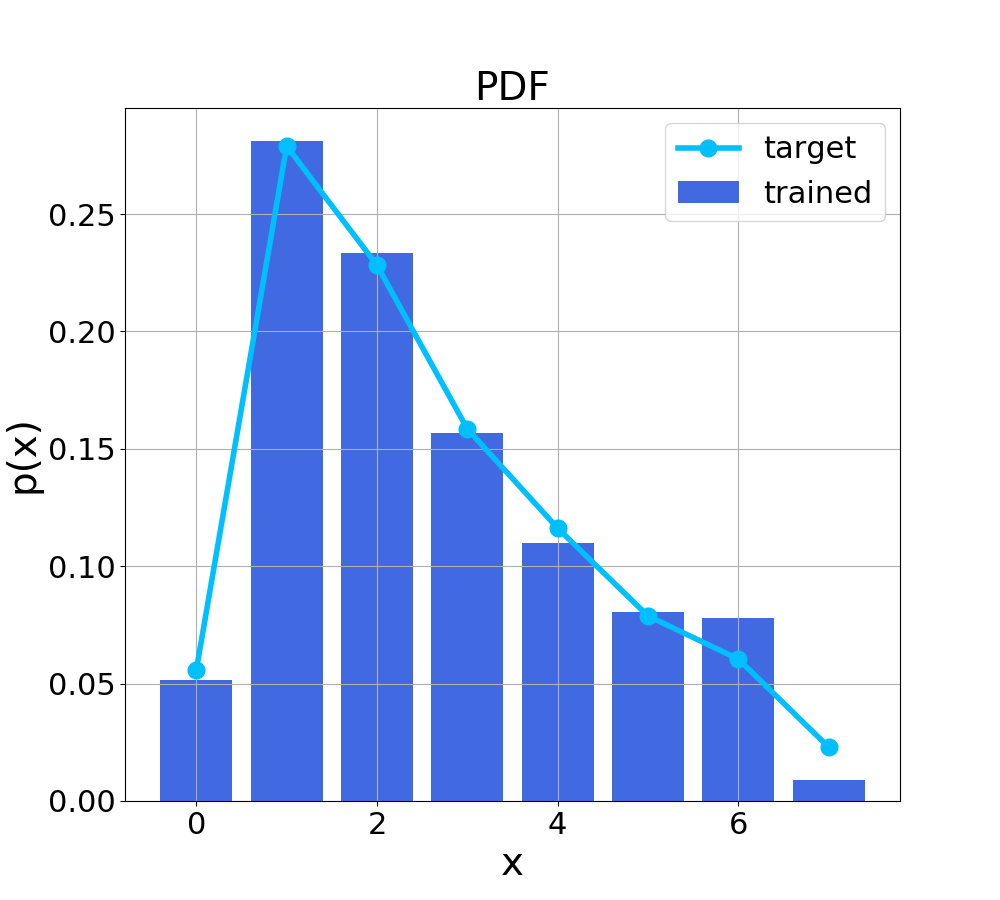}};
   \node at (-5,-5) {\includegraphics[width=0.35\linewidth]{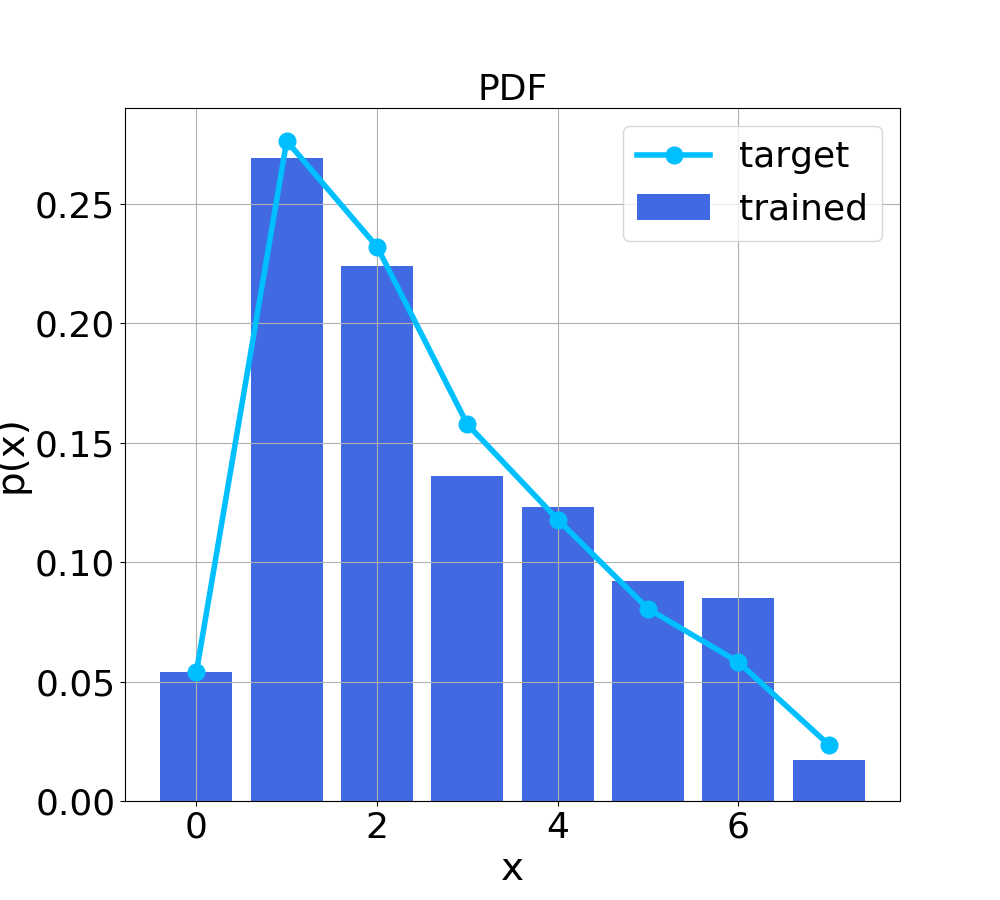}};
   \node at (2,0) {\includegraphics[width=0.57\linewidth]{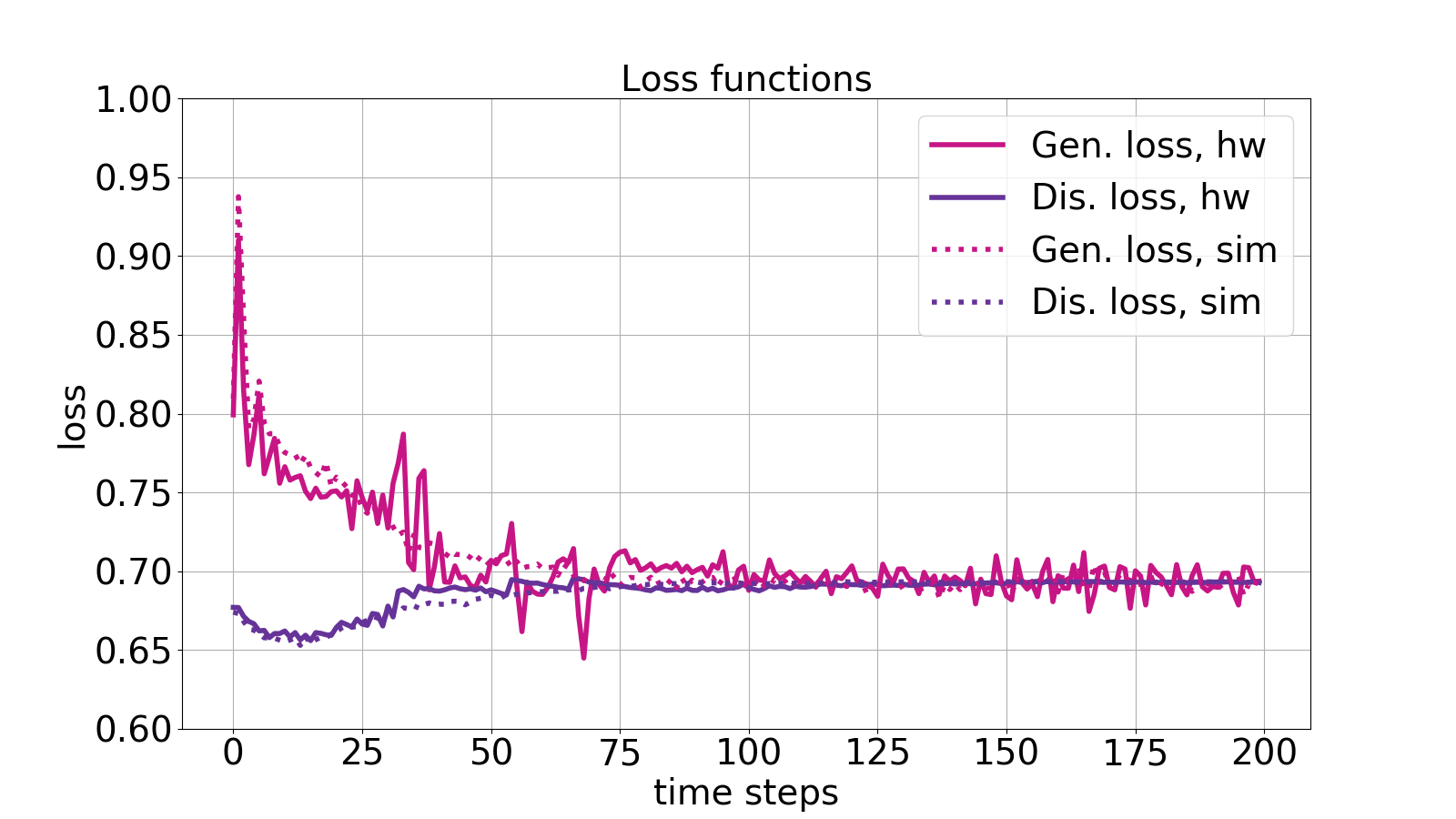}};
    \node at (2,-5) {\includegraphics[width=0.57\linewidth]{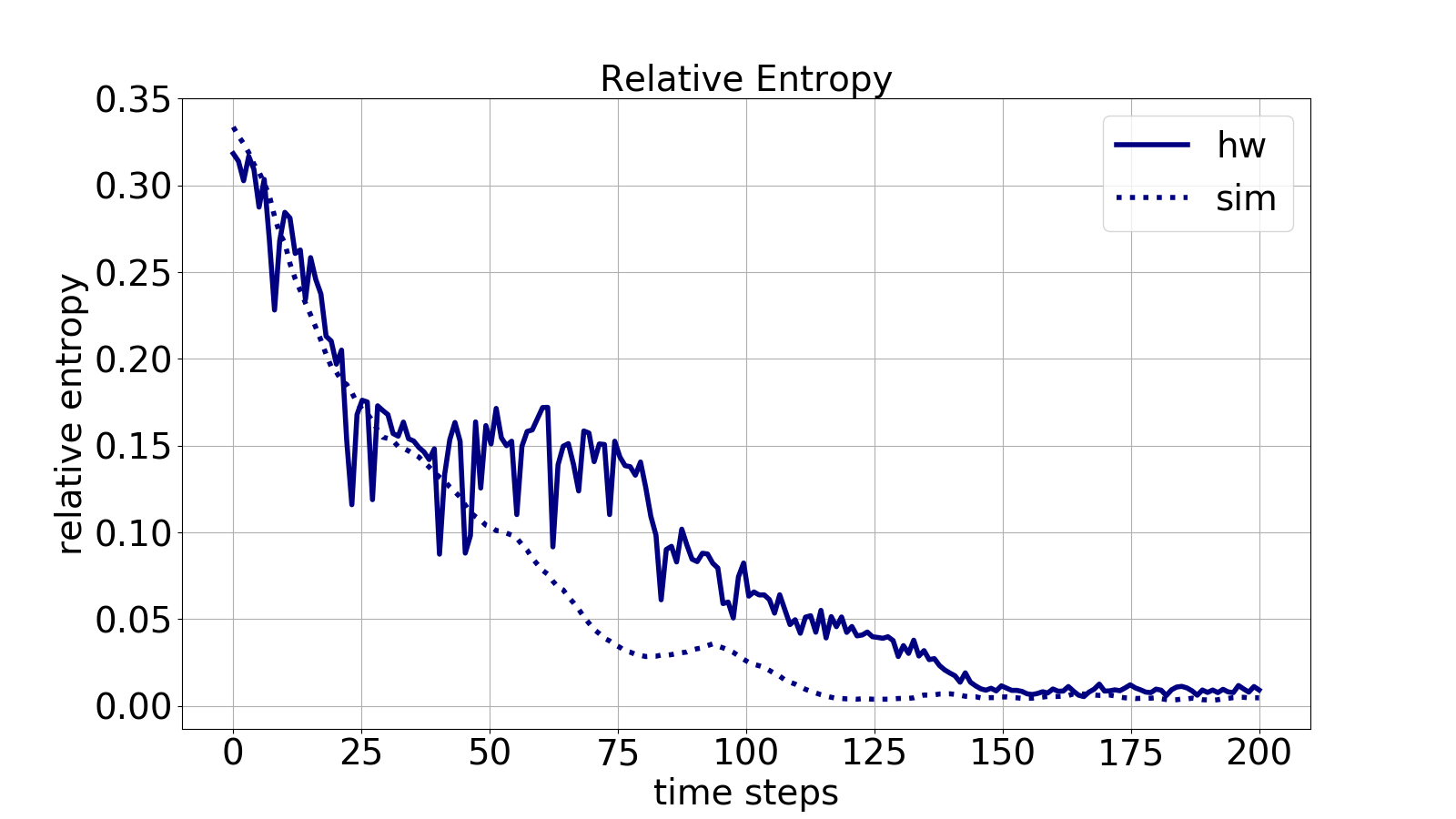}};
\node at (-7.3,2.2) {(a)};
\node at (-2.1, 2.2) {(b)};
\node at (-7.3, -2.8) {(c)};
\node at (-2.1, -2.8) {(d)};
\end{tikzpicture}
}
\caption{The presented PDFs correspond to $\ket{g^{\omega}}$ trained on (a) the IBM Quantum Boeblingen and (c) a quantum simulation with a model of the hardware noise.
Moreover, (b) shows the respective evolution of the loss functions and (d) illustrates the corresponding progress of the relative entropy during the training of the qGAN.
 \label{fig:pdf_real}}
\end{figure}

\begin{table}[!h]
\centering{
\begin{tabular}{c|c|c|c}
\textbf{Initialization} & \textbf{Backend}  &  $\bm{D}_{\text{\textbf{KS}}}$ & \textbf{Reject} \\
\hline
random & simulation & $0.0420$ & No \\
random & quantum computer & $0.0224$ & No
\end{tabular}
}
\caption{
The Kolmogorov-Smirnov statistic is computed for randomly chosen samples of $\ket{g^{\omega}}$ trained with a noisy quantum simulation and using the IBM Quantum Boeblingen device.}
\label{tbl:ks_hw}
\end{table}

Finally, we are going to present a numerical example which is trained on a real world data set, i.e., the first two principle components of multivariate, constant maturity treasury rates of US government bonds.
The training data set $X$ consists of more than $5000$ samples. To reduce the number of required qubits for a reasonable representation of the distribution, data samples smaller than the $5\%-$percentile and bigger than the $95\%-$percentile are discarded. 
Furthermore, depth $l \in \set{2, 3, 6}$ quantum generators are compared which act on $n=6$ qubits, i.e., $3$ qubits per dimension (principle component).
The input state $\ket{\psi_{\text{in}}}$ is prepared as a multivariate uniform distribution and the generator parameters $\bm{\omega}$ are initialized with random draws from a uniform distribution on the interval $[-\delta, +\delta]$ with $\delta = 10^{-1}$.
Here, the classical discriminator is composed of a $6$-node input layer, $512$-node hidden-layer, a $256$-node hidden-layer, and a single-node output layer. 
Equivalently to the discriminator described before, the hidden layers apply linear transformations followed by Leaky ReLU functions \cite{Pedamonti_2018_Non-linear_Activation_Functions_NN} and the output layer employs a linear transformation followed by a sigmoid function.
Furthermore, the optimization uses data batches of size $1200$ and is run for $20000$ training epochs.

The evolution of the relative entropy between the generated and the real probability distribution is shown in Fig.~\ref{fig:rel_ent_multi}. The plot reveals that higher depths, can lead to better results as well as faster convergence but also that the largest depth does not necessarily result in the best model. This indicates that good models can already be found using shallow quantum circuits as ansatz. 
\begin{figure}[!htb]
\captionsetup{singlelinecheck = false, format= hang, justification=raggedright, font=footnotesize, labelsep=space} 
\centering{
\begin{tikzpicture}
\node at(0,3.5){\textbf{Relative Entropy - US Government Bonds}};
\node at(0,0){\includegraphics[width=0.7\linewidth]{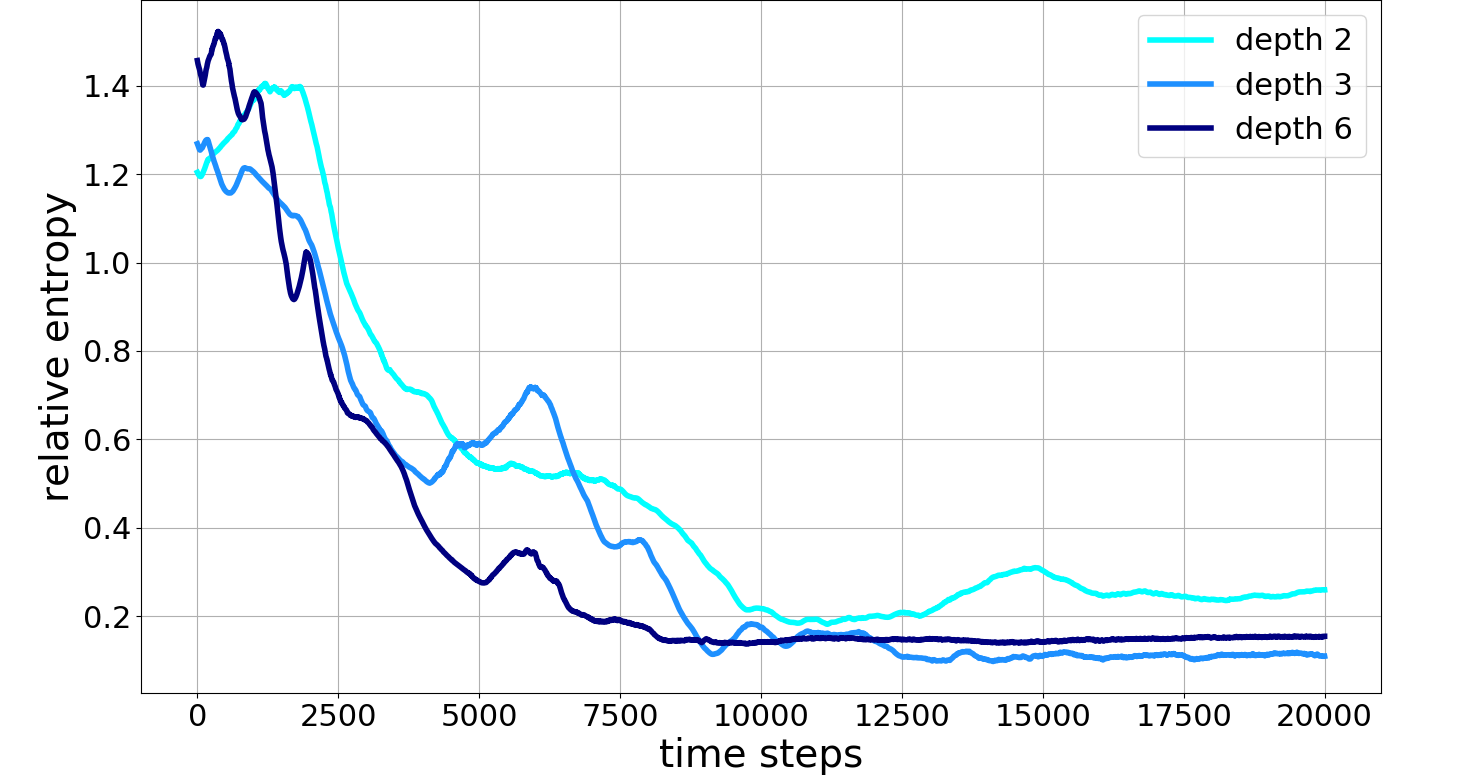}};
\end{tikzpicture}
}

\caption{The progress of the relative entropy between the quantum generator and the multivariate random distribution underlying the training data is shown for quantum generators with depth $l \in \set{2, 3, 6}$.}
 \label{fig:rel_ent_multi}
\end{figure}

%
\subsection{Discussion}
\label{sec:qganconclusion_Outlook}
This section presents qGANs, a generative QML algorithm which combines the use of quantum channels and an established, classical ML algorithm, i.e., GANs.
Moreover, it is described how qGANs are able to facilitate an efficient, approximate probability distribution learning and loading scheme that requires $\mathscr{O}\left(poly\left(n\right)\right)$ many gates. The efficiency of this method is illustrated on various examples which are run with simulators as well as actual quantum hardware.

One of the main open questions that is subject to future research is the analysis of optimal quantum generator structures as well as training strategies.
Like in classical machine learning, it is neither a priori clear what model structure is the most suitable for a given problem nor what training strategy may achieve the best results.

Another interesting subject for future research is the investigation of different loss functions for qGANs.
Classically, the Wasserstein distance \cite{arjovsky2017wassersteingan} has gained a lot of popularity as it tends to lead to stable GAN trainings. A quantum-equivalent to the successful classical Wasserstein distance \cite{kiani2021quantumEarthMover, chakrabarti2019wassersteinqGAN} can be used to achieve beneficial effects when learning the structure underlying quantum data with a qGAN \cite{chakrabarti2019wassersteinqGAN}. However, the evaluation of a quantum Wasserstein distance  is more involved. In practice, it may make sense to try different distance measures and study the trade-off between complexity and performance. 

Furthermore, one would intuitively expect that deeper quantum circuits which have more parameters and are, therefore, more expressive are a more desirable ansatz choice than shallow circuits. However, it is shown in \cite{holmes2021AnsatzExpressBarrenPlateaus} that the loss landscape induced by highly expressive ans\"atze is very flat and the respective models are, therefore, difficult to train. In the worst case, the loss function can flatten \cite{arrasmith2020effectbps_grad_freeOpt} and the gradients vanish \cite{Clean_2018_BarrenPlateaus} exponentially in the system size. Further details on vanishing gradients and related QML training issues as well as possible remedies are discussed in Sec.~\ref{sec:vanishing_grads}. In this context, it would be interesting to design and study problem-specific ans\"atze.


\section[Quantum Boltzmann Machines]{Quantum Boltzmann Machines\footnote{This section is reproduced in part, with permission, from C.~Zoufal, A.~Lucchi, S.~Woerner, "Variational Quantum Boltzmann Machines", \textit{Quantum Machine Intelligence}, vol. 3, Article Nr. 7, 2021}}
\label{sec:QBM}

Quantum Boltzmann machines are energy-based quantum machine learning models. \linebreak
More explicitly, the model information is encoded into a parameterized Hamiltonian, see Hamiltonian encoding in Def.~\ref{sec:gibbs_related_work}, whose structure can be chosen according to the problem structure \cite{BresslerMRFSSTructure13, piatkowski2019exponential, discreteProbDistApproxChow68}. Furthermore, the Hamiltonian is related to a probability distribution at thermal equilibrium via the so-called \emph{Gibbs state}. The aim of a quantum Boltzmann machine algorithm is to train the Hamiltonian parameters such that the sampling statistics of the respective Gibbs state approximate a target distribution.

This section introduces classical Boltzmann machines as well as the quantum counterpart, discusses various quantum Boltzmann machine implementations, and gives a detailed explanation of a quantum Boltzmann machine approach that relies on variational quantum imaginary time evolution for Gibbs state preparation. Lastly, we demonstrate the feasibility of variational quantum Boltzmann machines for generative modelling on various numerical examples.

\subsection{Introduction}
\label{sec:qbm_intro}

Boltzmann Machines (BMs) \cite{HintonBM1985, Du2019BM} offer a powerful framework for modelling probability distributions. In fact, they are known to be universal approximators of probability distributions \cite{montufar2015deep, YOUNES1996109}. 
BMs are widely applicable. Furthermore, they can even be formulated as discriminative as well as generative learning models \cite{Liu2010}.
Applications have been studied in a large variety of domains such as the analysis of quantum many-body systems, statistics, biochemistry, social networks, signal processing and finance, see, e.g., \cite{TorlaiNeuralNetwork18, CarleoRBMsQManyBody18,  CarleoRBMsQuantumManyBody17, YusukeRBM17, AnshuSample-efficientQManyBody, Melko2019, HRASKO2015RBMTimeSeries, Tubiana19RBMProteins, LiuRBMsSocialNetworks13, Mohamed10RBMSignal, Assis18RBMFin}. 

These types of neural networks use an undirected graph-structure to encode relevant information. 
More precisely, the respective information is stored in bias coefficients and connection weights of network nodes. The nodes are typically related to binary spin-systems and grouped into those that determine the output, the visible nodes, and those that act as latent variables, the hidden nodes.
Furthermore, the network structure is linked to an energy function which facilitates the definition of a probability distribution over the possible node configurations by using a concept from statistical mechanics, i.e., Gibbs states \cite{Boltzmann1877, gibbs02}.
The aim of Boltzmann machine training is to learn a set of weights such that the resulting model approximates a target probability distribution that may be implicitly given by training data $X$. 
However, Boltzmann machines are complicated to train in practice because the loss function's derivative requires the evaluation of a normalization factor, the partition function, that can become exponentially expensive and, thus, render the model infeasible.
Usually, the partition function is approximated using Markov Chain Monte Carlo methods which are prone to requiring long runtimes until convergence \cite{Hinton05CD, murphy2012machinelearning} or approximate gradient estimation using contrastive divergence \cite{Hinton2002TrainingPO} or pseudo-likelihood \cite{Besag1975} -- potentially leading to inaccurate results \cite{Tieleman08, Sutskever10}.

Quantum Boltzmann Machines (QBM) \cite{QBMAmin18} are a natural adaption of BMs to the quantum computing framework. They have the potential to exploit the intrinsic normalization of quantum states and may, thereby, help to avoid the evaluation of the exponentially expensive partition function.
Furthermore, QBMs are compatible with a larger class of Hamiltonians than their classical counterpart. More explicitly, classical simulation, e.g., based on quantum Monte Carlo methods \cite{Troyer05}, can suffer from the NP-hard \cite{LeeuwenTheoreticalCS90} so-called \emph{sign-problem} when computing an expectation value with respect to a QBM Hamiltonian  \cite{Hangleiter2019EasingTM, OkunishiSignProblem14, Li2016SignProblemFreeQMC, SignProblemAlet16, PhysRevB.91.241117}. Conducting the respective evaluation on a quantum computer instead could potentially help to avoid this problem \cite{OrtizQAFermionic01}.
Instead of an energy function with nodes being represented by binary spin values, QBMs define the underlying network using a Hermitian operator, a parameterized, $n$-qubit Hamiltonian
$H\left({\bm{\theta}}\right)$.
The respective qubits can be grouped into those which determine the model output, the \emph{visible} qubits, and those which act as latent variables the \emph{hidden} qubits.
The aim of the model is to learn Hamiltonian parameters such that the resulting Gibbs state reflects a given target system.


Unfortunately, several quantum Boltzmann machine implementations \cite{QBMAmin18, Anschtz2019RealizingQB, QBMWiebe17, Kappen18QBM, Wiebe2019GenerativeTO} are incompatible with efficient evaluation of the loss function's analytic gradients if the given model has hidden units and $\textstyle{\exists j: \:\left[H\left({\bm{\theta}}\right), \frac{\partial H\left({\bm{\theta}}\right)}{\partial\theta_j}\right] \neq 0}$.
Instead, the use of hidden qubits is either avoided, i.e., only fully-visible settings are considered \cite{QBMWiebe17, Kappen18QBM, Wiebe2019GenerativeTO}, or the gradients are computed with respect to an upper bound of the loss \cite{QBMAmin18, Anschtz2019RealizingQB, QBMWiebe17} that is based on the
Golden-Thompson inequality \cite{ThompsonInequality1965, Golden65Helm}.
It should be noted that training with an upper bound, renders the use of transverse Hamiltonian components, i.e., off-diagonal Pauli terms, difficult and imposes restrictions on the compatible models.

Therefore, we are going to present the details of a variational QBM implementation \cite{VarQBMZoufal20} that is compatible with generic Hamiltonians $H\left({\bm{\theta}}\right)$ for $\bm{\theta}\in \mathbb{R}^p$ and arbitrary Pauli terms $h_i$, and with near-term, gate-based quantum computers.
The respective method exploits \emph{variational quantum imaginary time evolution} \cite{VarSITEMcArdle19, Simon18TheoryVarQSim}, see Sec.~\ref{sec:varqite}, which is based on McLachlan's variational principle \cite{McLachlan64}, to not only prepare approximate Gibbs states, see Sec.~\ref{sec:varqite_gibbs}, but also to train the model with gradients of the actual loss function, see Sec.~\ref{sec:varqite_chainRule}.
This variational QBM algorithm (\varqbm) is inherently normalized which implies that the training does not require the explicit evaluation of the partition function. 
The following discussion is focused on the training of quantum Gibbs states whose sampling behavior reflects a classical probability distribution. Notably, the scheme could be easily adapted to approximate a quantum state instead by using the quantum relative entropy as a loss function \cite{QBMWiebe17, Kappen18QBM, Wiebe2019GenerativeTO}.

The remainder of this section is structured as follows. Firstly, we review classical Boltzmann machines in Sec.~\ref{sec:BM}.
Then, we explain QBMs and outline \varqbm{} in Sec.~\ref{sec:QBM}. Next, we present illustrate applications of VarQBMs for generative learning in Sec.~\ref{sec:qbm_examples}.
Finally, a conclusion and an outlook are given in Sec.~\ref{sec:qbm_conclusion_outlook}.

\subsection{Boltzmann Machines}
\label{sec:BM}

A BM \cite{HintonBM1985} represents a network model that stores the learned knowledge in bias and connection weights of network nodes. 
These weights are trained to generate outcomes according to a probability distribution of interest, e.g., to generate samples which are similar to given training samples $x$. 
Typically, this type of neural network is related to an Ising-type model \cite{Ising1925, peierls_1936} such that each node $i$ corresponds to a binary variable $z_i \in \set{-1, +1}$.
Now, the set of nodes may be split into visible $v$ and hidden $h$ nodes representing observed and latent variables, respectively.
Furthermore, a certain configuration $z = v \cup h$ determines an energy which is given as 
\begin{equation}
    E_z\left({\bm{\tilde\theta}, \bm{\theta}}\right)= -\sum\limits_i\tilde{\theta}_iz_i - \sum\limits_{i, j}\theta_{ij}z_iz_j,
\end{equation}
with $\tilde{\theta}_i, \theta_{ij}\in\mathbb{R}$ denoting the bias respectively connection weights and $z_i$ representing the value taken by node $i$. 
The probability to observe a configuration $v$ of the visible nodes -- which is related to an item in the training data set -- is defined as \begin{equation}
\label{eq:gibbs_dis}
    p_v\left({\bm{\tilde\theta}, \bm{\theta}}\right) = \frac{\sum_he^{-E_z\left({\bm{\tilde\theta}, \bm{\theta}}\right)/\left(\text{k}_{\text{B}}\text{T}\right)}}{Z}  = \frac{e^{-E_v\left({\bm{\tilde\theta}, \bm{\theta}}\right)/\left(\text{k}_{\text{B}}\text{T}\right)}}{Z},
\end{equation}
where $k_B$ is the Boltzmann constant, $T$ the system temperature and $Z$ the canonical partition function
\begin{equation}
    Z=\sum\limits_{z}e^{-E_z\left({\bm{\tilde\theta}, \bm{\theta}}\right)/\left(\text{k}_{\text{B}}\text{T}\right)}.
\end{equation}

Given training data items $x$ \footnote{Notably, the number of individual training data items $x$ must be smaller or equal to the number of configurations that are accessible to the visible nodes of the BM.} which are distributed according to a discrete probability distribution $p^{\text{data}}$,
the goal of a BM is to fit the target probability distribution $p^{\text{data}}$ with $p\left({\bm{\tilde\theta}, \bm{\theta}}\right)$. To that end, we identify the items $x$ in the training data set with configurations $v$ and optimize, e.g., the averaged negative log-likelihood function
\begin{equation}
\label{eq:crossEnt}
	L\left({\bm{\tilde\theta}, \bm{\theta}}\right) = -\sum\limits_{x}p_x\left({\bm{\tilde\theta}, \bm{\theta}}\right)\log{p_x^{\text{data}}},
\end{equation}
where $p_x^{\text{data}}$ denotes the occurrence probabilities of the data items $x$\footnote{Suppose the number of individual training data items $x$ is smaller than the number of possible configurations $v$. Then, it can happen that a configuration is sampled which is not identified with a training data sample such that $p_x^{\text{data}}=0$. To ensure that the loss function can still be computed, one can employ \emph{clipping}, i.e., $\max\{p_x^{\text{data}}, \epsilon \}$ for a small $\epsilon>0$.}.

In theory, fully-connected BMs have interesting representation capabilities \cite{HintonBM1985, YOUNES1996109, Fischer12RBM}, i.e., they are universal probability distribution approximators \cite{Roux10, Du2019BM}.
However, in practice they are difficult to train as the optimization easily gets expensive. 
Thus, it has become common practice to restrict the connectivity between nodes which relates to restricted Boltzmann Machines\cite{RBM_Montufar_2018}. 
Furthermore, several approximation techniques, such as contrastive divergence \cite{Hinton2002TrainingPO}, have been developed to facilitate BM training. 
However, these approximation techniques typically still face issues such as long computation time due to a large amount of required Markov chain steps or poor compatibility with multimodal probability distributions \cite{murphy2012machinelearning}.
Furthermore, BMs are actually compatible with generative as well as discriminative learning tasks, see Appendix \ref{app:qbms} for the latter.
For further details, we refer the interested reader to \cite{Hinton2012, Fischer2012, Fischer2015}.

\subsection{Quantum Algorithm}
\label{sec:qbm}

The nodes of a QBM can again be grouped into visible $v$ and hidden $h$ nodes (qubits). The output of the QBM is determined by the configuration of the visible qubits, whereas the hidden ones correspond to the latent variables.
The network structure is defined by a parameterized, $n$-qubit Hamiltonian
\begin{equation}
    H_z\left({\bm{\theta}}\right)=\sum_{i=0}^{p-1}\theta_ih_i,
\end{equation}
with $z = v \cup h$, $\bm{\theta}\in\mathbb{R}^p$ and $h_i=\bigotimes_{j=0}^{n-1}\sigma_i^j$ for $\sigma_i^j\in\set{I, X, Y, Z}$ acting on the $j^{\text{th}}$ qubit -- which represents the $j^{\text{th}}$ node. 
This Hamiltonian relates to a quantum Gibbs state,
\begin{align}
\label{eq:QGibbs}
    \rho_z(\bm{\theta}) = {e^{-H_z\left({\bm{\theta}}\right)/\left(\text{k}_{\text{B}}\text{T}\right)}}/{Z}
\end{align}
with the Boltzmann constant $\text{k}_{\text{B}}$, the system temperature $\text{T}$ and the partition function
\begin{equation}
    Z=\text{Tr}\left[e^{-H_z\left({\bm{\theta}}\right)/\left(\text{k}_{\text{B}}\text{T}\right)}\right].
\end{equation}
QBMs are often represented by an Ising model \cite{Ising1925}, i.e., a $2$-local system \cite{bravyi06LocalHam} with nearest-neighbor coupling that is defined with regard to a particular grid but generally any Hamiltonian could be used. 
Now, the probability to measure a configuration $\ket{x}$ of the visible qubits is defined with respect to a projective measurement $\Lambda_x = \proj{x}^v\otimes \mathds{1}^h$ on the quantum Gibbs state $\rho_z(\bm{\theta})$, i.e., the probability to measure $\ket{x}$ is given by 
\begin{equation}
     p_v\left(\bm{\theta}\right) = \text{Tr}\left[\Lambda_v\rho_z(\bm{\theta}) \right].
\end{equation}

Our goal is to train the Hamiltonian parameters $\bm{\theta}$ such that the sampling probabilities of the corresponding $\rho_z(\bm{\theta})$ reflect the probability distribution underlying given classical training data. For this purpose, the loss function is given in accordance with the classical case
\begin{equation}
\label{eq:loss_qbm}
	L\left(\bm{\theta}\right) = -\sum\limits_{x}p_x\left(\bm{\theta}\right)\log{p_x^{\text{data}} },
\end{equation}
where the training data items $x$ are identified with a configuration of the visible qubits $v$ and 
\begin{equation}
\label{eq:loss_derivative}
\begin{split}
	\nabla_{\bm{\theta}}L\left(\bm{\theta}\right) = - \sum\limits_{x}\nabla_{\bm{\theta}}p_x\left(\bm{\theta}\right)\log{p_x^{\text{data}} }.
		\end{split}
\end{equation}

Next, a variational QBM implementation (\varqbm{}) is introduced that facilitates the evaluation of $\nabla_{\bm{\theta}}L\left(\bm{\theta}\right)$ using automatic differentiation. The key component of this QML algorithm is variational Gibbs state preparation. 
To that end, variational quantum imaginary time evolution (VarQITE), introduced in Sec.~\ref{sec:varqite}, is employed to prepare an approximation to the Gibbs state, see Sec.~\ref{sec:varqite_gibbs}, using a parameterized quantum state $\rho\left(\bm{\omega}\left(\bm{\theta}\right)\right) \approx \rho_z\left(\bm{\theta}\right)$ with $\bm{\omega}\in\mathbb{R}^k$. Furthermore, $\nabla_{\bm{\theta}}\bm{\omega}$ can be calculated with the VarQITE chain rule introduced in Sec.~\ref{sec:varqite_chainRule}. 
Now, we use the quantum gradient techniques introduced in Sec.~\ref{sec:gradients} to evaluate $\nabla_{\bm{\omega}}p_x\left(\bm{\theta}\right)$ for
\begin{equation}
  p_x\left(\bm{\omega}\left(\bm{\theta}\right)\right) = \text{Tr}\left[\Lambda_x\rho\left(\bm{\omega}\left(\bm{\theta}\right)\right)\right].  
\end{equation}
Then, automatic differentiation, as described in Sec.~\ref{sec:varqite_chainRule}, gives
\begin{align}
    \nabla_{\bm{\theta}}p_x\left(\bm{\omega}\left(\bm{\theta}\right)\right) = \nabla_{\bm{\omega}}p_x\left(\bm{\omega}\left(\bm{\theta}\right)\right)\nabla_{\bm{\theta}}\bm{\omega},
\end{align}
which, in turn, enables the computation of $\nabla_{\bm{\theta}}L\left(\bm{\theta}\right)$.
The parameters $\bm{\theta}$ may, now, be updated using a classical optimizer, such as Truncated Newton \cite{TNCDembo1983} or Adam \cite{Kingmaadam14}.
The \varqbm{} training is illustrated in Fig.~\ref{fig:varqbm}.

\begin{figure}[h!]
\captionsetup{singlelinecheck = false, format= hang, justification=raggedright, font=footnotesize, labelsep=space}
\begin{center}
\begin{tikzpicture}
\node at (0, 3) {VarQBM training};
\node at (0, 0){   
\includegraphics[width=0.5\linewidth]{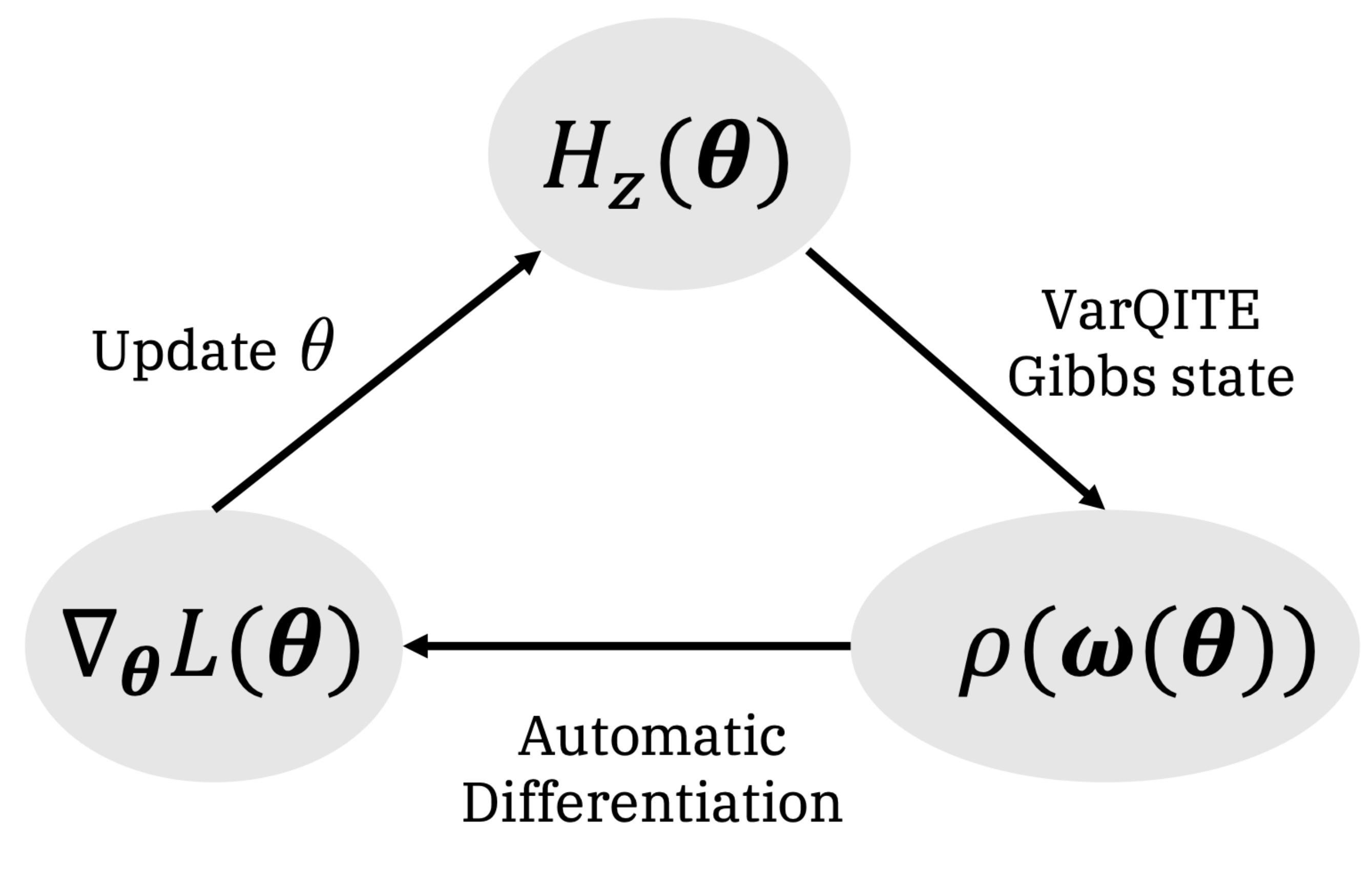}};
\end{tikzpicture}
\end{center}
\caption{The \varqbm{} training includes the following steps. First, we fix the Pauli terms for $H_z\left({\bm{\theta}}\right)$ and choose initial parameters $\bm{\theta}$. Then, VarQITE is used to generate $\rho\left(\bm{\omega}\left(\bm{\theta}\right)\right)$ and compute $\nabla_{\bm{\theta}}\bm{\omega}$. 
The quantum state and the derivative are needed to evaluate 
$\nabla_{\bm{\theta}}  p_x\left(\bm{\omega}\left(\bm{\theta}\right)\right)$. Now, we can find $\nabla_{\bm{\theta}}L$ to update the Hamiltonian parameters with a classical optimizer.}
\label{fig:varqbm}
\end{figure}

\subsection{Illustrative Examples}
\label{sec:qbm_examples}

 \begin{figure}[!htb]
\captionsetup{singlelinecheck = false, format= hang, justification=raggedright, font=footnotesize, labelsep=space}
\begin{center}
\begin{tikzpicture}
\node at (0, 0) {\includegraphics[width=0.67\linewidth]{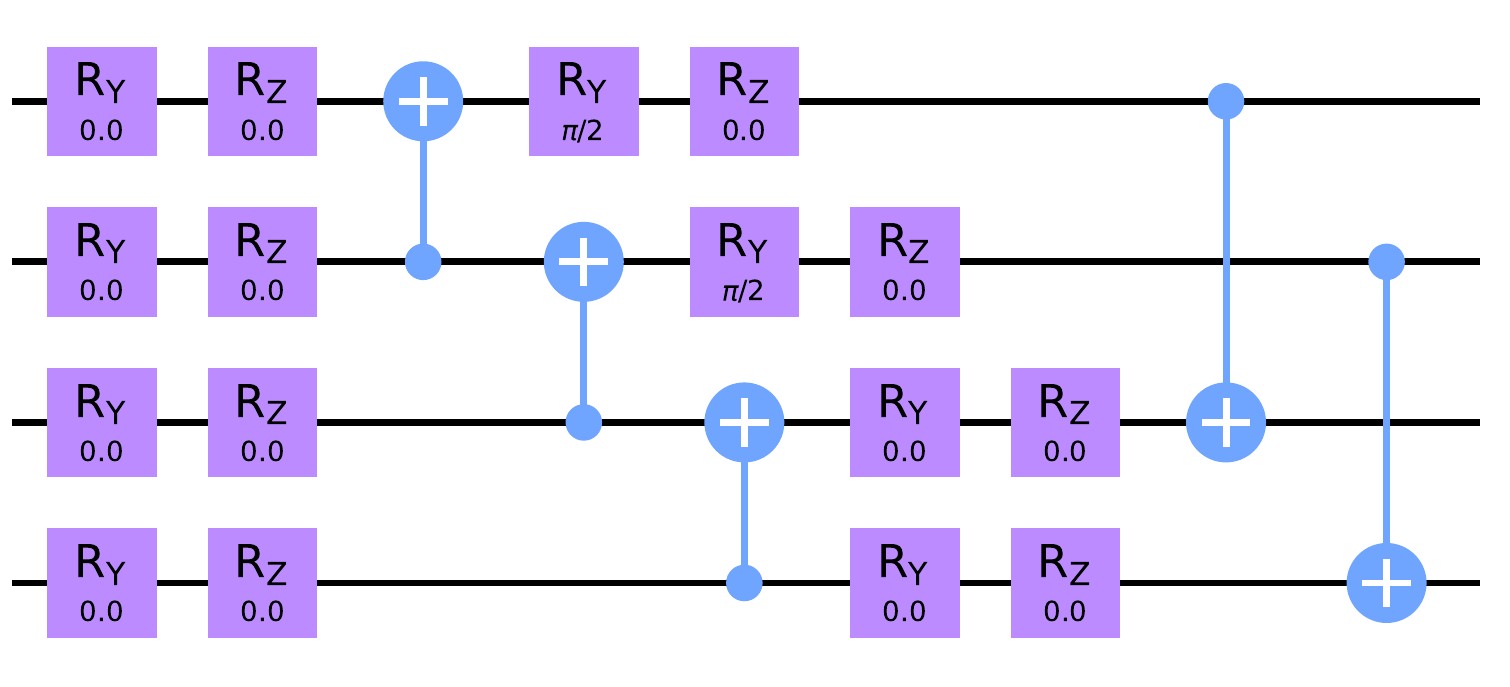}};
   \node at (-6.6, 0) {$\ket{0}^{\otimes 2n}$};
\draw[decorate, thick, decoration = {brace, amplitude=15pt}] (-5.3,-1.65) --  (-5.3,1.6);
\end{tikzpicture}
\end{center}
\caption{The \varqbm{} training uses VarQITE as a sub-routine for the preparation of an $n$-qubit Gibbs state. The examples presented in this section are based on the parameterized ansatz shown in this figure, whereby the given parameters are used to prepare the initial state.
Notably, this figure illustrates the $n=2$ setting. The ans\"atze for the $n=3,\, 4$ GHZ states are constructed analogously.}
\label{fig:ansatz_Bell}
\end{figure}

\begin{figure}[!htb]
\captionsetup{singlelinecheck = false, format= hang, justification=raggedright, font=footnotesize, labelsep=space}
\begin{center}
\begin{tikzpicture}
\node at (-4, 3.8){Bell and GHZ State Training};
\node at (-7, 0){   
\includegraphics[width=0.5\linewidth]{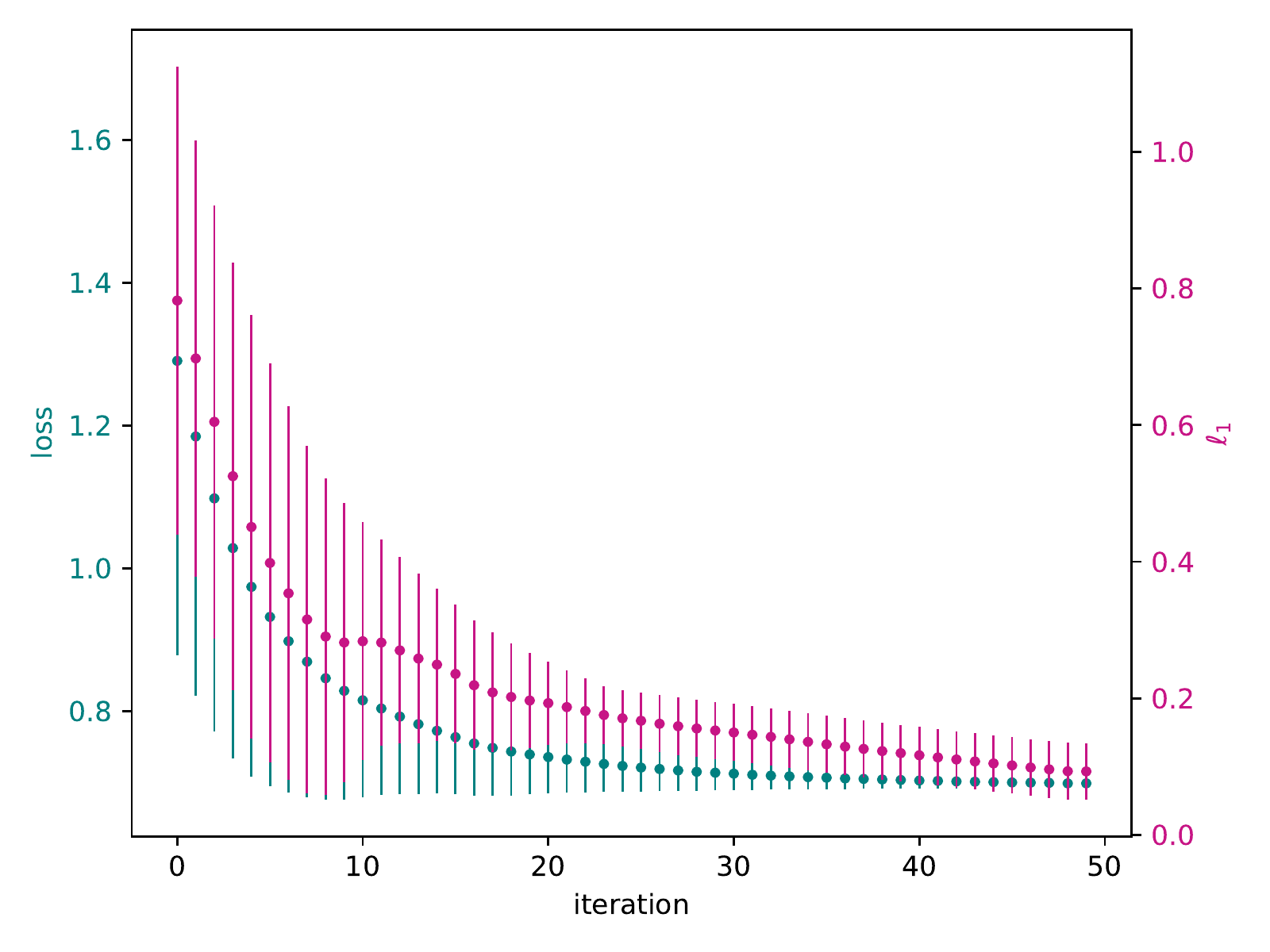}};
\node at (0, 0){   
\includegraphics[width=0.5\linewidth]{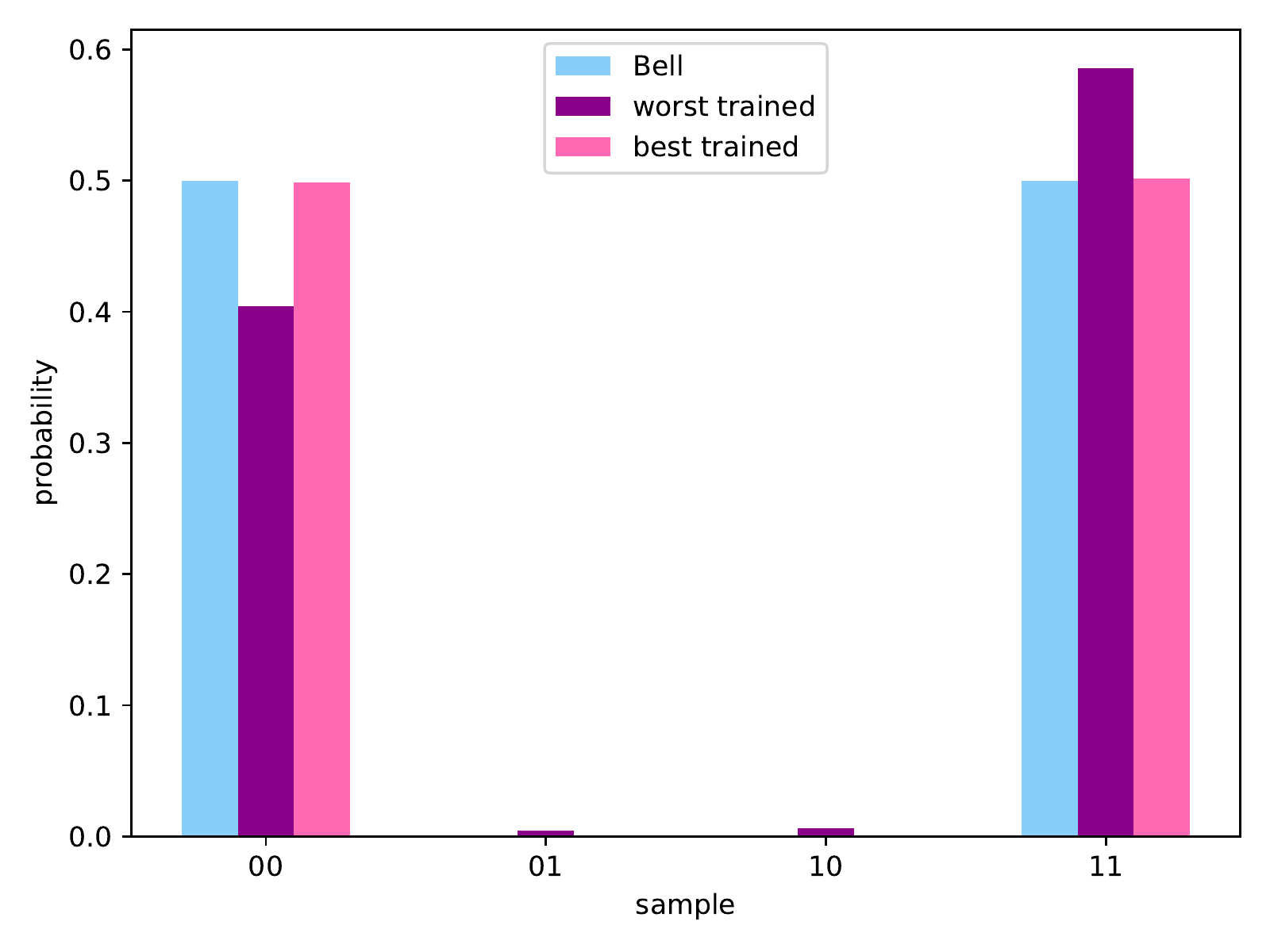}};
\node at (-3, -6){   
\includegraphics[width=0.5\linewidth]{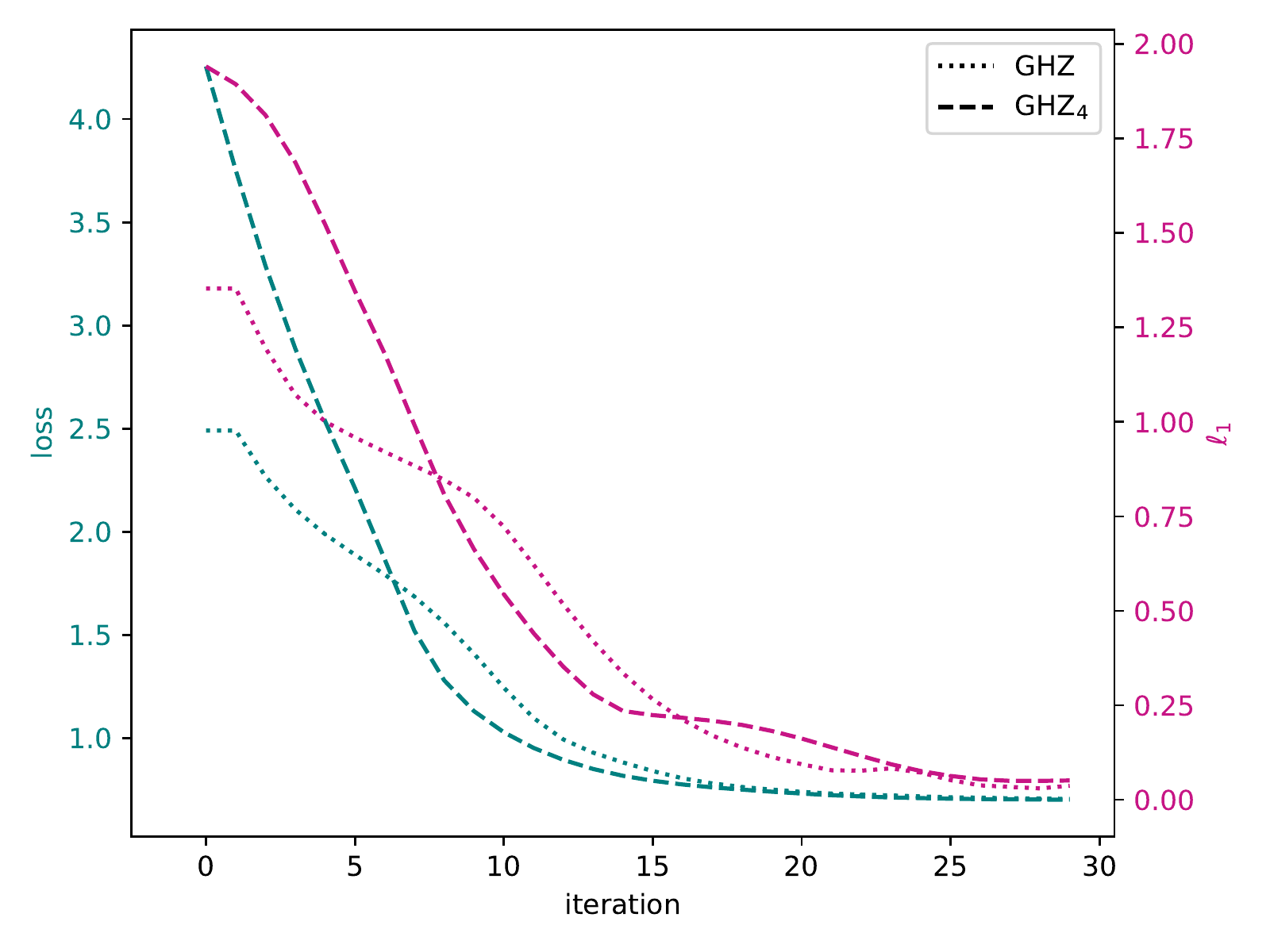}};
\node at (-10, 3){(a)};
\node at (-3, 3){(b)};
\node at (-6, -3){(c)};
\end{tikzpicture}
\end{center}
\caption{(a) The training progress of a fully-visible QBM model is illustrated which aims to represent the measurement distribution of a Bell state.
The green function corresponds to the loss and the pink function represents the distance between the trained and target distribution with respect to the $\ell_1$ norm at each step of the iteration.
The points represent the mean and the error bars the standard deviation of the results.
(b) The sampling probability of the Bell state (blue) as well as the best (pink) and worst (purple) probability distribution achieved from  $10$ different random seeds is presented. 
All results are computed for $10$ different random seeds. 
(c) The training progress of a fully-visible QBM model is shown that aims to represent the measurement distribution of a $3$ and a $4$ qubit GHZ state.
The green respectively pink dots (crosses) correspond to the loss respectively the distance between the trained and target distribution with respect to the $\ell_1$ norm at each step of the iteration for the $3$ ($4$) qubit GHZ state.}
\label{fig:qrbm}
\end{figure}
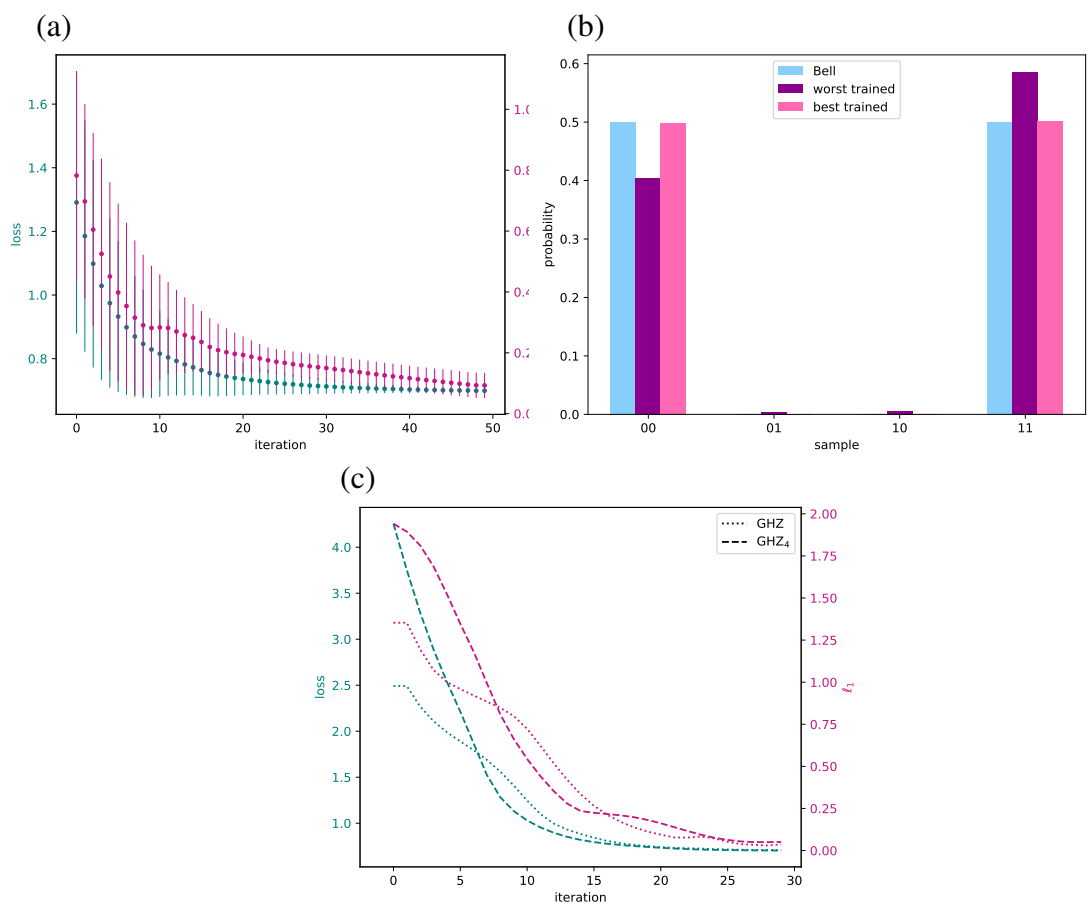



Now, we present three illustrative examples for generative QBMs \footnote{These examples are trained with $L\left(\bm{\theta}\right) = -\sum_{x}p_x^{\text{data}} \log{ p_x\left(\bm{\theta}\right)}$ which is convention in existing ML literature. Although this is feasible for the small system size of the presented examples, a sampling based evaluation of the respective loss function gradient would not scale.}.
We would like to point out that all of the presented examples learn the sampling statistics of maximally entangled states which exhibit non-local correlations. This setting is not compatible with every Gibbs state preparation method -- see Sec.~\ref{sec:data_encoding} for further discussion -- but specifically enabled by the use of VarQITE.
More precisely, we train QBMs to mimic the sampling statistics of a Bell state
    \begin{equation}
        \ket{\phi^{+}}=\frac{\ket{00} + \ket{11}}{\sqrt{2}},
    \end{equation}
    of a Greenberg-Horne-Zeiliinger (GHZ) state \cite{GHZ1989}
    \begin{equation}
        \ket{\text{GHZ}}=\frac{\ket{000} + \ket{111}}{\sqrt{2}},
    \end{equation}
    and a generalized $4$ qubit GHZ state 
    \begin{equation}
        \ket{\text{GHZ}_4}=\frac{\ket{0000} + \ket{1111}}{\sqrt{2}}.
    \end{equation}
The learning models are given by fully visible QBMs which are based on two-local Hamiltonians of the form
\begin{align}
     H\left(\bm{\tilde\theta}, \bm{\theta}\right)= -\left(\sum\limits_{i=0}^{n-1} \tilde\theta_i Z_j + \sum\limits_{i,j=0}^{n-1}\theta_{i,j} Z_i\otimes Z_j \right),
\end{align}
where $Z_i$ acts on qubit $i$ and $n$ corresponds to the number of state qubits.

The respective optimization runs on an ideal statevector simulation of a quantum computer using AMSGrad \cite{amsgrad} with initial learning rate $0.1$, first momentum $0.7$, and second momentum $0.99$ as optimization routine. The initial values of the Hamiltonian parameters $\bm{\theta}$ are drawn from a uniform distribution on $\left[-1, 1\right]$.
Furthermore, the Gibbs state preparations use forward Euler ODE solvers with $10$ steps per state preparation and $ \text{k}_{\text{B}}\text{T} = 1$. For a detailed discussion on the possible choices of ODE solvers for VarQITE, we refer to Sec.~\ref{sec:odesolvers}. The respective ansatz is chosen according to Fig.~\ref{fig:ansatz_Bell}.

First, we investigate the training of $p^{\text{Bell}}$ in detail. More specifically, we run $10$ experiments using different randomly drawn initial values for $\bm{\theta}$ and look into the statistics.
The averaged values of the loss function as well as the distance between the target distribution $p^{\text{Bell}}$ and the trained distribution $p^{\text{Bell}}\left(\bm{\omega}\left(\bm{\theta}\right)\right)$ with respect to the $\ell_1$ norm are illustrated over $50$ optimization iterations in Fig.~\ref{fig:qrbm} (a). 
The plot shows that loss and distance converge toward the same values for all sets of initial parameters. Likewise, the trained parameters $\bm{\theta}$ converge to similar values.
Moreover, Fig.~\ref{fig:qrbm} (b) illustrates the target probability distribution and for the best and worst of the trained distributions. The plot reveals that the model is able to train the respective distribution well.

Next, we present the results for the bigger \varqbm{} examples which learn $p^{\text{GHZ}}$ respectively $p^{\text{GHZ}_4}$. Fig.~\ref{fig:qrbm} (c) illustrates the progress of the loss function as well as the $\ell_1$ norm between the trained and the target distribution for $30$ training step. Both measures show good convergence behavior and, thus, indicate that the method has potential to scale well.

To sum up, the proof of concept examples presented in this section indicate that \varqbm{} has promising training properties. Since we use fully visible Hamiltonian models and shallow quantum circuits  neither entanglement-induced nor circuit-depth-induced barren plateaus, see Sec.~\ref{sec:vanishing_grads}, are expected to impact this training setting. However, the potential occurrence of vanishing gradients which is due to hardware noise and global cost functions require further investigation with quantum hardware and larger experiments, respectively.

Lastly, we would like to point out that QBMs -- like their classical counterpart -- are applicable to discriminative learning, as well. In this context, the QBM is used to learn a conditional distribution with respect to training data which is given by pairs of data samples and data labels. An example of a \varqbm{} model that is trained for a classification task on an artificial credit card transaction data set, where the quantum model outperforms a set of standard classifiers from scikit-learn \cite{scikit-learn2011}, is presented in Appendix \ref{appendix:discriminativeqbm}. Potential advantages of using QBMs for discriminative learning require additional investigation. In particular, it would be interesting to study whether QBMs enable the representation of a larger class of conditional distributions than their classical counterparts and might, therefore, represent more expressive classification models.

\subsection{Discussion}
\label{sec:qbm_conclusion_outlook}

This section discusses the application of QBMs for generative QML which can be used to facilitate approximate quantum data loading. Furthermore, we outline an implementation method which employs VarQITE for approximate Gibbs state preparation. Notably, this variational Gibbs state preparation method allows for an efficient, a posteriori error bound of the preparation accuracy, see Sec.~\ref{sec:error_qite}. We would like to emphasize that the evaluation of this error bound can be easily integrated into the algorithmic workflow of VarQBMs. A benchmark of the Gibbs state preparation accuracy with respect to the mentioned error bound remains to be studied in future research. This variational QBM implementation is compatible with generic Hamiltonians, enables automatic differentiation for gradient-based training and is suitable for execution on first-generation quantum hardware. Moreover, the presented scheme is not only compatible with local but also long-range correlations and arbitrary system temperatures for a sufficiently powerful ansatz. Lastly, we present proof of concept examples which demonstrate the power of this method.

Although generative QML models that train the parameters of a given quantum channel directly, such as the qGANs presented in Sec.~\ref{sec:qgan}, are simpler to implement than QBMs, which encode the information in the parameters of quantum operators, there are various promising aspects about the latter.
It is known that the energy functions of BMs and related energy-based models can be chosen so that they reflect the conditional independence structure of given training data \cite{BresslerMRFSSTructure13, piatkowski2019exponential, discreteProbDistApproxChow68, takashina2018structureMRF, vuffray2020efficientLearningGraphModels}. Similar approaches can enable us to make informed choices about the Hamiltonian structure of a QBM. Furthermore, the Hamiltonian structure may facilitate the construction of problem-specific ans\"atze that, in turn, might help to avoid exponentially vanishing gradients and improve training efficiency, see Sec.~\ref{sec:vanishing_grads}.
Due to their inherent structure, QBMs are also more likely to lead to well interpretable machine learning models than, e.g., qGANs. 

Furthermore, QBMs are not limited to generative applications but can also be used to train discriminative models, see Appendix \ref{appendix:discriminativeqbm}. Interestingly, the respective algorithm avoids the loading of the input training data into a quantum state. Instead, it encodes the data into the parameters of the Hamiltonian.

An open question for future research would be to conduct a thorough analysis of trial state ans\"atze which would help to improve our understanding of the model's representation capabilities and enable more informed choices of ansatz states. More specifically, the investigation of problem-specific ans\"atze could help to avoid vanishing gradient problems, see Sec.~\ref{sec:vanishing_grads} for further discussion.



\section[Approximate Data Loading for Quantum Amplitude Estimation]{Approximate Data Loading for Quantum Amplitude Estimation\footnote{This section is reproduced in part, with permission, from C.~Zoufal, A.~Lucchi, S.~Woerner, "Quantum Generative Adversarial Networks for Learning and Loading Random Distributions", \textit{npj Quantum Information}, vol. 5, Article Nr. 103, 2019}}
\label{sec:EuropeanCallOptionPricingexample}

Generative quantum machine learning may facilitate the study of potentially advantageous fault-tolerant quantum algorithms such as QAE \cite{brassardQAE02} and the HHL algorithm\footnote{Assuming that the respective matrix is well-conditioned and the right-hand-side has a suitable form \cite{Aaronson2015_finePrint}.} \cite{hhl}. These algorithms are stable up to small errors in the input state, i.e., small deviations in the input only lead to small deviations in the result and are, thus, compatible with approximate quantum data loading. 
The following section demonstrates how generative QML algorithms can help to study the potential advantages of quantum algorithms in the context of a real-world example. More specifically, we present a qGAN-based approximate quantum data loading for financial derivative pricing with QAE. To that end, qGANs are used to learn and load a model for the spot price of an asset underlying a European call option. This, in turn, enables the evaluation of characteristics of this model, such as the expected payoff, with QAE \cite{brassardQAE02, worQuantumRiskAnalysis19}.
To demonstrate and verify the efficiency of the suggested method, we implement a small illustrative example that is based on the analytically evaluable standard model for European option pricing, the Black-Scholes model \cite{BlackScholes}.
We would like to point out that in this context, QAE can achieve a quadratic improvement in the error scaling compared to standard classical Monte Carlo simulation.

The remaining section is structured as follows. First, we discuss the pricing of European call options with the Black-Scholes model and how to load a model for the spot price of an asset underlying a European call option into a quantum channel in Sec.~\ref{sec:europeancallOpt}. Then, we give an introduction to QAE in Sec.~\ref{sec:qae}. 
Finally, Sec.~\ref{subsec:europeanOptionPricingResults} presents the results of the call option pricing and Sec.~\ref{sec:generativeQMLconclusionapplication} gives conclusions and an outlook.

\subsection{European Call Option Pricing}
\label{sec:europeancallOpt}
Now, we are going to explain European call option pricing with respect to the Black-Scholes model \cite{BlackScholes}.
The owner of a European call option is permitted, but not obliged, to buy an underlying asset for a given strike price $K$ at a predefined future maturity date $T$, where the asset's spot price at maturity $S_T$ is assumed to be uncertain.
If $S_T \leq K$, i.e., the spot price is below the strike price, it is unreasonable to exercise the option as it is cheaper to buy the asset regularly.
However, if $S_T > K$, exercising the option to buy the asset for price $K$ and immediately selling it again for $S_T$ can realize a payoff $S_T - K$. In other words, the payoff of the option is defined as $\max\lbrace S_T - K, 0 \rbrace$ and our goal is to evaluate the expected payoff $\EX\left[\max\lbrace S_T - K,0 \rbrace\right]$. To that end, $S_T$ is assumed to follow a particular random distribution.
This corresponds to the fair option price before discounting. Notably, discounting is neglected to simplify the problem.

We would also like to point out that this model often over-simplifies the real circumstances.
In more realistic and complex cases, where the spot price follows a more generic stochastic process or where the payoff function has a more complicated structure, options are usually evaluated with Monte Carlo simulations \cite{Glasserman2003}.
A Monte Carlo simulation uses $N$ random samples drawn from the respective distribution to evaluate an estimate for a characteristic of the distribution, e.g., the expected payoff. The estimation error of this technique behaves like $\epsilon=\mathscr{O}(1/\sqrt{N})$.
When using $n$ evaluation qubits in QAE, which is explained next, this induces the evaluation of $N=2^n$ quantum samples to estimate the respective distribution characteristic. Now, this quantum algorithm achieves a Grover-type error scaling for option pricing, i.e., $\epsilon=\mathscr{O}(1/N)$ \cite{brassardQAE02, Rebentrost2018, worQuantumRiskAnalysis19}.

To evaluate an option's expected payoff with QAE, the problem must be encoded into a quantum channel that loads the respective probability distribution and implements the payoff function.
The respective model for an asset's spot price can be loaded approximately by training a generative QML algorithm such as a qGAN. To that end, we collect historic data about the asset's spot price and use it as training data for our generative QML model.
According to the Black-Scholes model \cite{BlackScholes}, the spot price at maturity $S_T$ for a European call option is log-normally distributed.
Thus, we, now, assume that the asset's spot price, which is typically unknown, is given distributed according to a log-normal distribution. This means that we can use the qGAN training results for the log-normal distribution presented in Sec.~\ref{sec:qganApplications} as a model for the asset's spot price.
More explicitly, we base our analysis on the model that is trained with an actual quantum computer, the IBM Quantum Boeblingen $20$ qubit chip \cite{ibmQX}, see Fig.~\ref{fig:pdf_real}.

\subsection{Quantum Amplitude Estimation}
\label{sec:qae}
Given a quantum channel 
\begin{equation}
\label{eq:qae}
\mathscr{A}\ket{0}^{\otimes n+1} = \sqrt{1-a}\ket{\psi_0}\ket{0} + \sqrt{a}\ket{\psi_1}\ket{1},
\end{equation}
where $\ket{\psi_0}$, $\ket{\psi_1}$ denote $n$-qubit states, the QAE algorithm \cite{brassardQAE02}, illustrated in Fig.~\ref{fig:qae}, enables the efficient evaluation of the amplitude $a$.
The algorithm requires $m$ additional evaluation qubits that control the applications of an operator $\mathscr{Q} = - \mathscr{A}\mathscr{S}_0\mathscr{A}^{\dagger}\mathscr{S}_{\psi_0}$ where $\mathscr{S}_0 = \mathbb{I}^{\otimes n +1} - 2\ket{0}\bra{0}^{\otimes n+1}$ and $\mathscr{S}_{\psi_0} = \mathbb{I}^{\otimes n +1} - 2\ket{\psi_0}\bra{\psi_0}\otimes\ket{0}\bra{0}$.
The error in the outcome -- ignoring higher terms -- can be bounded by $\frac{\pi}{2^m}$. Considering that $2^m$ is the number of quantum samples used for the estimate evaluation, this error scaling offers a quadratic improvement compared to classical Monte Carlo simulations.

\begin{figure}[h!]
\captionsetup{singlelinecheck = false, format= hang, justification=raggedright, font=footnotesize, labelsep=space}
\begin{center}
\begin{tikzpicture}
 \node at (0,4){Quantum Amplitude Estimation};
\node at (0,0){\includegraphics[width=0.8\linewidth]{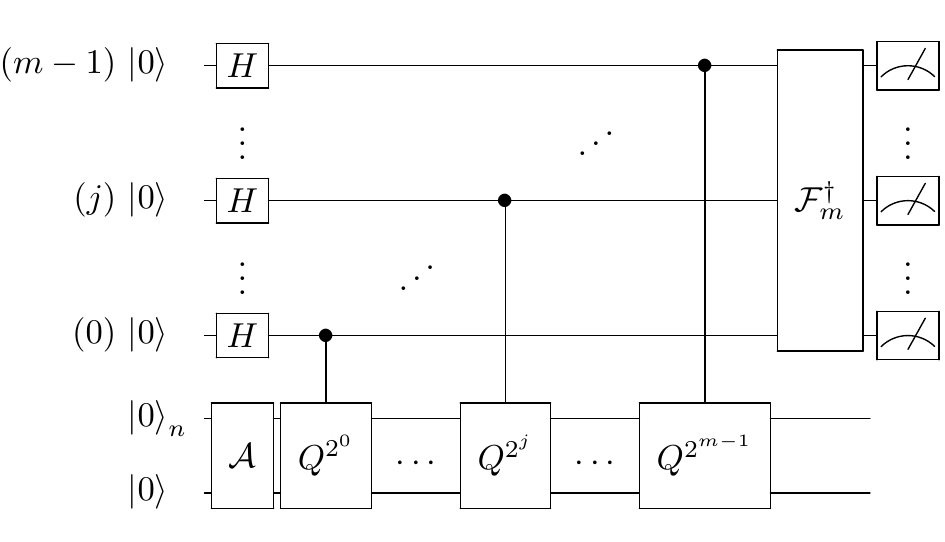}};
\end{tikzpicture}
\end{center}
\caption{The illustrated quantum circuit corresponds to the QAE algorithm with the inverse quantum Fourier transform \cite{nielsen10} being denoted by $\mathscr{F}_m^{\dagger}$.}
\label{fig:qae}
\end{figure}

To use QAE for the pricing of European options, we need to construct and implement a suitable oracle $\mathscr{A}$. First, we load a discretized form of the uncertainty distribution that represents the spot price $S_T$ of the underlying asset at the option's maturity $T$ into a quantum state $\sum_{i=0}^{2^n-1} \sqrt{p_i}\ket{i}$, where $\ket{i}$ represents a possible spot price. It should be noted that small errors in this state preparation only lead to small errors in the final result.
Then, we add a working qubit $\ket{0}$ and use a compare circuit which applies an $X$ gate to the working qubit if $i>K$, i.e.
\begin{eqnarray}
\ket{i}\ket{0} \mapsto \begin{cases}
\ket{i}\ket{0} & \text{, if} \;  i \leq K \\
\ket{i}\ket{1} &  \text{, if} \; i > K,
\end{cases}
\end{eqnarray}
where $K$ denotes the strike price. 
Now, the state reads
\begin{equation}
\sum_{i=0}^{K} \sqrt{p_i} \ket{i}\ket{0} + \sum_{i=K+1}^{2^n-1} \sqrt{p_i} \ket{i} \ket{1}.
\end{equation}
Next, we want to apply the payoff function 
\begin{align}
    f: i\mapsto\frac{i - K}{2^n - K - 1}
\end{align}
to the quantum state. To that end, we can use an additional working qubit to control the mapping of the payoff function as
\begin{align}
    F: \ket{i}\ket{0} \mapsto \sqrt{1 - f(i)}\ket{i} \ket{0} + \sqrt{f(i)}\ket{i} \ket{1}.
\end{align}
For practical reasons, we implement an efficient polynomial approximation to $F$ that is introduced in \cite{worQuantumRiskAnalysis19} and, thereby, avoid potentially expensive operations.
Now, the full construction for $\mathscr{A}$ prepares the quantum state
\begin{align}
\label{eq:qaeRot}
\mathscr{A}\ket{0}^{\otimes n+1} &= \sum_{i=0}^{K} \sqrt{p_i} \ket{i}\ket{0} \ket{0} \nonumber\\
&\hspace{5mm}+\sum_{i=K+1}^{2^n-1} \sqrt{p_i} \ket{i} \ket{1} \left( \sqrt{1 - f(i)} \ket{0} + \sqrt{f(i)} \ket{1} \right).
\end{align}

Eventually, the probability of measuring $\ket{1}$ in the last qubit is equal to
\begin{eqnarray}
\mathbb{P}[\ket{1}] &=&
\frac{1}{2^n - K - 1} \sum_{i=K+1}^{2^n-1} p_i (i - K) \\
&=& \frac{1}{2^n - K - 1} \mathbb{E}[\max\{0, S_T - K\}].
\end{eqnarray}
We can see from comparing Eq.~\eqref{eq:qae} and Eq.~\eqref{eq:qaeRot} that $\mathbb{P}[\ket{1}] = a$.
QAE can, thus, be used to efficiently evaluate $\mathbb{E}[\max\{0, S_T - K\}] = \mathbb{P}[\ket{1} ](2^n - K - 1)$.




\subsection{Results}
\label{subsec:europeanOptionPricingResults}

It remains to be demonstrated that the generative QML facilitates approximate quantum data loading in the context of quantum applications. To that end, we consider a European call option pricing task, where
the quantum generator approximates the spot price at maturity $S_T$.
More specifically, we integrate the distribution loading quantum channel into a quantum algorithm-based on QAE to evaluate the expected payoff $\EX\left[\max\left\{S_T - K, 0 \right\}\right]$ for $K=\$2$, illustrated in Fig.~\ref{fig:european_option}.  
We refer to \cite{worQuantumRiskAnalysis19, Stamatopoulos_2020} for further discussion on derivative pricing with QAE.

\begin{figure}[!h]
\centering{
\begin{tikzpicture}
   \node at (0, 3.3) {\textbf{Payoff Function European Call Option}};
 \node at (0, 0) { \includegraphics[width=0.6\linewidth]{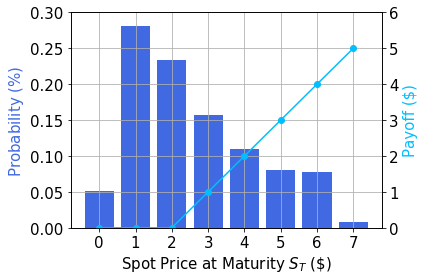}};
 \end{tikzpicture}
 }
\caption{
The illustration shows a probability distribution model of the spot price at maturity $S_T$ that has been learned with a randomly initialized qGAN run on the IBM Quantum Boeblingen chip and the corresponding payoff function for a European call option.
\label{fig:european_option} }
\end{figure}

The results for estimating $\EX\left[\max\left\{ S_T - K, 0 \right\}\right]$ are given in Tab.~\ref{tbl:europeanCallOption}, where we compare
\begin{itemize}
\item an analytic evaluation with the exact (truncated and discretized) log-normal distribution $p_{\text{real}}$,
\item a Monte Carlo simulation utilizing $\ket{g^{\omega}}$ trained and generated with the quantum hardware (IBM Quantum Boeblingen), i.e., $1024$ random samples of $S_T$ are drawn by measuring $\ket{g^{\omega}}$ and used to estimate the expected payoff, and

\item a classically simulated QAE-based evaluation using $m=8$ evaluation qubits, i.e., $2^8=256$ quantum samples, where the probability distribution $\ket{g^{\omega}}$ is trained with IBM Quantum Boeblingen chip.

\end{itemize}

The resulting confidence intervals (CI) are shown for a confidence level of $95\%$ for Monte Carlo simulation as well as QAE.
The CIs are of comparable size, although, because of better scaling, QAE requires only a fourth of the samples.
Since the distribution is approximated, both CIs are close to the exact value but do not actually contain it.
Note that the estimates and the CIs of the Monte Carlo and the QAE evaluation are not subject to the same level of noise effects. This is because the QAE evaluation uses the generator parameters trained with IBM Quantum Boeblingen but is run with a quantum simulator, whereas the Monte Carlo simulation is solely run on actual quantum hardware. To be able to run QAE on a quantum computer, further improvements are required, e.g., longer coherence times and higher gate fidelities.

\begin{table}
\captionsetup{singlelinecheck = false, format= hang, justification=raggedright, font=footnotesize, labelsep=space}
\centering{
\begin{tabular}{c|c|c|c|c}
\textbf{Approach} & \textbf{Distribution} & \textbf{Payoff} (\$) & \#\textbf{Samples}  & \textbf{CI} (\$)\\ 
\hline
analytic    & log-normal       &  1.0602  & -    & - \\
MC + QC    & $\ket{g^{\omega}}$ & 0.9740  & 1024& $\pm 0.0848$ \\
QAE     & $\ket{g^{\omega}}$ &  1.1391  &  256 & $\pm 0.0710$
\end{tabular}
}
\caption{
This table presents a comparison of different approaches to evaluate $\EX\left[\max\lbrace S_T - K, 0 \rbrace\right]$: an analytic evaluation of the log-normal model, a Monte Carlo simulation drawing samples from IBM Quantum Boeblingen, and a classically simulated QAE.
Furthermore, the $95\%$ confidence intervals of the estimates are shown.}
\label{tbl:europeanCallOption}
\end{table}

\subsection{Discussion}
\label{sec:generativeQMLconclusionapplication}
This section presents an explicit example where quantum data loading with generative QML facilitates a QAE application. More specifically, the quantum algorithm is used in the context of European call option pricing, where it has the potential to achieve a quadratic improvement compared to classical Monte Carlo simulations. To that end, a qGAN trains a quantum state such that it represents a model of the spot price at maturity underlying the European call option.

Of course, this setting can also be used to price a variety of other options as is described in \cite{Stamatopoulos_2020}. QAE applications are not limited to option pricing but can generally be explored for problems which are classically evaluated with Monte Carlo methods, such as simulation-based optimization \cite{Gacon_2020QSBO}. For all of these tasks and many others, generative QML offers a flexible and promising method to load approximations to the necessary data in an efficient way.

\addtocontents{toc}{}

%% file: conclusion_outlook.tex
\chapter{Conclusion and Outlook}
\label{conclusion} 

\index{conclusion}

\vspace{8mm}

\noindent

\section{Conclusion}
\label{sec:conclusion}

This thesis introduces and analyzes two generative QML algorithms, qGANs and QBMs, that can be implemented with variational quantum circuits and, are thus, suitable for execution on first-generation quantum hardware. Both paradigms have different advantages that were analyzed and demonstrated theoretically as well as practically.


QGANs represent a simple but flexible algorithm that is straight-forward to implement. However, the model offers limited inherent structure. The choice of the input state can be based on information about the data structure but the generator architecture as well as the discriminator architecture are usually chosen heuristically. Notably, the respective architectures must be selected carefully to avoid barren plateaus as well as an imbalanced training behavior.

QBMs represent a structured approach and can, thus, help with the development of problem-specific ans\"atze as well as facilitate interpretability. Nevertheless, the algorithm requires Gibbs state preparation -- or at least Gibbs sampling -- which is generally difficult and potentially expensive. 
As part of the presented VarQBM workflow, we overcome this hurdle by using VarQITE for approximate Gibbs state preparation.
We benchmark the power of this technique and present a novel, efficient, a posteriori error bound for the simulation accuracy of VarQITE. 

Furthermore, we discuss that a useful and feasible application of generative QML is given by approximate quantum data loading.
The respective practicality is accentuated with various examples that are executed using quantum simulations or actual quantum hardware.

\section{Outlook}
\label{sec:outlook}

Although quantum machine learning is considered a promising field for achieving a quantum advantage \cite{Huang_2021InformationTheoreticBoundQA, Huang_2021PowerDatainQML}, there are still many open questions about its strengths and weaknesses. 
It is assumed that, for certain problem instances,  parameterized quantum circuits may be employed as a beneficial \cite{BiamonteQML17} quantum counterpart to classical neural networks.
In order to understand the practical potential of quantum machine learning models, it is important to investigate properties such as:
\begin{itemize}
    \item the capacity, referring to the variety of functions that a model ansatz can represent,
    \item the generalization capabilities, expressing a model's ability that is trained on a given data set to generalize to unseen data instances,
    \item the expressivity, describing the fraction of the state space that can be accessed by the chosen ansatz,
    \item and the trainability, i.e., the speed of convergence.
\end{itemize}
However, there are only few research results that enable the investigation of these properties.
For example, the work presented in \cite{Abbas_2021PowerofQNNs} compares the capacity of a set of classical neural networks with comparable quantum counterparts using the \emph{effective dimension} as a capacity measure and finds that the parameterized quantum circuits exhibit advantageous properties. 
How to measure the generalization capabilities of quantum machine learning models using mutual information is studied in \cite{banchi2021generalizationQML}.
Statistical learning theory \cite{Vapnik1995StatLearningTheory} is used in \cite{du2021efficientMeasureExpressivityVQA} to efficiently quantify the expressivity of variational quantum circuits and, also, to study the generalization properties of these circuits.
Furthermore, the trainability of variational quantum circuits is, e.g., studied in
\cite{du2020learnabilityQNN} again using statistical learning theory.
So far, these measures have only been applied to study a limited number of examples. To understand the true practical power and bottlenecks of QML applications, thorough benchmarking is required.
To that end, it is vital to have the respective measures available in state-of-the-art quantum software frameworks to facilitate an easy access to model analysis as part of a quantum machine learning workflow.

Additionally, we need to get a better understanding of how we could scale up quantum machine learning models. As discussed in Sec.~\ref{sec:vanishing_grads}, using expressive quantum circuits or global cost functions can easily lead to exponentially vanishing gradients and an exponential flattening of the loss landscape and, thus, impede efficient training. Therefore, scalable quantum machine learning algorithms have to be defined carefully.
The beginnings of classical machine learning also faced similar problems and managed to resolve them with targeted ansatz structures, loss functions, and initialization strategies. 
We expect that the development and application of problem-specific ans\"atze with limited entanglement will play a crucial role in the future of quantum machine learning and are curious whether non-linear transformations could help to circumvent some of the discussed problems.

Lastly, we would like to point out that shallow, parameterized quantum circuits offer a suitable setting to study the practical performance of quantum algorithms with currently available simulators and hardware backends. They, also, act as a stepping stone, in the sense that they allow us to build up knowledge, start designing software platforms and form expectations about the resources that are required to facilitate full-scale quantum applications. 

\addtocontents{toc}{}

%% file: appendix_varqte.tex

\chapter[Variational Quantum Real Time Evolution]{Appendix: Variational Quantum Real Time Evolution\footnote{This section is reproduced in part, with permission, from C.~Zoufal, D.~Sutter, S.~Woerner, "Error Bounds for Variational Quantum Time Evolution", Preprint available at arXiv:2108.00022, 2021}}
\label{app:varqrte}


\textbf{Abstract.} 
This chapter introduces variational quantum real time evolution, as well as, an a posteriori error bound that enables to estimate the approximation accuracy of this method. To illustrate the performance of the presented variational quantum real time evolution and the respective error bounds, various numerical examples are presented.
\index{appendix VarQTE}

\vspace{8mm}

\noindent

Quantum time evolution describes the dynamics of a quantum state with respect to a given Hamiltonian $H$, see Sec.~\ref{sec:qte} for further discussion on this topic.
Real quantum time dynamics are described by the Schr\"odinger equation
\begin{align}
    i\hslash\ket{\dot\psi^*_t}= H\ket{\psi^*_t},
\end{align}
where $\ket{\dot\psi^*_t} = \frac{\partial \ket{\psi^*_t}}{\partial t}$.

This Appendix is structured as follows. First, Sec.~\ref{sec:varqrte_implementation} discusses how to implement quantum real time evolution with a variational quantum circuit. Then, Sec.~\ref{subsec:error_qrte} introduces an efficient, phase-agnostic, a posteriori error bound for the variational implementation. Lastly, Sec.~\ref{app:varqrte_experiments} illustrates the efficiency and practicality of the variational quantum real time method as well as the introduced error bound.

\section{Variational Approach}
\label{sec:varqrte_implementation}

The variational quantum real time evolution (VarQRTE) derived in the following employs McLachlan's variational principle \cite{McLachlan64} to simulate the actual time dynamics by propagating a parameterized state $\ket{\psi^{\omega}_t}$. The respective approach is agnostic to a potential time-dependent global phase \cite{Simon18TheoryVarQSim}.

Consider the real time evolution of a parameterized state with an explicit time-\linebreak dependent global phase parameter $\nu$, i.e., $\ket{\psi^{\nu}_t} = \ee^{-i\nu}\ket{\psi^{\omega}_t}$ for $\nu = \nu_t \in \mathbb{R}$,
where
\begin{align}
\ket{\dot\psi^\nu} = -i\dot\nu\ee^{-i\nu}\ket{\psi^{\omega}_t}+ \ee^{-i\nu}\ket{\dot{\psi}^{\omega}_t}.
\end{align}
The Schr\"odinger equation with respect to $\ket{\psi^{\nu}_t}$ reads
\begin{align}
\label{eq:phase_fix_Schroed}
    i\hslash\ket{\dot\psi^\nu} = H\ket{\psi^\nu}
\end{align}
and can be rewritten as 
\begin{align}
    \label{eq:Schroedinger_real_nu}
     i\hslash\ee^{-i\nu}\ket{\dot\psi_t^{\omega}} = \left(H-\hslash\dot\nu\mathds{1}\right)\ee^{-i\nu}\ket{\psi_t^{\omega}}.
\end{align}
To simplify the notation, we shall from now on refer to $\hslash\dot\nu\mathds{1}$ as $\hslash\dot\nu$.
Division by $\ee^{-i\nu}$ gives
\begin{align}
\label{eq:redefinedVarQRTE}
     i\hslash\ket{\dot\psi^{\omega}_t} = \left(H-\hslash\dot\nu\right)\ket{\psi^{\omega}_t}.
\end{align}
Sine, we employ a variational ansatz $\ket{\psi^{\omega}_t}$ Eq.~\eqref{eq:redefinedVarQRTE} may not hold exactly. The goal of McLachlan's variational principle \cite{McLachlan64} is, now, to find parameter updates $\bm{\dot\omega}$ which minimize the error in Eq.~\eqref{eq:redefinedVarQRTE} w.r.t.~the $\ell_2$-norm $\norm{x}_2 = \sqrt{\langle x,x\rangle}$, i.e.,
\begin{align}
\label{eq:MacLachlan_phase_real}
 \delta\norm{i\hslash\ket{\dot\psi^{\omega}_t} - \left(H-\hslash\dot\nu\right)\ket{\psi^{\omega}_t}}_2&= 0.
\end{align}
Furthermore, we can find an explicit expression for $\dot\nu$ by evaluating the corresponding variational principle, i.e.,
\begin{align}
     \delta_{\dot\nu}\norm{i\hslash\ket{\dot\psi^{\omega}_t} - \left(H-\hslash\dot\nu\right)\ket{\psi^{\omega}_t}}_2&= 0
\end{align}
which leads to
\begin{align} \label{eq_nu_dot}
     \dot\nu = \frac{1}{\hslash}E_t^{\omega} + \text{Im}\left(\braket{\dot\psi^{\omega}_t|\psi^{\omega}_t}\right),
\end{align}
where $E_t^{\omega}:=\bra{\psi^{\omega}_t}H\ket{\psi^{\omega}_t}$.

Finally, we can see that Eq.~\eqref{eq:Schroedinger_real_nu} defines an evolution for $\ket{\psi^{\omega}_t}$ that simulates the existence of the global phase parameter $\nu$ without actually integrating or tracking $\ee^{-i\nu}$. We, now, rewrite this equation as
\begin{align}
\label{eq:MacLachlan_phase_agnostic_real}
     i\hslash\ket{\dot\psi^{\dot\nu}_t}= H\ket{\psi^{\omega}_t},
\end{align}
where $\ket{\dot\psi^{\dot\nu}_t} := \ket{\dot\psi^{\omega}_t} -\frac{i}{\hslash}(E_t^{\omega}+\hslash \text{Im}(\braket{\dot\psi^{\omega}_t|\psi^{\omega}_t}))\ket{\psi^{\omega}_t}$ represents the effective state gradient.
McLachlan's variational principle now implies
\begin{align}
\label{eq:MacLachlan_mixed_real}
     \delta\| i\hslash\ket{\dot\psi^{\dot\nu}_t} - H\ket{\psi^{\omega}_t} \|_2&= 0.
\end{align}

Since $\ket{\psi^{\omega}_t}$ is given by a parameterized quantum circuit, solving Eq.~\eqref{eq:MacLachlan_mixed_real} with $\ket{\dot\psi^{\omega}_t} = \sum_i\dot\omega_i\frac{\partial\ket{\psi^{\omega}_t}}{\partial\omega_i}$
results in 
 \begin{align}
 \label{eq:McLachlanVarQRTE}
    \hslash \sum\limits_{j=0}^{k}\mathcal{F}^Q_{ij} \dot\omega_j = \text{Im}\left( C_i -  \frac{\partial \bra{\psi^{\omega}_t}}{\partial \omega_i}\ket{\psi^{\omega}_t}E_t^{\omega}\right),
\end{align}
where $C_i = \bra{\psi^{\omega}_t} H\frac{\partial \ket{\psi^{\omega}_t}}{\partial \omega_i}$ and
$\mathcal{F}_{ij}^Q$ denotes the $(i,j)$-entry of the Fubini-Study metric \cite{QFIMBraunstein94, meyer2021fisher} given by
\begin{align}
\mathcal{F}^Q_{ij} =
\text{Re}\left(\frac{\partial \bra{\psi^{\omega}_t}}{\partial \omega_i}\frac{\partial \ket{\psi^{\omega}_t}}{\partial \omega_j} - \frac{\partial \bra{\psi^{\omega}_t}}{\partial \omega_i}\proj{\psi^{\omega}_t}\frac{\partial \ket{\psi^{\omega}_t}}{\partial \omega_j}\right). \nonumber
\end{align}
The efficient evaluation of the terms in Eq.~\eqref{eq:McLachlanVarQRTE} is discussed in Sec.~\ref{sec:analytic_gradients}.

Solving the system of linear equations (SLE) \cite{Liesen2015LinAlg} from Eq.~\eqref{eq:McLachlanVarQRTE} for $\bm{\dot{\omega}}$ defines an ordinary differential equation (ODE) that describes the evolution of the parameters, i.e.,
\begin{align}
\label{eq:standardODE_varqrte}
    f_{\text{std}}\left(\bm{\omega}\right) = \frac{1}{\hslash}\left(\mathcal{F}^Q\right)^{-1}\text{Im}\left( \bm{C} -  \frac{\partial \bra{\psi^{\omega}_t}}{\partial \bm{\omega}}\ket{\psi^{\omega}_t}E_t^{\omega}\right),
\end{align}
where $\bm{C} = \left(C_0, \ldots, C_k\right)$.

Notably, the gradient $\ket{\dot\psi^{\dot\nu}_t}$ computed on the basis of this variational principle will not always be exact. Therefore, $\| \ket{e_t}\| > 0$
where
\begin{align}
\label{eq:gradientErrorVarQRTE}
     \ket{e_t} := \ket{\dot\psi^{\dot\nu}_t}+ \frac{i}{\hslash}H\ket{\psi^{\omega}_t}
\end{align}
denotes the gradient error.
This gradient error also enables an alternative definition of a VarQRTE ODE. More explicitly, we can define the ODE as optimization problem that searches for the minimizing argument of the squared gradient error norm,
\begin{es}
\label{eq:argminODE_varqrte}
  f_{\text{min}}\left(\bm{\omega}\right)=\underset{\boldsymbol{\dot{\omega}}\in\mathbb{R}^{k+1}}{\text{argmin}} \, \|\ket{e_{t}}\|_2 ^2,
\end{es}
which can also be written as
\begin{align}
\label{eq:qrte_error}
    \norm{\ket{e_{t}}}_2  ^2 &= \Var\left(H\right)_{\psi^{\omega}_t} + \left( \braket{\dot\psi^{\omega}_t|\dot\psi^{\omega}_t} - \braket{\dot\psi^{\omega}_t|\psi^{\omega}_t}\braket{\psi^{\omega}_t|\dot\psi^{\omega}_t}\right)  \nonumber \\
    &\hspace{5mm}- 2\textnormal{Im}\left(\bra{\dot\psi^{\omega}_t}H\ket{\psi^{\omega}_t}-E_t^{\omega} \braket{\dot\psi^{\omega}_t|\psi^{\omega}_t} \right),
\end{align}
with $\Var(H)_{\psi^{\omega}_t} = \bra{\psi^{\omega}_t} H^2\ket{\psi^{\omega}_t} - (E_t^{\omega})^2$ and $\textstyle{
 2\text{Re}(\braket{\psi^{\omega}_t|\dot\psi^{\omega}_t}) = \frac{\partial\braket{\psi^{\omega}_t|\psi^{\omega}_t}}{\partial t} = 0}.$

Since the time-dependence of $\ket{\psi^{\omega}_t}$ is encoded in the real parameters $\bm{\omega}$,
Eq.~\eqref{eq:qrte_error} can, therefore, be rewritten as
\begin{align}
\label{eq:error_rqte_var}
    \|\ket{e_t}\|_2^2 &= \Var(H)_{\psi^{\omega}_t} + \sum\limits_{i, j}\dot \omega_i \dot \omega_j \mathcal{F}^Q_{ij} \nonumber \\ 
    &\hspace{5mm}- 2\sum_i\dot\omega_i\text{Im}\left(C_i -  \frac{\partial \bra{\psi^{\omega}_t}}{\partial \omega_i}\ket{\psi^{\omega}_t}E_t^{\omega}\right)
\end{align}
using
\begin{align}
     &\text{Im}\left(\bra{\dot \psi^{\omega}_t}H\ket{\psi^{\omega}_t}-E_t^{\omega}\braket{\dot \psi^{\omega}_t|\psi^{\omega}_t}\right) \nonumber \\ 
     &\hspace{10mm}= \sum_i\dot\omega_i\text{Im}\left( C_i -  \frac{\partial \bra{\psi^{\omega}_t}}{\partial \omega_i}\ket{\psi^{\omega}_t}E_t^{\omega}\right) 
\end{align}
and
\begin{align}
   \braket{\dot\psi^{\omega}_t|\dot\psi^{\omega}_t} - \braket{\dot\psi^{\omega}_t|\psi^{\omega}_t}\braket{\psi^{\omega}_t|\dot\psi^{\omega}_t} = \sum_{i,j}\dot \omega_i \dot \omega_j \mathcal{F}^Q_{ij}.
\end{align}
This reformulation enables the evaluation of the norm since the individual terms can be efficiently computed.

\section{Error Bound}
\label{subsec:error_qrte}
An error bound for VarQRTE in terms of the $\ell_2$-norm has been derived in~\cite{MartinazzoErrorVarQuantumDyn20}. We extend this bound to the Bures metric given in Eq.~\eqref{eq:Bures} and, thereby, avoid that a physically irrelevant mismatch in the global phase influences the bound.

\begin{theorem} \label{thm_VQRT}
For $T>0$, let $\ket{\psi^*_T}$ be the exact solution to Eq.~\eqref{eq:MacLachlan_phase_agnostic_real} and $\ket{\psi^{\omega}_T}$ correspond to the VarQRTE approximation. Then
\begin{align}
 B\left(\proj{\psi^*_T}, \proj{\psi^{\omega}_T}\right) \leq \epsilon_{T},
\end{align}
for $\epsilon_{T} = \int_{0}^{T} \dot\varepsilon_t \di t$ with the error rate given by the gradient error, i.e.,
\begin{align}
    \dot\varepsilon_t = \norm{\ket{e_t}}_2. \label{eq:varqrte_error_grad}
\end{align}
\end{theorem}

\begin{proof}
For $\delta_t > 0$, let the state evolution be defined with respect to the effective gradient given in Eq.~\eqref{eq:MacLachlan_phase_agnostic_real}
    \begin{align} 
  \ket{\psi^{\omega}_{t+\delta_t}  }&= \ket{\psi^{\omega}_{t}  } + \delta_t\ket{\dot\psi^{\dot\nu}_t} \nonumber \\
  &=\ket{\psi^{\omega}_{t} } \!+\! \delta_t\!\left(\!\ket{\dot\psi^{\omega}_{t}}\!-\!\frac{i}{\hslash}E_t^{\omega}  \!-\!i \text{Im}\left(\!\braket{\dot\psi^{\omega}_t|\psi^{\omega}_t}\!\right)\ket{\psi^{\omega}_t} \!\right)\!. \label{eq_step1}
    \end{align}
For the remainder of the proof, we use the simplified notation
$B\left(\ket{v}, \ket{w}\right)$ when referring to $B\left(\proj{v}, \proj{w}\right)$.
Combining Eq.~\eqref{eq_step1} with the triangle inequality gives 
\begin{align}
    B\left(\ket{\psi^{\omega}_{t+\delta_t}}, \ket{\psi^*_{t+\delta_t}}\right) 
      &\leq B\Big(\ket{\psi^{\omega}_{t+\delta_t}}, \big(\mathds{1}\!-\!\frac{i\delta_t}{\hslash}H\big)\ket{\psi^{\omega}_{t}}\! \Big)\nonumber  \\
      &\hspace{-6mm}+B\Big(\big(\mathds{1}-\frac{i\delta_t}{\hslash}H\big)\ket{\psi^{\omega}_{t}}, \ket{\psi^*_{t+\delta_t}}\Big). \label{eq_triangle_start}
\end{align}
Using Eq.~\eqref{eq_buresPhase} and neglecting terms of order $\mathscr{O}(\delta_t^2)$ gives
\begin{align}
     &B\Big(\ket{\psi^{\omega}_{t+\delta_t}}, \big(\mathds{1}-\frac{i\delta_t}{\hslash}H\big)\ket{\psi^{\omega}_{t}} \Big) \nonumber \\
     &\hspace{5mm}=\min_{\phi \in [0,2\pi]}\norm{\ee^{i\phi}\big(\ket{\psi^{\omega}_{t+\delta_t}} \big) - \big(\mathds{1}-\frac{i\delta_t}{\hslash}H\big)\ket{\psi _{t}}}_2 \nonumber\\
     &\hspace{5mm}\leq \norm{\ket{\psi^{\omega}_{t+\delta_t}}- \big(\mathds{1}-\frac{i\delta_t}{\hslash}H\big)\ket{\psi _{t}}}_2 \nonumber\\
     &\hspace{5mm}=\delta_t \norm{\ket{\dot\psi^{\omega}_t}  + \frac{i}{\hslash}\big(H\!-\!E_t^{\omega}\!-\!\hslash\,\text{Im}(\braket{\dot\psi^{\omega}_t|\psi^{\omega}_t}) \big)\ket{\psi^{\omega}_t}
    }_2 \nonumber \\
     &\hspace{5mm}=: \delta_t\norm{\ket{e_t}}_2 , \label{eq_triangle_part1}
\end{align}
where the penultimate step uses Eq.~\eqref{eq_nu_dot}.

To bound the second term in Eq.~\eqref{eq_triangle_start}, we again use Eq.~\eqref{eq_buresPhase}, neglecting terms of order $\mathscr{O}(\delta_t^2)$ to obtain
\begin{align}
     &B\Big(\big(\mathds{1}-\frac{i\delta_t}{\hslash}H\big)\ket{\psi^{\omega}_t}, \big(\mathds{1}-\frac{i\delta_t}{\hslash}H\big)\ket{\psi^*_t}\Big) \nonumber  \\
     &\hspace{5mm}= \min_{\phi \in [0,2\pi]} \norm{\big(\mathds{1}-\frac{i\delta_t}{\hslash}H\big)\big(\ee^{i\phi}\ket{\psi^{\omega}_t}-\ket{\psi^*_t}\big)}_2 \nonumber  \\
     &\hspace{5mm}\leq \min_{\phi \in [0,2\pi]} \norm{\ee^{-\frac{i\delta_t}{\hslash}H}\big(\ee^{i\phi}\ket{\psi^{\omega}_t}-\ket{\psi^*_t}\big)}_2\nonumber  \\
     &\hspace{10mm}+\norm{\Big(\big(\mathds{1}-\frac{i\delta_t}{\hslash}H\big)-\ee^{-\frac{i\delta_t}{\hslash}H} \big)\big(\ee^{i\phi}\ket{\psi^{\omega}_t}-\ket{\psi^*_t}\big)\Big)}_2\nonumber   \\
     &\hspace{5mm}=\min_{\phi \in [0,2\pi]} \norm{\ee^{-\frac{i\delta_t}{\hslash}H}\big(\ee^{i\phi}\ket{\psi^{\omega}_t}-\ket{\psi^*_t}\big)}_2 \nonumber \\
     &\hspace{5mm}= B(\ket{\psi^{\omega}_t}, \ket{\psi_t^*}). \label{eq_triangle_part2}
\end{align}
The penultimate step uses that $\mathds{1} -i \frac{\delta_t}{\hslash} H $ is the first Taylor expansion of \smash{$\exp(-i \frac{\delta_t}{\hslash} H)$} and drops terms of order $\mathscr{O}\left(\delta_t^2\right)$.

Combining Eqs~\eqref{eq_triangle_start},~\eqref{eq_triangle_part1} and~\eqref{eq_triangle_part2} gives
\begin{align}
     B\left(\ket{\psi^{\omega}_{t+\delta_t}}, \ket{\psi^*_{t+\delta_t}}\right)
     \leq B\left(\ket{\psi^{\omega}_t}, \ket{\psi_t^*}\right) + \delta_t \norm{\ket{e_t}}_2.
\end{align}
Assuming that $B\left(\ket{\psi_0}, \ket{\psi_0^*}\right) = 0$, we can evolve
\begin{align}
B\left(\ket{\psi^{\omega}_{T}}, \ket{\psi_{T}^*}\right) = \delta_t\sum\limits_{k=0}^{K}\|\ket{e_{k\delta_t}}\|_2,
\end{align}
where $K$ corresponds to the number of time steps.
Finally setting $\delta_t = T/t$ leads to
\begin{align}
     B\left(\ket{\psi^{\omega}_{T}}, \ket{\psi_{T}^*}\right) \leq \int_{0}^{T} \|\ket{e_t}\|_2 \, \di t := \epsilon_T,
\end{align}
which proves the assertion.\qed
\end{proof}

It should be noted that the error bound is compatible with practical VarQRTE implementations which use, e.g., regularized least squares methods or pseudo-inversion methods to solve the system of linear equations from Eq.~\eqref{eq:McLachlanVarQRTE}.

\begin{figure}[!h]
    \centering
    \captionsetup{singlelinecheck = false, format= hang,justification=centerlast, font=footnotesize, labelsep=space}
    \begin{tikzpicture}
\node at (1, 1.1) {\textbf{VarQRTE: $H_{\text{illustrative}}$ with different ODE solvers}};
\node[inner sep=0pt, anchor=north west] at (-6, -0.12) {\includegraphics[width=0.4\textwidth]{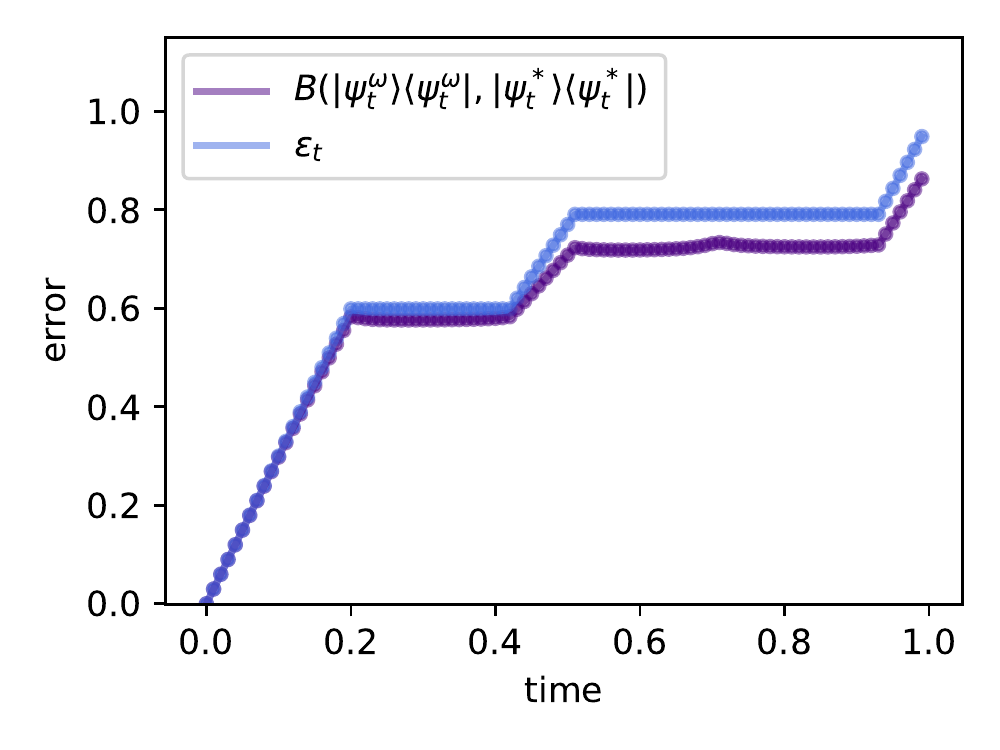}};
\node[anchor=north west] at (1,0) {    \includegraphics[width=0.4\textwidth]{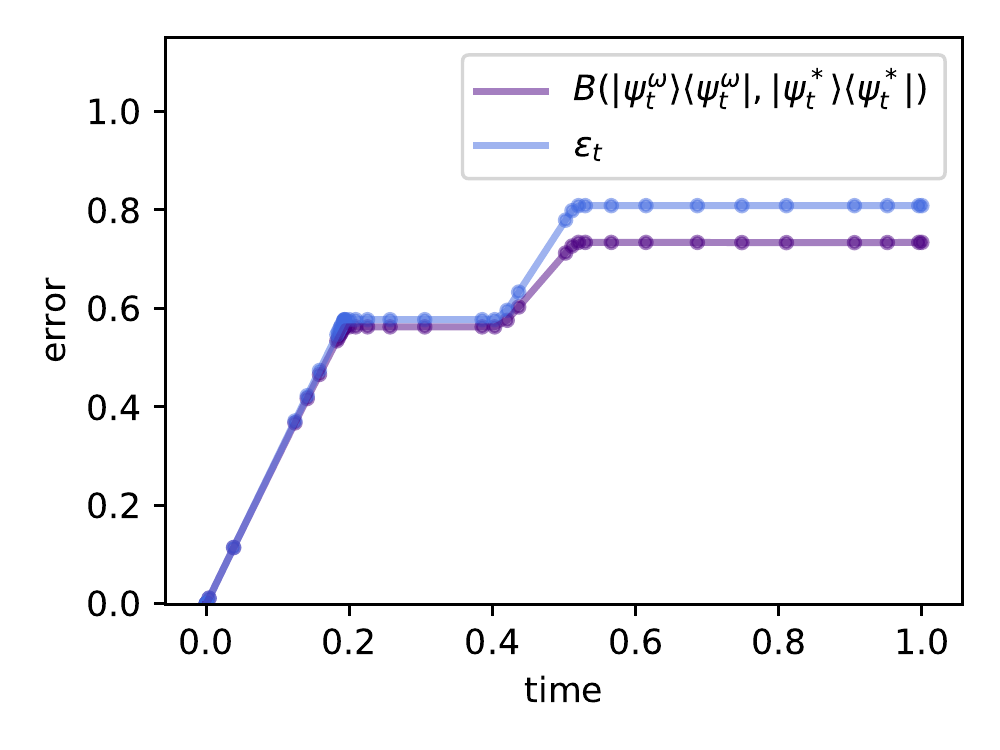}};
\node[anchor=north west] at (-6, -4.7) {\includegraphics[width=0.4\textwidth]{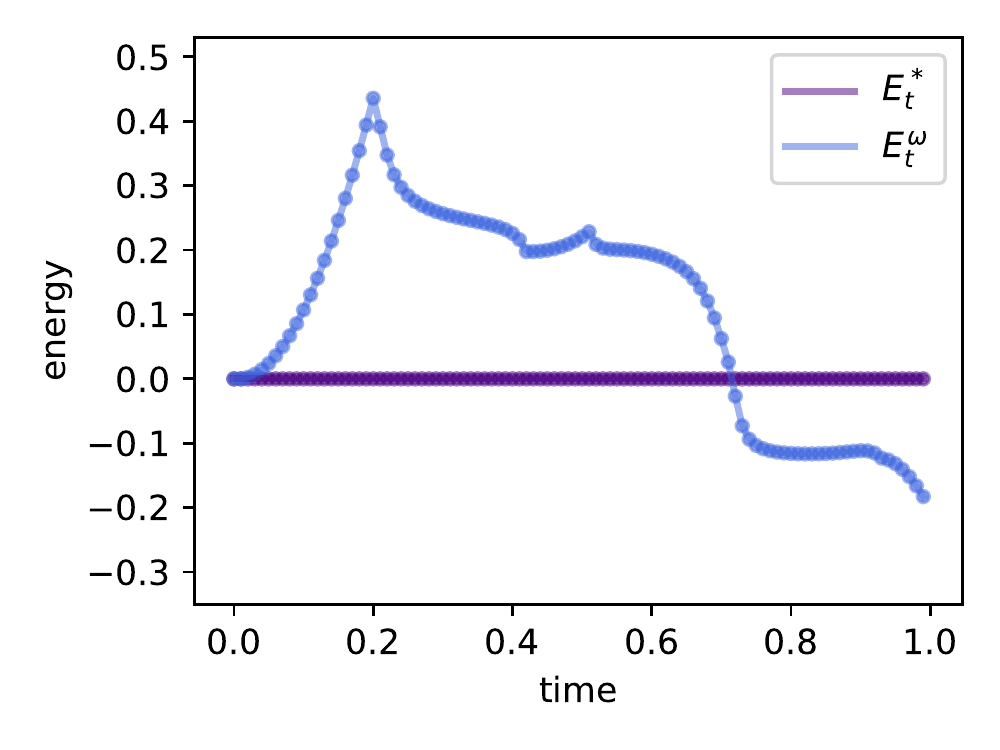}};
\node[anchor=north west] at (1,-4.7) {\includegraphics[width=0.4\textwidth]{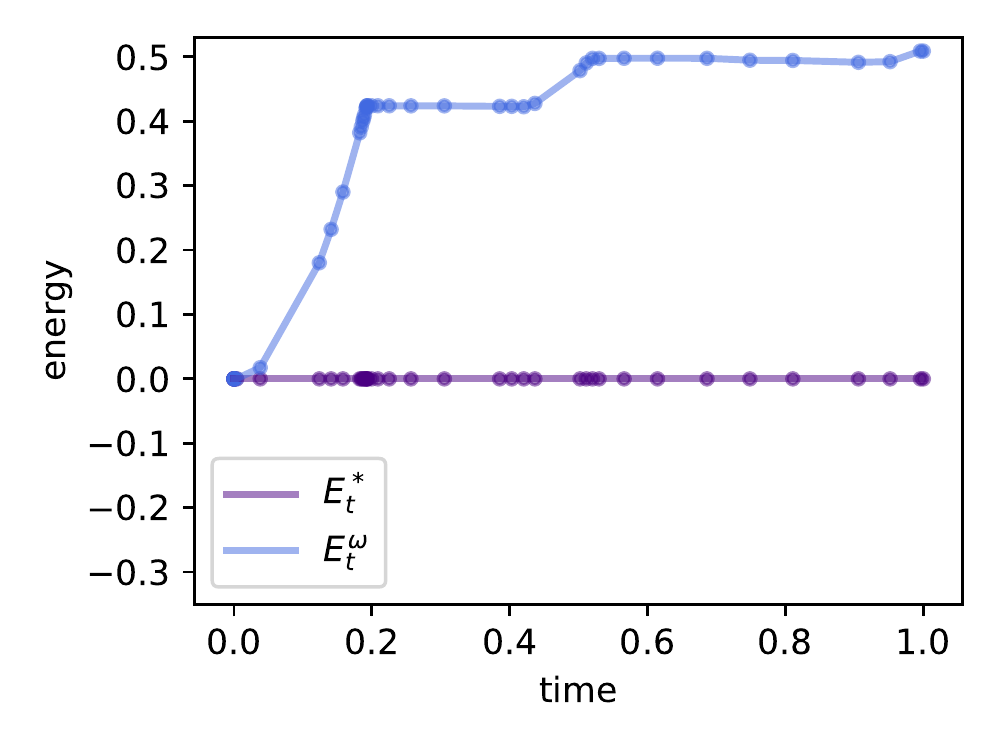}};
\node at (-3.2, 0) {(a) \small{State Error (Euler, $f_{\text{std}}$) }};
\node at (4.3, 0) {(b)  \small{State Error (RK54, $f_{\text{std}}$)} };
\node at (-3.2, -4.7) {(c)  \small{Energy (Euler, $f_{\text{std}}$)} };
\node at (4.3, -4.7) {(d)\small{ Energy (RK54, $f_{\text{std}}$) }};
\end{tikzpicture}
    \caption{VarQRTE for $\ket{\psi_0}=\ket{++}$, $H_{\text{illustrative}}$ and $T=1$  with the standard ODE. (a), (c) employ forward Euler. (b), (d) use RK54. -  (a), (b) illustrate the error bounds $\epsilon_t$ and the actual Bures metric between $\ket{\psi^*_t}$ and  $\ket{\psi^{\omega}_t}$.  (c), (d) show the corresponding energies.}
    \label{fig:illustrative_varqrte}
\end{figure}
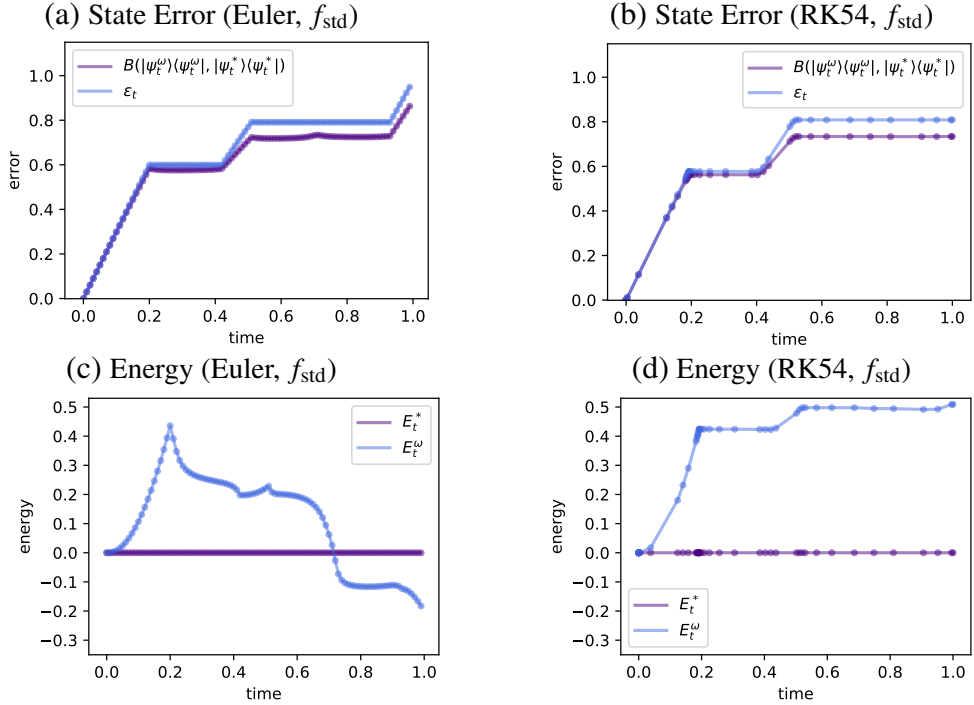

\section{Examples}
\label{app:varqrte_experiments}
\begin{figure}[!ht]
    \centering
    \captionsetup{singlelinecheck = false, format= hang, justification=centerlast, font=footnotesize, labelsep=space}
    \begin{tikzpicture}
\node at (1, 1.1) {\textbf{VarQRTE: $H_{\text{Ising}}$ with different ODE definitions}};
\node[inner sep=0pt, anchor=north west] at (-6, -0.12) {\includegraphics[width=0.4\textwidth]{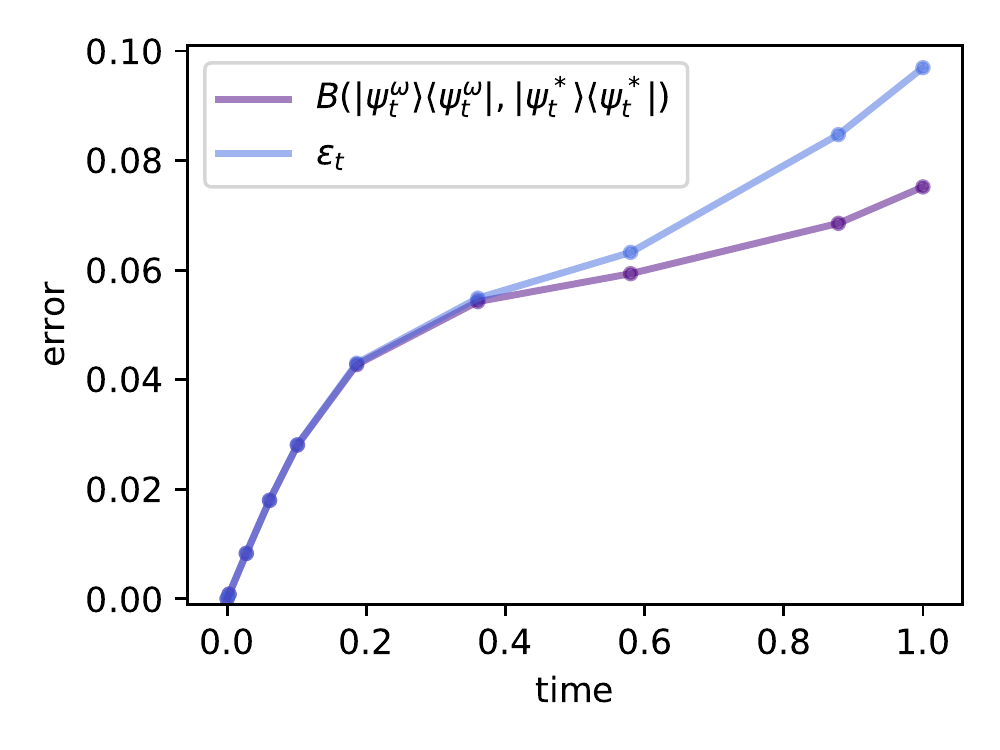}};
\node[anchor=north west] at (1,0) {    \includegraphics[width=0.4\textwidth]{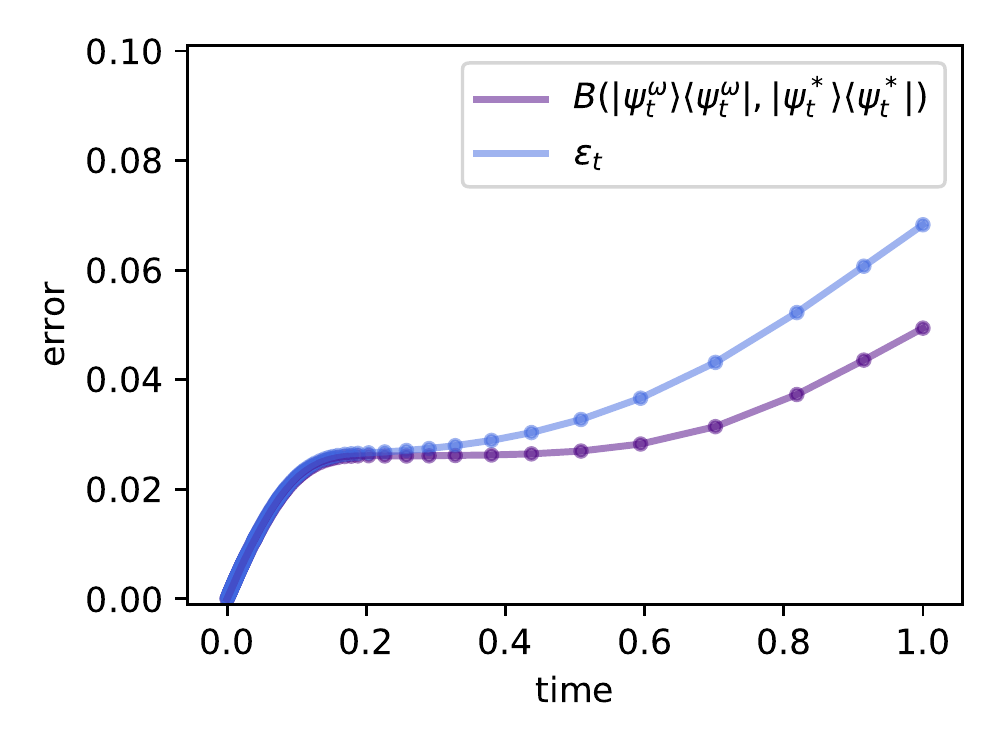}};
\node[anchor=north west] at (-6, -4.7) {\includegraphics[width=0.4\textwidth]{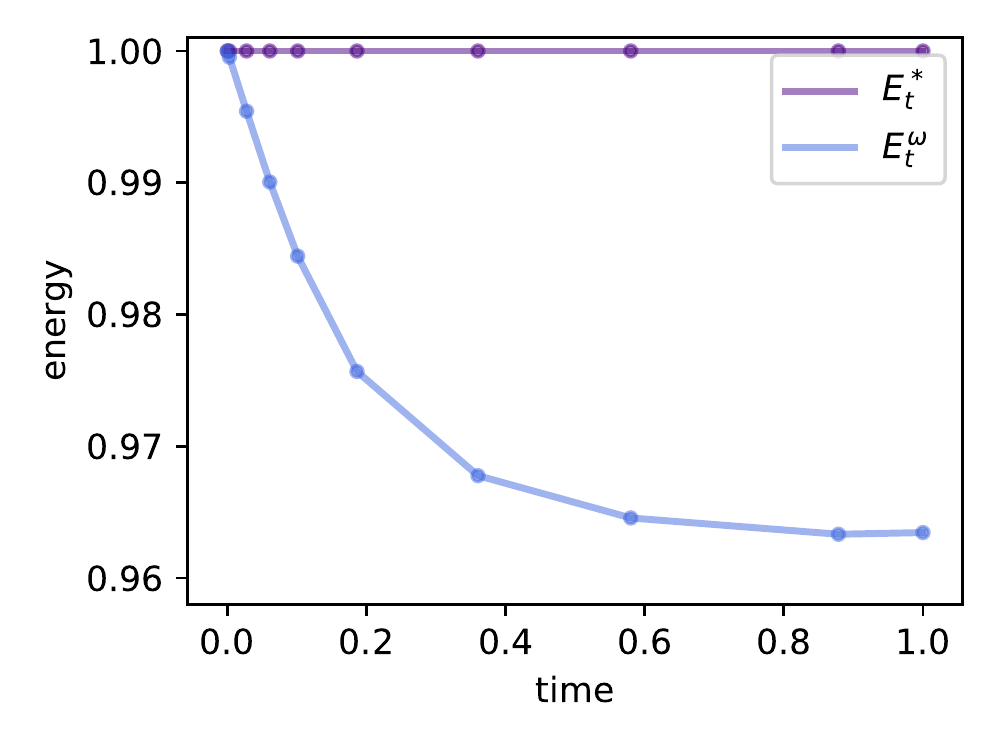}};
\node[anchor=north west] at (1,-4.7) {\includegraphics[width=0.4\textwidth]{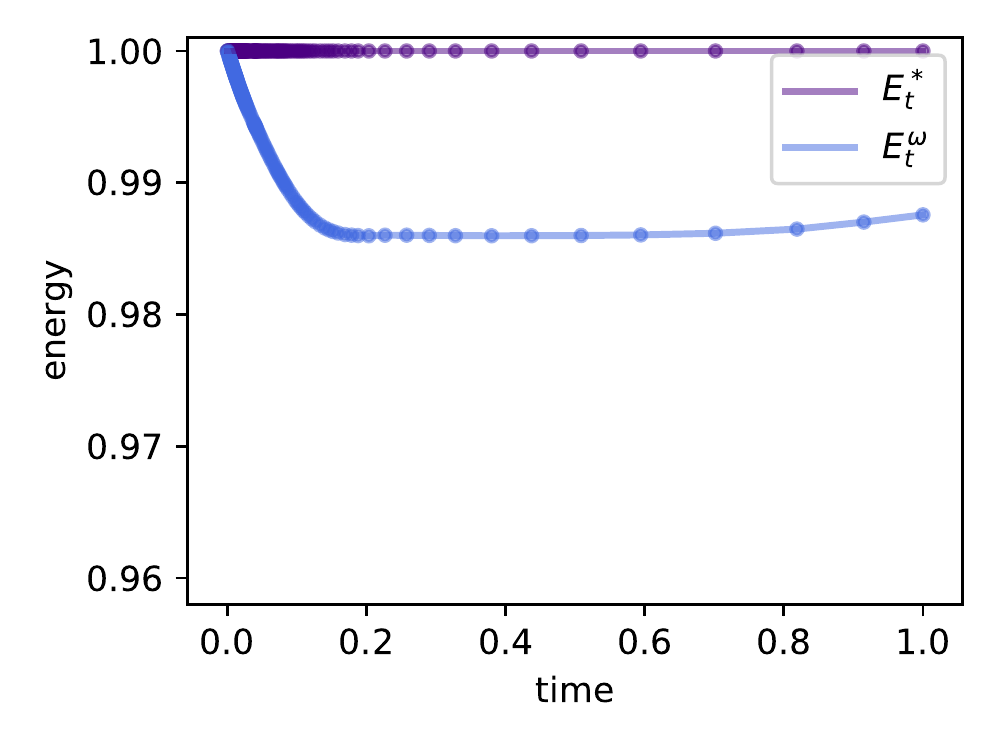}};
\node at (-3.2, 0) {(a) \small{State Error (RK54, $f_{\text{std}}$)}};
\node at (4.3, 0) {(b) \small{State Error (RK54, $f_{\text{min}}$)}};
\node at (-3.2, -4.7) {(c) \small{Energy (RK54, $f_{\text{std}}$)}};
\node at (4.3, -4.7) {(d) \small{Energy (RK54, $f_{\text{min}}$)}};
\end{tikzpicture}
    \caption{VarQRTE for $\ket{\psi_0}=\ee^{-i\alpha}\ket{000}$, $H_{\text{Ising}}$ and $T=1$ with RK54. (a), (c) are based on the standard ODE.  (b), (d) use the argmin ODE. - (a), (b) illustrate the error bounds $\epsilon_t$ and the actual Bures metric. (c), (d) show the energies of the prepared and the target state.}
    \label{fig:ising_varqrte}
\end{figure}

In the following, we present a set of numerical experiments and investigate the error bounds for VarQRTE with a particular focus on the comparison of different ODE formulations and solvers which are presented in Sec.~\ref{sec:methods}. For that purpose, we use the three example settings for $H_{\text{illustrative}}$, $H_{\text{Ising}}$ and $H_{\text{hydrogen}}$ described in Sec.~\ref{enumerate:Examples}.

Firstly, we apply the forward Euler method with $100$ time steps as well as an adaptive step size RK54 ODE solver to the illustrative example using the standard ODE. The resulting SLEs are solved using a least square solver provided by NumPy \cite{Numpy2020}.
The results shown in Fig.~\ref{fig:illustrative_varqrte} illustrate that the error bounds are very tight and, thus, relevant for practical accuracy estimations.  
Furthermore, one can see that RK54 achieves better error bounds as well as smaller fluctuations in the system energy while using significantly less time steps. The plateaus are due to exact local gradients, i.e., $\|\ket{e_{t}}\|_2 = 0$.
Furthermore, we would like to point out that the energy should actually be preserved for a real time evolution under a time-independent Hamiltonian but McLachlan's variational principle does not guarantee energy preservation.

\begin{figure*}[!ht]
    \captionsetup{singlelinecheck = false, format= hang, justification=centerlast, font=footnotesize, labelsep=space}
    \centering
    \begin{tikzpicture}
     \node at (1, 0.9) {\textbf{VarQRTE: $H_{\text{hydrogen}}$ with different ODE solvers and different ODE types}};
     \node[anchor=north west] at (-6, -0.12) {\includegraphics[width=0.4\textwidth]{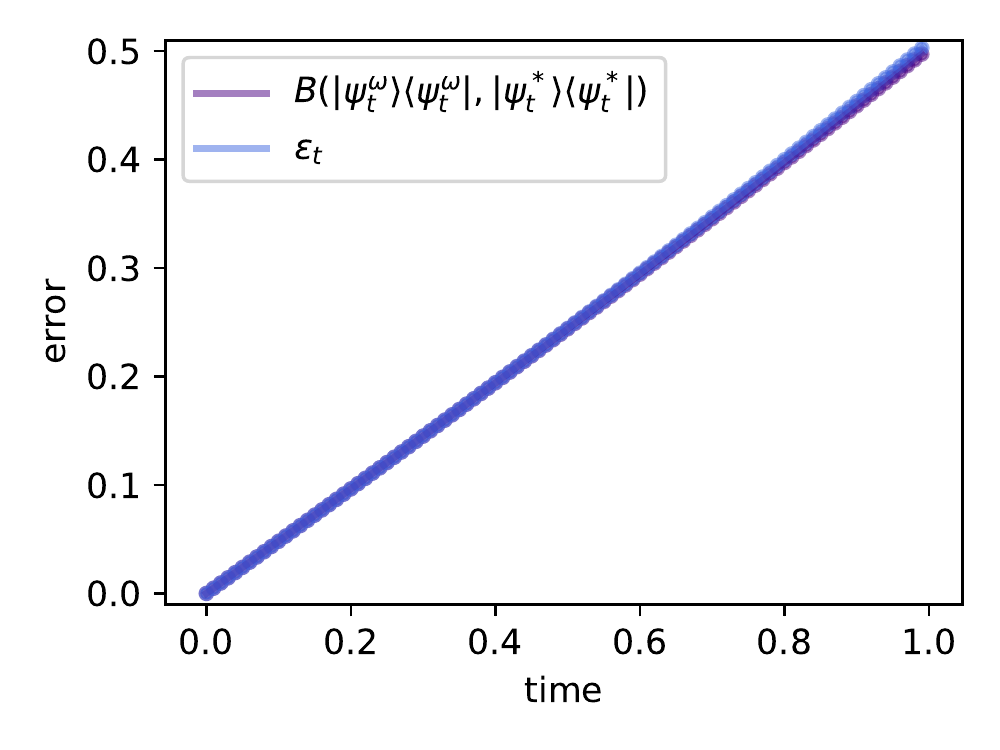}};
\node[anchor=north west] at (-6, -4.8) {\includegraphics[width=0.4\textwidth]{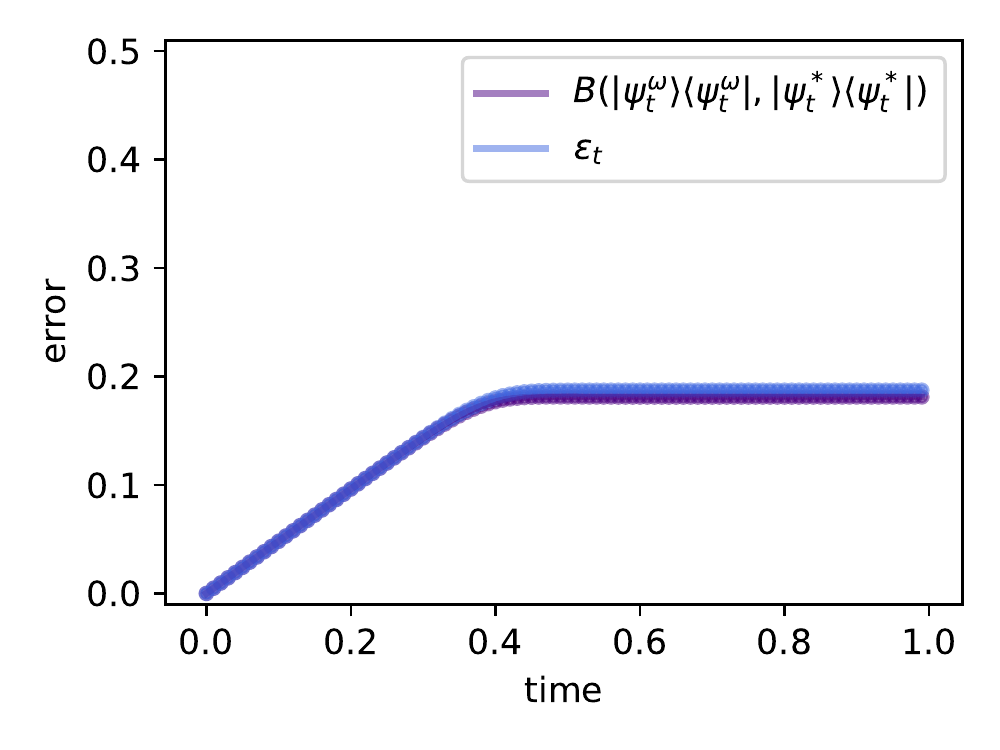}};
\node[anchor=north west] at (-6, -9.6) {\includegraphics[width=0.4\textwidth]{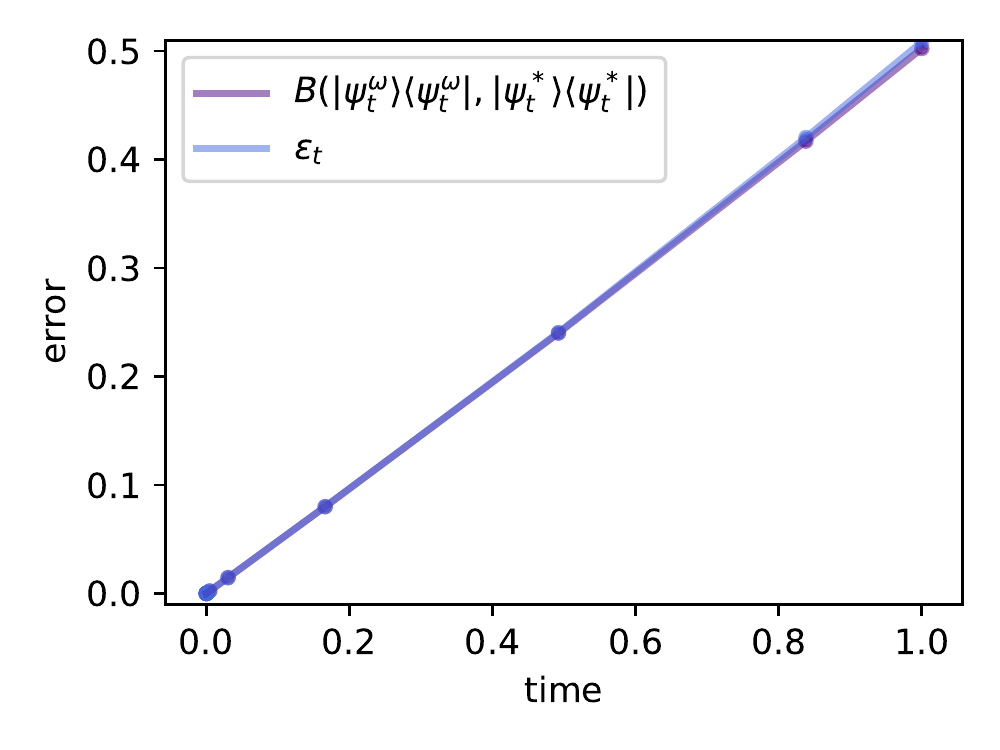}};
\node[anchor=north west] at (-6, -14.4) {\includegraphics[width=0.4\textwidth]{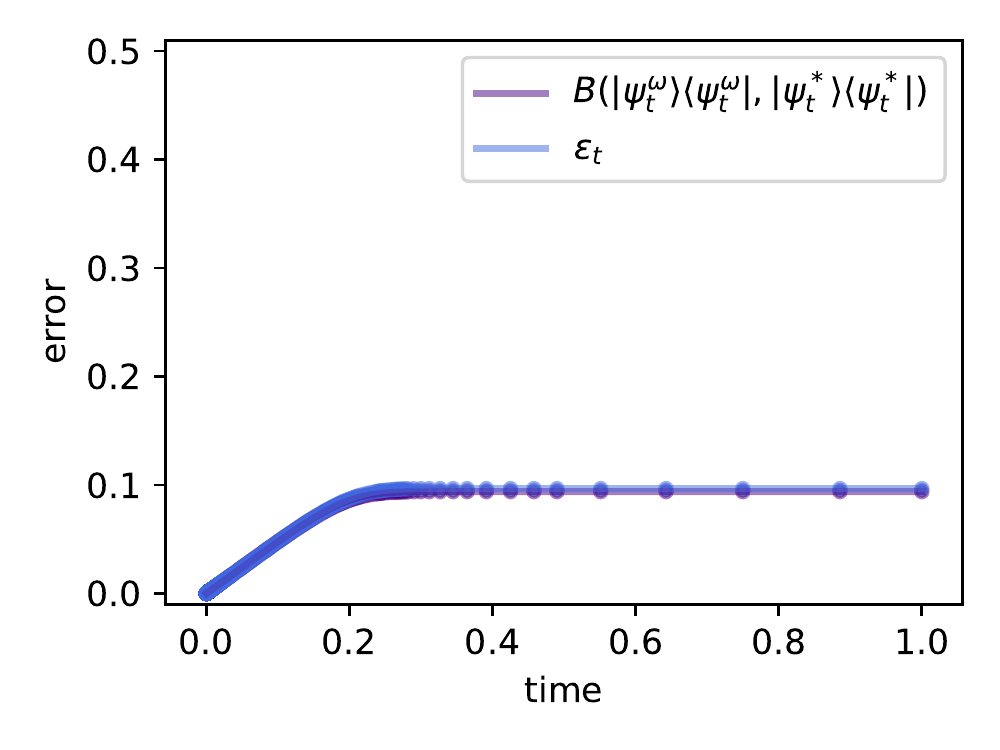}};
\node[anchor=north west] at (1, -0) {\includegraphics[width=0.4\textwidth]{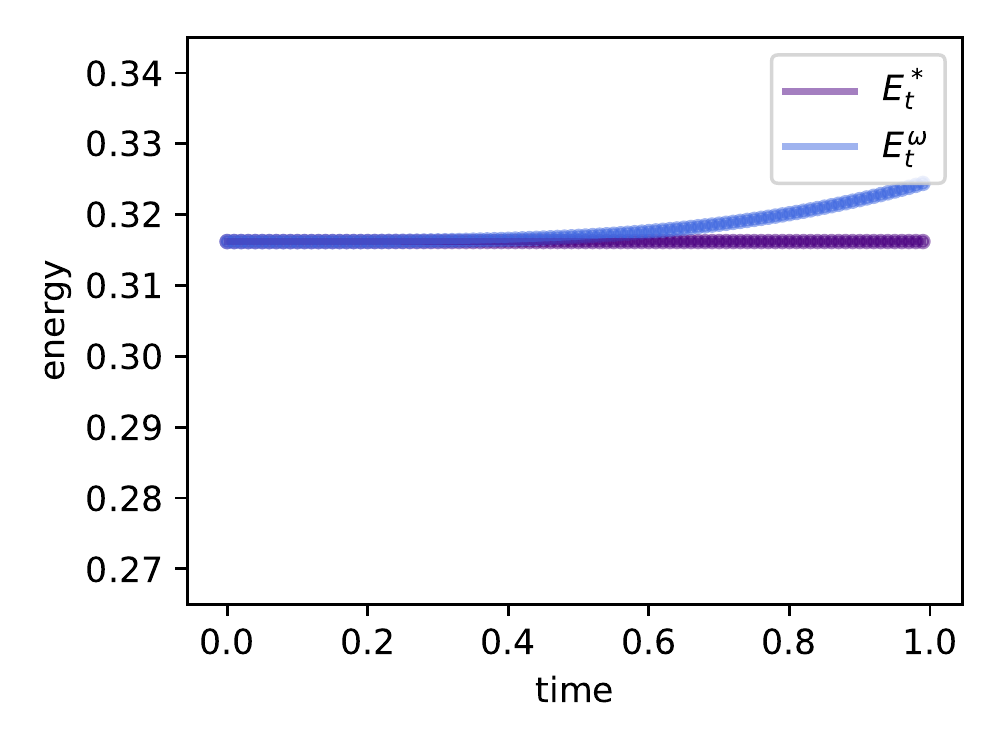}};
\node[anchor=north west] at (1, -4.8) {\includegraphics[width=0.4\textwidth]{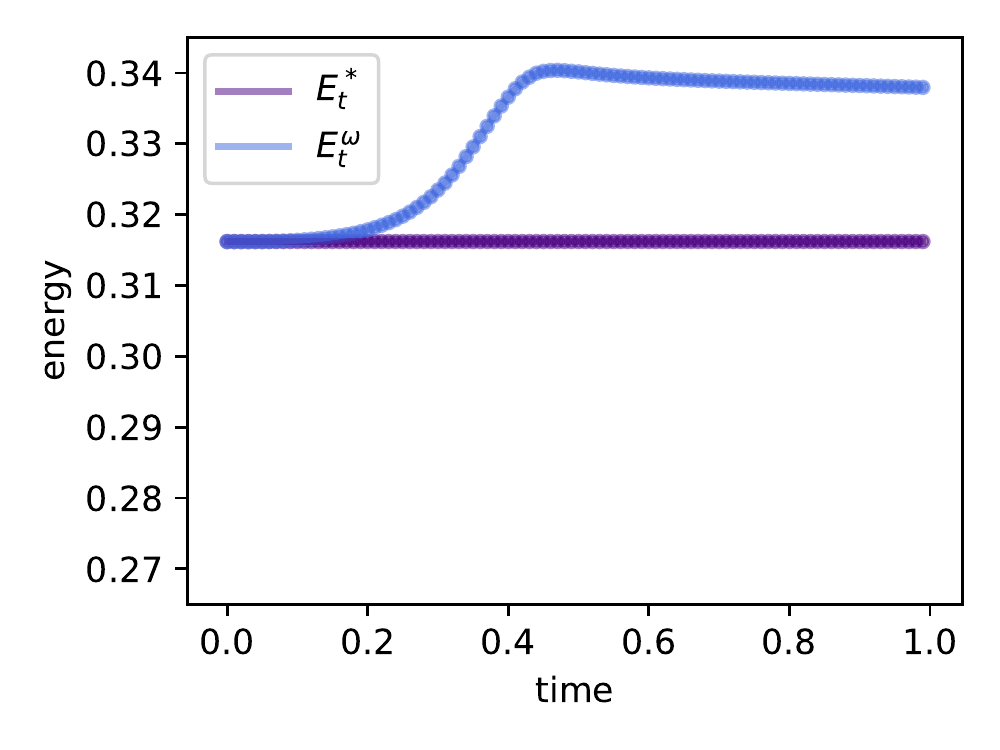}};
\node[anchor=north west] at (1, -9.6) {\includegraphics[width=0.4\textwidth]{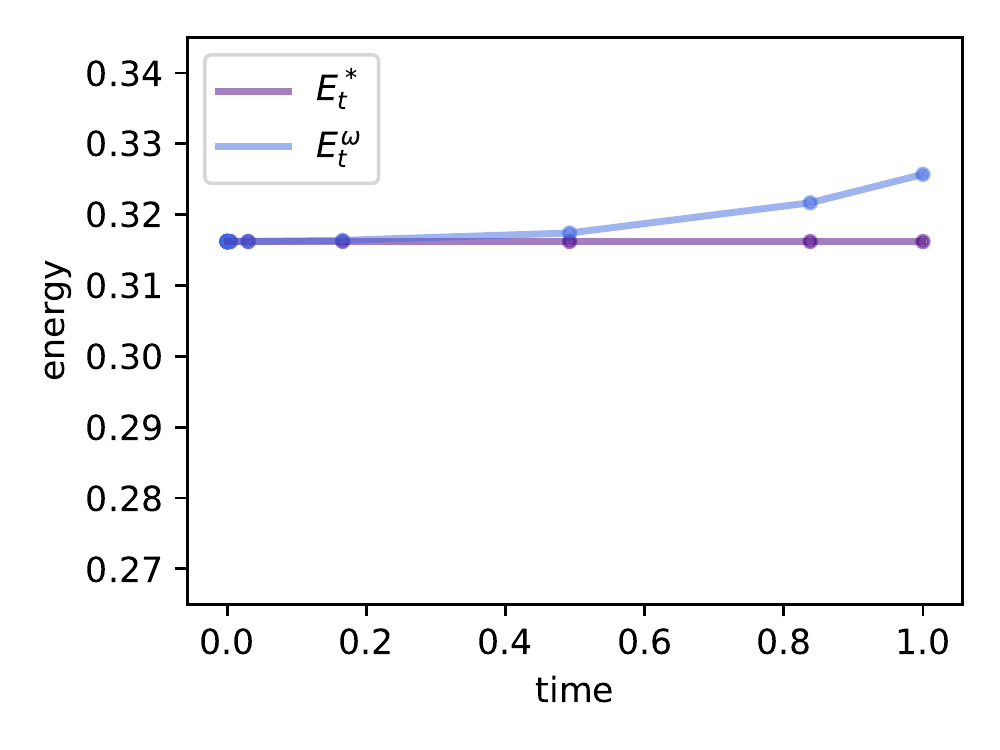}};
\node[anchor=north west] at (1, -14.4) {\includegraphics[width=0.4\textwidth]{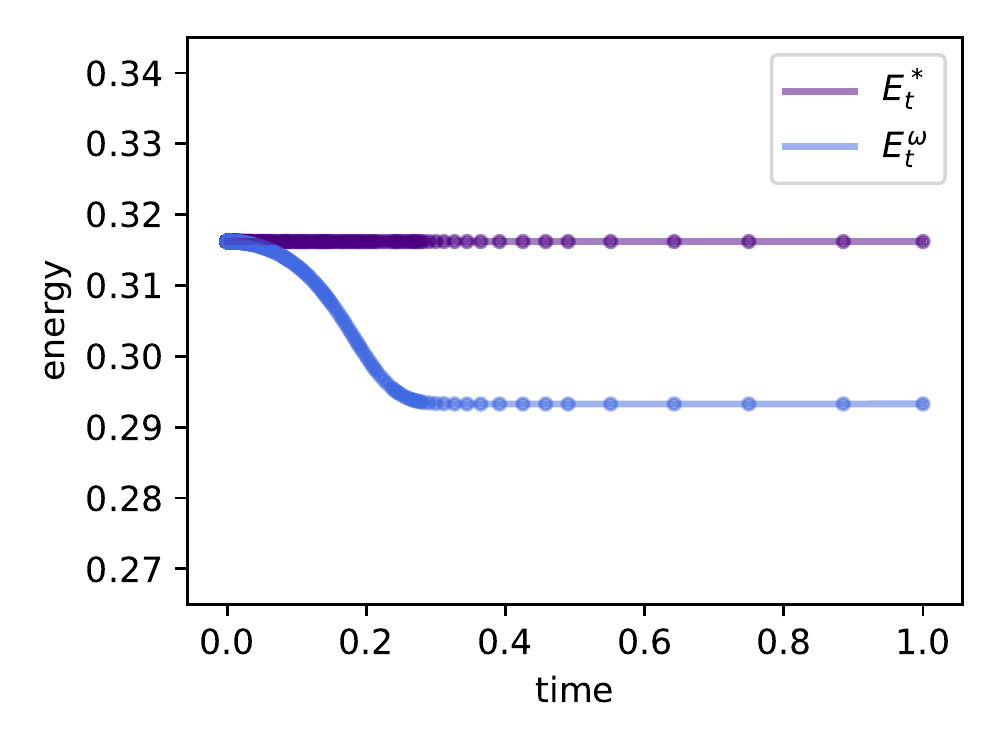}};
\footnotesize{
\node at (-3.2, 0) {(a) \small{State Error (Euler, $f_{\text{std}}$) }};
\node at (-3.2, -4.8) {(c) \small{State Error (Euler, $f_{\text{min}}$)}};
\node at (-3.2, -9.6) {(e) \small{State Error (RK54, $f_{\text{std}}$)}};
\node at (-3.2, -14.4) {(g) \small{State Error (RK54, $f_{\text{min}}$)}};
\node at (4.3, 0) {(b) \small{Energy (Euler, $f_{\text{std}}$)}};
\node at (4.3, -4.8) {(d) \small{Energy (Euler, $f_{\text{min}}$)}};
\node at (4.3, -9.6) {(f) \small{Energy (RK54, $f_{\text{std}}$)} };
\node at (4.3, -14.4) {(h) \small{Energy (RK54, $f_{\text{min}}$)}};}
\end{tikzpicture}
    \caption{VarQRTE for $\ket{\psi_0}=\ket{++}$, $H_{\text{hydrogen}}$ and $T=1$. (a), (b), (c), (d) employ a forward Euler solver. (e), (f), (g), (h) use an RK54 ODE solver. (a), (b), (e), (f) use the standard ODE. (c), (d), (g), (h) rely on the argmin ODE. -
    (a), (c), (e), (g) illustrate the error bounds $\epsilon_{t}$ and the true Bures metric. (b), (d), (f), (h) show $E_t^{\omega}$ and $E_t^*$.}
    \label{fig:hydrogen_varqrte_energy}
\end{figure*}

Next, we compare the implications of the ODE formulation on the Ising model using RK54.
The initial points for the optimization of the argmin ODE are chosen as the solution to the respective SLE at time $t$.
Fig.~\ref{fig:ising_varqrte} presents the Bures metrics as well as the respective bounds for the prepared $\ket{\psi^{\omega}_t}$ and  the target state $\ket{\psi^*_t}$. The errors show that the argmin ODE leads to smaller  errors than the standard ODE. Furthermore, it can be seen that also the system energy changes less for the former.


Lastly, the error bounds for the hydrogen Hamiltonian from Eq.~\eqref{eq:hydrogen} are compared the standard and argmin ODE as well as forward Euler and RK54 ODE solvers. 
The results are presented in Fig.~\ref{fig:hydrogen_varqrte_energy}. Notably, the experiment which uses RK54 and the argmin ODE leads to the best results, i.e., the smallest state error as well as error bound. 
In general, one can see that the argmin ODE achieves better errors compared to the standard ODE, the error seems to saturate for the former while it keeps increasing with the latter.
Furthermore, RK54 improves the errors as well as the error bounds for both ODE definitions while using significantly less time steps.
We would like to point out that the setting which gives to the smallest error $\epsilon_t$ does not necessarily lead to the smallest discrepancy between $E_t^{\omega}$ and $E_t^*$, as can be seen when comparing the RK54 results.

To sum this up, the numerical results reveal that the error bounds represent good estimates for the actual errors. The experiments indicate further that an adaptive step size ODE solver such as RK54 significantly improves the simulation results while reducing the computational costs. Moreover, it was shown that replacing the standard ODE by the argmin ODE has also a a positive influence on the simulation accuracy. 
Lastly, the results reveal that the lack of energy conservation in McLachlan's variational principle can lead to significant energy fluctuations.

\addtocontents{toc}{}

%% file: appendix_qbm.tex

\chapter[Discriminative Learning with Quantum Boltzmann Machines]{Appendix: Discriminative Learning with Quantum Boltzmann Machines\footnote{This section is reproduced in part, with permission, from C.~Zoufal, A.~Lucchi, S.~Woerner, "Variational Quantum Boltzmann Machines", \textit{Quantum Machine Intelligence}, vol. 3, Article Nr. 7, 2021}}
\label{app:qbms} 

\textbf{Abstract.} 
We outline the application of quantum Boltzmann machines to discriminative learning and present a classification example where a variational quantum Boltzmann machine learns to identify fraudulent credit card data.
\index{appendix QBMs}

\vspace{8mm}

\noindent

Classical and quantum Boltzmann machines are not only compatible with generative, see Sec.~\ref{sec:QBM}, but also with discriminative learning. To that end, the algorithm needs to be adapted such that it can learn a conditional probability distribution, i.e., the probability to measure a certain label given a particular input data sample.

The following Appendix, first, explains how QBMs can be used as discriminative models in Sec.~\ref{appendix:discriminativeqbm}. Then, Sec.~\ref{app:discr_qbm_results} gives an example of QBM classification w.r.t.~a credit card transactions data set and compares the model to a set of standard classifiers.
\section{Discriminative Quantum Boltzmann Machine}
\label{appendix:discriminativeqbm}

Given a training data set that corresponds to tuples of input data samples and labels $\left(\bm{x}, y\right)$ with
$\bm{x}\in\mathbb{R}^s$, discriminative learning with QBMs can be enabled by encoding the input data samples $\bm{x}$ into Hamiltonian weights, i.e., we employ Hamiltonian encoding, see Def.~\ref{sec:gibbs_related_work}. The nodes $z$ of the system constitute again of hidden $h$ and visible $v$ units such that $z=v\cup h$. The parameters of the Hamiltonian
\begin{align}
    H_z\left(\bm{\theta}, \bm{x}\right)=\sum\limits_{i}f_i\left(\bm{\theta}_i, \bm{x}\right)h_i
\end{align}
are given by a function $f_i\left(\bm{\theta}_i, \bm{x}\right)$ which maps the parameter vector $\bm{\theta}_i \in \mathbb{R}^s$ and $\bm{x}$ to a scalar in $\mathbb{R}$ and the corresponding Gibbs state for temperature $T$ reads
\begin{align}
    \rho_z\left(\bm{\theta}; \bm{x}\right) =\frac{e^{-H_z\left(\bm{\theta}, \bm{x}\right)/\left(\text{k}_{\text{B}}\text{T}\right)}}{Z},
\end{align}
where $k_B$ refers to the Boltzmann constant and $Z$ represents the partition function
\begin{equation}
    Z=\text{Tr}\left[e^{-H_z\left(\bm{\theta}, \bm{x}\right)/\left(\text{k}_{\text{B}}\text{T}\right)}\right].
\end{equation}

Furthermore, the conditional probability to sample label $y$ when given data sample $\bm{x}$ is given by
\begin{align}
   p_{y|x}\left(\bm{\theta}\right) = \text{Tr}\left[\Lambda_y\rho(\bm{\theta}; \bm{x})\right],
\end{align}
where $\Lambda_y = \proj{y}^v\otimes\mathds{1}^h$ represents a projector onto a configuration $y$ of the visible units. Notably, the goal of the algorithm is to train the parameters of $H_z\left(\bm{\theta}, \bm{x}\right)$ such that the probabilities to sample the correct labels for given data instances is maximized.
The respective loss function reads
\begin{align}
\label{eq:loss_sup}
	L\left(\bm{\theta}\right) = -\sum\limits_{x}p_x^{\text{data}}\sum\limits_{y}p_{y|x}^{\text{data}}\log{p_{y|x}\left(\bm{\theta}\right)},
\end{align}
where $p_x^{\text{data}}$ can be identified with the occurrence frequency of data point $\bm{x}$ in the training data set and 
\begin{align}
    p_{y|x}^{\text{data}} =     \begin{cases}
        1, \text{ if } y \:\text{is the correct label to}\: \bm{x}, \\
        0, \text{ otherwise.}
    \end{cases}
\end{align}

\section{Results}
\label{app:discr_qbm_results}

The following QBM training is based on $500$ artificially created credit card transactions \cite{Altman2019} where the labels denote whether an instance is fraudulent or non-fraudulent. In the given data set about $15\%$  correspond to fraudulent transactions. To avoid redundant state preparation, the training is run for all unique item instances in the data set and the results are averaged according to the item's occurrence frequency.
The data set includes the following features: location (ZIP code), time, amount, and Merchant Category Code (MCC) of the transactions.
To facilitate the training, the features of the given data set are discretized and normalized as follows. Using k-means clustering, each of the first three features are independently discretized to $3$ reasonable bins. Furthermore, we consider MCCs $<10 000$ and group them into $10$ different categories.
The discretization is discussed in more detail in Tbl.~\ref{tbl:discr_data_preproc}.
Furthermore, for each feature, we map the values $x$ to $x' = \frac{x-\mu}{\sigma}$ with $\mu$ denoting the mean and $\sigma$ denoting the standard deviation. 
\begin{table}[!ht]
\captionsetup{singlelinecheck = false, format= hang, justification=raggedright, font=footnotesize, labelsep=space}
\begin{center}
\footnotesize{
\begin{tabular}{c|c|c}
\textbf{Feature} & \textbf{Condition} & \textbf{Value} \\
 \hline
\multirow{3}{*}{time} & $0$AM $- 11$AM  & 0 \\
& $11$AM $- 6$PM  & 1 \\
& $6$PM - $0$AM & 2 \\
\hline
\multirow{3}{*}{amount} & $<\$50 $ & 0 \\
& $\$50-150 $ & 1 \\ 
& $>\$150$ & 2 \\
\hline
\multirow{3}{*}{ZIP} & east & 0 \\ 
& central & 1 \\ 
& west & 2 \\
\end{tabular}
}
\end{center}
\caption{The table discusses the clustering of a transaction fraud data set which is used to train a discriminative QBM model. MCC refers to the merchant category code and ZIP to zone improvement plan.}
\label{tbl:discr_data_preproc}
\end{table}

The complexity of this model demands a Hamiltonian that has sufficient representation capabilities. We investigate a diagonal Hamiltonian
\begin{equation}
\label{eq:H_disc_no_diags}
\begin{split}
    H^{(0)}_z\left(\bm{\theta}, \bm{x}\right) =&\:f^{(0)}_0\left(\bm{\theta}, \bm{x}\right) \, Z\otimes Z + f^{(0)}_1\left(\bm{\theta}, \bm{x}\right)\,Z\otimes I  +\\
    &\:f^{(0)}_2\left(\bm{\theta}, \bm{x}\right)\, I\otimes Z,
\end{split}
\end{equation}
a Hamiltonian with off-diagonal terms whose parameters are fixed to $0.1$
\begin{equation}
\label{eq:H_disc_diag_fixed}
\begin{split}
    H^{(1)}_z\left(\bm{\theta}, \bm{x}\right) =&\:f^{(1)}_0\left(\bm{\theta}, \bm{x}\right) \, Z\otimes Z + f_1^{(1)}\left(\bm{\theta}, \bm{x}\right)\, Z\otimes I  +\\
    &\:f^{(1)}_2\left(\bm{\theta}, \bm{x}\right)\, I\otimes Z + 0.1\, X\otimes I + 0.1 \, I\otimes X,
\end{split}
\end{equation}
and a Hamiltonian with parameterized off-diagonals
\begin{equation}
\label{eq:H_disc}
\begin{split}
    H^{(2)}_z\left(\bm{\theta}, \bm{x}\right) =&\:f^{(2)}_0\left(\bm{\theta}, \bm{x}\right) \, Z\otimes Z + f^{(2)}_1\left(\bm{\theta}, \bm{x}\right)\, Z\otimes I  +\\
    &\:f^{(2)}_2\left(\bm{\theta}, \bm{x}\right)\, I\otimes Z + f^{(2)}_3\left(\bm{\theta}, \bm{x}\right)\, X\otimes I + f^{(2)}_4\left(\bm{\theta}, \bm{x}\right)\, I\otimes X,
\end{split}
\end{equation}
where $f^{(j)}_i\big(\bm{\theta}^{(j)}_i, \bm{x}\big) = \bm{\theta}^{(j)}_i \cdot \bm{x} $ corresponds to the dot product of the vector corresponding to the data item $\bm{x} $ and a parameter vector $\bm{\theta}^{(j)}_i $ of equal length. 
The first and second qubit of all systems correspond to a hidden $h$ and visible $v$ node, respectively. Given a transaction instance, the measurement output of the visible units labels it as being either fraudulent or valid.



The training uses a Conjugate Gradient optimization routine with a maximum iteration number of $100$. 
Furthermore, the initial values for the Hamiltonian parameters are drawn from a uniform distribution on $\left[-1, 1\right]$.

 \begin{figure}[!htb]
\captionsetup{singlelinecheck = false, format= hang, justification=raggedright, font=footnotesize, labelsep=space}
\begin{center}
\begin{tikzpicture}

\node at (0,0){\includegraphics[width=0.7\linewidth]{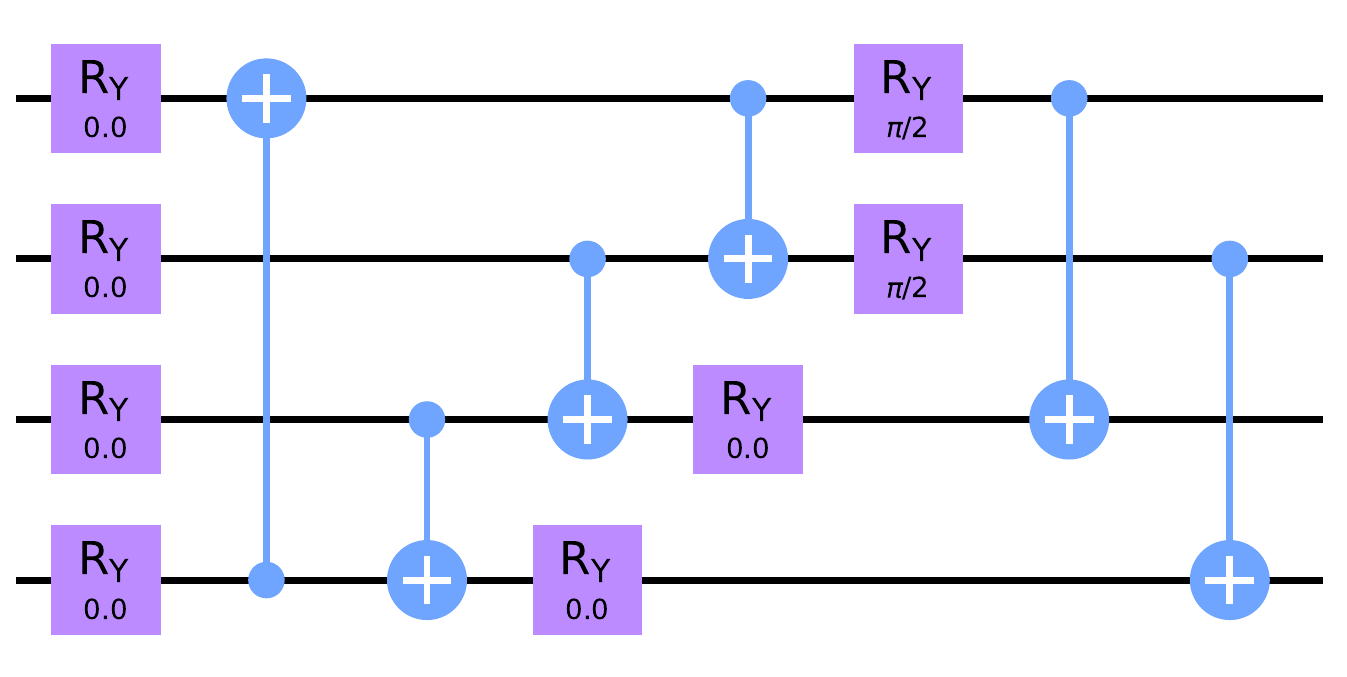}};
   \node at (-7, 0) {$\ket{0}^{\otimes 4}$};
   \draw[decorate, thick, decoration = {brace, amplitude=15pt}] (-5.5,-1.9) --  (-5.5,1.9);
\end{tikzpicture}
\end{center}
\caption{The Gibbs state preparation with VarQITE uses the illustrated parameterized quantum circuit as ansatz. The shown parameters correspond to the initial parameters of the Gibbs state preparation. The first qubit is the visible unit that determines the QBM output, the second qubit represents the hidden unit, and the last two qubits are working qubits needed to generate a purification of a maximally-mixed state as starting state for the evolution.}
\label{fig:ansatz_Disc}
\end{figure}
\pagebreak
In the following, we are going to employ the \varqbm{} implementation introduced in Sec.~\ref{sec:qbm} where VarQITE, see Sec.~\ref{sec:varqite}, is used to approximate the underlying Gibbs state $\rho(\bm{\theta};\bm{\omega}; \bm{x}) \approx \rho_z(\bm{\theta}; \bm{x})$ and to facilitate automatic differentiation for the training \footnote{These examples are trained with $L\left(\bm{\theta}\right) = -\sum_{x}p_x^{\text{data}} \log{ p_x\left(\bm{\theta}\right)}$ which is convention in existing ML literature. Although this is feasible for the small system size of the presented examples, a sampling based evaluation of the respective loss function gradient would not scale.}.
Given a test data set consisting of $250$ instances, with about $10\%$ fraudulent transactions, the Gibbs states, corresponding to the unique items of the test data, are approximated using VarQITE with the trained parameters $\bm{\theta}$ and the trial state shown in Fig.~\ref{fig:ansatz_Disc}. 
To predict the labels of the data instances, we sample from the states $\rho(\bm{\theta};\bm{\omega}; \bm{x})$ and choose the label with the highest sampling probability.
These results are, then, used to evaluate the accuracy (Acc), recall (Rec), precision (Pre) and F$_1$ score.
It should be noted that we choose a relatively simple quantum circuit to keep the simulation cost small. However, it can be expected that a more complex parameterized quantum circuit would lead to further improvement in the training results.

\begin{table}[h!tb]
\footnotesize
\centering{
\captionsetup{singlelinecheck = false, format= hang, justification=raggedright, font=footnotesize, labelsep=space}
{\renewcommand{\arraystretch}{1.2}
\begin{tabular}{c | c | c | c | c }
\textbf{Model} & \textbf{Accuracy} & \textbf{Recall} & \textbf{Precision} & \textbf{F}$_{\bm{1}}$\\
\hline
Nearest Neighbours & $0.94$ &  $0.54$ & $0.72$ & $0.62$ \\
Linear SVM & $0.90 $& $0$ & $0$ & $0$ \\
RBF SVM & $0.94 $& $0.42$ & $0.83$ & $0.56$  \\
Gaussian Process & $0.94$ & $0.46$ & $0.85$ & $0.60$\\
Gaussian Naive Bayes  & $0.91$ &$ 0.42 $& $0.56$ & $0.48$\\
Decision Tree & $0.94$ & $0.42$ & $0.83$ & $0.56$\\
Random Forest & $0.93 $& $0.29$ & $1.00$ & $0.45$\\
Multi-layer Perceptron  & $0.94$ &  $0.38$ &$0.9$ & $0.53$\\
AdaBoost & $0.94$ & $0.54$ & $0.81$ & $0.65$\\
QDA & $0.92$ & $0.46$ & $0.61$ & $0.52$\\
RBM & $0.90$ & $0$ & $0$ & $0$ \\
\textbf{\varqbm$^{(0)}$ } & $\mathbf{0.93}$ & 
$\mathbf{0.50}$ & $\mathbf{0.67}$ & $\mathbf{0.57}$\\
\textbf{\varqbm$^{(1)}$} & $\mathbf{0.93}$ & 
$\mathbf{0.50}$ & $\mathbf{0.67}$ & $\mathbf{0.57}$\\
\rowcolor{lightgray}
\textbf{\varqbm}$^{(2)}$ & $\mathbf{0.95}$ & 
$\mathbf{0.58}$ & $\mathbf{0.88}$ & $\mathbf{0.70}$
\end{tabular}
}
}
\caption{This table presents performance measures for scikit-learn standard classifiers, as well as the trained variational QBM models. The nearest neighbours classifier uses a $3$ nearest neighbours vote. The linear and RBF support vector machine (SVM) are based on a linear and radial kernel, respectively. The linear SVM uses a regularization term of $0.25$ and for the RBF SVM the kernel coefficient is set to $2$. The maximum depth of the Decision Tree, as well as the random forest is set to 5. Furthermore, the random forest classifier uses $10$ trees and uses $1$ feature to search for the best spit. The multi-layer perceptron uses $\ell_2$ regularization with coefficient $1$ and a maximum iteration number of $1000$. QDA refers to quadratic discriminant analysis. It should be noted that the remaining classifier properties are default settings.}
\label{tbl:measuresQBM}
\end{table}

The resulting values are compared to a set of standard classifiers defined in a \emph{scikit-learn} \cite{scikit-learn2011} classifier comparison tutorial \cite{scikitClassifierComp}, as well as a classifier consisting of a restricted Boltzmann machine and a logistic regression (RBM) \cite{scikitRBMClassifier}, see Tbl.~\ref{tbl:measuresQBM}. The respective classifiers are used with the hyper parameters defined in the tutorials.
Neither the linear SVM nor the RBM classify any test data item as fraudulent and, thus, the classifier sets precision and recall score to $0$.
The comparison reveals that all variational QBMs perform similarly well to the classical classifiers considering accuracy and are competitive regarding precision. In terms of recall and F$_1$ score, the \varqbm{} models without off-diagonal elements and with fixed off-diagonal elements perform well. Notably, the addition of fixed-off diagonal terms in the Hamiltonian does not help to improve the model performance compared to a Hamiltonian without these terms. Finally, the \varqbm{} with trained off-diagonal elements achieves the best results in terms of recall and F$_1$ score.

\addtocontents{toc}{}